\documentclass[10pt,a4paper]{amsart}

\usepackage{times}
\usepackage[left=1.15in,right=1.15in,top=1.15in,bottom=1.15in]{geometry} 
\usepackage{amsmath}
\usepackage{amssymb}
\usepackage{amsthm}
\usepackage{amscd}
\usepackage[percent]{overpic}
\usepackage[ansinew]{inputenc}
\usepackage{cite}
\usepackage{bbm}
\usepackage{color}
\usepackage[english=american]{csquotes}
\usepackage{hyperref}
\usepackage[all]{xy}
\usepackage{color}
\usepackage{scrextend}
\usepackage[ampersand]{easylist}
\usepackage{enumitem}
\usepackage{tikz-cd}
\usepackage{tikz,tikz-cd}

\newtheorem{lemma}{Lemma}
\newtheorem{proposition}[lemma]{Proposition}
\newtheorem{theorem}[lemma]{Theorem}
\numberwithin{lemma}{section}

\theoremstyle{definition}
\newtheorem{definition}[lemma]{Definition}

\theoremstyle{definition}
\newtheorem{remark}[lemma]{Remark}

\theoremstyle{definition}

\theoremstyle{definition}
\newtheorem{example}[lemma]{Example}

\theoremstyle{definition}

\graphicspath{{Pictures/}}

\numberwithin{equation}{section}

\newcommand{\Dcal}{\mathcal{D}}

\newcommand{\Fcal}{\mathcal{F}}

\newcommand{\Mcal}{\mathcal{M}}

\newcommand{\Rbb}{\mathbb{R}}

\newcommand{\Cb}{\mathbb{C}}
\newcommand{\Nb}{\mathbb{N}}
\newcommand{\Rb}{\mathbb{R}}

\newcommand{\Zb}{\mathbb{Z}}

\DeclareMathOperator{\diag}{diag}

\DeclareMathOperator{\rk}{rk}
\DeclareMathOperator{\SCI}{SCI}
\DeclareMathOperator{\spc}{sp}

\DeclareMathOperator{\spne}{sp_{N,\epsilon}}

\newcommand{\tr}{\textit{true}}
\newcommand{\f}{\textit{false}}

\DeclareMathOperator{\clos}{cl}

\DeclareMathOperator{\diam}{diam}

\DeclareMathOperator{\dist}{dist}

\newcommand{\Kc}{\mathcal{K}}

\newcommand{\ii}{\mathrm{i}}

\newcommand{\N}{\mathbb{N}}
\newcommand{\R}{\mathbb{R}}
\newcommand{\C}{\mathbb{C}}

\newcommand{\Yes}{\textit{Yes}}
\newcommand{\No}{\textit{No}}

\def\XXint#1#2#3{{\setbox0=\hbox{$#1{#2#3}{\int}$} 
\vcenter{\hbox{$#2#3$}}\kern-.5\wd0}}

\newcommand{\be}{\begin{equation}}
\newcommand{\en}{\end{equation}}

\newcommand{\rat}{\operatorname{Rat}}

\newcommand*{\MA}[1]{{\color{magenta}#1}}

\DefineNamedColor{named}{JungleGreen} {cmyk}{0.99,0,0.52,0}

\begin{document}

\title[Computing spectra]{Computing spectra - On the Solvability Complexity Index Hierarchy \\and Towers of Algorithms}

\author{J. Ben-Artzi} 
\address{School of Mathematics, Cardiff University}
\email{Ben-ArtziJ@cardiff.ac.uk}

\author{M. J. Colbrook} 
\address{DAMTP, University of Cambridge}
\email{m.colbrook@damtp.cam.ac.uk}

\author{A. C. Hansen} 
\address{DAMTP, University of Cambridge}
\email{a.hansen@damtp.cam.ac.uk}

\author{O. Nevanlinna} 
\address{Department of Mathematics and Systems Analysis, Aalto University}
\email{olavi.nevanlinna@aalto.fi}

\author{M. Seidel} 
\address{West Saxon University of Applied Sciences, Zwickau}
\email{markus.seidel@fh-zwickau.de}


\maketitle





\section{Introduction}\label{sec:intro}
This paper resolves the long-standing computational spectral problem. That is to determine the existence of algorithms that can compute spectra $\mathrm{sp}(A)$ of classes of bounded operators $A = \{a_{ij}\}_{i,j \in \mathbb{N}} \in \mathcal{B}(l^2(\mathbb{N}))$, given the matrix elements $\{a_{ij}\}_{i,j \in \mathbb{N}}$, that are sharp in the sense that they realise the boundaries of what a digital computer can achieve. Similarly, for a Schr\"odinger operator $H = -\Delta+V$, determine the existence of algorithms that can compute the spectrum $\mathrm{sp}(H)$ given point samples of the potential function $V$. In order to solve the problems, we establish the Solvability Complexity Index (SCI) hierarchy, based on the SCI introduced by one of the authors in \cite{Hansen_JAMS}. This is a classification hierarchy for all types of problems in computational mathematics that allows for classifications determining the boundaries of what computers can achieve in scientific computing. As a consequence, the SCI hierarchy 
provides classifications of computational problems that can be used in computer-assisted proofs, see \S \ref{sec:comp_ass_proofs} and \S \ref{sec:role_SCI_comp_ass}. 

The SCI hierarchy captures many key computational issues in the history of mathematics including the insolvability of the quintic, Smale's problem on the existence of iterative generally convergent algorithm for polynomial root finding, the computational spectral problem, inverse problems, optimisation etc.

Given the many applications in mathematical physics, analysis, quantum chemistry, statistical mechanics, quantum mechanics, quasicrystals, optics etc., the problem of computing spectra of operators has fascinated and frustrated mathematicians since the early work by H. Goldstine, F. Murray and J. von Neumann \cite{Goldstine} in the 1950s, yielding a vast literature (see \S \ref{sec:prev_work}). W. Arveson \cite{Arveson_role_of94} pointed out in the early 1990s that: "{\it Unfortunately, there is a dearth of literature on this basic problem, and so far
as we have been able to tell, there are no proven techniques}" (see also A. B{\"o}ttcher's Problem I in \cite{Albrecht_Fields}).   Arveson considered computing spectra from matrix elements $\{a_{ij}\}_{i,j \in \mathbb{N}} \in \mathcal{B}(l^2(\mathbb{N}))$, however, the situation is not better for the Schr\"odinger case. In particular, despite more than 90 years of quantum mechanics, it is still unknown how to compute spectra of $-\Delta_{\mathrm{discrete}} + V$ on lattices and $-\Delta+V$ on $\mathrm{L}^2(\mathbb{R}^d)$ given point samples from the potential function $V$.  

We solve these problems by providing algorithms that compute spectra and approximate eigenvectors, allowing for problems that were previously out of reach. We prove lower bounds yielding sharp classification results and optimality of the algorithms. The results may be surprising and link to many areas of mathematics. 

\begin{labeling}{\,\,}
\item[\,\,\,\,\,\textit{Classifications and new algorithms}:] 
The SCI hierarchy induces a total ordering $\leq_{\mathrm{SCI}}$ (see Remark \ref{rem:SCI_order}) on the family of computational spectral problems describing their difficulty. For example, given infinite matrices of the form $A = \{a_{ij}\}_{i,j \in \mathbb{N}} \in \mathcal{B}(l^2(\mathbb{N}))$, we prove the following:
\begin{equation}\label{eq:SCI_ordering}
\begin{split}
\text{Computing } \mathrm{sp}(A), \, A \text{ is diagonal}  &=_{\mathrm{SCI}} \text{ computing } \mathrm{sp}(-\Delta+V) \text{ with bounded } V\\
&=_{\mathrm{SCI}} \text{ computing } \mathrm{sp}(-\Delta_{\mathrm{discrete}}+V) \text{ on any lattice}\\
&<_{\mathrm{SCI}} \text{ computing } \mathrm{sp}(A), \, A \text{ is compact}\\
&=_{\mathrm{SCI}} \text{ computing } \mathrm{sp}(-\Delta+V) \text{ with } V \text{ blowing up at } \infty\\
&<_{\mathrm{SCI}} \text{ computing } \mathrm{sp}(A), \, A \text{ is self-adjoint}.\\
\end{split}
\end{equation}
Indeed, \eqref{eq:SCI_ordering} shows that computing spectra of Schr\"odinger operators (the first equalities hold also in many non-Hermitian cases) from point samples of a bounded potential function $V$ is not harder than computing the spectrum of a diagonal infinite matrix, the simplest of the non-trivial infinite-dimensional spectral problems. Paradoxically, the problem of computing spectra of compact operators, for which the method has been known for decades, is strictly harder than the problem of computing spectra of such Schr\"odinger operators, which has been open for more than half a century. The new algorithms and classification results finally solve this problem allowing computations that before were unachievable, see \S \ref{numerics}.

\vspace{1mm}

\item[\,\,\,\,\,\textit{Higher part of the SCI hierarchy - why algorithms were not found}:] We prove that in order to compute spectra or essential spectra of arbitrary infinite matrices one needs three limits in the computation, and it is impossible with two limits - these problems are very high up in the SCI hierarchy. 
In particular, there does exist a family of algorithms $\{\Gamma_{n_3,n_2,n_1}\}$ such that for all $A = \{a_{ij}\}_{i,j \in \mathbb{N}} \in \mathcal{B}(l^2(\mathbb{N})$, 
\[
\lim_{n_3\rightarrow\infty}\lim_{n_2\rightarrow\infty}\lim_{n_1\rightarrow\infty}\Gamma_{n_3,n_2,n_1}(A) = \mathrm{sp}(A).
\]
Yet, for any family of algorithms $\{\Gamma_{n_2,n_1}\}$ based on two limits there is an $A$ such that 
\[
\lim_{n_2\rightarrow\infty}\lim_{n_1\rightarrow\infty}\Gamma_{n_2,n_1}(A) \neq \mathrm{sp}(A).
\]
In the self-adjoint case, however, one needs two limits. These phenomena, that are similar to the solution to Smale's problem (see below), explain Arveson's comment, why there have been no known techniques for the general cases, and why it has taken substantial time to resolve the computational spectral problem. Indeed, classical approaches (see \S \ref{sec:prev_work}), including the $C^*$-algebra techniques (see W. Arveson \cite{Arveson_cnum_lin94, Arveson_noncommute93,Arveson_role_of94,Arveson_Improper93,  Arveson_discrete91} and N. Brown \cite{brown2007quasi, Brown_2006, Brown_Memoars, brown2002}) also used for the Schr\"odinger case, yield algorithms based on one limit. By the results above, algorithms based on one limit can never capture the general problem even in the self-adjoint case. However, classical approaches yield invaluable classification results in the lower part of the SCI hierarchy. 

\vspace{1mm}

\item[\,\,\,\,\,\textit{Computer-assisted proofs}:]  As we point out in \S \ref{sec:role_SCI_comp_ass}, the recent proof of Kepler's conjecture (Hilbert's 18th problem) \cite{Hales_Annals, hales_Pi}, led by T. Hales, is a striking example of a computer-assisted proof relying on computing non-computable problems (in the Turing sense). This may seem paradoxical, however, as the SCI hierarchy reveals and explains, there are many computational problems that are non-computable, that still can be used in computer-assisted proofs. Another example of non-computable problems used in computer-assisted proofs is the Dirac--Schwinger conjecture in spectral theory proved by C. Fefferman and L. Seco  \cite{fefferman1990, fefferman1992, fefferman1993aperiodicity,  fefferman1994, fefferman1994_2, fefferman1995, fefferman1996interval, fefferman1996, fefferman1997}. The SCI hierarchy provides a natural framework for determining which computational problems are suitable for computer-assisted proofs and explains why, for example, Kepler's conjecture and the Dirac--Schwinger conjecture can be resolved despite the above mentioned paradox. In fact, in both of the proofs of these conjectures one implicitly proves $\Sigma_1$ classifications (see \S \ref{sec:informal}) in the SCI hierarchy. Moreover, our classification results and algorithms for the computational spectral problem open up for new use of computer-assisted proofs in mathematical physics since the $\Sigma_1$ classifications yield algorithms that will never make mistakes.     

\vspace{1mm}

\item[\,\,\,\,\,\textit{Smale's problem on the existence of iterative generally convergent algorithm}:] 
An example of how the SCI hierarchy encompasses important foundational results is the question of computing zeros of polynomials with a rational map applied iteratively (such as Newton's method \cite{Smale2}). The problem with Newton's method is that it may not converge. This problem prompted S. Smale  to ask whether there exists an alternative to Newton's method, namely, a purely iterative generally convergent algorithm (see \S \ref{roots_pols}). Smale asked \cite{smale_question}: ``\emph{Is there any
purely iterative generally convergent algorithm for polynomial zero
finding?}" His conjecture was that the answer is `no'.  This problem was settled by C. McMullen in \cite{McMullen1} as follows: yes, if the degree is three; no, if the degree is higher (see also \cite{McMullen2,Smale_McMullen}). However, in \cite{Doyle_McMullen} P. Doyle and C. McMullen demonstrated a striking phenomenon: this problem can be solved in the case of the quartic and the quintic using several limits. Indeed, Smale's question and Doyle and McMullen's results are classification problems in the SCI hierarchy (see \S \ref{roots_pols}).

\end{labeling}


\tableofcontents





\subsection{The SCI hierarchy - an informal introduction}\label{sec:informal} 
We give an informal description of the SCI hierarchy in order to present the main results. The detailed definitions can be found in \S \ref{sec:background}.
The SCI hierarchy is based on the concept of a computational problem. This is described by a function 
\[
\Xi:\Omega \rightarrow \mathcal{M}
\] that we want to compute, where $\Omega$ is some domain, and $(\mathcal{M},d)$ is a metric space. For example, $\Xi(T) = \mathrm{sp}(T)$ (the spectrum) for some bounded operator $T \in \Omega$ and $\mathcal{M}$ is the collection of non-empty compact subsets of $\mathbb{C}$ equipped with the Hausdorff metric. The SCI was first introduced in the paper {\it ``On the Solvability Complexity Index, the {$n$}-pseudospectrum and approximations of spectra of operators''} \cite{Hansen_JAMS} for spectral problems in order to introduce the concept of several limits for spectral computation. The SCI of a spectral problem is the smallest number of limits needed in order to compute the solution.  However, in the paper above, the main issue was left open: is it necessary to use several limits? In other words, could the SCI collapse to one for all spectral problems, or in fact for all problems in scientific computing? Moreover, as is easily seen, a hierarchy based on only the number of limits needed would not be refined enough to capture the boundaries of what is possible in spectral computation. 
    
In this paper we introduce the general SCI hierarchy (see \S \ref{sec:background} for the formal definition) for all types of computational problems, and the mainstay of the hierarchy are the $\Delta^{\alpha}_k$ classes. The $\alpha$ is related to the model of computation as explained below. Informally, we have the following description. Given a collection $\mathcal{C}$ of computational problems, then
\begin{itemize}
\item[(i)] $\Delta^{\alpha}_0$ is the set of problems that can be computed in finite time, the SCI $=0$.
\item[(ii)] $\Delta^{\alpha}_1$ is the set of problems that can be computed using one limit (the SCI $=1$) with control of the error, i.e. $\exists$ a sequence of algorithms $\{\Gamma_n\}$ such that $d(\Gamma_n(A), \Xi(A)) \leq 2^{-n}, \, \forall A \in \Omega$.
\item[(iii)] $\Delta^{\alpha}_2$ is the set of problems that can be computed using one limit (the SCI $=1$) without error control, i.e. $\exists$ a sequence of algorithms $\{\Gamma_n\}$ such that $\lim_{n\rightarrow \infty}\Gamma_n(A) = \Xi(A), \, \forall A \in \Omega$.
\item[(iv)] $\Delta^{\alpha}_{m+1}$, for $m \in \mathbb{N}$, is the set of problems that can be computed by using $m$ limits, (the SCI $\leq m$), i.e. $\exists$ a family of algorithms $\{\Gamma_{n_m, \hdots, n_1}\}$ such that 
$$
\lim_{n_m \rightarrow\infty}\hdots \lim_{n_1\rightarrow\infty}\Gamma_{n_m,\hdots, n_1}(A) = \Xi(A), \, \forall A \in \Omega.
$$
 \end{itemize}
In general, this hierarchy cannot be refined unless there is some extra structure on the metric space $\mathcal{M}.$ The hierarchy typically does not collapse, and we have:
\begin{equation}\label{SCI1}
 \Delta_0^{\alpha} \subsetneq \Delta_1^{\alpha} \subsetneq \Delta_2^{\alpha} \subsetneq \hdots \subsetneq \Delta^{\alpha}_{m} \subsetneq \hdots.
\end{equation}
However, depending on the collection $\mathcal{C}$ of computational problems, the hierarchy \eqref{SCI1} may terminate for a finite $m$, or it may continue for arbitrary large $m$. For computational spectral problems the hierarchy terminates, see Figure \ref{fig:SCI_bounded_matrix} and Figure \ref{fig:SCI_schrod}.   

\begin{remark}[Clash of notation]
The $\Delta$ notation for the Laplacian and the $\Delta^{\alpha}_{m}$ notation for the classes in the SCI hierarchy is a slight mismatch, however, the meaning will always be clear from the context. 
\end{remark}
 
 The SCI hierarchy can be refined if the metric space $\mathcal{M}$ allows for convergence from ``above'' and ``below'', for example when considering the Hausdorff metric, which is natural for spectral problems. The motivation behind the refinement is to characterise the intricate classifications of different problems. For example, consider $\Omega$ to be the class of all diagonal operators $T \in \mathcal{B}(l^2(\mathbb{N}))$ of the form
 \begin{equation}\label{eq:diag_oper}
 T =\begin{pmatrix}
a_1& & & \\
 &a_2& &  \\
 & &a_3&  \\
 & & & \ddots \\
\end{pmatrix},
\qquad a_j \in \mathbb{C}.
 \end{equation}
 The problem of computing the spectrum $\mathrm{sp}(T)$ of such $T$s is trivially not in $\Delta^{\alpha}_1$. However, one can simply choose an algorithm $\Gamma_n$ to collect $\{a_j\}_{j=1}^n$ and then one has that $\Gamma_n(T) \rightarrow \mathrm{sp}(T)$ as $n \rightarrow \infty$. Thus, the problem of computing spectra of operators in $\Omega$ is in $\Delta^{\alpha}_2$. However, we clearly have an extra feature that is not captured by the hierarchy \eqref{SCI1}. Indeed, we have that 
\[
\Gamma_n(T) \subset \mathrm{sp}(T), \quad n \in \mathbb{N}.
\]     
In particular, we have convergence from below, and this is much stronger than just convergence since $\Gamma_n(T)$ always produces a correct output. Such type of convergence becomes incredibly important as it provides an error control from below.
Moreover, clearly, the hierarchy \eqref{SCI1} does not capture this important feature. This gives the motivation behind the $\Sigma^{\alpha}_1$ class, which captures the concept of convergence from below. Similarly, the  $\Pi^{\alpha}_1$ class captures a convergence from above. Informally, for spectral problems we have the following additions to \eqref{SCI1}:  
\begin{itemize}
\item[(1)] $\Delta^{\alpha}_0 = \Pi^{\alpha}_0 = \Sigma^{\alpha}_0$ is the set of problems that can be solved in finite time, the SCI $=0$.
\item[(2)] $\Sigma^{\alpha}_1$: We have $\Delta^{\alpha}_1 \subset \Sigma^{\alpha}_1 \subset \Delta^{\alpha}_2 $ and $\Sigma^{\alpha}_1$ is the set of problems for which there exists a sequence of algorithms $\{\Gamma_n\}$ such that for every $A \in \Omega$ we have $\Gamma_n(A) \rightarrow \Xi(A)$ as $n \rightarrow \infty$. However, $\Gamma_n(A)$ is always contained in the $2^{-n}$ neighbourhood of $\Xi(A)$.
\item[(3)] $\Pi^{\alpha}_1$: We have $\Delta^{\alpha}_1 \subset \Pi^{\alpha}_1 \subset \Delta^{\alpha}_2 $ and $\Pi^{\alpha}_1$ is the set of problems for which there exists a sequence of algorithms $\{\Gamma_n\}$ such that for every $A \in \Omega$ we have $\Gamma_n(A) \rightarrow \Xi(A)$ as $n \rightarrow \infty$. However, the $2^{-n}$ neighbourhood of $\Gamma_n(A)$ always contains $\Xi(A)$.
\item[(4)] $\Sigma^{\alpha}_{m}$ is the set of problems that can be computed by passing to $m$ limits, and computing the $m$-th limit is a $\Sigma^{\alpha}_1$ problem.
\item[(5)] $\Pi^{\alpha}_{m}$ is the set of problems that can be computed by passing to $m$ limits, and computing the $m$-th limit is a $\Pi^{\alpha}_1$ problem. 
 \end{itemize}

 \begin{figure}[t]
\centering
\includegraphics[width=0.8\textwidth]{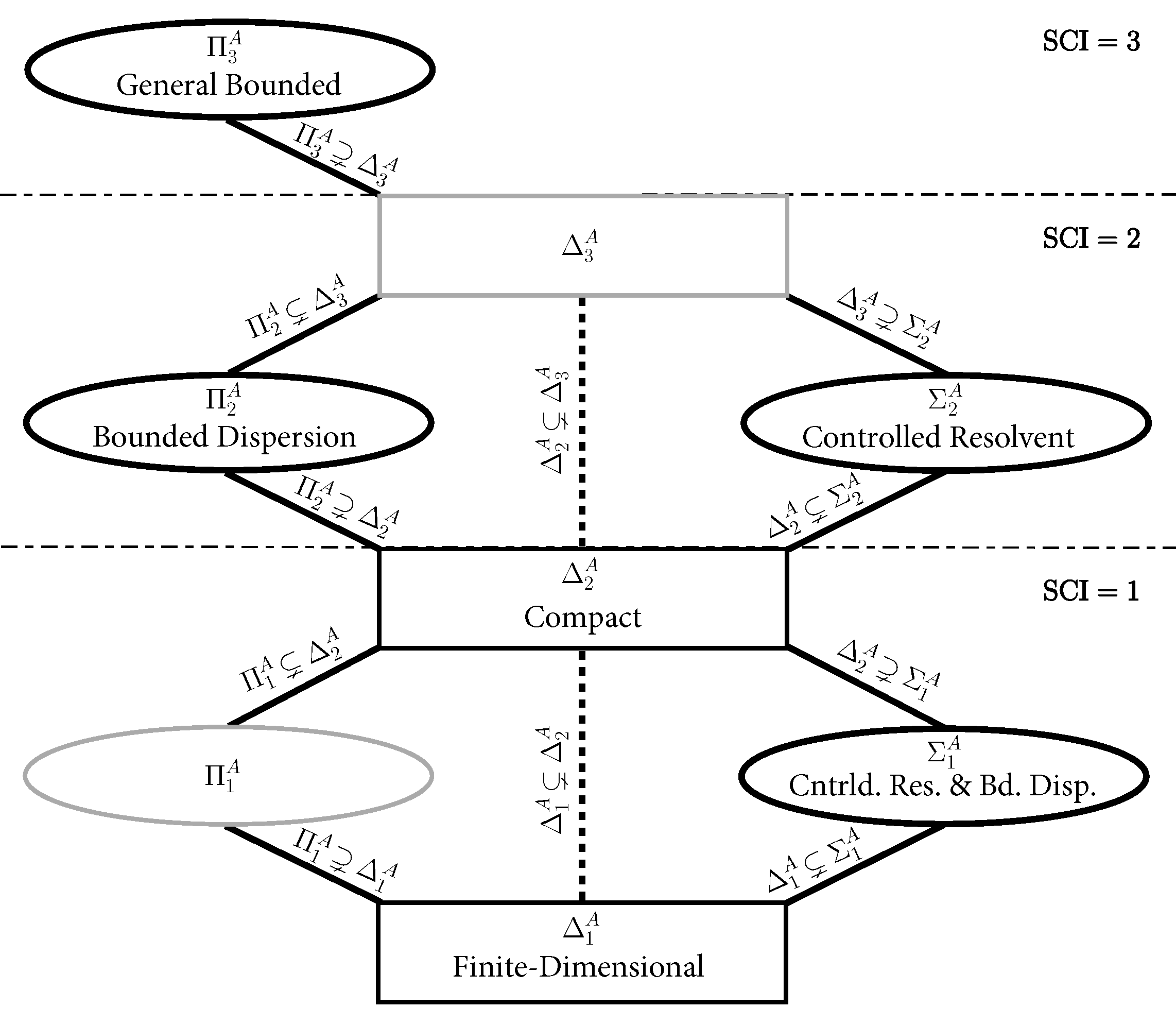}
\caption{Main results (Theorem \ref{spec_thm_main}); solution to the computational spectral problem: The SCI hierarchy for the computational problem of computing spectra of bounded infinite matrices acting on $l^2(\mathbb{N})$. Note that $\Sigma^A_3\setminus \Delta^A_2 = \emptyset$ and hence $\Sigma^A_3$ is not a part of this hierarchy. However, the classes enclosed in grey may contain spectral problems. The figure contains only some of the main classification results in Theorem \ref{spec_thm_main}.}
\label{fig:SCI_bounded_matrix}
\end{figure}
 
\begin{remark}[The general SCI hierarchy]\label{rem:general_SCI} The above sketch of the SCI hierarchy with convergence from below and above is well suited when considering the Hausdorff metric. However, the SCI hierarchy extends immediately to any metric space where there is a total ordering, for example, for $\mathcal{M} = \mathbb{R}$ and for decision problems where $\mathcal{M} = \{0,1\} = \{\mathrm{No}, \mathrm{Yes}\}$. For example, for decision problems a $\Sigma^{\alpha}_1$ (similarly $\Pi^{\alpha}_1$) classification of a computational problem with domain $\Omega$ means that there is a sequence of algorithms $\{\Gamma_n\}$ such that for $A \in \Omega$, $\Gamma_n(A)$ will provide the correct output for large $n$ (however, we do not know how big $n$ must be), but if $\Gamma_n(A) = \mathrm{Yes}$ ($\Gamma_n(A) = \mathrm{No}$ in the $\Pi^{\alpha}_1$ case), then the answer to the decision problem is $\mathrm{Yes}$ ($\mathrm{No}$ in the $\Pi^{\alpha}_1$ case).
\end{remark}

Schematically, the general SCI hierarchy can be viewed in the following way.
\begin{equation}\label{SCI_hierarchy}
\begin{tikzpicture}[baseline=(current  bounding  box.center)]
  \matrix (m) [matrix of math nodes,row sep=1.2em,column sep=1.5em] {
  \Pi_0^{\alpha}   &                    & \MA{\Pi_1^{\alpha}} &    &  \Pi_2^{\alpha}&  & {}\\
  \Delta_0^{\alpha}&  \Delta_1^{\alpha} & \Sigma_1^{\alpha}\cup\Pi_1^{\alpha} & \Delta_2^{\alpha}&      \Sigma_2^{\alpha}\cup\Pi_2^{\alpha} & \Delta_3^{\alpha}& \cdots\\
	\Sigma_0^{\alpha}&                    & \MA{\Sigma_1^{\alpha}} & &  \Sigma_2^{\alpha}&  &{} \\
  };
 \path[-stealth, auto] (m-1-1) edge[draw=none]
                                    node [sloped, auto=false,
                                     allow upside down] {$=$} (m-2-1)
																		(m-3-1) edge[draw=none]
                                    node [sloped, auto=false,
                                     allow upside down] {$=$} (m-2-1)
																		
																		(m-2-2) edge[draw=none]
                                    node [sloped, auto=false,
                                     allow upside down] {$\subsetneq$} (m-2-3)
																		(m-2-3) edge[draw=none]
                                    node [sloped, auto=false,
                                     allow upside down] {$\subsetneq$} (m-2-4)
																		(m-2-4) edge[draw=none]
                                    node [sloped, auto=false,
                                     allow upside down] {$\subsetneq$} (m-2-5)
																		(m-2-5) edge[draw=none]
                                    node [sloped, auto=false,
                                     allow upside down] {$\subsetneq$} (m-2-6)
																		(m-2-6) edge[draw=none]
                                    node [sloped, auto=false,
                                     allow upside down] {$\subsetneq$} (m-2-7)

												(m-2-1) edge[draw=none]
                                    node [sloped, auto=false,
                                     allow upside down] {$\subsetneq$} (m-2-2)
											 (m-2-2) edge[draw=none]
                                    node [sloped, auto=false,
                                     allow upside down] {$\subsetneq$} (m-1-3)
											(m-2-2) edge[draw=none]
                                    node [sloped, auto=false,
                                     allow upside down] {$\subsetneq$} (m-3-3)
											 (m-1-3) edge[draw=none]
                                    node [sloped, auto=false,
                                     allow upside down] {$\subsetneq$} (m-2-4)
																		(m-3-3) edge[draw=none]
                                    node [sloped, auto=false,
                                     allow upside down] {$\subsetneq$} (m-2-4)
																		(m-2-4) edge[draw=none]
                                    node [sloped, auto=false,
                                     allow upside down] {$\subsetneq$} (m-1-5)
											(m-2-4) edge[draw=none]
                                    node [sloped, auto=false,
                                     allow upside down] {$\subsetneq$} (m-3-5)
																		(m-1-5) edge[draw=none]
                                    node [sloped, auto=false,
                                     allow upside down] {$\subsetneq$} (m-2-6)
																		(m-3-5) edge[draw=none]
                                    node [sloped, auto=false,
                                     allow upside down] {$\subsetneq$} (m-2-6)
											(m-2-6) edge[draw=none]
                                    node [sloped, auto=false,
                                     allow upside down] {$\subsetneq$} (m-1-7)
																		(m-2-6) edge[draw=none]
                                    node [sloped, auto=false,
                                     allow upside down] {$\subsetneq$} (m-3-7);
																		
\end{tikzpicture}
\end{equation}
Note that the highlighted $\Sigma_1^{\alpha}$ and $\Pi_1^{\alpha}$ classes are crucial as they guarantee existence of algorithms that will never make mistakes, thus they become crucial in computer-assisted proofs, see \S \ref{sec:comp_ass_proofs} and \S \ref{sec:role_SCI_comp_ass}.

 \begin{figure}
\centering
\includegraphics[width=0.8\textwidth]{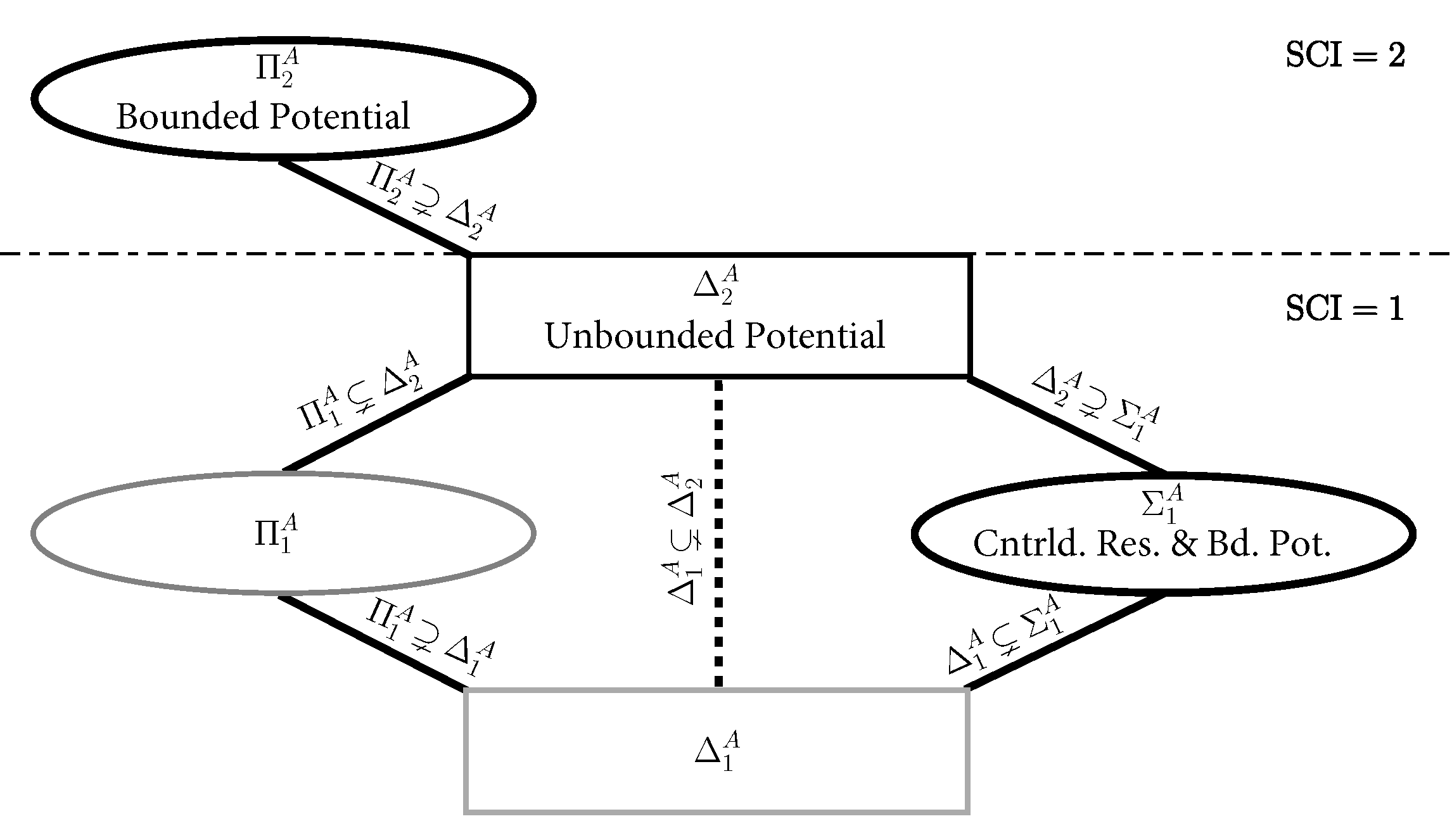}
\caption{Main results (Theorem \ref{main_self_adjoint} and Theorem \ref{thm:comp-res}): The SCI hierarchy for the problem of computing spectra of Schr\"odinger operators. $\Sigma^A_2\setminus \Delta^A_1 = \emptyset$, however, the classes enclosed in grey may contain spectral problems. The paradox of the $\Sigma^A_1$ result: Despite being a long-standing open problem, computing spectra of such Schr\"odinger operators is not harder than computing spectra of diagonal infinite-matrices (the simplest problem).}
\label{fig:SCI_schrod}
\end{figure}

\begin{remark}[The meaning of the $\alpha$, the model of computation] The $\alpha$ in the superscript indicates the model of computation, which is described in \S \ref{sec:background}. For $\alpha = G$, the underlying algorithm is general and can use any tools at its disposal. The purpose is to assure that lower bounds become universal regardless of the model of computation. The reader may think of a Blum--Shub--Smale (BSS) \cite{Smale_book} machine or a Turing machine \cite{Turing_Machine} with access to any oracle, although a general algorithm is even more powerful. However, for $\alpha = A$ this means that only arithmetic operations and comparisons are allowed. In particular, if rational inputs are considered, the algorithm is a Turing machine, and in the case of real inputs, a BSS machine. Hence, a result of the form
$
\notin \Delta_k^G \text{ is stronger than } \notin \Delta_k^A.
$  
Indeed, a $\notin \Delta_k^G$ result is universal and holds for any model of computation. Moreover, 
$
\in \Delta_k^A \text{ is stronger than } \in \Delta_k^G,
$  
and similarly for the $\Pi^{\alpha}_k$ and $\Sigma^{\alpha}_k$ classes. 
Note that classical hierarchies, such as the arithmetical hierarchy \cite{Odi}, become special cases of the SCI hierarchy (see Proposition \ref{thrm:prop_SCI_00} discussed in the appendix for completeness), and hence we keep the similar notation.
\end{remark}

\begin{remark}[The SCI ordering]\label{rem:SCI_order}
Note that the SCI hierarchy immediately implies a total ordering on the set of problems in the hierarchy. This is obvious when we only consider the $\Delta^{\alpha}_k$ classes, but can also be extended to the general case by considering $\Sigma^{\alpha}_k \cup \Pi^{\alpha}_k$ as one class in between the $\Delta^{\alpha}_k$s. This is the ordering  
$\leq_{\mathrm{SCI}}$ referred to in \S \ref{sec:intro}.
\end{remark}

\subsection{Smale's problem on iterative generally convergent algorithms and the SCI}\label{sec:Smale}
S. Smale initiated a comprehensive program on the foundations of computational mathematics in the 1980s \cite{Smale2,Smale_book}, focusing on problems in scientific computing rather than classical computer science. One of the key problems and algorithms Smale considered was polynomial root finding as well as Newton's method. As Newton's method may not converge, even for a cubic polynomial, a natural question would be if there exists an alternative approach. This question was formulated in terms of the existence of iterative generally convergent algorithms \cite{Smale2}. C. McMullen \cite{McMullen1, McMullen2, Smale_McMullen} solved the problem in the negative and, together with P. Doyle, realised that the problem of existence could be resolved by allowing more limits resulting in several iterative convergent algorithms used consecutively \cite{Doyle_McMullen}. They introduced a {\it tower of algorithm} in order to make the mathematical statement precise and also realised that for polynomials of degree $6$ and higher, one could not handle the problem regardless of the height of the tower (number of limits used). We have adopted the name towers of algorithms, however, we have made the concept general. The original towers of algorithms are now referred to as {\it Doyle--McMullen towers}, see \S \ref{roots_pols}. In \S \ref{roots_pols} we show how Smale's problem on the existence of iterative generally convergent algorithms and the theory of McMullen and Doyle become classification problems in the SCI hierarchy.

\subsection{Computer-assisted proofs in spectral theory and the SCI hierarchy}\label{sec:comp_ass_proofs}
The SCI hierarchy classifications determine the boundaries of what computers can achieve in scientific computing. As a consequence, the SCI hierarchy provides classifications of computational problems that can be used in computer-assisted proofs (for a detailed account see \S \ref{sec:role_SCI_comp_ass}), and these problems may be non-computable i.e. higher up than $\Delta^A_1$. Moreover, typically a computer assisted proof based on numerical calculation requires a classification result in the SCI hierarchy. An example of this in spectral theory is the proof of the Dirac--Schwinger conjecture.  

\textit{\textbf{Dirac--Schwinger conjecture - SCI classification: $\in \Sigma^A_1$, $\notin \Delta^G_1$}}:
The Dirac--Schwinger conjecture was proven in a series of papers by C. Fefferman and L. Seco \cite{fefferman1990, fefferman1992, fefferman1993aperiodicity,  fefferman1994, fefferman1994_2, fefferman1995, fefferman1996interval, fefferman1996, fefferman1997}, where one uses numerical computations to obtain asymptotic results on the ground state of an atom.  Consider the following Schr\"odinger operator
\[
H_{NZ} = \sum_{k=1}^N(-\Delta_{x_k} - Z|x_k|^{-1}) + \sum_{1\leq j < k \leq N}|x_j-x_k|^{-1}
\]
acting on antisymmetric functions in $\mathrm{L}^2(\mathbb{R}^{3N})$. The ground state energy $E(N,Z)$ for $N$ electrons and a nucleus of charge $Z$ is then defined by 
$
E(N,Z) := \inf\{\lambda \in \mathrm{sp}(H_{NZ})\}.
$ The ground state energy of an atom is then defined as $E(Z) := \min_{N\geq 1} E(N,Z)$. The key result of C. Fefferman and L. Seco was to show asymptotic behaviour of $E(Z)$ for large $Z$. In particular, 
\[
E(Z) = -c_0Z^{7/3} + \frac{1}{8}Z^2 - c_1Z^{5/3} + \mathcal{O}(Z^{5/3 - 1/2835}),
\] 
 for some explicitly defined constants $c_0$ and $c_1$.  The highly intricate computer-assisted proof hinges on several problems that are $\notin \Delta^G_1$ but are in $\Sigma^A_1$ (see for example Algorithm 3.7 and Algorithm 3.8 in \cite{fefferman1996interval}), and a crucial part of the proof implicitly establishes the $\Sigma^A_1$ classification in the SCI hierrchy. 

\textit{\textbf{Spectral problems that can be used in computer-assisted proofs}}: Our main results in Theorem \ref{spec_thm_main} and Theorem \ref{main_self_adjoint} provide the necessary $\Sigma^A_1$ classifications showing that computational spectral problems with any Jacobi operators with known growth of the resolvent can be used in computer-assisted proofs. This is also the case of Schr\"odinger operators $-\Delta + V$ where $V$ is bounded and of bounded variation. However, by Theorem \ref{thm:comp-res}, if we only know that $V$ blows up at infinity, the spectral problem $\notin (\Sigma^G_1 \cup \Pi^G_1)$ so such a Schr\"odinger operators cannot be used in a computer-assisted proof unless stronger assumptions are available.


\section{The main results}
The introduction of the SCI hierarchy implies an infinite classification theory even for the computational spectral problem by considering different classes of operators, and we provide the first foundations here. The precise formulations can be found in Theorem \ref{spec_thm_main}, Theorem \ref{main_self_adjoint}, Theorem \ref{thm:comp-res}, Theorem \ref{linear_systems_thrm} and Theorem \ref{thrm:norm_inverse}, however, we provide an informal and easy to read summary in this section. The fundamental question is as follows:
\begin{displayquote}
\normalsize
{\it Given a computational problem with a domain $\Omega$ and a problem function $\Xi: \Omega \rightarrow \mathcal{M}$, where in the SCI hierarchy is the problem when $\Xi$ represents the spectrum, essential spectrum, pseudospectrum or even a solution to an inverse problem?}
\end{displayquote}
Our results describing where a computational problem lies in the SCI hierarchy are mainly of the form: {\it computational problem} $\in \mathcal{S}$ and {\it computational problem} $\notin \mathcal{R}$, where $\mathcal{R}, \mathcal{S}$ are of the form $ \Sigma_k^{\alpha}, \Pi_k^{\alpha}, \Delta_k^{\alpha}$. This is typically written as
\[
\mathcal{R} \not\owns \text{ computational problem} \in \mathcal{S},
\]
where 
\[
\mathcal{R} = \Sigma_k^G, \Pi_k^G, \Delta_k^G, \qquad \mathcal{S} = \Sigma_j^A, \Pi_j^A, \Delta_j^A, \qquad k \leq j.
\]

\subsection{The main contribution of the paper}
The results are summarised as follows:

\begin{labeling}{\,\,\, {\it Theorem \ref{spec_thm_main}}:}
\item[\,\,\,{\bf  \emph{Theorem} \ref{spec_thm_main}}:] (Computational spectral problem, bounded operators). An informal summary follows in \S \ref{sec:bounded}, and the precise formulation is in \S \ref{finding_spectra}. 
\item[\,\,\,{\bf \emph{Theorem \ref{main_self_adjoint} \& Theorem \ref{thm:comp-res}}}:] (Computational spectral problem, Schr\"odinger operators). \S \ref{sec:schrodiner_main} provides an introductory summary, however, the precise statements are in \S \ref{quantum_mech}. 

\item[\,\,\,{\bf \emph{Theorem \ref{linear_systems_thrm} \& Theorem \ref{thrm:norm_inverse}}}:]  (Inverse problems in the SCI hierarchy).  A synopsis follows in \S \ref{sec:syn_inverse} whereas the exact formulations can be found in \S \ref{linear_systems}.

\item[\,\,\,{\bf \emph{New algorithms}}:]  (Spectral computations in mathematics and the sciences).  The proofs of the upper bounds in the theorems above yield new algorithms allowing for previously untouched problems both in the sciences and potentially in computer assisted proofs. Several examples can be found in \S \ref{numerics}. 
\end{labeling}

\begin{remark}[Model of recursiveness]
All our upper bounds hold in both the Turing model and the BSS model, thus we do not make any distinction when stating the main results and the theorems. When considering the Turing model with inputs with irrational (computable) numbers, the input to the algorithm representing such a number is an infinite string of numbers approximating the irrational number to any precision, see \S \ref{sec:inexact_input}. Lower bounds are universal for any model of computation.   
\end{remark}

\subsection{Computing spectra and approximate eigenvectors of bounded operators} \label{sec:bounded}
We are given operators $T \in \mathcal{B}(l^2(\mathbb{N}))$ and the task is to compute spectral properties from the matrix elements of $T$. We consider the following six problems.

\begin{labeling}{\,\,\, {\it Problem 1}:}
\item[\,\,\,{\it Problem 1}:] Compute spectra/essential spectra/pseudospectra of general operators.
\item[\,\,\,{\it Problem 2}:] Compute spectra/essential spectra/pseudospectra of self-adjoint/normal/known growth of resolvent (see Definition \ref{def:res_growth}) operators.
\item[\,\,\,{\it Problem 3}:]  Compute spectra/essential spectra/pseudospectra of operators with off-diagonal decay (see Definition \ref{def:disp}). 
\item[\,\,\,{\it Problem 4}:] Compute spectra and approximate eigenvectors of normal operators (with off-diagonal decay). 
\item[\,\,\,{\it Problem 5}:] Compute spectra/pseudospectra of compact operators.
\item[\,\,\,{\it Problem 6}:] Determine if a given point $z\in\mathbb{C}$ lies in the spectrum.
\end{labeling}

To avoid trivialities, when considering self-adjoint classes of operators we will restrict to $z\in\mathbb{R}$ and when considering compact operators we will restrict to $z\neq 0$. Moreover, by the essential spectrum we mean the spectrum that is invariant under compact perturbation, and the pseudospectrum is defined in \ref{eq:pseudo}.
We prove the following classifications. 
\begin{align}
&  \Delta^G_3 \not\owns \text{Prob 1 (sp.)} \in \Pi^A_3
&&  \Delta^G_3 \not\owns \text{Prob 1 (ess-sp.) } \in \Pi^A_3
&& \Delta^G_2 \not\owns \text{Prob 1 (pseudosp.)} \in \Sigma^A_2,
\label{eq:Prob1}\\
&  \Delta^G_2 \not\owns \text{Prob 2 (sp.)} \in \Sigma^A_2
&&  \Delta^G_3 \not\owns \text{Prob 2 (ess-sp.) } \in \Pi^A_3
&& \Delta^G_2 \not\owns \text{Prob 2 (pseudosp.)} \in \Sigma^A_2.
\label{eq:Prob2}\\
&  \Delta^G_2 \not\owns \text{Prob 3 (sp.)} \in \Pi^A_2
&&  \Delta^G_2 \not\owns \text{Prob 3 (ess-sp.) } \in \Pi^A_2
&& \Delta^G_1 \not\owns \text{Prob 3 (pseudosp.)} \in \Sigma^A_1.
\label{eq:Prob3}
\end{align}
Note that \eqref{eq:Prob2} means that the classification is the same for self-adjoint operators, normal operators and operators with known growth of the resolvent. These classes of operators are obviously increasingly included in each other. 

Problem 4 will only make sense for normal operators and for problems that are already in $\Sigma^{\alpha}_1$. Hence, we define the following set. 
\begin{itemize}
\item[] $\Sigma^{\alpha,\mathrm{eigv}}_1$: We have $\Sigma^{\alpha,\mathrm{eigv}}_1 \subset \Sigma^{\alpha}_1$ and $\Sigma^{\alpha,\mathrm{eigv}}_1$ is the set of problems for which there exists a sequence of algorithms $\{\Gamma_n\}$ such that for every $A \in \Omega$ we have $\Gamma_n(A) = \{(\lambda_{1,n}, \xi_{1,n}), \hdots, (\lambda_{K,n}, \xi_{K,n})\}$ for some $K = K(n) \in \mathbb{N}$, where $\lambda_{j,n}$ is contained in the $2^{-n}$ neighbourhood of $\mathrm{sp}(A)$ and $\|A\xi_{j,n}-\lambda_{j,n} \xi_{j,n}\| \leq 2^{-n}$ with $\| \xi_{j,n}\| = 1 + a_n$, $|a_n| \leq 2^{-n}$ for all $j\leq K$. Moreover, $\bigcup_{j=1}^K \lambda_{j,n} \rightarrow \mathrm{sp}(A)$ as $n \rightarrow \infty.$
\end{itemize}
In words $\Sigma^{\alpha,\mathrm{eigv}}_1$ can be described as follows.
\vspace{1mm}
\begin{displayquote}
\normalsize
{\it $\Sigma^{\alpha,\mathrm{eigv}}_1$ is the collection of computational spectral problems concerning normal operators that are in $\Sigma^{\alpha}_{1}$, where there exists an algorithm that can also compute approximate eigenvectors. }
\end{displayquote}
\vspace{1mm}
We prove for Problem 4 that 
\begin{align*}
 \text{Prob 4} \in \Sigma^{A,\mathrm{eigv}}_1.
\end{align*}
Note that we define Problem 4 for operators with off-diagonal decay. Indeed, if we do not have this assumption the computational spectral problem is not in $\Sigma^{G}_1$ by \eqref{eq:Prob2}, and hence the definition of $\Sigma^{\alpha,\mathrm{eigv}}_1$ would not make any sense.

Continuing, we prove for Problem 5 that 
\begin{align}
\Sigma_1^G\cup\Pi_1^G \not\owns \text{Prob 5 (sp.)} \in \Delta^A_2, \qquad \Sigma_1^G\cup\Pi_1^G \not\owns \text{Prob 5 (pseudosp.) } \in \Delta^A_2.
\label{eq:Prob4}
\end{align}
As for Problem 6 we prove the following:
\begin{align}
&\Delta_2^G \not\owns \text{Prob 6 (diagonal/compact/off-diagonal decay)} \in \Pi^A_2,
\label{eq:class1}\\
&\Delta_3^G \not\owns \text{Prob 6 (general/self-adjoint)} \in \Pi^A_3.
\label{eq:class2}
\end{align}
Finally, combining Problem 2 and Problem 3 we have 
\begin{equation}\label{eq:classification_Sigma_1}
 \Delta^G_1 \not\owns \text{Prob 2 $\cap$ Prob 3 (sp.)} \in \Sigma^A_1.
\end{equation}
The detailed statements can be found in Theorem \ref{spec_thm_main}.   

\begin{remark}[Solutions to the computational spectral problem for bounded operators]
Apart from Problem 5, all the problems above have been open since the beginning of spectral computations that dates back to the work of H. Goldstine, F. Murray and J. von Neumann \cite{Goldstine} in the 1950s. Their work implies $\Delta^A_1$ classifications on certain self-adjoint finite-dimensional problems. Problem 5 can be handled by finite-section method techniques (see \S \ref{sec:prev_work}), and has been known for decades, however we have included a short proof for completeness. The new lower bounds demonstrate that the finite section method for compact operators is optimal. 
\end{remark}

\begin{remark}[New algorithms and computer-assisted proofs]\label{re:new_alg}
The $\Sigma_1^A$ classifications above means that spectra of, for example, Jacobi operators of the form 
\begin{equation}\label{eq:Jac}
J :=\begin{pmatrix}
b_1&c_1 & & & \\
a_1 &b_2&c_2 & & \\
& a_2&b_3& c_3& \\
& &  a_3&b_4& \ddots \\
& & & \ddots & \ddots\\
\end{pmatrix},
\end{equation}
with known growth of the resolvent can be used in potential computer-assisted proofs. However, note the rather subtle result that
\[
\Delta^G_2 \not\owns \text{ Computing essential spectra of diagonal self-adjoint operators} \in \Pi^A_2. 
\]
Thus, this suggests that one must have very specific assumptions on the class of operators in order to be able to use computer-assisted proofs regarding essential spectra. In particular, the essential spectrum is much harder to compute than the spectrum. Note also that \eqref{eq:Prob4} reveals that general compact operators are not suited for computer-assisted proofs in spectral theory. The question is which extra assumptions in addition to compactness are needed to get lower in the SCI hierarchy.   
\end{remark}

\subsection{Computing spectra and approximate eigenvectors of Schr\"odinger operators on $\mathrm{L}^2(\mathbb{R}^d)$}\label{sec:schrodiner_main} 

The problem of computing the spectrum of a Schr\"odinger operator 
	\begin{equation}\label{eq:schrodinger}
	 H=-\Delta+V,\qquad V:\R^d\to\C,
	\end{equation}
is a classical problem in computational quantum mechanics. We consider computing spectra/pseudospectra of closed Schr\"odinger operators from point samples of the potential $V(x)$, in particular, the following problems:
\begin{labeling}{\,\,\, {\it Problem I}:}
\item[\,\,\,{\it Problem I}:] Compute spectrum/pseudospectrum of $H$ when $\|V\|_{\infty} \leq M < \infty$ and $V \in  \mathrm{BV}_{\mathrm{loc}}(\mathbb{R}^d)$ (locally bounded total variation).
\item[\,\,\,{\it Problem II}:] Compute spectrum/pseudospectrum of $H$ when $\|V\|_{\infty} \leq M < \infty$, $V \in  \mathrm{BV}_{\mathrm{loc}}(\mathbb{R}^d)$ and there is known growth of the resolvent.
\item[\,\,\,{\it Problem III}:] Compute spectrum/pseudospectrum of $H$ when $V$ is continuous, takes values in a sector  of the complex plane (not containing the negative real line) and blows up at infinity.
\item[\,\,\,{\it Problem IV}:] Compute spectra and approximate eigenvectors of $H$ when $H$ is self-adjoint and  $\|V\|_{\infty} \leq M < \infty$, $V \in  \mathrm{BV}_{\mathrm{loc}}(\mathbb{R}^d)$.
\end{labeling}
Note that the assumption that $V \in  \mathrm{BV}_{\mathrm{loc}}(\mathbb{R}^d)$, the set of functions with locally bounded variation, is very mild as this class includes discontinuous functions and functions with arbitrary wild oscillations at infinity.  Note also that only requiring $V \in \mathrm{L}^{\infty}(\mathbb{R}^d)$ and $\|V\|_{\infty} \leq M$ is impossible as the concept of point samples of $V$ would not be well defined.  
We prove the following classifications.
\begin{align}
&  \Delta^G_1 \not\owns \text{Problem I (spectrum)} \in \Pi^A_2
&&  \Delta^G_1 \not\owns \text{Problem I (pseudosp.)} \in \Sigma^A_1,
\label{eq:3ProbI}
\\
&  \Delta^G_1 \not\owns \text{Problem II (spectrum)} \in \Sigma^A_1
&&  \Delta^G_1 \not\owns \text{Problem II (pseudosp.)} \in \Sigma^A_1,
\label{eq:3ProbII}
\\
&  \Sigma_1^G\cup\Pi_1^G \not\owns \text{Problem III (spectrum)} \in \Delta^A_2
&&  \Sigma_1^G\cup\Pi_1^G \not\owns \text{Problem III (pseudosp.)} \in \Delta^A_2,
\label{eq:3ProbIII}
\\
& \text{Problem IV} \in \Sigma^{A,\mathrm{eigv}}_1.
\label{eq:3ProbIV}
\end{align}
The detailed statements can be found in Theorem \ref{main_self_adjoint} and Theorem \ref{thm:comp-res}.

\begin{remark}[Non-Hermitian Hamiltonians]
We emphasise that the results above are valid also for non-Hermitian quantum systems. This level of generality is important as we want the theory to include non-Hermitian quantum mechanics \cite{Bender, Bender3, Hatano_97, Hatano_96} and the theory of resonances \cite{Zworski2, ZworskiSjostrand}.
\end{remark}

\begin{remark}[The solutions to the computational spectral problem for Schr\"odinger operators]
The results in \eqref{eq:3ProbI}, \eqref{eq:3ProbII} and \eqref{eq:3ProbIV}   provide solutions to the problem of computing spectra of Schr\"odinger operators on $\mathrm{L}^2(\mathbb{R}^d)$ with bounded potential. In view of \eqref{eq:Prob4} the results in \eqref{eq:3ProbII} and \eqref{eq:3ProbIV} may be surprising. Indeed, despite being open since the 1950s, the sharp $\Sigma^A_1$ classification implies that the problem of computing spectra and pseudospectra of even non-Hermitian Schr\"odinger operators in \eqref{eq:3ProbII} is not harder than computing the spectrum of a diagonal infinite matrix, the simplest of the non-trivial infinite-dimensional spectral problems. Moreover, the problem of computing spectra and pseudospectra of even non-Hermitian Schr\"odinger operators in \eqref{eq:3ProbII} is actually strictly easier than computing spectra of compact operators on $l^2(\mathbb{N})$, a computational problem for which successful algorithms have been known for decades. 
Finally, since we achieve the $\Sigma^A_1$ classification in several cases, computer-assisted proofs may be a possibility.
\end{remark}

\subsection{Computational inverse problems}\label{sec:syn_inverse} 
Just as finding spectra of operators and roots of polynomials, the problem of solving linear systems of equations is at the heart of computational mathematics. For the finite-dimensional case, it is easy to find an algorithm that can perform the task, but what about the infinite-dimensional case? We consider the inverse problem 
\[
Ax=y \qquad A \in \mathcal{B}(l^2(\mathbb{N})), \quad x,y \in l^2(\mathbb{N}),
\]
where we want to compute various quantities such as $x$ from the matrix values of $A$ and vector components of $y$ when $A$ is known to be invertible. In summary, we consider the following problems.

\begin{labeling}{\,\,\, {\it Problem a}:}
\item[\,\,\,{\it Problem a}:] Compute $x$ when $A$ and $y$ are arbitrary.
\item[\,\,\,{\it Problem b}:] Compute $x$ when $A$ is self-adjoint and $y$ is arbitrary.
\item[\,\,\,{\it Problem c}:] Compute $x$ when $A$ has known off-diagonal decay and $y$ is arbitrary. 
\item[\,\,\,{\it Problem d}:] Compute $x$ when $A$ has known off-diagonal decay and $y$ has known decay.
\item[\,\,\,{\it Problem e}:] Compute the norm of the inverse $\|A^{-1}\|^{-1}$.
\item[\,\,\,{\it Problem f}:] Determine if $A$ is invertible.
\end{labeling}

When computing solutions to general inverse problems, as there is no concept of convergence from above and below, we only have the initial $\Delta^{\alpha}_k$ classes. However, when it comes to computing the norm of the inverse and the decision problem of determining whether $A$ is invertible or not, we do have the $\Sigma^{\alpha}_k$ and $\Pi^{\alpha}_k$ classes. In particular, we prove the following classifications. 
\begin{align}
&  \Delta^G_2 \not\owns \text{Problem a} \in \Delta^A_3
&& \Delta^G_2 \not\owns \text{Problem b} \in \Delta^A_3,
\label{eq:5ProbI}
\\
&  \Delta^G_1 \not\owns \text{Problem c} \in \Delta^A_2
&&  \Delta^G_0 \not\owns \text{Problem d} \in \Delta^A_1.
\label{eq:5ProbII}
\end{align}
Moreover, these are the classifications for Problem e.
\begin{align}
&\Delta_2^G \not\owns \text{Problem e (general/self-adjoint)} \in \Pi^A_2,
\label{eq:4class1}\\
&\Delta_1^G \not\owns \text{Problem e (off-diagonal decay)} \in \Pi^A_1.
\label{eq:4class2}
\end{align}
Note that Problem f is a special case of Problem 5 in \S \ref{sec:bounded}. Thus, we have that 
\begin{align*}
&\Delta_2^G \not\owns \text{Problem f (diagonal/compact/off-diagonal decay)} \in \Pi^A_2,
\\
&\Delta_3^G \not\owns \text{Problem f (general/self-adjoint)} \in \Pi^A_3.
\end{align*}
The detailed statements can be found in Theorem \ref{linear_systems_thrm} and Theorem \ref{thrm:norm_inverse}.

\begin{remark}[Finite section in inverse problems]
Note that the results in \eqref{eq:5ProbI} and \eqref{eq:5ProbII} provide a simple explanation of why the finite section method or any of its variants could never solve the general inverse problem. Indeed, such methods would imply at least a $\Delta^A_2$ results, which are impossible. However, note that we immediately get that the class of problems for which the finite section method works are in $\Delta^A_2$. This demonstrates the importance of the vast literature on the finite section method for classifications in the SCI hierarchy, see \S \ref{sec:prev_work}.
\end{remark}

\section{Computing the non-computable - The role of the SCI hierarchy in computer-assisted proofs}\label{sec:role_SCI_comp_ass}
Computer-assisted proofs using numerical approximations have become essential in mathematics. There are an increasing number of famous conjectures and theorems that have been proven using computer assisted proofs.
 A highly incomplete list in alphabetical order includes the Dirac-Schwinger conjecture \cite{fefferman1990, fefferman1992, fefferman1993aperiodicity,  fefferman1994, fefferman1994_2, fefferman1995, fefferman1996interval, fefferman1996, fefferman1997}, the Double-Bubble conjecture \cite{double_bubble}, Kepler's conjecture (Hilbert's 18th problem) \cite{Hales_Annals, hales_Pi}, Smale's 14th problem \cite{Tucker2002}, the 290-theorem \cite{290}, the weak Goldbach conjecture \cite{weak_Goldbach} etc. In all of these cases the proofs are based on using numerical computations with approximations. Hence, a key question will always be; given a problem that needs to be computed in order to secure a computer-assisted proof, can the computation be done with verification that is $100\%$ reliable? Or asked more broadly:
\begin{displayquote}
\normalsize
{\it Question I: Which computational problems are suited for use in computer-assisted proofs?}
\end{displayquote}
The instinct would normally be that the computational problem must be in $\Delta^A_1$, or computable in the words of Turing. This is not the case. The computer-assisted proof of Kepler's conjecture is done by computing non-computable problems, i.e. $\notin \Delta^G_1$, as explained below. There are several cases of important conjectures that have been solved by computer-assisted proof, where the computational problem is higher up in the SCI hierarchy than $\Delta^G_1$. Hence, the SCI hierarchy is instrumental in answering Question I above as follows.

\begin{labeling}{\,\,}
\item[\,\,\,\,\,(i) \textbf{Classifications in the SCI hierarchy - Which problems are safe in computer-assisted proof?}] 
In addition to problems in the class $\Delta_1^A$, problems in the classes $\Sigma_1^A$ and $\Pi_1^A$ can be used in computer assisted proofs regardless of the metric space $\mathcal{M}$ (see Remark \ref{rem:general_SCI}) that induces the different classifications. However, the use of problems in $\Sigma_1^A$ or $\Pi_1^A$ depends on the phrasing of the conjecture. For example, suppose the conjectured statement is that spectra of operators in a certain class of self-adjoint discrete Schr\"odinger operators never intersect a certain open interval $I$. Such a statement can be falsified given the new $\Sigma_1^A$ classification of the computational spectral problem concerning discrete Schr\"odinger operators. Indeed, suppose one has located a candidate Schr\"odinger operator $H$ for a counterexample, however, one does not know the spectrum of $H$. One can use one of the new algorithms realising the $\Sigma_1^A$ classification, and if $\mathrm{sp}(H) \cap I \neq \emptyset$, the algorithm will eventually demonstrate this intersection with a $100\%$ guarantee, thus falsifying the conjecture. Similarly, decision problems in $\Sigma_1^A$ and $\Pi_1^A$ can be used in computer-assisted proof, however, problems in $\Delta_2^A\setminus (\Sigma_1^A \cup \Pi_1^A)$ and higher up in the SCI hierarchy will in general be unsuitable. 

\item[\,\,\,\,\,(ii) \textbf{A computer-assisted proof typically requires an SCI hierarchy classification}.] 
A computer-assisted proof that relies on numerical computations will typically require a proof of a $\Delta_1^A$, $\Sigma_1^A$ or $\Pi_1^A$ classification in the SCI hierarchy. Indeed, a mathematician facing a computational problem in order to complete a computer-assisted proof will likely have to ask: where in the SCI hierarchy is the problem? If this is not already known one must prove it, and, as the examples below suggest, this classification is typically done implicitly in the proofs. Sometimes this is trivial, however, sometimes this may be very delicate as in the proof of Kepler's conjecture and intricate and technical as in the proof of the Dirac--Schwinger conjecture (see below). 

\item[\,\,\,\,\,(iii) \textbf{Understanding the higher end of the SCI hierarchy helps answering Question I}.] One may ask about the value of studying the higher end of the SCI hierarchy, in particular the classes $\Delta^A_j, \Sigma^A_j, \Pi^A_j$ for $j \geq 2$, as these classes may seem of only theoretical interest. This is not the case. Answering Question I above becomes an infinite classification theory. Hence, given a particular computational problems that is desirable to use in a computer-assisted proof, one may not know the answer to the question whether this problem is in an appropriate class of the SCI hierarchy. However, one may have knowledge of an upper bound, say $\Pi^A_3$. The question is whether extra features of the computational problem would allow for a classification lower in the SCI hierarchy. Existing results on the higher end of the SCI hierarchy may therefore be invaluable. In fact, the solution to the problem of computing spectra of Schr\"odinger operators evolved in this way, where initially there was a crude classification of $\Pi^A_3$. By gradually learning which extra assumptions were needed in order to achieve classifications further down in the hierarchy, we were eventually able to reach the sharp $\Sigma^A_1$ classification, yielding a classification suitable for computer-assisted proofs.  
\end{labeling}

Below follow examples of successful computer-assisted proofs with the corresponding SCI hierarchy classification of the main computational problem. 

\begin{labeling}{\,\,}
\item[\,\,\,\,\,\textit{\textbf{Kepler's Conjecture (Hilbert's 18th problem) - SCI classification: $\in \Sigma^A_1$, $\notin \Delta^G_1$ }}:] Kepler conjectured that no packing of congruent balls
in Euclidean three space has density greater than that of cubic close packing and hexagonal close packing arrangements. The Flyspeck program, led by T. Hales \cite{Hales_Annals, hales_Pi}, provides a fully computer-assisted verification, where parts of the numerical computations in the computer-assisted proof are based on deciding about 50000 linear programs with irrational inputs. The computational problem is to decide whether 
\begin{equation}\label{eq:LP_Kepler}
M \geq \max_{x} \langle x , c \rangle  \text{ subject to } Ax \leq b, 
\end{equation}
where 
$M$ is an irrational number and $A \in \mathbb{R}^{m \times n},$ $b \in \mathbb{R}^m$ can contain irrational numbers. One needs affirmative answers on all the linear programs in order to verify the conjecture. The irrational input makes deciding  \eqref{eq:LP_Kepler} quite delicate. One may consider the dual problem and ask whether there exists a $y$ such that 
\begin{equation}\label{Dual_LP}
\langle y , b \rangle \leq M \text{ subject to } A^*y = c, \quad y \geq 0.
\end{equation}
In that case, producing a candidate that satisfies \eqref{Dual_LP} would immediately imply a positive answer to \eqref{eq:LP_Kepler}, and this is the main idea behind the verification of \eqref{eq:LP_Kepler} in the proof of Kepler's conjecture. However, given the irrational input, deciding \eqref{Dual_LP} is a bit optimistic, and thus one goes for the approximate version of deciding whether there is a $y \in \mathbb{R}^m$ such that 
\begin{equation}\label{Dual_LP_extended}
\langle y , b \rangle_K \leq M \text{ subject to } A^*y = c, \quad x \geq 0,
\end{equation}
where 
$
\langle y, b \rangle_K = \lfloor  10^{K} \langle y , b \rangle \rfloor  10^{-K}
$
and $K \in \mathbb{N}$.
Informally, we could think of $\langle y , b \rangle_K$ as $\langle y , b \rangle$ computed with $K$ digits accuracy ($K = 6$ is used in the proof).
 The fact that there are irrational input numbers means that $A$ and $b$
are only known approximately, however, to any precision one wants (think of either a Turing machine or a BSS machine that can access $A \in \mathbb{R}^{m\times N}$ in form of an oracle $\mathcal{O}_A$ such that $|\mathcal{O}_A(i,j,k)-A_{i,j}| \leq 2^{-k}$). 
There are several facts about the problem \eqref{Dual_LP_extended} and its classification in the SCI hierarchy that may be surprising given that Kepler's conjecture is successfully proven. In a companion paper \cite{SCI_optimization} to our results, as a part of the extended Smale's 9th problem, the following is proven.
\vspace{1mm}
\begin{displayquote}
\normalsize
{\it  For any integer $\tilde K > 1$ there exists a class of inputs $\Omega$ such that the problem \eqref{Dual_LP_extended} with $K = \tilde K$ is $\notin \Sigma_1^G$.  However, with the same input class $\Omega$, we have that the problem \eqref{Dual_LP_extended}, with $K = \tilde K -1$, is in $\Delta_1^A$. Similarly, deciding \eqref{eq:LP_Kepler} is $\notin \Sigma_1^G$.}
\end{displayquote}
\vspace{1mm}

One may ask how the computer-assisted proof of Kepler's conjecture was at all possible, given that one needs to decide \eqref{Dual_LP_extended} for $K = 6$. Indeed, the $\notin \Sigma_1^G$ fact would suggest that no positive verification could be possible. However, if the inequality $\langle y , b \rangle_K \leq M$ in \eqref{Dual_LP_extended} is replaced by a strict inequality $\langle y , b \rangle_K < M$, then there are classes of inputs $\Omega$ such that deciding \eqref{Dual_LP_extended} is in $\Sigma_1^A$ and hence also deciding \eqref{Dual_LP} is in $\Sigma_1^A$.

In the proof of Kepler's conjecture, the process that chooses the approximations of the irrational numbers in the input $A$, $b$ is fully automated, and so is the process that makes a suggestion for a candidate for \eqref{Dual_LP_extended} and the formal verification. Thus, one may view the fully automated process for deciding \eqref{eq:LP_Kepler} in the proof of Kepler's conjecture as an algorithm, where its domain is the class $\Omega$ of inputs for which the algorithm either halts with output `yes' or runs forever if the answer to the decision problem is `no'. This algorithm thus yields a $\Sigma_1^A$ classification. 
Note that, due to the fact that the decision problem are $\notin \Delta^G_1$ one could have had the following outcome. If there had been a case where $M = \max_{x} \langle x , c \rangle  \text{ subject to } Ax \leq b$ in \eqref{eq:LP_Kepler} the decision algorithm would have run forever and the Flyspeck program would never have resolved Kepler's conjecture. This could also have been the case if the answer to the decision problem was `no' i.e a case where $M < \max_{x} \langle x , c \rangle  \text{ subject to } Ax \leq b$ in \eqref{eq:LP_Kepler}. 

\vspace{1mm}
\item[\,\,\,\,\,\textit{\textbf{Dirac--Schwinger conjecture - SCI classification: $\in \Sigma^A_1$, $\notin \Delta^G_1$}}:]
We discussed the details in \S \ref{sec:comp_ass_proofs}.

 \vspace{1mm}
\item[\,\,\,\,\,\textit{\textbf{Boolean Pythagorean triples problem - SCI classification: $\in \Pi^A_1$, $\not \in \Delta_1^G$}}:] The Boolean Pythagorean triples problem asks if it is possible to colour each of the positive integers either red or blue, so that no Pythagorean triple of integers $a, b, c$, satisfying 
$a^{2}+b^{2}=c^{2}$ are all the same colour. For example, in the Pythagorean triple $3$, $4$ and $5$ (
$3^{2}+4^{2}=5^{2}$), if $3$ and $4$ are coloured red, then $5$ must be coloured blue.
 This is true for integers up to $n = 7824$. The computer-assisted proof, performed by M. Heule, O. Kullmann, and V. Marek (2016) \cite{Bool},
is based on showing that this is not true for $n = 7825$. While it is a combinatorial task checking the problem for any finite set of integer (and hence $\in \Delta_0^A$), it is clearly not $\in \Delta_0^G$ for infinite sets of integers. Yet, the problem is clearly $\in \Pi_1^A$, which is why it was possible to verify the counterexample. 
\vspace{1mm}
 \item[\,\,\,\,\,\textit{\textbf{Group theory: $\mathrm{Aut}(\mathbb{F}_5)$ has property $(T)$ - SCI classification 
: $\in \Sigma^A_1$, $\notin \Delta^G_1$}}:] The fact that the automorphism group of the free group on five generators has Kazhdan's property $(T)$, was shown by M. Kaluba, P. Nowak and N. Ozawa \cite{Kaluba}. 
The proof relies on a decision problem involving a minimiser of a semi-definite program (actually a root of a positive definite matrix that is a minimiser). The minimiser is computed using floating point arithmetic. Hence, it is, at best (if one could do a backward error analysis), equivalent to solving the semi-definite program with inexact input. Computing minimisers to semi-definite programs with inexact, yet arbitrary small precision is $\notin \Delta^G_1$ \cite{SCI_optimization}. Showing that the final decision problem used on  \cite{Kaluba} is $\in \Sigma^A_1$ requires an argument, which we do not repeat here, however, it is of similar spirit to the argument above arguing why Kepler's conjecture could be verified. 
\end{labeling}

\section{The history of the SCI hierarchy in mathematics}\label{sec:role_SCI}
Although we formally establish the SCI hierarchy in this paper, it implicitly shows up throughout the history of mathematics. Moreover, it encompasses many key computational problems and has applications in many computational areas of the mathematical sciences. This is summarised as follows.
\begin{itemize}
\item[(i)] {\it Computer-assisted proofs:} This is discussed in detail in \S \ref{sec:role_SCI_comp_ass}.
\item[(ii)] {\it Insolvability of the quintic:} The insolvability of the quintic becomes a classification problem in the SCI hierarchy. In particular, showing that the SCI of the problem of computing the zeros of a polynomial, when one can use arithmetic operations and radicals, is greater than 0 for polynomials of degree 5 is equivalent to the insolvability of the quintic.  
\item[(iii)] {\it Smale's problem on the existence of generally convergent algorithms and McMullen's solutions:} \S \ref{sec:Smale} summarises how the results by McMullen and Doyle \& McMullen are classification results in the SCI hierarchy.  A more detail account can be found in \S \ref{roots_pols}.
\item[(iv)] {\it Optimisation (compressed sensing and the extended Smale's 9th problem, statistical estimation, machine learning):}
As discussed in \S \ref{sec:comp_ass_proofs} and proved in a companion paper \cite{SCI_optimization}, deciding feasibility of linear programs given irrational inputs is not only undecidable ($\notin \Delta^G_1$) but $\notin \Sigma^G_1$. As shown in \cite{SCI_optimization}, using the framework of the SCI hierarchy, similar phenomena extend to many key 
problems in optimisation such as finding minimisers of Basis pursuit and Lasso. These form the basis of compressed sensing, statistical estimation, areas of machine learning etc.  Moreover, there is a link to the extended Smale's 9th problem \cite{SCI_optimization}. 
\item[(v)] {\it Computing the exit flag (validating the output of an algorithm):} Often computational routines come with a certification, a so-called exit flag, that determines if the computed solution is trustworthy or not. An example is MATLABs popular routine \texttt{linprog} for solving linear programs. Paradoxically, as shown in \cite{SCI_optimization}, this exit flag is not trustworthy, and the problem of computing the exit flag is higher up in the SCI hierarchy than computing the original problem itself. This also occurs for the spectral computation problem, where the problem of deciding if spectral pollution (for finite section) occurs on an interval for self-adjoint operators is strictly harder than computing the spectrum \cite{colbrook4}.
\item[(vi)] {\it Spectral problems:} Arveson's comment (recall \S \ref{sec:intro}) on the lack of algorithms for general spectral problems can be explained by the SCI hierarchy. As many computational spectral problems are high up in the hierarchy, all attempts with standard methods would fail. Moreover, the standard methods were based on one limit approaches, and would therefore never capture the depth of the computational spectral problem. Finally, the new $\Sigma_1^A$ classifications provide new algorithms that will never produce incorrect outputs. 
\item[(vii)] {\it Inverse problems:} As established in \S \ref{sec:syn_inverse}, inverse problems have a rich classification theory in the SCI hierarchy. 
\item[(viii)] {\it Foundations of computational mathematics:} The SCI hierarchy can be viewed as a direct continuation of Smale's program on the foundations of scientific computing, however, it allows for any computational model and any computational problem. 
\item[(ix)] {\it Hierarchies in logic:} Classical hierarchies in logic such as the arithmetical hierarchy become special cases of the SCI hierarchy (see Proposition \ref{thrm:prop_SCI_00}). This is not a paper in logic and computer science, however, a short discussion on connections to logic can be found in the appendix. 
\end{itemize}


\section{Connections to previous work}\label{sec:prev_work}
We split the comments into four categories: foundations of computational mathematics, spectral computation, computer-assisted proofs and inverse problems.  
\begin{labeling}{}

\item[\,\,\,{\bf  \emph{Foundations of computational mathematics}}:]   S. Smale's seminal work \cite{Smale2, Smale_Acta_Numerica} and his program on the foundations of computational mathematics and scientific computing initiated 
the pioneering work by C. McMullen \cite{McMullen1, McMullen2, Smale_McMullen} and P. Doyle \& C. McMullen \cite{Doyle_McMullen} on polynomial root-finding. These are classification results in the SCI hierarchy, and our contribution is motivated by this program and the pioneering work by L. Blum, F. Cucker, M. Shub \& S. Smale  \cite{Smale_book}.  
Other results in this program on hierarchies include the work of F. Cucker \cite{Cucker_AH_real} that becomes a part of the  SCI hierarchy (see \S \ref{ext_dec_1} and \S \ref{sec:Cucker}). 

\item[\,\,\,{\bf \emph{Spectral computations}}:] The literature on computing spectra is enormous, thus, we will only emphasise the work that has been most influential on this paper. The ideas of using computational and algorithmic approaches to obtain spectral information date back to leading physicists and mathematicians such as E. Schr\"odinger \cite{schrodinger1940method}, T. Kato \cite{kato1949upper} and J. Schwinger \cite{Schwinger}. Schwinger introduced finite-dimensional approximations to quantum systems in infinite-dimensional spaces that allow for spectral computations. An interesting observation is that Schwinger's ideas were already present in the work of H. Weyl \cite{weyl1950theory}.
The work by H. Goldstine, F. Murray and J. von Neumann \cite{Goldstine} was one of the first to establish rigorous convergence results, and their work yields $\Delta^A_1$ classification for certain self-adjoint finite-dimensional problems. The work of N. Aronzajn \cite{Aronszajn_PNAS1, Aronszajn_PNAS2} inspired a comprehensive research program continued by H. Weinberger \cite{Weinberger1974}.
In \cite{Digernes} T. Digernes, V. S. Varadarajan  and S. R. S. Varadhan proved convergence of spectra of Schwinger's finite-dimensional discretisation matrices for a specific class of Schr\"{o}dinger operators with certain types of potential, which yields $\Delta^A_2$ classification in the SCI hierarchy. 

The finite-section method, which has been intensely studied for spectral computation, and has often been viewed in connection with Toeplitz theory, is very similar to Schwinger's idea of approximating in a finite-dimensional subspace. The reader may want to consult the pioneering work by A. B{\"o}ttcher \cite{Albrecht_Fields, Bottcher_pseu} and A. B{\"o}ttcher \& B. Silberman \cite{Bottcher_book, bottcher2006analysis}, see also A. B{\"o}ttcher, H. Brunner, A. Iserles \& S. N{\o}rsett \cite{Arieh2}, M. Marletta \cite{Marletta_pollution} and M. Marletta \& R. Scheichl \cite{Marletta_Spec_gaps}. The latter papers also discuss the failure of the finite section approach for certain classes of operators, see also \cite{Hansen_PRS, Hansen_JFA}. Note that in the cases where the finite section method works, it will typically yield $\Delta_2^A$ classifications in the SCI hierarchy, and occasionally $\Delta_1^A$ classifications. E. B. Davies considered second order spectra methods
\cite{Davies00, Davies_sec_order}, and E. Shargorodsky \cite{Shargorodsky1} demonstrated how second order spectra methods \cite{Davies_sec_order} will never recover the whole spectrum. 

W. Arveson  \cite{Arveson_cnum_lin94, Arveson_noncommute93,Arveson_role_of94,Arveson_Improper93,  Arveson_discrete91} 
and N. Brown \cite{brown2007quasi, Brown_2006, Brown_Memoars} pioneered the combination of spectral computation and the $C^*$-algebra literature (which dates back to the work by A. B{\"o}ttcher \& B. Silberman \cite{Albrecht1983}), both for the general spectral computation problem as well as for Schr\"{o}dinger operators. See also the work by N. Brown, K. Dykema, and D. Shlyakhtenko \cite{brown2002}, where variants of finite section analysis are implicitly used. Arveson also considered spectral computation in terms of densities, which is related to Szeg\"{o}'s work \cite{Szego} on finite section approximations. Similar results are also obtained by A. Laptev and Y. Safarov \cite{Laptev}. 
Typically, when applied to appropriate subclasses of operators, finite section approaches yield $\Delta^A_2$ classification results. There are also other approaches based on the infinite QR algorithm in connection with Toda flows with infinitely many variables pioneered by P. Deift, L. C. Li, and C. Tomei \cite{deift}. See also the work by P. Deift, J. Demmel, C. Li, and C. Tomei \cite{deift1991bidiagonal}.

The seminal work of C. Fefferman and L. Seco \cite{fefferman1990, fefferman1992, fefferman1993aperiodicity,  fefferman1994, fefferman1994_2, fefferman1995, fefferman1996interval, fefferman1996, fefferman1997}  on proving the Dirac-Schwinger conjecture is a striking example of computations used in order to obtain complete information about the asymptotical behaviour of the ground state of a family of Schr\"{o}dinger operators. The computer-assisted proof implicitly proves $\Sigma^A_1$ classifications in the SCI hierarchy. Moreover, the paper \cite{feffermanCPAM} by C. Fefferman is based on similar approaches using numerical calculation of eigenvalues. See also the paper \cite{Fefferman1979} by C. Fefferman and D. H. Phong on numerically computing the lowest eigenvalue of pseudo-differential operators.
 We also want to highlight the recent pioneering work by M. Zworski \cite{Zworski1, Zworski2} on computing resonances that can be viewed in terms of the SCI hierarchy. In particular, the computational approach \cite{Zworski1} is based on expressing the resonances as limits of non-self-adjoint spectral problems, and hence the SCI hierarchy is inevitable, see also \cite{ZworskiSjostrand}.  
Finally, the reader may want to consult the important contributions by C. Lubich \cite{lubich2008quantum} in his monograph of computational quantum mechanics on time dependent Schr\"{o}dinger problems. 
 
 We remark in passing that our results and theory are very different from the work of J. Richards and M. Pour--El \cite{Pour-El}. Their results rely on strong assumptions that are unrealistic for actual computations and therefore miss the whole SCI hierarchy. Moreover, such unrealistic assumptions have limited the impact on computational spectral theory, as suggested by the quote from W. Arveson in \S \ref{sec:intro}.

\item[\,\,\,{\bf \emph{Computer-assisted proofs}}:] 
Some of the issues have already been discussed in \S \ref{sec:comp_ass_proofs} and \S \ref{sec:role_SCI_comp_ass}.
The number of examples of computer-assisted proofs in the literature using numerical calculations is substantial. What most of them have in common is that in order to prove that the computational proof is 100\% accurate one implicitly has to prove a classification in the SCI hierarchy. The work by Fefferman and Seco \cite{fefferman1990, fefferman1992, fefferman1993aperiodicity,  fefferman1994, fefferman1994_2, fefferman1995, fefferman1996interval, fefferman1996, fefferman1997} can both be viewed from a computational spectral theory point of view as well as a computer-assisted proofs angle, and the $\Sigma^A_1$ classification is crucial. Similarly, the computer-assisted proof of Kepler's conjecture, via Hales' Flyspeck program, is also relying on  $\Sigma^A_1$ classification. Note that these are examples of computer-assisted proofs done by non-computable problems, however, there are many examples of computer-assisted proofs based on $\Delta_1^A$ classifications as well. A great example of this is the work of D. Gabai, R. Meyerhoff, and P. Milley  \cite{Gabai} on hyperbolic three-manifolds.

\item[\,\,\,{\bf \emph{Inverse Problems}}:]

There is a vast literature on computing solutions to certain infinite-dimensional inverse problems in one limit, typically by using the finite section method. The connection to Toeplitz theory is important and the reader may consult the foundational results in the books by A. B{\"o}ttcher \& B. Silberman \cite{Bottcher_book, bottcher2006analysis} as well as the monograph by Lindner \cite{Lindner} and the references therein. Note that two-limit algorithms have been suggested by K. Gr\"ochenig, Z. Rzeszotnik, and T. Strohmer in \cite{Charly1}, see also \cite{Charly2}.
\end{labeling}

\subsection*{Acknowledgements}
The authors would like to thank Percy Deift, Charlie Fefferman, Tom Hales, Ari Laptev, Steve Smale and Maciej Zworski for helpful discussions.
ACH would like to thank Caroline Series for pointing out the connection between the results of Doyle and McMullen in \cite{Doyle_McMullen} and the work in \cite{Hansen_JAMS}. It was this connection that initiated the work leading to this paper. The authors are grateful to Roman Bogdan for producing the data for Figure \ref{fQ_fig}. MJC acknowledges support from the UK Engineering and Physical Sciences Research Council (EPSRC) grant EP/L016516/1. ACH acknowledges support from a Royal Society University Research Fellowship, the Leverhulme Prize 2017, as well as the UK Engineering and Physical Sciences Research Council (EPSRC) grant EP/L003457/1.


\section{The Solvability Complexity Index Hierarchy and towers of algorithms}\label{sec:background}

Throughout this paper we assume the following:
\begin{subequations}\label{setup}
\begin{equation}
\Omega\text{ is some set, called the \emph{domain}},
\end{equation}
\vspace{-20pt}
\begin{equation}\label{Lambda_Omega}
\Lambda \text{ is a set of complex valued functions on $\Omega$, called the \emph{evaluation} set},
\end{equation}
\vspace{-20pt}
\begin{equation}
\Mcal  \text{ is a metric space},
\end{equation}
\vspace{-20pt}
\begin{equation}
\text{The mapping }\Xi:\Omega\to\Mcal,  \text{called the \emph{problem} function}.
\end{equation}
\end{subequations}

The set $\Omega$ is the collection of objects that give rise to our computational problems. It can be a family of matrices (infinite or finite), a collection of polynomials, a family of Schr\"{o}dinger (or Dirac) operators with a certain potential etc. The problem function $\Xi : \Omega\to\Mcal$ is what we are interested in computing. It could be the set of eigenvalues of an $n\times n$ matrix, the spectrum of a Hilbert (or Banach) space operator, root(s) of a polynomial etc. 
Finally, the set $\Lambda$ is the collection of functions that provide us with the information we are allowed to read, say matrix elements, polynomial coefficients or pointwise values of a potential function of a Schr\"{o}dinger operator, for example. 
In most cases it is convenient to consider a metric space $\mathcal{M}$, however, in the case of polynomials it may be more useful to use a pseudo metric space (see Example \ref{Ex1} (III) ).  To explain this rather abstract set-up in (\ref{setup}) we commence with the following examples:
\begin{example}\label{Ex1}

\begin{itemize}

\item[]

\item[(I)]{\bf (Spectral problems)}
Let $\Omega=\mathcal{B}(\mathcal{H})$, the set of all bounded linear operators on a separable Hilbert space $\mathcal{H}$, and the problem function $\Xi$ be the mapping $A\mapsto\mathrm{sp}(A)$ (the spectrum of $A$). Here $(\Mcal,d)$ is the set of all non-empty compact subsets of $\Cb$ provided with the Hausdorff metric $d=d_{\mathrm{H}}$ (defined precisely in \eqref{Hausdorff}). The evaluation functions in $\Lambda$ could for example consist of the family of all functions $f_{i,j}: A\mapsto \langle Ae_j,e_i\rangle$, $i,j\in\Nb$, which provide the entries of the matrix representation of $A$ w.r.t. an orthonormal basis $\{e_i\}_{i\in\N}$. Of course, $\Omega$ could be a strict subset of $\mathcal{B}(\mathcal{H})$, for example the set of self-adjoint or normal operators, and $\Xi$ could have represented the pseudospectrum, the essential spectrum or any other interesting information about the operator.

\item[(II)]{\bf (Inverse problems)}  Let $\Omega=\mathcal{B}_{\mathrm{inv}}(\mathcal{H}) \times \mathcal{H}$, where 
$\mathcal{B}_{\mathrm{inv}}(\mathcal{H})$ denotes the set of all bounded invertible operators on $\mathcal{H}$, and let the problem function $\Xi$ be the mapping $(A,b) \mapsto A^{-1}b$, which assigns to a linear problem $Ax=b$ its solution $x$. The metric space $\mathcal{M}$ would simply be $\mathcal{H}$ and $\Lambda$ the collection of mappings $\{f_{i,j}\}_{i\in \mathbb{N}, j \in \mathbb{Z}_+}$ where $f_{i,j}: (A,b)\mapsto \langle Ae_j,e_i\rangle$ for $j \in \mathbb{N}$ and $f_{i,0} : (A,b)\mapsto \langle b,e_i\rangle$. Also here $\Omega$ could consist of operators with specific properties (off diagonal decay, self-adjointness, isometric properties). 

\item[(III)]{\bf (Polynomial root finding)} Let $\Omega = \mathbb P_s$, the set of polynomials of degree $\leq s$ over $\C$ and let the problem function $\Xi$ be the mapping $p \mapsto  \{\alpha\in\C\ | \ p(\alpha)=0\}$ (the roots of $p$). Let $(\mathcal{M},d)$ denote the collection of finite sets of points in $\mathbb{C}$ equipped with the pseudo metric 
$d : \mathcal{M} \times \mathcal{M} \rightarrow [0,\infty]$, defined by
	$
	d(x,y) = \min_{1 \leq i \leq n,1 \leq j \leq m}|x_j - y_i|,
	$
where $x = \{x_1,\hdots, x_n\}, y = \{y_1,\hdots, y_m\} \in \mathcal{M}$. The reason for the pseudo metric is that the techniques of Doyle and McMullen that we will consider are based on computing a single root of a polynomial (as for example Newton's method does). In this case $\Lambda$ is the finite set of functions $\{f_j\}_{j=1}^s$ where $f_j : p \mapsto \alpha_j$ for $p(t) = \sum_{k=1}^s \alpha_k t^k$.

\item[(IV)]{\bf (Computational quantum mechanics)} Let $\Omega = \mathrm{L}^{\infty}(\mathbb{R}^d) \cap \mathrm{C}(\mathbb{R}^d)$ and let $\Xi: V \mapsto \mathrm{sp}(-\Delta + V),$ where the domain $\mathcal{D}(-\Delta + V) = \mathrm{W}^{2,2}(\mathbb{R}^d)$ (the standard Sobolev space) and $-\Delta + V$ is the usual Schr\"{o}dinger operator. Given that the spectra are unbounded, we cannot use the Hausdorff metric anymore, but will let $(\mathcal{M},d_{\mathrm{AW}})$ denote the set of non-empty closed subsets of $\mathbb{C}$ equipped with the \emph{Attouch--Wets} metric (see (\ref{eq:attouch-wets-metric})). In this case a natural choice of $\Lambda$ would be
 the set of all evaluations $f_x: V \mapsto V(x)$, $x\in\mathbb{Q}^d$. 
 
\item[(V)]{\bf (Decision making)} Let $\Omega$ denote the set of infinite matrices with values in $\{0,1\}$ and $\Xi:\Omega\to \mathcal{M}=\{\Yes, \No\}$ where $\mathcal{M}$ is equipped with the discrete metric $d_{\mathrm{disc}}$. The evaluation functions would naturally be $f_{i,j}: A\mapsto A_{i,j}$, $i,j\in\Nb$, the $(i,j)$th matrix coordinate of $A$. A typical example of $\Xi$ could be: $\Xi(\{A_{i,j}\})$: Does $\{A_{i,j}\}$ have a column containing infinitely many non-zero entries? Naturally, $\Omega$ can be replaced with the natural numbers including zero $\mathbb{Z}_+$, and $\Xi$ could be a question about membership in a certain set, as in classical recursion theory. In this case the evaluation set would be $\Lambda=\{\lambda\}$ consisting of the function $\lambda:\Zb_+\to\Cb,\, x\mapsto x$.
 \end{itemize}
\end{example}

Given this set-up and motivation, we can now define what we mean by a computational problem.
\begin{definition}[Computational problem]\label{def:comp_prob}
Given a domain $\Omega$; an evaluation set $\Lambda$, such that for $A_1, A_2 \in \Omega$ then $A_1 = A_2$ if and only if $f(A_1) = f(A_2)$ for all $f \in \Lambda$; a metric space $\mathcal{M}$; and a problem function $\Xi:\Omega\to\Mcal$, we call the collection $\{\Xi,\Omega,\mathcal{M},\Lambda\}$ a computational problem.
\end{definition}

Our aim is to find and to study families of functions (that we will sometimes refer to as algorithms) which permit us to approximate the function $\Xi$. The main pillar of our framework is the concept of a tower of algorithms. However, before that we will define a general algorithm.
\begin{definition}[General Algorithm]\label{alg}
Given a  computational problem $\{\Xi,\Omega,\mathcal{M},\Lambda\}$, a \emph{general algorithm} is a mapping $\Gamma:\Omega\to\Mcal$ such that for each $A\in\Omega$:
\begin{itemize}
\item[(i)] there exists a finite subset of evaluations $\Lambda_\Gamma(A) \subset\Lambda$,
\item[(ii)]  the action of $\,\Gamma$ on $A$ only depends on $\{A_f\}_{f \in \Lambda_\Gamma(A)}$ where $A_f := f(A),$
\item[(iii)]  for every $B\in\Omega$ such that $B_f=A_f$ for every $f\in\Lambda_\Gamma(A)$, it holds that $\Lambda_\Gamma(B)=\Lambda_\Gamma(A)$.
\end{itemize}
We will sometimes write $\Gamma(\{A_f\}_{f \in \Lambda_\Gamma(A)})$, in order to emphasise that $\Gamma(A)$ only depends on the results $\{A_f\}_{f \in \Lambda_\Gamma(A)}$ of finitely many evaluations. 
\end{definition}

Note that for a general algorithm there are no restrictions on the operations allowed. The only restriction is that it can only take a finite amount of information, though it \emph{is} allowed to \emph{adaptively} choose the finite amount of information it reads depending on the input (which may very well be infinite, say an infinite matrix, or a function).
The condition (iii) just ensures that the algorithm is well defined and consistent since, put in simple words, changing the input $A$ shall not affect the algorithm's action as long as the change does not affect the output of the relevant evaluations in $\Lambda_\Gamma(A)$. 

\begin{remark}[The purpose of a general algorithm]
The purpose of a general algorithm is to have a definition that will encompass any model of computation, and that will allow lower bounds and impossibility results to become universal. Given that there are several non equivalent models of computation, lower bounds will be shown with a general definition of an algorithm. Upper bounds will always be done with more structure on the algorithms for example using a Turing machine or a Blum--Shub--Smale (BSS) machine. 
\end{remark}

The concept of a general algorithm, however, is not enough to describe the world of computational problems. For that we need the concept of \emph{towers of algorithms}. 

\begin{definition}[Tower of algorithms]\label{tower_funct}
Given a computational problem $\{\Xi,\Omega,\mathcal{M},\Lambda\}$, a \emph{tower of algorithms of height $k$
 for $\{\Xi,\Omega,\mathcal{M},\Lambda\}$} is a collection of sequences of functions 
\begin{equation*}
\begin{split}
\Gamma_{n_k}:\Omega
\rightarrow \mathcal{M}, \quad 
\Gamma_{n_k, n_{k-1}} :\Omega
\rightarrow \mathcal{M}, \, \hdots \,, 
\Gamma_{n_k, \hdots, n_1}:\Omega \rightarrow \mathcal{M}, 
\end{split}
\end{equation*}
where $n_k,\hdots,n_1 \in \mathbb{N}$ and the functions $\Gamma_{n_k, \hdots, n_1}$ at the lowest level in the tower are general algorithms in the sense of Definition \ref{alg}. Moreover, for every $A \in \Omega$,
\begin{equation}\label{conv}
\begin{split}
\Xi(A) &= \lim_{n_k \rightarrow \infty} \Gamma_{n_k}(A), \\
\Gamma_{n_k}(A) &=
  \lim_{n_{k-1} \rightarrow \infty} \Gamma_{n_k, n_{k-1}}(A),\\
& \, \, \, \, \vdots\\
\Gamma_{n_k, \hdots, n_2}(A) &=
  \lim_{n_1 \rightarrow \infty} \Gamma_{n_k, \hdots, n_1}(A),
\end{split}
\end{equation}
where $S = \lim_{n \rightarrow \infty}S_n$ means convergence $S_n\to S$ in the (pseudo) metric space $\mathcal{M}$.
For simplicity, and with a slight abuse of notation, we will often refer to $\{\Gamma_{n_k, \hdots, n_1}\}$ as a tower of algorithms, implicitly meaning the whole collection as described above.  
\end{definition}

In this paper we will discuss several types of towers: {\it General towers}, when there is no extra structure on the functions at the lowest level in the tower; {\it Doyle--McMullen towers}, that are used for Smale's problem on polynomial root finding (see \S \ref{roots_pols}); {\it Arithmetic towers}, that restricts the algorithm to arithmetic operations and comparisons; {\it Radical towers}, that also allows the operation of $\sqrt{\cdot}$ of a real number. 
A General tower will refer to the very general definition in Definition \ref{tower_funct} specifying that there are no further restrictions as will be the case for the other towers. When we specify the type of tower, we specify requirements on the functions $\Gamma_{n_k, \hdots, n_1}$, in particular, what kind of operations may be allowed. We can now define an {\it arithmetic tower of algorithms} and a {\it radical tower of algorithms}. 

\begin{definition}[Arithmetic towers]
Given a computational problem $\{\Xi,\Omega,\mathcal{M},\Lambda\}$, where $\Lambda$ is countable, we define the following: An \emph{Arithmetic tower of algorithms} of height $k$
 for $\{\Xi,\Omega,\mathcal{M},\Lambda\}$ is a tower of algorithms where the lowest functions $\Gamma = \Gamma_{n_k, \hdots, n_1} :\Omega \rightarrow \mathcal{M}$ satisfy the following:
 For each $A\in\Omega$ the mapping $(n_k, \hdots, n_1) \mapsto \Gamma_{n_k, \hdots, n_1}(A) = \Gamma_{n_k, \hdots, n_1}(\{A_f\}_{f \in \Lambda})$ is recursive, and $\Gamma_{n_k, \hdots, n_1}(A)$ is a finite string of complex numbers that can be identified with an element in $\mathcal{M}$. For arithmetic towers we let $\alpha = A$ 
\end{definition} 

\begin{remark}[Recursiveness] By recursive we mean the following. If $f(A) \in \mathbb{Q}$ for all $f \in \Lambda$, $A \in \Omega$, and $\Lambda$ is countable, then $\Gamma_{n_k, \hdots, n_1}(\{A_f\}_{f \in \Lambda})$ can be executed by a Turing machine \cite{Turing_Machine}, that takes $(n_k, \hdots, n_1)$ as input, and that has an oracle tape consisting of $\{A_f\}_{f \in \Lambda}$. If $f(A) \in \mathbb{R}$ (or $\mathbb{C}$) for all $f \in \Lambda$, then $\Gamma_{n_k, \hdots, n_1}(\{A_f\}_{f \in \Lambda})$ can be executed by a Blum-Shub-Smale (BSS) machine \cite{Smale_book} that takes $(n_k, \hdots, n_1)$, as input, and that has an oracle that can access any $A_f$ for $f \in \Lambda$. 
\end{remark}

\begin{remark}[Radical towers and beyond - the SCI and the insolvability of the quintic]
Similarly to the definition of an arithmetic tower, one could define a radical tower, where we let $\alpha = R$, by allowing, in addition to the arithmetic operations and comparisons, the operation $\sqrt{\cdot}$ on real numbers. In that case the recursiveness requirement above would mean recursive in the sense of a BSS machine with an oracle for the operation of computing $\sqrt{\cdot}$. Note that in this case the insolvability of the quintic becomes a question of the SCI with respect to a radical tower of algorithms. Similarly, one could define other towers by allowing other operations.  
\end{remark}

Given the definition of a tower of algorithms, we can now define the main concept of this paper: the Solvability Complexity Index (SCI). The SCI was first discussed in \cite{Hansen_JAMS} for a specific spectral problem, however, this definition extends to include general problems in computations.

\begin{definition}[Solvability Complexity Index]\label{complex_ind}
Given a computational problem $\{\Xi,\Omega,\mathcal{M},\Lambda\}$, it is said to have \emph{Solvability Complexity Index $\mathrm{SCI}(\Xi,\Omega,\mathcal{M},\Lambda)_{\alpha} = k$} with respect to a tower of 
algorithms of type $\alpha$ if $k$ is the smallest integer for which there exists a tower of algorithms of type 
$\alpha$ of height $k$. If no such tower exists then $\mathrm{SCI}(\Xi,\Omega,\mathcal{M},\Lambda)_{\alpha} = \infty.$ If 
there exists a tower $\{\Gamma_n\}_{n\in\Nb}$ of type $\alpha$ and height one such that $\Xi = \Gamma_{n_1}$ for some $n_1 < \infty$, then we define $\mathrm{SCI}(\Xi,\Omega,\mathcal{M},\Lambda)_{\alpha} = 0$.
\end{definition}

With the definition of the SCI, we can define the SCI hierarchy, for which any computational problem can be classified. Without any extra structure on the metric space $\mathcal{M}$, the $\Delta^{\alpha}_k$ classes are the finest refinement we can obtain in terms of the SCI. However, as described below, when more structure is present, the hierarchy becomes much richer.

\begin{definition}[The Solvability Complexity Index hierarchy]
\label{1st_SCI}
Consider a collection $\mathcal{C}$ of computational problems and let $\mathcal{T}$ be the collection of all towers of algorithms of type $\alpha$ for the computational problems in $\mathcal{C}$.
Define 
\begin{equation*}
\begin{split}
\Delta^{\alpha}_0 &:= \{\{\Xi,\Omega\} \in \mathcal{C} \ \vert \   \mathrm{SCI}(\Xi,\Omega)_{\alpha} = 0\}\\
\Delta^{\alpha}_{m+1} &:= \{\{\Xi,\Omega\}  \in \mathcal{C} \ \vert \   \mathrm{SCI}(\Xi,\Omega)_{\alpha} \leq m\}, \qquad \quad m \in \mathbb{N},
\end{split}
\end{equation*}
as well as
\[
\Delta^{\alpha}_{1} := \{\{\Xi,\Omega\}  \in \mathcal{C}   \  \vert \ \exists \ \{\Gamma_n\} \in \mathcal{T}\text{ s.t. } \forall A\in\Omega \ d(\Gamma_n(A),\Xi(A)) \leq 2^{-n}\}. 
\]
\end{definition}

\subsection{Extending the hierarchy for totally ordered $\mathcal{M}$}
\label{ext_dec_1}

When there is extra structure on the metric space $\mathcal{M}$, say $\mathcal{M} = \mathbb{R}$ or $\mathcal{M} = \{0,1\}$ with the standard metric, one may be able to define convergence of functions from above or below. This is an extra form of structure that allows for a type of error control. As we argue below, this is important, for example, in computer-assisted proofs, and of course, crucial in scientific computing.

\begin{definition}[The SCI Hierarchy (totally ordered set)]\label{def:tot_ord}
Given the set-up in Definition \ref{1st_SCI} and suppose in addition that $\mathcal{M}$ is a totally ordered set. 
Define 
\begin{equation*}
\begin{split}
\Sigma^{\alpha}_0 &= \Pi^{\alpha}_0 = \Delta^{\alpha}_0,\\
\Sigma^{\alpha}_{1} &= \{\{\Xi,\Omega\} \in \Delta_{2}^{\alpha} \ \vert \  \exists \ \{\Gamma_{n}\} \in \mathcal{T} \text{ s.t. } \Gamma_{n}(A) \nearrow \Xi(A) \ \, \forall A \in \Omega\}, 
\\
\Pi^{\alpha}_{1} &= \{\{\Xi,\Omega\} \in \Delta_{2}^{\alpha} \ \vert \  \exists \ \{\Gamma_{n}\} \in \mathcal{T} \text{ s.t. } \Gamma_{n}(A) \searrow \Xi(A) \ \, \forall A \in \Omega\},
\end{split}
\end{equation*}
where $\nearrow$ and $\searrow$ denotes convergence from below and above respectively,
as well as, for $m \in \mathbb{N}$, 
\begin{equation*}
\begin{split}
\Sigma^{\alpha}_{m+1} &= \{\{\Xi,\Omega\} \in \Delta_{m+2}^{\alpha} \ \vert \  \exists \ \{\Gamma_{n_{m+1}, \hdots, n_1}\} \in \mathcal{T} \text{ s.t. }\Gamma_{n_{m+1}}(A) \nearrow \Xi(A) \ \, \forall A \in \Omega\}, \\
\Pi^{\alpha}_{m+1} &= \{\{\Xi,\Omega\} \in \Delta_{m+2}^{\alpha} \ \vert \  \exists \ \{\Gamma_{n_{m+1}, \hdots, n_1}\} \in \mathcal{T} \text{ s.t. }\Gamma_{n_{m+1}}(A) \searrow \Xi(A) \ \, \forall A \in \Omega\}.
\end{split}
\end{equation*}
\end{definition}

If the metric space $\mathcal{M} = \{0,1\}$, it is clearly a totally ordered set and hence, from Definition \ref{def:tot_ord}, we get the SCI hierarchy for arbitrary decision problems. 

\subsection{Extending the hierarchy for spectral problems}

In the case where $\mathcal{M}$ is the collection of non-empty closed subsets of another metric space $(\mathcal{M}^{\prime},d')$ 
it is custom to equip $\mathcal{M}$ with the Hausdorff metric (bounded case)
\begin{equation}\label{Hausdorff}
d_\mathrm{H}(X,Y) = \max\left\{\sup_{x \in X} \inf_{y \in Y} d'(x,y), \sup_{y \in Y} \inf_{x \in X} d'(x,y) \right\},
\end{equation}
 or the Attouch--Wets metric (unbounded case) 
 \begin{equation}\label{eq:attouch-wets-metric}
	d_{\mathrm{AW}}(A,B)=\sum_{m=1}^\infty2^{-m}\min\left\{1,\sup_{d'(x,x_0)<m}\left|\mathrm{dist}(x,A)-\mathrm{dist}(x,B)\right|\right\},
	\end{equation}
where $A$ and $B$ are non-empty closed subsets of $\mathcal{M}^{\prime}$, and where $\mathrm{dist}(x,A)$ denotes the distance between the point $x \in \mathcal{M}^{\prime}$ and $A \subset \mathcal{M}^{\prime}$, and where $x_0\in\mathcal{M}^{\prime}$ can be chosen arbitrarily.

\begin{definition}[The SCI Hierarchy (Attouch--Wets/Hausdorff metric)]\label{def:SCI_Haus}
Given the set-up in Definition \ref{1st_SCI}, and suppose in addition that $(\mathcal{M},d)$ has the Attouch--Wets or the Hausdorff metric induced by another metric space $(\mathcal{M}^{\prime},d')$,
define, for $m \in \mathbb{N}$,
\begin{align*}
\Sigma^{\alpha}_0 &= \Pi^{\alpha}_0 = \Delta^{\alpha}_0,\\
\Sigma^{\alpha}_{1} &= \{\{\Xi,\Omega\} \in \Delta_{2}^{\alpha} \ \vert \  \exists \ \{\Gamma_{n}\} \in \mathcal{T}, \  \{X_{n}(A)\}\subset\mathcal{M} \text{ s.t. }  \ \Gamma_{n}(A)  \mathop{\subset}_{\mathcal{M}^{\prime}} X_n(A),\\
 & \qquad \qquad \qquad \qquad \lim_{n\rightarrow\infty}\Gamma_{n}(A)=\Xi(A),\ \ d(X_{n}(A),\Xi(A))\leq 2^{-n} \ \ \forall A \in \Omega\}, \\
\Pi^{\alpha}_{1} &= \{\{\Xi,\Omega\} \in \Delta_{2}^{\alpha} \ \vert \  \exists \ \{\Gamma_{n}\} \in \mathcal{T}, \  \{X_{n}(A)\}\subset\mathcal{M} \text{ s.t. }  \ \Xi(A)  \mathop{\subset}_{\mathcal{M}^{\prime}} X_{n}(A),\\
& \qquad \qquad \qquad \qquad \lim_{n\rightarrow\infty}\Gamma_{n}(A)=\Xi(A),\ \ d(X_{n}(A),\Gamma_n(A))\leq 2^{-n} \ \ \forall A \in \Omega\},
\end{align*}
where $\mathop{\subset}_{\mathcal{M}^{\prime}}$ means inclusion in the metric space $\mathcal{M}^{\prime}$, and $\{X_{n}(A)\}$ is a sequence where $X_n(A) \in \mathcal{M}$ depends on $A$. Moreover, 
\begin{equation*}
\begin{split}
\Sigma^{\alpha}_{m+1} = \{\{\Xi,\Omega\} \in \Delta_{m+2}^{\alpha} \ &\vert \  \exists \ \{\Gamma_{n_{m+1},...,n_1}\} \in \mathcal{T}, \  \{X_{n_{m+1}}(A)\}\subset\mathcal{M} \text{ s.t. }  \ \Gamma_{n_{m+1}}(A)  \mathop{\subset}_{\mathcal{M}^{\prime}} X_{n_{m+1}}(A),\\
 & \lim_{n_{m+1}\rightarrow\infty}\Gamma_{n_{m+1}}(A)=\Xi(A),\ \ d(X_{n_{m+1}}(A),\Xi(A))\leq 2^{-n_{m+1}} \ \ \forall A \in \Omega\}, \\
\Pi^{\alpha}_{m+1} = \{\{\Xi,\Omega\} \in \Delta_{m+2}^{\alpha} \ &\vert \  \exists \ \{\Gamma_{n_{m+1},...,n_1}\} \in \mathcal{T}, \  \{X_{n_{m+1}}(A)\}\subset\mathcal{M} \text{ s.t. }  \ \Xi(A)  \mathop{\subset}_{\mathcal{M}^{\prime}} X_{n_{m+1}}(A),\\
& \lim_{n_{m+1}\rightarrow\infty}\Gamma_{n_{m+1}}(A)=\Xi(A),\ \ d(X_{n_{m+1}}(A),\Gamma_{n_{m+1}}(A))\leq 2^{-n_{m+1}} \ \ \forall A \in \Omega\},
\end{split}
\end{equation*}
where $d$ can be either $d_{\mathrm{H}}$ or $d_{\mathrm{AW}}$.
\end{definition}

\begin{remark}[Convergence from below and above]
Intuitively, Definition \ref{def:SCI_Haus} captures convergence from below or above respectively, up to a small error parameter $2^{-n}$. 
Indeed, in the case of the Hausdorff metric case it is easy to see that the above Definition \ref{def:SCI_Haus} yields
\begin{equation*}
\begin{split}
\Sigma^{\alpha}_{1} &= \{\{\Xi,\Omega\} \in \Delta_{2}^{\alpha} \ \vert \  \exists \ \{\Gamma_{n}\} \in \mathcal{T}\text{ s.t. }  \ \Gamma_{n}(A) \subset \bar{\mathcal{N}}_{2^{-n}}(\Xi(A)) \ \ \forall A \in \Omega\},\\
\Pi^{\alpha}_{1} &= \{\{\Xi,\Omega\} \in \Delta_{2}^{\alpha} \ \vert \  \exists \ \{\Gamma_{n}\}  \in \mathcal{T} \text{ s.t. } \ \bar{\mathcal{N}}_{2^{-n}}(\Gamma_{n}(A)) \supset \Xi(A) \ \ \forall A \in \Omega\},
\end{split}
\end{equation*}
where $\bar{\mathcal{N}}_{\delta}(\omega)$ denotes the closed $\delta$-neighbourhood of $\omega \subset \mathcal{M}^{\prime}$, and similar definitions for $\Sigma^{\alpha}_{m+1}$ and $\Pi^{\alpha}_{m+1}$.
\end{remark}

Note that to build a $\Sigma_1$ algorithm, it is enough, by taking subsequences, to construct $\Gamma_n(A)$ such that $\Gamma_{n}(A) \subset \bar{\mathcal{N}}_{E_n(A)}(\Xi(A))$ with some computable $E_n(A)$ that converges to zero. 

\begin{definition}
Given a totally ordered metric space $(\mathcal{M},d)$, we say that the metric is order respecting if for any $a,b,c\in\mathcal{M}$ with $a\leq b\leq c$ we have $d(a,b)\leq d(a,c)$.
\end{definition}

\begin{proposition}[Properties of the SCI hierarchy]\label{thrm:prop_SCI}
Given the set-up, let $(\mathcal{M},d)$ be either the Hausdorff or Attouch--Wets metric or a totally ordered metric space with order respecting metric. Let $k=1,2$ or $3$, then we have the following.
\begin{itemize}
\item[(i)] $\Delta^{G}_k = \Sigma^{G}_k \cap \Pi^{G}_k$. In particular, if for a problem $
\Xi:\Omega\rightarrow \mathcal{M}$ we have $\Delta_k^G\not\ni\{\Xi,\Omega\}\in X_k^{\alpha}$, where $X=\Sigma$ or $\Pi$ and $\alpha$ denotes any type of tower, then $\{\Xi,\Omega\}\not\in Y_k^{\alpha}$, where $Y=\Pi$ or $\Sigma$ respectively.
\item[(ii)] Suppose for a computational problem $\Xi:\Omega\rightarrow \mathcal{M}$ we have a corresponding convergent $\Sigma_k^A$ tower $\Gamma_{n_k,...,n_1}^1$ and a corresponding convergent $\Pi_k^A$ tower $\Gamma_{n_k,...,n_1}^2$. Suppose also that we can compute for every $A\in\Omega$ the distance $d(\Gamma_{n_k,...,n_1}^1(A),\Gamma_{n_k,...,n_1}^2(A))$ to arbitrary precision using finitely many arithmetic operations and comparisons. Then $\{\Xi,\Omega\}\in\Delta_k^{A}$.
\end{itemize}
Finally, we also have the following property:
\begin{itemize}
\item[(iii)] When $\mathcal{M} = \{0,1\}$, $\Delta^{\alpha}_k = \Sigma^{\alpha}_k \cap \Pi^{\alpha}_k$ for all $k \in \mathbb{N}$ and $\alpha = G,A$.
\end{itemize}
\end{proposition}
The proof of Proposition \ref{thrm:prop_SCI} can be found in \S \ref{KS_results}.

\begin{remark}
Part (i) of Proposition \ref{thrm:prop_SCI} shows that the classifications obtained in this paper are sharp in the SCI hierarchy.
\end{remark}

\subsubsection{Computing approximate eigenvectors}
Let $\mathcal{C}$ denote the collection of computation spectral problems $\{\Xi,\Omega,\mathcal{M},\Lambda\}$ where $\Omega$ is a collection of normal operators on some Hilbert space $\mathcal{H}$ and $\Xi(A) = \mathrm{sp}(A)$ . If we consider bounded operators, $\mathcal{M}$ is the collection of compact subsets of $\mathbb{C}$ equipped with the Hausdorff, and in the unbounded case $\mathcal{M}$ is the collection of closed subsets of $\mathbb{C}$ with the Attouch--Wets metric.  
\begin{equation*}
\begin{split}
&\Sigma^{\alpha,\mathrm{eigv}}_1 = \{\{\Xi,\Omega\} \in \Sigma^{\alpha}_{1} \ \vert \  \exists \ \{\Gamma_{n}\} \in \mathcal{T}\text{ s.t. } \Gamma_n(A) = \{(\lambda_{1,n}, \xi_{1,n}), \hdots, (\lambda_{K,n}, \xi_{K,n})\}, \\
& \qquad \qquad \quad K = K(n) \in \mathbb{N}, \,\, \lambda_{j,n} \in \bar{\mathcal{N}}_{2^{-n}}(\mathrm{sp}(A)), \,\, \|A\xi_{j,n}-\lambda_{j,n} \xi_{j,n}\| \leq 2^{-n},  \\ 
&\qquad \qquad \quad \| \xi_{j,n}\| = 1 + a_n, \, |a_n| \leq 2^{-n}\, \forall j, \,\cup_{j=1}^K \lambda_{j,n} \rightarrow \mathrm{sp}(A), n \rightarrow \infty, \, \forall A \in \Omega \}.
\end{split}
\end{equation*}
In words $\Sigma^{\alpha,\mathrm{eigv}}_1$ can be described as follows.
\vspace{1mm}
\begin{displayquote}
\normalsize
{\it $\Sigma^{\alpha,\mathrm{eigv}}_1$ is the collection of computational spectral problems concerning normal operators that are in $\Sigma^{\alpha}_{1}$, where there exists an algorithm that can also compute approximate eigenvectors. }
\end{displayquote}
\vspace{1mm}

\subsection{Inexact input}\label{sec:inexact_input} 
Suppose we are given a computational problem $\{\Xi, \Omega, \mathcal{M}, \Lambda\}$, and that $\Lambda = \{f_j\}_{j \in \beta}$, where $\beta$ is some index set that can be finite or infinite.  However, obtaining $f_j$ may be a computational task on its own, which is exactly the problem in most areas of 
computational mathematics. In particular, for $A \in \Omega$, $f_j(A)$ could be the number $e^{\frac{\pi}{j} i }$ for example. Hence, we cannot access $f_j(A)$, but rather $f_{j,n}(A)$ where $f_{j,n}(A) \rightarrow f_{j}(A)$ as $n \rightarrow \infty$. 
Or, just as for problems that are high up in the SCI hierarchy, it could be that we need several limits, in particular one may need mappings
$f_{j,n_m,\hdots, n_1}: \Omega \rightarrow \mathbb{D} + i\mathbb{D}$, where $\mathbb{D}$ denotes the dyadic rational numbers, such that 
\begin{equation}\label{Lambda_limits}
\lim_{n_m \rightarrow \infty} \hdots \lim_{n_1 \rightarrow \infty} \|\{f_{j,n_m,\hdots, n_1}(A)\}_{j\in\beta} - \{f_j(A)\}_{j\in\beta}\|_{\infty} = 0 \quad  \forall A \in \Omega.
\end{equation}

In particular, we may view the problem of obtaining $f_j(A)$ as a problem in the SCI hierarchy, where $\Delta_1$ classification would correspond to the existence of mappings $f_{j,n}: \Omega \rightarrow \mathbb{D} + i \mathbb{D}$
such that 
 \begin{equation}\label{Lambda_limits2}
 \|\{f_{j,n}(A)\}_{j\in\beta} - \{f_j(A)\}_{j\in\beta}\|_{\infty} \leq 2^{-n} \quad \forall A \in \Omega.
 \end{equation}

This idea is formalised in the following definition.

\begin{definition}[$\Delta_{m}$-information]\label{definition:Lambda_limits}
	Let $\{\Xi, \Omega, \mathcal{M}, \Lambda\}$ be a computational problem. For $m \in \mathbb{N}$ we say that $\Lambda$ has $\Delta_{m+1}$-information if each $f_j \in \Lambda$ is not available, however, there are mappings $f_{j,n_m,\hdots, n_1}: \Omega \rightarrow \mathbb{D} + i \mathbb{D}$ such that \eqref{Lambda_limits} holds. Similarly, for $m = 0$ there are mappings $f_{j,n}: \Omega \rightarrow \mathbb{D} + i \mathbb{D}$
	such that \eqref{Lambda_limits2} holds. Finally, if $k \in \mathbb{N}$ and $\hat \Lambda$ is a collection of such functions described above such that $\Lambda$ has $\Delta_k$-information, we say that $\hat \Lambda$ provides $\Delta_k$ information for $\Lambda$. Moreover, we denote the family of all such $\hat \Lambda$ by $\mathcal{L}^k(\Lambda)$. 
\end{definition}

Note that we want to have algorithms that can handle all  computational problems $\{\Xi,\Omega,\mathcal{M},\hat \Lambda\}$ when $\hat \Lambda \in  \mathcal{L}^m(\Lambda)$. In order to formalise this we define what we mean by a computational problem with $\Delta_m$ information.

\begin{definition}[Computational problem with $\Delta_m$ information]
	Given $m \in \mathbb{N}$, a computational problem where $\Lambda$ has $\Delta_m$-information is denoted by 
	$
	\{\Xi,\Omega,\mathcal{M},\Lambda\}^{\Delta_m} := \{\tilde \Xi,\tilde \Omega,\mathcal{M},\tilde \Lambda\},
	$ 
where 
\[
\tilde \Omega = \left\{ \tilde A = \{f_{j,n_m,\hdots, n_1}(A)\}_{j,n_m,\hdots, n_1 \in \beta \times \mathbb{N}^m} \, \vert \, A \in \Omega, \{f_j\}_{j \in \beta} = \Lambda, f_{j,n_m,\hdots, n_1} \text{ satisfy (*)} \right\},
\]
and (*) denotes \eqref{Lambda_limits} if $m > 1$ and (*) denotes \eqref{Lambda_limits2} if $m = 1$. Moreover, $\tilde \Xi(\tilde A) = \Xi(A)$, and we have  
$\tilde \Lambda = \{\tilde f_{j,n_m,\hdots, n_1}\}_{j,n_m,\hdots, n_1 \in \beta \times \mathbb{N}^m}$ where $\tilde f_{j,n_m,\hdots, n_1}(\tilde A) = f_{j,n_m,\hdots, n_1}(A)$. Note that $\tilde \Xi$ is well defined by Definition \ref{def:comp_prob} of a computational problem. 
\end{definition}

The SCI and the SCI hierarchy, given $\Delta_m$-information, is then defined in the standard obvious way. We will use the notation 
\[
\{\Xi,\Omega,\mathcal{M},\Lambda\}^{\Delta_m} \in \Delta_k^{\alpha}
\]
to denote that the computational problem is in $\Delta_k^{\alpha}$ given $\Delta_m$-information. When $\mathcal{M}$ and $\Lambda$ are obvious then we will write $\{\Xi,\Omega\}^{\Delta_m} \in \Delta_k^{\alpha}$ for short.

\section{Main theorem on the general computational spectral problem}\label{finding_spectra}

For $A \in \Omega$, where $\Omega$ is an appropriate domain of operators, we define the problem functions
\begin{align}
&\Xi_{\mathrm{sp}}(A) := \mathrm{sp}(A) \quad (\text{spectrum}),
&&\Xi_{\mathrm{e}\textrm{-}\mathrm{sp}}(A) := \mathrm{sp}_{\mathrm{ess}}(A) \quad (\text{essential spectrum}) \label{eq:different_spectra} 
\\
&\Xi^N_{\mathrm{sp},\epsilon}(A) := \mathrm{sp}_{N,\epsilon}(A)  \quad (\text{pseudospectrum})
&&\Xi^z_{\mathrm{sp}}(A) := \Yes \text{ if } z \in \mathrm{sp}(A), \, \No \text{ otherwise}. \label{eq:different_spectra2} 
\end{align}

Here $\mathrm{sp}(A)$ denotes the spectrum, $\mathrm{sp}_{\mathrm{ess}}(A)$ the essential spectrum (invariant under compact perturbations) and $\mathrm{sp}_{N,\epsilon}(A)$ denotes the $(N,\epsilon)$-pseudospectrum \cite{Trefethen_Embree, Hansen_JFA, Bottcher_book}
\begin{equation}\label{eq:pseudo}
\mathrm{sp}_{N,\epsilon}(A) := \mathrm{cl}\left(\left\{z\in\Cb: \|(A-z I)^{-{2^N}}\|^{2^{-N}} > 1/\epsilon \right\}\right), \qquad N \in \mathbb{Z}_{\geq 0}, \epsilon > 0,
\end{equation}
where we use the convention that $\|(A-z I)^{-{2^N}}\| = \infty$ when $z \in \mathrm{sp}(A)$. This set has been popular in spectral theory, analysis of pseudo differential operators and non-Hermitian quantum mechanics. For computing the spectrum/essential spectrum/$(N,\epsilon)$-pseudospectrum, we consider computational problems $\{\Xi,\Omega,\mathcal{M},\Lambda\}$ a la the ones in Example \ref{Ex1} in \S \ref{sec:background} (i.e. with respect to the Hausdorff metric). For the final problem of determining if $z\in\mathrm{sp}(A)$, the metric space becomes the discrete metric on $\{\Yes,\No\}$. To avoid trivialities for this final problem, when considering self-adjoint classes of operators we will restrict to $z\in\mathbb{R}$ and when considering compact operators we will restrict to $z\neq 0$.
The key question then becomes:
\vspace{1mm}
\begin{displayquote}
\normalsize
{\it Given a problem function $\Xi$ of the form \eqref{eq:different_spectra} or \eqref{eq:different_spectra2}  with a domain $\Omega$ and evaluation set $\Lambda$, where in the SCI hierarchy is the computational problem $\{\Xi,\Omega,\mathcal{M},\Lambda\}$?}
\end{displayquote}
\vspace{1mm}

Since one can consider different classes of operators, the above question obviously becomes an infinite classification theory, however, we will establish some of the foundations. In order to do that we consider certain key domains such as the set of bounded self-adjoint or normal operators on $l^2(\mathbb{N})$, compact operators, operators on $l^2(\mathbb{N})$ with off-diagonal decay (bounded dispersion), operators with controlled growth of the resolvent etc. To define such domains we need a couple of definitions. 
 \begin{definition}[Dispersion]\label{def:disp}
We say that the dispersion of an operator $A \in \mathcal{B}(l^2(\mathbb{N}))$ is bounded by
the function $f:\Nb\to\Nb$ if
\[
D_{f,m}(A):=\max\{\|(I-P_{f(m)})AP_m\|,\|P_mA(I-P_{f(m)})\|\}\to 0\quad\text{as }m\to\infty.
\]
\end{definition}
Note that for every operator $A$ there is always a function $f$ which is a bound
for its dispersion since $AP_m$, $P_mA$ are compact and $\{P_n\}$ converges strongly to the 
identity. But there is no function $f$ which acts as a uniform bound for 
all operators. Nevertheless, there are important (sub)classes of operators having
well known uniform bounds, which should be mentioned:
\begin{itemize}
\item[(i)] Banded operators with bandwidth less than $d$: $f(k)= k+d$. More generally, we can consider operators with sparse matrices (only finitely many non-zero entries in each row and column) where $f$ captures the sparsity pattern. For example, for discrete Schr\"odinger operators on $l^2(\mathbb{Z}^2)$, we can choose an ordering of the lattice sites so that $f(k)-k=\mathcal{O}(\sqrt{k})$.
\item[(ii)] Band-dominated and weakly band-dominated operators: $f(k) = 2k$. For definitions and 
			properties of band and band-dominated operators see \cite{Roch_Silbermann_LimitOps, Lindner, SeiSurvey}.
			Weakly band-dominated operators can be found in \cite{Masc_Santos_Seidel}.
\item[(iii)] Laurent/Toeplitz operators with piecewise continuous generating function: $f(k)=k^2$
			(cf. \cite{bottcher2006analysis} and \cite[Proposition 5.4]{JuMaSe}).
\item[(iv)] Let $\mathcal{F}$ be a family of bounded operators with a common bound $f$. Then
			$\tilde{f}$, given by $\tilde{f}(k)=f(k)+k$, is a common bound for all operators in the Banach 
			algebra which is generated by $\mathcal{F}$.
\end{itemize}
Without loss of generality, we assume that $f$ is strictly increasing and $f(n)>n$.
We are also interested in operators where the control of the growth of the resolvent is bounded.  
\begin{definition}[Controlled growth of the resolvent]\label{def:res_growth}
Let $g:[0,\infty)\to[0,\infty)$ be a continuous function, vanishing only at $x=0$ and tending to infinity as $x\to\infty$ with $g(x)\leq x$. 
We say that a closed operator $A$ with non-empty spectrum on the Hilbert space $\mathcal{H}$ has controlled growth of the resolvent by $g$ if
\begin{equation}\label{EqSpDist}
\|(A-zI)^{-1}\|^{-1} \geq g(\dist(z,\spc(A)))\quad \forall  z\in \Cb,
\end{equation}
where we use the convention $\|B^{-1}\|^{-1}:=0$ if $B$ has no bounded inverse.
\end{definition}
Notice that for every bounded operator $A$ there always exists such a $g$ (define $g(\alpha):=\min\{\|(A-zI)^{-1}\|^{-1}:z\in\Cb\text{ with }\dist(z,\spc(A))=\alpha\}$, taking continuity and compactness into account) although there is no $g$ which works for all $A$. 

\begin{remark}[Assumptions on $\Lambda$]\label{rem:Lambda_bounded}
In order to make the ``additional knowledge'' $g$ available for the algorithms we assume that $\Lambda$ also contains the constant functions $g_{i,j}: A\mapsto g(i/j)$ ($i,j\in\Nb$), which provide the values of $g$ in all positive rational numbers. When considering the case of $\Delta_1$-information and arithmetic algorithms over $\mathbb{Q}$, we assume that $g$ maps $\mathbb{Q}_{\geq 0}$ to $\mathbb{Q}_{\geq 0}$ without loss of generality (e.g. by replacing $g$ with a suitable piecewise linear function). In the case when the dispersion of the operator is known, the values $f(m)$ $(m\in \Nb)$ shall be available to the algorithms as constant evaluation functions. When computing problems with $\mathrm{SCI}=1$ for $\Omega_f$ (and $\Omega_{fg}$), our algorithms also require the knowledge of a null sequence $\{c_m\}_{m\in\mathbb{N}}\subset\mathbb{Q}$ such that $D_{f,m}(A)\leq c_m$. 
\end{remark}

We consider the following domains defined below. In the cases of bounded dispersion or controlled growth of the resolvent we assume that we are given either $f$ or $g$ as above.
\begin{align*}
&\Omega_\mathrm{B} := \text{bounded operators}
&&\Omega_{\mathrm{N}} := \text{bounded normal operators},
\\
&\Omega_{\mathrm{SA}} := \text{bounded self-adjoint operators}
&&\Omega_{\mathrm{C}} := \text{compact operators},
\\
&\Omega_{f} := \text{bounded oper. w/ dispersion bounded by $f$}
&&\Omega_{g} := \text{bounded oper. w/ contr. res. growth by $g$}.
\\
&\Omega_{fg} := \Omega_{f} \cap \Omega_{g}
&&\Omega_{\mathrm{D}} := \text{bounded, diagonal, self-adjoint operators}.
\end{align*}
Note that to avoid trivialities, in the case of $\{\Xi^z_{\mathrm{sp}},\Omega_{\mathrm{D}}\}$ or $\{\Xi^z_{\mathrm{sp}},\Omega_{\mathrm{SA}}\}$ we take $z$ to be real, and in the case of $\{\Xi^z_{\mathrm{sp}},\Omega_{\mathrm{C}}\}$ we take $z\neq 0$. Given the different domains, we can now state the main theorem for bounded operators.

\begin{remark}[The upper bounds hold both in the Turing and BSS model]
Note that the results in Theorem \ref{spec_thm_main} hold with inexact input ($\Delta_1$ information) as well as with exact input. Hence, our results are valid in both the Turing and the BSS model. To avoid extra notation we will simply write  
$\{\Xi,\Omega\} \in  \Delta/\Pi/\Sigma$ rather than the correct notation $\{\Xi,\Omega\}^{\Delta_1} \in  \Delta/\Pi/\Sigma$. 
\end{remark}

\begin{theorem}[The bounded computational spectral problem]
\label{spec_thm_main}
Given the set-up above we have the following classification results in the SCI hierarchy.
\begin{itemize}
\item[(i)] Spectrum:
  \begin{align*}
& \Delta^G_3 \not\owns \{\Xi_{\mathrm{sp}},\Omega_{\mathrm{B}}\} \in \Pi^A_3\,\, \, (\text{all oper.}), 
&& \Delta^G_2 \not\owns \{\Xi_{\mathrm{sp}},\Omega_{\mathrm{N}}\} \in \Sigma^A_2\,\, \, (\text{normal}), \\
& \Delta^G_2 \not\owns \{\Xi_{\mathrm{sp}},\Omega_{\mathrm{SA}}\} \in \Sigma^A_2\,\, \, (\text{self-adj.}), 
&& \Sigma_1^G\cup\Pi_1^G\not\owns\{\Xi_{\mathrm{sp}}, \Omega_{\mathrm{C}} \} \in \Delta^{A}_2\,\, \, (\text{compact}),\\ 
&\Delta^G_2 \not\owns \{\Xi_{\mathrm{sp}},\Omega_f\} \in \Pi^A_2 \,\, \, (\text{disp. bound. by $f$}),
&&\Delta^G_2 \not\owns \{\Xi_{\mathrm{sp}},\Omega_{g}\} \in \Sigma^A_2\,\, \, (\text{resolvent growth bound. by $g$}),
\\
&\Delta^G_1 \not\owns \{\Xi_{\mathrm{sp}},\Omega_{fg}\} \in \Sigma^{A}_1\,\, \, 
&& \Delta^G_1 \not\owns\{\Xi_{\mathrm{sp}},\Omega_{f}\cap \Omega_{\mathrm{N}}\}\in\Sigma_1^{A,\mathrm{eigv}}.
\end{align*}

\item[(ii)] Essential spectrum:
\begin{align*}
& \Delta^G_3 \not\owns \{\Xi_{\mathrm{e\textrm{-}sp}},\Omega_{\mathrm{B}}\} \in \Pi^A_3\,\, \, (\text{all oper.}), 
&& \Delta^G_3 \not\owns \{\Xi_{\mathrm{e\textrm{-}sp}},\Omega_{\mathrm{N}}\} \in \Pi^A_3\,\, \, (\text{normal}), \\
& \Delta^G_3 \not\owns \{\Xi_{\mathrm{e\textrm{-}sp}},\Omega_{\mathrm{SA}}\} \in \Pi^A_3\,\, \, (\text{self-adj.}), 
&& \Delta^G_2 \not\owns \{\Xi_{\mathrm{e\textrm{-}sp}},\Omega_{\mathrm{D}}\} \in \Pi^A_2 \,\, \, (\text{self-adj. diag.}),\\ 
&\Delta^G_2 \not\owns \{\Xi_{\mathrm{e\textrm{-}sp}},\Omega_f\} \in \Pi^A_2 \,\, \, (\text{disp. bound. by $f$}),
&&\Delta^G_3 \not\owns \{\Xi_{\mathrm{e\textrm{-}sp}},\Omega_{g}\} \in \Pi^A_3\,\, \, (\text{resolvent growth bound. by $g$}),\\
&\Delta^G_2 \not\owns \{\Xi_{\mathrm{sp}},\Omega_{fg}\} \in \Pi^A_2\,\, \, (\text{res. growth bound. by $g$}
&& \hspace{-2mm}\text{and disp. bound. by $f$}).
\end{align*}

\item[(iii)] Pseudospectrum:
\begin{align*}
& \Delta^G_2 \not\owns \{\Xi^N_{\mathrm{sp},\epsilon},\Omega_{\mathrm{B}}\} \in \Sigma^A_2\,\, \, (\text{all oper.}), 
&& \Delta^G_2 \not\owns \{\Xi^N_{\mathrm{sp},\epsilon},\Omega_{\mathrm{N}}\} \in \Sigma^A_2\,\, \, (\text{normal}), \\
& \Delta^G_2 \not\owns \{\Xi^N_{\mathrm{sp},\epsilon},\Omega_{\mathrm{SA}}\} \in \Sigma^A_2\,\, \, (\text{self-adj.}), 
&& \Sigma_1^G\cup\Pi_1^G\not\owns\{\Xi^N_{\mathrm{sp},\epsilon}, \Omega_{\mathrm{C}} \} \in \Delta^{A}_2\,\, \, (\text{compact}),\\ 
&\Delta^G_1 \not\owns \{\Xi^N_{\mathrm{sp},\epsilon},\Omega_f\} \in \Sigma^A_1 \,\, \, (\text{disp. bound. by $f$}),
&&\Delta^G_2 \not\owns \{\Xi^N_{\mathrm{sp},\epsilon},\Omega_{g}\} \in \Sigma^A_2\,\, \, (\text{resolvent growth bound. by $g$}),\\
&\Delta^G_1 \not\owns \{\Xi_{\mathrm{sp}},\Omega_{fg}\} \in \Sigma^A_1\,\, \, (\text{res. growth bound. by $g$}
&& \hspace{-2mm}\text{and disp. bound. by $f$}).
\end{align*}

\item[(iv)] Is $z$ in the spectrum?:
\begin{align*}
& \Delta^G_3 \not\owns \{\Xi^z_{\mathrm{sp}},\Omega_{\mathrm{B}}\} \in \Pi^A_3\,\, \, (\text{all oper.}), 
&& \Delta^G_3 \not\owns \{\Xi^z_{\mathrm{sp}},\Omega_{\mathrm{N}}\} \in \Pi^A_3\,\, \, (\text{normal}), \\
& \Delta^G_3 \not\owns \{\Xi^z_{\mathrm{sp}},\Omega_{\mathrm{SA}}\} \in \Pi^A_3\,\, \, (\text{self-adj.}), 
&& \Delta^G_2\not\owns\{\Xi^z_{\mathrm{sp}}, \Omega_{\mathrm{C}} \} \in \Pi^A_2\,\, \, (\text{compact}),\\ 
&\Delta^G_2 \not\owns \{\Xi^z_{\mathrm{sp}},\Omega_f\} \in \Pi^A_2 \,\, \, (\text{disp. bound. by $f$}),
&&\Delta^G_3 \not\owns \{\Xi^z_{\mathrm{sp}},\Omega_{g}\} \in \Pi^A_3\,\, \, (\text{resolvent growth bound. by $g$}),\\
&\Delta^G_2 \not\owns \{\Xi^z_{\mathrm{sp}},\Omega_\mathrm{D}\} \in \Pi^A_2 \,\, \, (\text{self-adj. diag.}),&& \Delta^G_2 \not\owns \{\Xi_{\mathrm{sp}},\Omega_{fg}\} \in \Pi^A_2\,\, \, (\text{res. growth bound. by $g$ and disp. bound. by $f$}).
\end{align*}
\end{itemize}
\end{theorem}

\begin{remark}
\label{c_ns}
In order to gain the $\Sigma_1^A$ algorithms for $\Xi^N_{\mathrm{sp},\epsilon}$ we need an upper bound for $\|A\|$ when $N>0$ (without which we gain a $\Delta_2^A$ classification). No such knowledge is needed for the other towers of algorithms.
\end{remark}

\begin{remark}
\label{compact_special}
The proofs also show that the above lower bounds for compact operators hold when considering self-adjoint compact operators.
\end{remark}

\section{Main theorems on computational quantum mechanics}\label{quantum_mech}

Here we formally state the results summarised in \S \ref{sec:schrodiner_main}. We consider the spectral and pseudospectral mappings $\Xi_{\mathrm{sp}}$, $\Xi_{\mathrm{sp},\epsilon}$ from \eqref{eq:different_spectra} for Schr\"odinger operators: 
	\begin{equation}\label{eq:schrodinger2}
	 H=-\Delta+V,\qquad V:\R^d\to\C.
	\end{equation}
We assume that the information the algorithm can read are point samples $V(x)$ for $x \in \mathbb{Q}^d$. In particular, $\Lambda$ is as in \ref{Ex1} in \S \ref{sec:background}. Moreover, $\mathcal{M}$ is the collection of non-empty closed subsets of $\mathbb{C}$ with the standard Attouch--Wets metric \eqref{eq:attouch-wets-metric}. If we fix the domain of $H$ such that it is appropriate for a class of potentials $V$, the spectrum of $H$ is uniquely determined by $V$. The basic question is therefore: 
\vspace{1mm}
\begin{displayquote}
\normalsize
{\it Given a class of Schr\"{o}dinger operators $-\Delta + V \in \Omega$, let $\Xi$ be either $\Xi_{\mathrm{sp}}$ or $\Xi_{\mathrm{sp},\epsilon}$, $\Lambda$ and $\mathcal{M}$ as above, where in the SCI hierarchy is the computational problem $\{\Xi,\Omega,\mathcal{M},\Lambda\}$?}
\end{displayquote}
\vspace{1mm}

Though we have stuck to the Hilbert space $\mathrm{L}^2(\mathbb{R}^d)$ for simplicity, the algorithms we construct can also be adapted for other spaces commonly found in applications such as $\mathrm{L}^2(\mathbb{R}_{>0})$.

\vspace{2mm}

\noindent {\bf Bounded Potentials.}
We first consider cases with bounded potential. In particular, let $\phi:[0,\infty)\to[0,\infty)$ be some increasing function and $M > 0$, define
\begin{align*}
\Omega_{\phi} & :=  \{H : \Dcal(H)=\mathrm{W}^{2,2}(\Rbb^d), V \in \mathrm{BV}_{\phi}(\mathbb{R}^d), \|V\|_{\infty} \leq M\}, \\
\Omega_{\phi, g} & := \{H\in\Omega_{\phi}:\|(-\Delta+V-zI)^{-1}\|^{-1} \geq g(\dist(z,\spc(H)))\},
\end{align*}
where 
\begin{equation}
\label{BV_bound}
\mathrm{BV}_{\phi}(\mathbb{R}^d) = \{f : \mathrm{TV}(f_{[-a,a]^d}) \leq \phi(a)\},
\end{equation}
($f_{[-a,a]^d}$ means $f$ restricted to the box $[-a,a]^d$)
with $\mathrm{TV}$ being the total variation of a function in the sense of Hardy and Krause (see \cite{Niederreiter}). Here as in \S \ref{finding_spectra}, $g:[0,\infty)\to[0,\infty)$ is a continuous strictly increasing function with $g(x)\leq x$, vanishing only at $x=0$ and tending to infinity as $x\to\infty$.

Note that the set $\Omega_{\phi}$ requires a little bit more than $V$ just being locally of bounded variation. There is a universal upper bound across the class on the growth of the total variation of the potential function as we restrict the function to a larger set. The class $\Omega_{\phi, g}$ obviously includes self-adjoint Sch\"odinger operators in $\Omega_{\phi}$, however, it is much larger. We denote the class of self-adjoint Sch\"odinger operators in $\Omega_{\phi}$ by $\Omega_{\phi,\mathrm{SA}}$.

\begin{remark}[Assumptions on $\Lambda$]\label{Asmpt_on_Lam}
In addition to containing the point sampling functions $f_x$ such that $f_x(V) = V(x)$ for $x \in \mathbb{Q}^d$ we have the following. As done in the case of bounded Hilbert space operators discussed in Remark \ref{rem:Lambda_bounded}, the additional knowledge of $g$, describing the growth of the resolvent, is available for the algorithms by assuming that $\Lambda$ also contains the constant functions $g_{i,j}: V \mapsto g(i/j)$ ($i,j\in\Nb$), which provide the values of $g$ in all positive rational numbers (again in the case of $\Delta_1$-information and arithmetic algorithms over $\mathbb{Q}$, we assume that $g(\mathbb{Q}_{\geq 0})\subset\mathbb{Q}_{\geq 0}$ without loss of generality). Moreover, $\Lambda$ contains the constant functions $\phi_n: V \mapsto \phi(n)$ for $n \in \mathbb{N}$ and we assume without loss of generality that $\phi(n)\in\mathbb{Q}$.
\end{remark}

\begin{remark}[The upper bounds hold both in the Turing and BSS model]
Note that the results in Theorem \ref{main_self_adjoint} and Theorem \ref{thm:comp-res} hold with inexact input ($\Delta_1$ information) as well as with exact input. Hence, our results are valid in both the Turing and the BSS model. To avoid extra notation we will simply write  
$\{\Xi,\Omega\} \in  \Delta/\Pi/\Sigma$ rather than the correct notation $\{\Xi,\Omega\}^{\Delta_1} \in  \Delta/\Pi/\Sigma$. 
\end{remark}

\begin{theorem}[Bounded potential]\label{main_self_adjoint}
Given the above set-up, we have the following classification results.
\begin{align*}
&\Delta_1^G\not\ni\{\Xi_{\mathrm{sp}},\Omega_{\phi}\}\in\Pi_2^A, \hspace{-20mm}&&\Delta_1^G\not\ni\{\Xi_{\mathrm{sp},\epsilon},\Omega_{\phi}\}\in\Sigma_1^A,\\
&\Delta_1^G\not\ni\{\Xi_{\mathrm{sp}},\Omega_{\phi,g}\}\in  \Sigma^{A}_1, \hspace{-20mm}&& \Delta_1^G\not\ni\{\Xi_{\mathrm{sp},\epsilon},\Omega_{\phi,g}\}\in\Sigma_1^A,\\
& \Delta_1^G\not\ni\{\Xi_{\mathrm{sp}},\Omega_{\phi,\mathrm{SA}}\}\in  \Sigma^{A,\mathrm{eigv}}_1.
\end{align*}
\end{theorem}

\begin{remark}
When considering the problem of computing approximate eigenvectors by arithmetic algorithms, we need a suitable way of encoding the space. We choose do to so via computing coefficients of a function with respect to an orthonormal basis in $\mathrm{L}(\mathbb{R}^d)$, where each of these is a simple function consisting of trigonometric and rational functions.
\end{remark}

As will be evident from the proof techniques, one can build towers of algorithms for operators with more general classes of potentials (for example $\mathrm{L}^1(\mathbb{R}^d) \cap \mathrm{BV}_{\mathrm{loc}}(\mathbb{R}^d)$ or $\mathrm{L}^2(\mathbb{R}^d) \cap \mathrm{BV}_{\mathrm{loc}}(\mathbb{R}^d)$), however, the height of these towers will be higher than the ones considered in this paper. The main future task is to obtain exact values of the SCI of the spectrum, given the different potential classes.

\vspace{2mm}

\noindent {\bf Unbounded Potentials.}
We get a rather intriguing phenomenon for sectorial operators. Namely, the SCI of both the spectrum and the pseudospectrum is one, but no type of error control is possible. In particular, suppose that  we have non-negative $\theta_1,\theta_2$  such that $ \theta_1+\theta_2<\pi$.
Define
\begin{equation}\label{eq:sector}
\Omega_{\infty} = \{V \in \mathrm{C}(\mathbb{R}^d): \forall x \,\mathrm{arg}(V(x)) \in [-\theta_2,\theta_1], |V(x)| \rightarrow \infty  \text{ as }  x \rightarrow \infty \}.
\end{equation}
We define the operator $H$ via the minimal operator  $h$ as:
	$
	H=h^{**},
	$ 
	$h=-\Delta+V,$ $\Dcal(h)=\mathrm{C}^{\infty}_c(\mathbb{R}^d).
$
When $V \in \Omega_{\infty}$ it follows that $H$ has compact resolvent, a result that we also establish as a part of the proof of the following theorem.

Interestingly, no constant functions are needed in $\Lambda$ in order to obtain the results in the following theorem, as opposed to the case where we have a bounded potential. 

\begin{theorem}[Unbounded potential]\label{thm:comp-res}
Given the above set-up, we have the following classification results 
$$
\Sigma_1^G\cup\Pi_1^G\not\ni\{\Xi_{\mathrm{sp}},\Omega_\infty\}\in\Delta^A_2,\quad \Sigma_1^G\cup\Pi_1^G\not\ni\{\Xi_{\mathrm{sp},\epsilon},\Omega_\infty\}\in\Delta^A_2.
$$
\end{theorem}

Note that the key to this result is the compact resolvent of $H$. It is therefore natural that these problems have the same SCI classification as for compact operators $\Omega_\mathrm{C}$ (see Theorem \ref{spec_thm_main} in \S \ref{finding_spectra}).
The continuity assumption on $V$ in Theorem \ref{thm:comp-res} is to make sure that the discretisation used converges. However, by tweaking with the approximation, this assumption can be weakened to include potentials that have certain discontinuities.

\section{Main theorems on solving linear systems}\label{linear_systems}

Suppose that 
$b \in l^2(\mathbb{N})$, $A \in \mathcal{B}_{\mathrm{inv}}(l^2(\mathbb{N}))$ (the set of bounded invertible operators) and $\Omega \subset  \mathcal{B}_{\mathrm{inv}}(l^2(\mathbb{N})) \times l^2(\mathbb{N})$ and we define the mappings 
$
\Xi_{\mathrm{inv}}: \Omega \ni (A,b) \mapsto A^{-1}b,
$
and $
\Xi_{\mathrm{norm}}: A \mapsto \|A^{-1}\|^{-1}.
$
Depending on the problem function, $\mathcal{M}$ is either $l^2(\mathbb{N})$ or $\mathbb{R}$ with the canonical metrics. 
We ask the following basic question:
\begin{displayquote}
{\it Where in the SCI hierarchy are the computational problems $\{\Xi,\Omega,\mathcal{M},\Lambda\}$ for different domains $\Omega$ when $\Xi$ is either $\Xi_{\mathrm{inv}}$ or $\Xi_{\mathrm{norm}}$, with the appropriate choices of $\mathcal{M}$?}
\end{displayquote}

\begin{remark}[Assumptions on $\Lambda$]
Here, as in Example \ref{Ex1}, we again suppose that the set $\Lambda$ of evaluations consists of the functions which read the matrix elements $\{\langle Ae_j,e_i \rangle\}_{i,j \in \mathbb{N}}$ and the sequence entries $\{\langle b,e_k \rangle\}_{k \in \mathbb{N}}$ of  $(A,b)\in\Omega$.
Also, in the case when the dispersion of the operator is known, the values $f(m)$ $(m\in \Nb)$ shall be available to the algorithms as constant evaluation functions. However, if the dispersion is not known, then $\Lambda$ will not contain any constant functions in the theorems below.
\end{remark}

\begin{remark}[The upper bounds hold both in the Turing and BSS model]
Note that the results in Theorem \ref{linear_systems_thrm} and Theorem \ref{thrm:norm_inverse} hold with inexact input ($\Delta_1$ information) as well as with exact input. Hence, our results are valid in both the Turing and the BSS model. To avoid extra notation we will simply write  
$\{\Xi,\Omega\} \in \Delta/\Pi/\Sigma$ rather than the correct notation $\{\Xi,\Omega\}^{\Delta_1} \in \Delta/\Pi/\Sigma$. 
\end{remark}

\begin{theorem}[Solving linear systems]\label{linear_systems_thrm}
Let $\mathcal{B}_{\mathrm{inv},f}(l^2(\mathbb{N}))$ denote the class of bounded invertible operators with dispersion bounded by $f : \mathbb{N} \rightarrow \mathbb{N}$, $\mathcal{B}_{\mathrm{inv},f}^M(l^2(\mathbb{N}))$ denote the class of operators $A\in\mathcal{B}_{\mathrm{inv},f}(l^2(\mathbb{N}))$ with the $\|A^{-1}\|\leq M$, $\mathcal{B}_{\mathrm{inv},sa}(l^2(\mathbb{N}))$ denote the class of bounded invertible self-adjoint operators, and define the domains $\Omega_1 =  \mathcal{B}_{\mathrm{inv}}(l^2(\mathbb{N})) \times l^2(\mathbb{N})$,  
$\Omega_2 =  \mathcal{B}_{\mathrm{inv},sa}(l^2(\mathbb{N})) \times l^2(\mathbb{N})$ and
$\Omega_3 =  \mathcal{B}_{\mathrm{inv},f}(l^2(\mathbb{N})) \times l^2(\mathbb{N})$.
\begin{align*}
&\Delta^G_2\not\ni\{\Xi_{\mathrm{inv}},\Omega_1\}\in\Delta_3^A,\\
&\Delta^G_2\not\ni\{\Xi_{\mathrm{inv}},\Omega_2\}\in\Delta_3^A,\\
&\Delta^G_1\not\ni\{\Xi_{\mathrm{inv}},\Omega_3\}\in\Delta_2^A.
\end{align*}
Furthermore, if we define $\Omega_4 =  \mathcal{B}_{\mathrm{inv},f}^M(l^2(\mathbb{N})) \times l^2(\mathbb{N})$, and in this particular case we assume knowledge of a null sequence $\{c_m\}_{m\in\mathbb{N}}$ such that $D_{f,m}(A)\leq c_m$ and $\|b-P_mb\|\leq c_m$ then we have the error control
\begin{equation}
\Delta_0^G\not\owns\{\Xi_{\mathrm{inv}},\Omega_4\}\in\Delta_1^A.
\end{equation}
\end{theorem}

Another problem of interest when dealing with solutions of linear systems of equations is the computation of the norm of the inverse. This is obviously related to the stability of the problem. The task of computing the norm of the inverse of an operator can also be analysed in terms of the SCI, and that is the topic of the next theorem. Note that since our metric space is $\mathbb{R}$ with the usual metric, we have a notion of $\Sigma$ or $\Pi$ convergence.

\begin{theorem}[Computing norm of the inverse]\label{thrm:norm_inverse}
Let $\Omega_1=\mathcal{B}(l^2(\mathbb{N}))$, $\Omega_2$ the subset of self-adjoint operators, $\Omega_3$ the subset of operators with dispersion bounded by an $f : \mathbb{N} \rightarrow \mathbb{N}$, and let 
$
\Xi_{\mathrm{norm}}: A \mapsto \|A^{-1}\|^{-1}.
$
\footnote{As usual, $\|A^{-1}\|^{-1}:=0$ if $A$ is not invertible. We could have equally chosen to compute $\|A^{-1}\|$ with the point at infinity added to a suitable metrisation of $\mathbb{R}$. In this case we would get a $\Sigma$ rather than a $\Pi$ classification.}
Then 
\begin{equation}\label{SCI_lin2}
\Delta^G_2\not\ni\{\Xi_{\mathrm{norm}},\Omega_1\}\in\Pi_2^A,\quad \Delta^G_2\not\ni\{\Xi_{\mathrm{norm}},\Omega_2\}\in\Pi_2^A,\quad\Delta^G_1\not\ni\{\Xi_{\mathrm{norm}},\Omega_3\}\in\Pi_1^A.
\end{equation}
\end{theorem}

\begin{remark}
As in the spectral case, we require the knowledge of a null sequence $c_m$ such that $D_{f,m}(A)\leq c_m$ in order to gain $\{\Xi_{\mathrm{norm}},\Omega_3\}\in\Pi_1^A$. Without this knowledge the constructed algorithm gives a $\Delta^A_2$ classification.
\end{remark}

\section{Proof of Theorem \ref{spec_thm_main}}\label{finding_spectra_proofs}

We start the sections on the proofs of our main results with a simple but fundamental observation on the smallest singular values $\sigma_1(B)$ of finite matrices $B\in\Cb^{m\times n}$, which constitutes one of the cornerstones for most of the general algorithms we will construct in the subsequent proofs. Note that when dealing with infinite-dimensional operators, we will also use the notation $\sigma_1$ to denote the injection modulus defined, for $A \in \mathcal{B}(\mathcal{H})$ on some Hilbert space $\mathcal{H}$, as
\[
\sigma_1(A):=\inf_{\|x\|=1}\|Ax\|.
\]
\begin{proposition}\label{PCholesky}
Given a matrix $B\in\Cb^{m\times n}$ and a number $\epsilon>0$ one can test with finitely many arithmetic operations of the entries of $B$ whether the smallest singular value $\sigma_1(B)$ of $B$ is greater than $\epsilon$.
\end{proposition}
\begin{proof}
The matrix $B^*B$ is self-adjoint and positive semi-definite, hence has its eigenvalues in $[0,\infty)$. The singular values of $B$ are the square roots of these eigenvalues of $B^*B$.
The smallest singular value is greater than $\epsilon$ if and only if the smallest eigenvalue of $B^*B$ is greater than $\epsilon^2$, which is the case if and only if $C:=B^*B-\epsilon^2I$ is positive definite.
It is well known that $C$ is positive definite if and only if the pivots left after Gaussian elimination (without row exchange) are all positive. Thus, if $C$ is positive definite, Gaussian elimination leads to pivots that are all positive, and this requires finitely many arithmetic operations. If $C$ is not positive definite, then at some point a pivot is zero or negative, at this point the algorithm aborts. An alternative is the Cholesky decomposition. Although forming the lower triangular $L \in \mathbb{C}^{n \times n}$ (if it exists) such that $C = LL^*$ requires the use of radicals, the existence of $L$ can be determined using finitely many arithmetic operations. This follows from the standard Cholesky algorithm, and we omit the details.  
\end{proof}

\begin{proposition}
\label{REC_SING}
Given a matrix $B\in\Cb^{m\times n}$ with $\Delta_1$-information for the matrix entries of $B$, and $\eta>0$, we can compute $\sigma_1(B)$ to accuracy $\eta$ using finitely many arithmetic operations and comparisons over $\mathbb{Q}$.
\end{proposition}

\begin{proof}
Without loss of generality, we can assume that $\eta\in\mathbb{Q}$. Let $\hat B$ be a rational approximation of $B$, obtained using $\Delta_1$-information, such that $\|B-\hat B\|\leq\eta/2$. Note that we can bound the operator norm by the Frobenius norm and hence can guarantee $\|B-\hat B\|\leq\eta/2$ if each matrix entry of $\hat B$ is accurate to $\eta(2\sqrt{mn})^{-1}$ (we can choose a smaller rational accuracy parameter). It then follows that
$
|\sigma_1(B)-\sigma_1(\hat B)|\leq \|B-\hat B\|\leq\frac{\eta}{2}.
$
The proposition follows if we can compute $\sigma_1(\hat{B})$ to accuracy $\eta/2$. To do this, let $M\in\mathbb{N}$ be such that $M^{-1}<\eta/2$. Using Proposition \ref{PCholesky} (note that this only requires arithmetic operations and comparisons over $\mathbb{Q}$) and applying successive tests to $\epsilon=1/M,2/M,...$, we can compute the smallest $k\in\mathbb{N}$ such that $\sigma_1(\hat B)\leq k/M$. Our approximation is then given by $k/M$.
\end{proof}

\begin{remark}[Proofs of $\{\Xi,\Omega\}^{\Delta_1} \in \Delta/\Pi/\Sigma$]
All our theorems are valid regardless of inexact input ($\Delta_1$ information), and the main reason is Proposition \ref{REC_SING}. There are only minor alterations that need to be done in the proofs in order to deal with inexact input, and there will be guidelines specifying where the changes are needed. 
Note that there are much more numerically efficient procedures than the one in the proof of Proposition \ref{REC_SING}. However, the purpose of Proposition \ref{REC_SING} is to show that the algorithms we construct in this paper can be made to work in a 100\% rigorous manner on a Turing machine with inexact $\Delta_1$-information.
\end{remark}

We will split the proof of Theorem \ref{spec_thm_main} into several parts, and a brief roadmap for the proof is as follows. We first deal with computing the spectra and pseudospectra of compact operators since the constructive parts of the proof uses a different (most likely more familiar) method, the finite section method, than the proof for the other classes of operators. Step I of this part also contains one of the arguments used to prove lower bounds throughout this paper and is written out in detail for the reader's convenience. We then move onto pseudospectra where variants on the method of uneven sections are used to approximate the relevant resolvent norms. In some cases, these towers are used directly to provide (with an additional limit) towers of algorithms for the spectra. The proof that $\{\Xi_{\mathrm{sp}},\Omega_g\}\in\Sigma_2^A$ uses a very different method to those usually found in the literature, a local estimation of the resolvent norm (using similar ideas to \S \ref{pf_pseduo}) together with the function $g$ gives rise to upper bounds on the distance of a point to the spectrum. This is then used in a local search routine to compute the spectrum. The proof that $\{\Xi_\mathrm{sp},\Omega_\mathrm{B}\}\notin\Delta_3^G$ relies on reducing a decision problem, known to require three limits, to $\{\Xi_\mathrm{sp},\Omega_\mathrm{B}\}$. Proof that the decision problem requires three limits is provided in \S \ref{dec_sec} via a Baire category argument. The constructive proofs for essential spectra build on the towers of algorithms for computing spectra but are more involved. We end with the problem $\Xi^{z}_{\mathrm{sp}}$ where the proof of lower bounds uses similar arguments for the other problem functions, and the construction of towers of algorithms uses the towers constructed in \S \ref{spec_proof} for the spectrum.

\subsection{Spectra and pseudospectra of compact operators}

\begin{proof}[Proof of Theorem \ref{spec_thm_main} for compact operators]

\textbf{Step I}: $\{\Xi_{\mathrm{sp}},\Omega_{\mathrm{C}}\}\notin\Sigma_1^G$. We argue by contradiction and suppose that there is a sequence $\{\Gamma_n\}$ of general algorithms such that, for every $A\in\Omega_{\mathrm{C}}$, $\Gamma_n(A) \rightarrow \mathrm{sp}(A)$ with $\Gamma_n(A)\subset\mathrm{sp}(A)+B_{2^{-n}}(0)$, and in particular each $\Lambda_{\Gamma_n}(A)$ is finite. Thus, we define 
$
N(A,n) := \max\{i,j \, \vert \, f_{i,j} \in \Lambda_{\Gamma_n}(A)\}.
$
 We consider an operator of the type
\begin{equation*}
A:= A_{k}\oplus\mathrm{diag}\{0,0,...\} \quad\text{with } 
A_{k}:=\begin{pmatrix}
1& & & &1\\
 &0& & & \\
 & &\ddots& & \\
 & & &0& \\
1& & & &1\\
\end{pmatrix}
\in\mathbb{C}^{k\times k},
\end{equation*}
where we will choose the specific value of $k$ later. 
Let $C=\mathrm{diag}\{1,0,0,...\}$ then $\mathrm{sp}(C)=\{0,1\}$ and clearly $A$ is compact with $\mathrm{sp}(A)=\{0,2\}$. We choose $k$ to gain a contradiction as follows. There exists $n$ such that
$
\Gamma_n(C)\cap B_{1/4}(1)\neq\emptyset.
$
Let $k>N(C,n)$. By this construction, it follows that 
$
\Gamma_{n}(C)=\Gamma_n(A).
$
Indeed, since any evaluation function $f_{i,j} \in \Lambda$ just provides the $(i,j)$-th matrix element, it follows by the choice of $k$ that for any evaluation functions $f_{i,j} \in \Lambda_{\Gamma_{n}}(A)$ we have that 
that 
$
f_{i,j}(A) = f_{i,j}(C).
$ Thus, 
by assumption (iii) in the definition of a general algorithm (Definition \ref{alg}), we get that $\Lambda_{\Gamma_{n}}(A) = \Lambda_{\Gamma_{n}}(C)$ which, by assumption (ii) in Definition \ref{alg}, yields $
\Gamma_{n}(C)=\Gamma_n(A).
$ But then $\Gamma_n(A)\cap B_{1/4}(1)\neq\emptyset$, which is impossible since 
$
\Gamma_n(A)\subset \{0,2\}+B_{2^{-n}}(0), 
$
a contradiction.

\textbf{Step II}: $\{\Xi_{\mathrm{sp}},\Omega_{\mathrm{C}}\}\notin\Pi_1^G$. This is essentially the same argument. Assume that there exists $\Gamma_n$ such that $\mathrm{sp}(A)\subset\Gamma_n(A)+B_{2^{-n}}(0)$. Let $A$ and $C$ be as before. But now we know that there exists $n$ such that
$
\Gamma_n(C)\cap B_{3/4}(2)=\emptyset.
$
We argue as before, choosing $k>N(C,n)$, to get $\Gamma_{n}(C)=\Gamma_n(A)$. But we must have $2\in\Gamma_n(A)+B_{2^{-n}}(0)$, a contradiction.

\textbf{Step III}: $\{\Xi^N_{\mathrm{sp},\epsilon},\Omega_{\mathrm{C}}\}\notin\Pi_1^G\cup\Sigma_1^G$. For sufficiently small $\epsilon$ we have the required separation such that the above argument works for $\Xi^N_{\mathrm{sp},\epsilon}$. For larger $\epsilon$ we simply rescale the operators in the argument appropriately.

\textbf{Step IV}: $\{\Xi_{\mathrm{sp}},\Omega_{\mathrm{C}}\}\in\Delta^A_2$. For $n\in\mathbb{N}$, let
$
G_n=\frac{1}{n}(\mathbb{Z}+i\mathbb{Z})\cap B_{n}(0).
$
For $A\in\Omega_{\mathrm{C}}$ let
$
\Gamma_n(A)=\{z\in G_n:\sigma_1(P_n(A-zI)P_n)\leq 1/n\},
$
where $P_n$ denotes the orthogonal projection onto the linear span of the first $n$ basis vectors. By Proposition \ref{PCholesky}, it is clear that this can be computed in finitely many arithmetical operations and comparisons.
Hence we are done if we can prove convergence, the proof of which will make clear that we can make $\Gamma_n(A)$ non-empty by replacing $\Gamma_n(A)$ with $\Gamma_{m(n)}(A)$ such that $m(n)\geq n$ is minimal with $\Gamma_{m(n)}(A)\neq\emptyset$. Let $\epsilon>0$, then choose $N>2/\epsilon$. If $n\geq N$ and $z\in\Gamma_n(A)$ then we must have $\sigma_1(P_n(A-zI){P_n})\leq \epsilon/2$. Hence there exists $x_n\in l^2(\mathbb{N})$ of norm $1$ and with $x_n=P_nx_n$ such that
$
\left\|(P_nA-zI)x_n\right\|\leq\epsilon/2.
$
$A$ is compact and hence we can choose $N$ large if necessary to ensure that $\left\|(I-P_n)A\right\|\leq\epsilon/2$. It follows that $\left\|(A-zI)x_n\right\|\leq\epsilon$ and hence $z$ is in $\mathrm{sp}_\epsilon(A)$. Note that $N$ does not depend on the point $z$ so for large $n$ we have $\Gamma_n(A)\subset\mathrm{sp}_\epsilon(A).$

Conversely, let $z\in\mathrm{sp}(A)$. The method of finite section converges for compact operators and hence there exists $z_n\in \mathrm{sp}(P_nA{P_n})$ with $z_n\rightarrow z$. Let $w_n\in G_n$ be of minimal distance to $z_n$ then for large $n$ we must have $\left|w_n-z_n\right|\leq {1}/(\sqrt{2}n)$ and hence 
$
\sigma_1(P_n(A-w_nI)P_n)\leq {1}/(\sqrt{2}n)<1/n. 
$
It follows that $w_n\in\Gamma_n(A)$. Let $\epsilon>0$, then we can choose a finite set 
$
S_{\epsilon}\subset\mathrm{sp}(A)
$
with 
$
d_{\mathrm{H}}(S_{\epsilon},\mathrm{sp}(A))<\epsilon/2.
$
Applying the above argument to all points in $S_\epsilon$ implies that for large $n$ we must have that 
$
\mathrm{sp}(A)\subset \Gamma_n(A)+B_{\epsilon}(0).
$
Hence, since $\epsilon>0$ was arbitrary, the fact that $\Gamma_n(A)\subset\mathrm{sp}_\epsilon(A)$ implies the required convergence.

\textbf{Step V}: $\{\Xi^N_{\mathrm{sp},\epsilon},\Omega_{\mathrm{C}}\}\in\Delta^A_2$. This will follow from the classification of $\{\Xi^N_{\mathrm{sp},\epsilon},\Omega_f\}$ since we can use a dispersion bounding function $f(n)=n+1$. Note that we do not necessarily know the dispersion bound (in the form of the null sequence $\{c_n\}$) and hence (see Remark \ref{c_ns}) this provides a $\Delta_2^A$ tower (however not the $\Sigma^A_1$ classification).
\end{proof}

\begin{remark}
To deal with $\Delta_1$-information in the above construction (Step IV), we can replace $\sigma_1(P_n(A-zI)P_n)$ by a rational approximation accurate to $1/n^2$ (see Proposition \ref{REC_SING}) and the proof follows through with minor changes.
\end{remark}

\subsection{Pseudospectrum}
\label{pf_pseduo}

Since $\Omega_{\mathrm{SA}}\subset\Omega_{\mathrm{N}}\subset\Omega_{g}\subset\Omega_{\mathrm{B}}$, $\Omega_{fg}\subset\Omega_f$ and we have already dealt with compact operators, we only need to show that $\{\Xi^N_{\mathrm{sp},\epsilon},\Omega_{\mathrm{B}}\}\in\Sigma_2^A$, $\{\Xi^N_{\mathrm{sp},\epsilon},\Omega_f\}\in\Sigma_1^A$, $\{\Xi^N_{\mathrm{sp},\epsilon},\Omega_{\mathrm{SA}}\}\notin\Delta_2^G$ and $\{\Xi^N_{\mathrm{sp},\epsilon},\Omega_{fg}\}\notin\Delta_1^G$.

\begin{proof}[Proof of Theorem \ref{spec_thm_main} for the pseudospectrum]
{\bf Step I:} $\{\Xi^N_{\mathrm{sp},\epsilon},\Omega_{\mathrm{B}}\}\in\Sigma_2^A$.
Let $A\in\Omega_{\mathrm{B}}$, and $\epsilon>0$. We introduce the following continuous functions $\gamma^N:\mathbb{C} \rightarrow \mathbb{R}_+$, 
$\gamma^N_m:  \mathbb{C} \rightarrow \mathbb{R}_+$ and $\gamma^N_{m,n}:  \mathbb{C} \rightarrow \mathbb{R}_+$,
\begin{equation*}
\begin{split}
\gamma^N(z)&:=\left(\min\left\{\sigma_1\left((A-zI)^{2^N}\right),\sigma_1\left((A^*-\bar{z}I)^{2^N}\right)\right\}\right)^{2^{-N}}=\left\|(A-zI)^{-2^N}\right\|^{-2^{-N}}\\
\gamma_m^N(z)&:=\left(\min\left\{\sigma_1\left((A-zI)^{2^N}P_m\right),\sigma_1\left((A^*-\bar{z}I)^{2^N}P_m\right)\right\}\right)^{2^{-N}}\\
\gamma_{m,n}^N(z)&:=\left(\min\left\{\sigma_1\left((P_n(A-zI)P_n)^{2^N}P_m\right),\sigma_1\left((P_n(A^*-\bar{z}I)P_n)^{2^N}P_m\right)\right\}\right)^{2^{-N}},
\end{split}
\end{equation*}
where $\sigma_1(B)$ denotes the injection modulus of $B$, and in the terms such as
$\sigma_1(P_nBP_m)$ the operator $P_nBP_m$ is regarded as element of $\mathcal{B}(\mathrm{Ran}(P_m),\mathrm{Ran}(P_n))$.
For the proof that $\gamma^N(z)=\|(A-zI)^{-2^N}\|^{-2^{-N}}$ see \cite{Hansen_JAMS}.
We define initial approximations $\hat\Gamma_{m,n}(A)$ for $\mathrm{sp}_{N,\epsilon}(A)$ by
$
\hat\Gamma_{m,n}(A):=\{z\in G_n:\gamma_{m,n}^N(z)\leq\epsilon\},
$
where $G_j:=(j^{-1}(\Zb+\ii\Zb))\cap B_j(0)$. Writing $\gamma^N_{m,n}(z)\leq\epsilon$ as $(\gamma^N_{m,n}(z))^{2^N}\leq\epsilon^{2^N}$ and due to Proposition \ref{PCholesky} it is clear that the computation of $\hat\Gamma_{m,n}(A)$ requires only finitely many arithmetic operations on finitely many evaluations $\{\left\langle Ae_j,e_i\right\rangle: i,j=1,\ldots,n\}$ of $A$. The problem with this tower is that it might produce the empty set.
To get round this and construct our $\Sigma_2^A$ arithmetical tower, there are several facts we will state that can be found in \cite{Hansen_JAMS}. First, $\gamma_{m,n}^N$ converges uniformly to $\gamma_{m}^N$ on compact subsets of $\mathbb{C}$ as $n\rightarrow\infty$. Second, $\gamma_{m}^N$ is non-increasing in $m$ and converges uniformly to $\gamma^N$ on compact subsets of $\mathbb{C}$ as $m\rightarrow\infty$. Finally, we have
\begin{equation}
\label{useful_closure}
\mathrm{cl}\{z\in\mathbb{C}:\gamma_m^N(z)<\epsilon\}=\{z\in\mathbb{C}:\gamma_{m}^N(A)\leq \epsilon\}
\end{equation}
for all $\epsilon>0$. Now it is straightforward to show via a Neumann series argument (see the proof that $\{\Xi_{\mathrm{sp}},\Omega_g\}\in\Sigma_2^A$ below) that there exists a compact ball $K$ such that if $z\notin K$ then $\gamma_{m,n}^N(z)>2\epsilon$ for all $m,n$. In particular, by considering the minimum of $\gamma_m^N(\cdot)$, this together with the above closure property, shows that the minimum is zero and $\{z\in\mathbb{C}:\gamma_{m}^N(A)\leq \epsilon\}\neq \emptyset$. 

Now let $z_0\in\{z\in\mathbb{C}:\gamma_m^N(z)<\epsilon\}$. On the compact set $K$, and for any $m$, the functions $\gamma_{m,n}^N$ and $\gamma_m^N$ are Lipschitz continuous with a uniform Lipschitz constant. Using this and (\ref{useful_closure}), it follows that for large enough $n$, there exists $z_n\in\hat\Gamma_{m,n}(A)$ with $z_n\rightarrow z_0$. Furthermore, if $z_n\in\hat\Gamma_{m,n}(A)$ and we select a subsequence such that $z_{n_j}\rightarrow z$ as $n_j\rightarrow\infty$, we see that $\gamma_m^N(z)\leq\epsilon$. This observations together imply that
$$
\lim_{n\rightarrow\infty}\hat\Gamma_{m,n}(A)=\{z\in\mathbb{C}:\gamma_{m}^N(A)\leq\epsilon\}\subset\mathrm{sp}_{N,\epsilon}(A).
$$
Since $\gamma_m^N$ converges to $\gamma^N$ uniformly on compact sets and are uniformly Lipschitz, it is easy to show that $\lim_{m\rightarrow\infty}\{z\in\mathbb{C}:\gamma_{m}^N(A)\leq\epsilon\}=\mathrm{sp}_{N,\epsilon}(A)$. Hence in order to construct our $\Sigma_2^A$ arithmetical tower we define
$
\Gamma_{m,n}(A)=\hat\Gamma_{m,j(m,n)}(A),
$
where $j(m,n)\geq n$ is minimal such that $\hat\Gamma_{m,j(m,n)}(A)\neq\emptyset$. Such a $j(m,n)$ is guaranteed to exist and can be found by successively computing finitely many of the $\hat\Gamma_{m,k}(A)$'s.

{\bf Step II:} $\{\Xi^N_{\mathrm{sp},\epsilon},\Omega_f\}\in\Sigma_1^A$.
Let $A$ be such that $f$ is a bound for its dispersion, and $\epsilon>0$. Recall that $f(n)\geq n+1$ for every $n$. Define the composition $F^N:=f\circ\cdots\circ f$ of $2^N$ copies of $f$. Besides the already defined functions $\gamma^N$, $\gamma_m^N$ and $\gamma_{m,n}^N$ we additionally introduce $\psi_m^N:=\gamma_{m,F^N(m)}^N$, i.e.
\begin{equation*}
\begin{split}
\psi_m^N(z)&:=\left(\min\left\{\sigma_1\left((P_{F^N(m)}(A-zI)P_{F^N(m)})^{2^N}P_m\right),\sigma_1\left((P_{F^N(m)}(A^*-\bar{z}I)P_{F^N(m)})^{2^N}P_m\right)\right\}\right)^{2^{-N}},
\end{split}
\end{equation*}
and we define the desired approximations $\hat\Gamma_{m}(A)$ for $\mathrm{sp}_{N,\epsilon}(A)$ by
$
\hat\Gamma_{m}(A):=\{z\in G_m:\psi^N_m(z)\leq\epsilon\}.
$
Writing $\psi^N_m(z)\leq\epsilon$ as $(\psi^N_m(z))^{2^N}\leq\epsilon^{2^N}$ and using Proposition \ref{PCholesky}, we see that again the computation of $\hat\Gamma_{m}(A)$ requires only finitely many arithmetic operations on finitely many evaluations $\{\left\langle Ae_j,e_i\right\rangle: i,j=1,\ldots,F^N(m)\}$ of $A$. 

Again, there exists a compact ball $K\subset\Cb$ such that $\gamma^N_m(z)>2\epsilon$ and $\psi^N_m(z)>2\epsilon$ for all $z\in\Cb\setminus K$ and all $m$. Further note that $\psi_m^N$ converges to $\gamma^N_m$ uniformly on $K$. Indeed, since all $z\mapsto(P_{F^N(m)}(A-zI)P_{F^N(m)})^{2^N}P_m$ and $z\mapsto(A-zI)^{2^N}P_m$ are operator-valued polynomials of the same degree whose coefficients converge in the norm 
due to the choice of the function $F^N$, we can take into account that $|\sigma_1(B+C)-\sigma_1(B)|\leq\|C\|$ holds for arbitrary bounded operators $B,C$, and we arrive at the conclusion that $|\gamma^N_m(z)-\psi^N_m(z)|\to 0$ as $m\to\infty$ uniformly with respect to $z\in K$. To construct a $\Sigma_1^A$ tower we bound this difference using the sequence $\{c_n\}$ and the constant $\left\|A\right\|$ (for the case $N>0$ as follows).

If $N=0$ then clearly we have
$
\|P_{f(m)}(A-zI)P_m-(A-zI)P_m\|\leq c_m
$
by definition of the $\{c_n\}$. Suppose that we have a bound
\begin{equation}
\label{horrible1}
\|(P_{F^N(m)}(A-zI)P_{F^N(m)})^{2^N}P_m-(A-zI)^{2^N}P_m\|\leq \alpha(N,m,z),
\end{equation}
for some function $\alpha(N,m,z)$.
We can write
\begin{align*}
(P_{F^{N+1}(m)}&(A-zI)P_{F^{N+1}(m)})^{2^{N+1}}P_m-(A-zI)^{2^{N+1}}P_m\\
&=\big((P_{F^{N+1}(m)}(A-zI)P_{F^{N+1}(m)})^{2^{N}}-(A-zI)^{2^{N}}\big)(P_{F^{N+1}(m)}(A-zI)P_{F^{N+1}(m)})^{2^{N}}P_m\\
&-(A-zI)^{2^{N}}\big((A-zI)^{2^{N}}-(P_{F^{N+1}(m)}(A-zI)P_{F^{N+1}(m)})^{2^{N}}\big)P_m.
\end{align*}
Using the fact that $F^{N+1}(m)=F^{N}(F^{N}(m))$ and $P_{F^{N}(m)}P_{F^{N+1}(m)}=P_{F^{N+1}(m)}$, we can bound the first of the above terms in norm by
$
\alpha(N,F^{N}(m),z)(\left\|A\right\|+\left|z\right|)^{2^N}.
$
Arguing similarly, we can bound the second term in norm by the same quantity. It follows that we can choose
\begin{align*}
\alpha(N,m,z)&=2\alpha(N-1,F^{N-1}(m),z)(\left\|A\right\|+\left|z\right|)^{2^{N-1}}
\end{align*}
and iterating this $N$ times we can take
\begin{align*}
\alpha(N,m,z)
&=2^N c_n \, (\left\|A\right\|+\left|z\right|)^{2^N-1}, \quad n = F^{\frac{N(N-1)}{2}}(m),
\end{align*}
such that (\ref{horrible1}) holds. Note that this estimate can be computed with finitely many arithmetic operations and comparisons from the given data.

In order to simplify the notation we choose a sequence $(\delta_m)$ which converges monotonically to zero such that 
\[\gamma^N_m(z)+\delta_m\geq \psi^N_m(z)\geq \gamma^N_m(z)-\delta_m\text{ for every $m$ and every $z\in K$.}\]
Moreover, we point out that each of the functions $z\mapsto\psi_m^N(z)$ is continuous on the compact set $K$, hence even uniformly continuous, and we can assume without loss of generality that, for every $m$,
\begin{equation}\label{EUnifC}
|\psi_m^N(z)-\psi_m^N(y)|<\delta_m\text{ for arbitrary $z,y\in K$, $|z-y|<1/m$}.
\end{equation}

Now let $\zeta_{\epsilon}(A):=\{z\in\Cb: \gamma^N(z)\leq\epsilon\}$,
$
\zeta_{\epsilon,m}(A):=\{z\in\Cb:\gamma^N_m(z)\leq\epsilon\},
$
and
$
\Psi_{\epsilon,m}(A):=\{z\in\Cb:\psi^N_m(z)\leq\epsilon\}.
$
 By the discussion above, we conclude for all 
$m\geq k$ that
\begin{equation}\label{nested}
\zeta_{\epsilon+\delta_k,m}(A)\supset \zeta_{\epsilon+\delta_m,m}(A)\supset
 \Psi_{\epsilon,m}(A)\supset\zeta_{\epsilon-\delta_m,m}(A)\supset\zeta_{\epsilon-\delta_k,m}(A).
 \end{equation}
Since, $P_m \leq P_{m+1}$ and $P_m \rightarrow I$ strongly, $\gamma^N_m \rightarrow \gamma^N$ monotonically from above pointwise (and hence locally uniformly by Dini's Theorem). Thus, by \cite{Hansen_JAMS}, $\zeta_{\epsilon+\delta_k,m}(A) \rightarrow \zeta_{\epsilon+\delta_k}(A) = \mathrm{sp}_{N,\epsilon + \delta_k}(A)$ and $\zeta_{\epsilon-\delta_k,m}(A) \rightarrow \zeta_{\epsilon-\delta_k}(A) = \mathrm{sp}_{N,\epsilon-\delta_k}(A)$ as $m \rightarrow \infty$. Hence, since $\spc_{N,\epsilon\pm\delta_k}(A)\to \mathrm{sp}_{N,\epsilon}(A)$ 
as $k\to\infty$,  (\ref{nested}) yields
$
\lim_{m\to\infty} \Psi_{\epsilon,m}(A) = \mathrm{sp}_{N,\epsilon}(A).
$
To finish the convergence proof we observe that it is clear that on the one hand $\Psi_{\epsilon,m}(A)\supset\hat\Gamma_{m}(A)$.
On the other hand, for sufficiently large $m$ it holds true that for every point 
$x\in\Psi_{\epsilon-\delta_m,m}(A)$ there is a point $y_x\in G_m$ with $\left|x-y_x\right|<1/m$ 
and, by \eqref{EUnifC} 
we get $|\psi_m^N(y_x)-\psi_m^N(x)|< \delta_m$ that is
$y_x$ even belongs to $\hat\Gamma_{m}(A)$. Thus,
$
\hat\Gamma_{m}(A) + B_{1/m}(0)\supset \Psi_{\epsilon-\delta_m,m}(A)
$ for sufficiently
large $m$. Combining this, we arrive at
\[\Psi_{\epsilon,m}(A) + B_{1/k}(0)\supset \hat\Gamma_{m}(A) + B_{1/m}(0)\supset
\Psi_{\epsilon-\delta_m,m}(A)\supset\Psi_{\epsilon-\delta_k,m}(A),\]
for $m\geq k$ large. By the above, the sets on the left converge to $\mathrm{sp}_{N,\epsilon}(A)+ B_{1/k}(0)$ as 
$m\to\infty$, and the sets on the right converge to $\mathrm{sp}_{N,\epsilon-\delta_k}(A)$ for every
$k$. Since both of these sets converge to $\mathrm{sp}_{N,\epsilon}(A)$ as $k\to\infty$ this 
provides $\lim_{m\to\infty} \hat\Gamma_{m}(A) =\mathrm{sp}_{N,\epsilon}(A).$ This shows that (upon altering as in Step I to avoid the empty set), we can gain convergence in one limit without the knowledge of $\{c_n\}$ and $\left\|A\right\|$.

Now we have that
$
|(\psi^N_m(z))^{2^N}-(\gamma_m^N(z))^{2^N}|\leq \alpha(N,m,z).
$
Hence we define
\[
\tilde\Gamma_{m}(A):=\{z\in G_m:(\psi^N_m(z))^{2^N}\leq\epsilon^{2^N}-\alpha(N,m,z), \epsilon^{2^N}-\alpha(N,m,z)>0\},
\]
which can be computed in finitely many arithmetic operations and comparisons.
Of course this may be empty but it has the property that
$
\tilde\Gamma_{m}(A)\subset \mathrm{sp}_{N,\epsilon}(A).
$
Suppose for a contradiction that we do not have convergence to $\mathrm{sp}_{N,\epsilon}(A)$. Without loss of generality, by taking a subsequence if necessary, there exists $z_m\in\mathrm{sp}_{N,\epsilon}(A)$, $z\in \mathrm{sp}_{N,\epsilon}(A)$ and $\delta>0$ such that $\gamma^N(z)<\epsilon$, $z_m\rightarrow z$ but $\mathrm{dist}(z_m,\tilde\Gamma_{m}(A))\geq\delta$. Let $\hat z_m\in G_m$ with $\hat z_m\rightarrow z$. Then for large $m$ we must have $\gamma^N(\hat z_m)<\epsilon$. But $\alpha(N,m,\hat z_m)\rightarrow 0$ and hence $\hat z_m\in\tilde\Gamma_{m}(A)$ for large $m$, the required contradiction. To finish we simply define
$
\Gamma_{m}(A)=\hat\Gamma_{j(m)}(A),
$
where $j(m)\geq m$ is minimal such that $\hat\Gamma_{j(m)}(A)\neq\emptyset$. Such a $j(m)$ must exist and we hence avoid the empty set. Finally, the fact that $\tilde\Gamma_{m}(A)\subset \mathrm{sp}_{N,\epsilon}(A)$ ensures we have $\Sigma_1^A$ convergence.

{\bf Step III:} $\{\Xi^N_{\mathrm{sp},\epsilon},\Omega_{\mathrm{SA}}\}\notin\Delta_2^G$. Assume for a contradiction that there is a sequence $\{\Gamma_k\}$ of general algorithms such that $\Gamma_k(A) \rightarrow \mathrm{sp}_{N,\epsilon}(A)$ for all $A\in\Omega_{\mathrm{SA}}$, and consider operators of the type
\begin{equation}\label{Eq_Blocks}
A:= \bigoplus_{r=1}^{\infty} A_{l_r} \qquad\text{with $\{l_r\}\subset\Nb$ and } 
A_n:=\begin{pmatrix}
1& & & &1\\
 &0& & & \\
 & &\ddots& & \\
 & & &0& \\
1& & & &1\\
\end{pmatrix}
\in\Cb^{n\times n}.
\end{equation}
Then $\spc(A_n)=\{0,2\}$, hence $A$ is bounded, self-adjoint, and $\spc(A)=\{0,2\}$ as well. For sufficiently small $\epsilon$ the $(N,\epsilon)$-pseudospectrum is a certain neighbourhood of $\{0,2\}$ disjoint from $B_{\frac{1}{2}}(1)$, independently of the choice of $\{l_r\}$. In order to find a counterexample we simply construct an appropriate sequence $\{l_r\}\subset\Nb$ by induction: For $C:=\diag\{1,0,0,0,\ldots\}$ one obviously has $\spc(C)=\{0,1\}$. Choose $k_0:=1$ and $l_1>N(C,k_0)$, where $N(C,n) = \max\{i,j \, \vert \, f_{i,j} \in \Lambda_{\Gamma_n}(C)\}$ for $n \in \mathbb{N}$.
Now, suppose that $l_1,\ldots,l_n$ are already chosen. Then we obviously have that 
$
\spc\left(A_{l_1} \oplus \cdots \oplus  A_{l_n} \oplus C\right) = \{0,1,2\},
$ 
hence there exists a $k_n$ such that 
$
\Gamma_k\left(A_{l_1} \oplus \cdots \oplus  A_{l_n} \oplus C\right)\cap B_{\frac{1}{n}}(1) \neq\emptyset
$
for every $k\geq k_n$, where $B_{\frac{1}{n}}(1)$ denotes the closed ball of radius $1/n$ and centre $1$. Now, choose 
\begin{equation}\label{the_l}
l_{n+1}>N(A_{l_1} \oplus \cdots \oplus  A_{l_n} \oplus C,k_n)-l_1-l_2-\ldots-l_n.
\end{equation}
By construction, it follows that 
\begin{equation}\label{equality_G}
\Gamma_{k_n}(\oplus_{r=1}^{\infty} A_{l_r})\cap B_{\frac{1}{n}}(1)
=\Gamma_{k_n}(A_{l_1} \oplus \ldots \oplus A_{l_n}\oplus C)\cap B_{\frac{1}{n}}(1)\neq\emptyset \quad\forall\; n\in\Nb.
\end{equation}
Indeed, since any evaluation function $f_{i,j} \in \Lambda$ just provides the $(i,j)$-th matrix element, it follows by (\ref{the_l}) that for any evaluation functions $f_{i,j} \in \Lambda_{\Gamma_{k_n}}(A_{l_1} \oplus \cdots \oplus  A_{l_n} \oplus C)$ we have that 
$
f_{i,j}(A_{l_1} \oplus \cdots \oplus  A_{l_n} \oplus C) = f_{i,j}(\oplus_{r=1}^{\infty} A_{l_r}).
$ Thus, 
by assumption (iii) in the definition of a General algorithm (Definition \ref{alg}), we get that $\Lambda_{\Gamma_{k_n}}(A_{l_1} \oplus \cdots \oplus  A_{l_n} \oplus C) = \Lambda_{\Gamma_{k_n}}(\oplus_{r=1}^{\infty} A_{l_r})$ which, by assumption (ii) in Definition \ref{alg}, yields (\ref{equality_G}).
So, from (\ref{equality_G}), we see that the point $1$ is contained in the partial limiting set of the sequence 
$\{\Gamma_{k}(\oplus_{r=1}^{\infty} A_{l_r})\}_{k=1}^\infty$ which approximates $\mathrm{sp}_{N,\epsilon}(A)$, a contradiction. For general $N$ and $\epsilon$, we apply the above argument after appropriate re-scaling.

{\bf Step IV:} $\{\Xi^N_{\mathrm{sp},\epsilon},\Omega_{fg}\}\notin\Delta_1^G$. This is clear by considering diagonal operators. The point is that given any general $\Delta_1^G$ tower, $\Gamma_n$, and any $n$, $\Gamma_n(A)$ uses only finitely many matrix evaluations $\{f_{i,j}(A):i,j\leq N_0(n,A)\}$. We can choose $m$ large such that $m>N_0(1,0)$ and set $f_{m,m}(A)=2\epsilon+2$. Then $\Gamma_1(A)=\Gamma_1(0)\subset B_{1/2+\epsilon}(0)$, a contradiction since $2\epsilon+2\in\mathrm{sp}_{N,\epsilon}(A)$.
\end{proof}

\begin{remark}
To deal with $\Delta_1$-information in Step I of the above proof, we can simply replace $(\gamma^N_{m,n}(z))^{2^N}$ by a suitable rational approximation accurate to $1/n$ (see Proposition \ref{REC_SING}). For Step II, we can replace $(\psi^N_m(z))^{2^N}$ and $\alpha(N,m,z)$ by rational approximations from above accurate to $1/m$. If $\epsilon$ is not rational, then we approximate with a rational from below accurate to $1/n^2$ in Step I and $1/m^2$ in Step II.
\end{remark}

\subsection{Spectrum}
\label{spec_proof}
Again, using the inclusions $\Omega_{\mathrm{SA}}\subset\Omega_{\mathrm{N}}\subset\Omega_{{g}}$,when considering the spectrum we only need to show that $\{\Xi_{\mathrm{sp}},\Omega_{fg}\}\in\Sigma_1^A$, $\{\Xi_{\mathrm{sp}},\Omega_f\}\in\Pi_2^A$, $\{\Xi_{\mathrm{sp}},\Omega_g\}\in\Sigma_2^A$, $\{\Xi_{\mathrm{sp}},\Omega_{\mathrm{B}}\}\in\Pi_3^A$, $\{\Xi_{\mathrm{sp}},\Omega_\mathrm{SA}\}\notin\Delta_2^G$, $\{\Xi_\mathrm{sp},\Omega_{f}\}\notin\Delta_2^G$ and $\{\Xi_\mathrm{sp},\Omega_\mathrm{B}\}\notin\Delta_3^G$ (the fact that $\{\Xi_\mathrm{sp},\Omega_{fg}\}\notin\Delta_1^G$ is clear by considering diagonal operators). We then prove that $\{\Xi_{\mathrm{sp}},\Omega_{f}\cap \Omega_{\mathrm{N}}\}\in\Sigma_1^{A,\mathrm{eigv}}$ separately since the argument easily extends to the Schr\"odinger case in \S\ref{bounded_potential}. The proof that $\{\Xi_\mathrm{sp},\Omega_\mathrm{B}\}\notin\Delta_3^G$ relies on some results from decision making problems which we shall prove in Section \ref{dec_sec}.
\begin{proof}[Proof of Theorem \ref{spec_thm_main} for the spectrum]
{\bf Step I:} We begin with the easy cases that $\{\Xi_{\mathrm{sp}},\Omega_f\}\in\Pi_2^A$ and $\{\Xi_{\mathrm{sp}},\Omega_{\mathrm{B}}\}\in\Pi_3^A$. To prove that $\{\Xi_{\mathrm{sp}},\Omega_f\}\in\Pi_2^A$, let $\epsilon>0$ and let $\Gamma_n^{\epsilon}$ denote the height one arithmetic tower to compute the (classical) pseudospectrum of operators in $\Omega_f$. Using the fact that $\spne(A)$ are continuous with respect to the parameter $\epsilon>0$, 
and converge to $\spc(A)$ as $\epsilon\to 0$ for every $A$, we simply set
$
\Gamma_{m,n}(A)=\Gamma_n^{1/m}(A).
$
This is a $\Pi_2^A$ tower since $\mathrm{sp}_{0,1/m}(A)$ contains $\mathrm{sp}(A)$. $\{\Xi_{\mathrm{sp}},\Omega_\mathrm{B}\}\in\Pi_3^A$ is similar and just requires the additional first limit.

{\bf Step II:} $\{\Xi_{\mathrm{sp}},\Omega_g\}\in\Sigma_2^A$.
Let $g:[0,\infty)\to[0,\infty)$ be as in Definition \ref{EqSpDist}, in particular, continuous, vanishing only at $x=0$ and diverging to $\infty$ as $x\to\infty$. Note that $g(x)\leq x$ for all $x$ and without loss of generality we can also assume that $g$ is strictly increasing. Then the inverse function $h(y):=g^{-1}(y):[0,\infty)\to[0,\infty)$ is well defined, continuous, strictly increasing, $h(y)\geq y$ for every $y$, and $\lim_{y\to 0} h(y)=0$.

Let $K\subset \Cb$ be a compact set and $\delta>0$.  We introduce a $\delta$-grid for $K$ by 
$G^\delta(K):=(K+B_\delta(0))\cap(\delta(\Zb+\ii\Zb)),$ 
where $B_\delta(0)$ denotes the closed ball of radius $\delta$ centred at $0$. Without loss of generality we may assume that $\delta^{-1}$ is an integer, and obviously, $G^\delta(K)$ is finite. Moreover, introduce $h_\delta(y):=\min\{k\delta: k\in\Nb, g(k\delta)> y\}$ and observe that for each $y$, evaluating $h_\delta(y)$ requires only finitely many evaluations of $g$. Also, notice that $h(y)\leq h_\delta(y)\leq h(y)+\delta$.
For a given function $\zeta:\Cb\to[0,\infty)$ we define sets $\Upsilon_K^\delta(\zeta)$ as follows:
For each $z\in G^\delta(K)$ let $I_z:=B_{h_\delta(\zeta(z))}(z)\cap(\delta(\Zb+\ii\Zb))$. Further
\begin{itemize}
\item If $\zeta(z)\leq 1$ then introduce the set $M_z$ of all $w\in I_z$ for which $\zeta(w)\leq \zeta(v)$ holds for all $v\in I_z$. 
\item Otherwise, if $\zeta(z)>1$, just set $M_z:=\emptyset$. 
\end{itemize}
Now define 
\begin{equation}\label{Upsilon}
\Upsilon_K^\delta(\zeta):=\bigcup_{z\in G^\delta(K)}M_z.
\end{equation}
Notice that for the computation of $\Upsilon_K^\delta(\zeta)$ only finitely many evaluations of $\zeta$ and $g$ are required.

{\bf Claim:} Let $K$ be a compact set containing the spectrum of $A$ and $0<\delta<\epsilon<1/2$.
Further assume that $\zeta$ is a function with $\|\zeta-\gamma\|_{\infty,\hat{K}}:=\|(\zeta-\gamma)\chi_{\hat{K}}\|_\infty < \epsilon$ on 
$\hat{K}:=(K+B_{h(\diam(K)+2\epsilon)+\epsilon}(0))$, where $\chi_{\hat{K}}$ denotes the characteristic function of $\hat K$. Finally, let 
\begin{equation}\label{u_function}
u(\xi):=\max\{h(3\xi+h(t+\xi)-h(t))+\xi:t\in[0,1]\}.
\end{equation} 
Then we have that 
$
d_\mathrm{H}(\Upsilon_K^\delta(\zeta),\spc(A))\leq u(\epsilon)
$
and 
$
\lim_{\xi\to 0}u(\xi)=0.
$
\begin{proof}[Proof of claim:] To prove the claim, let $z\in G^\delta(K)$ and notice that $I_z\subset\hat{K}$ since, for every $v\in I_z$,
\begin{equation}\label{I_z}
\begin{split}
|z-v| &\leq h_\delta(\zeta(z))\leq h_\delta(\gamma(z)+\epsilon) \leq h(\dist(z,\spc(A))+\epsilon)+\delta \\
&\leq h(\diam(K)+\delta+\epsilon)+\delta.
\end{split}
\end{equation}
Suppose that $M_z \neq \emptyset$. Note that by \eqref{EqSpDist}, the monotonicity of $h$, and the compactness of $\spc(A)$ there is a $y\in\spc(A)$ of minimal distance to $z$ with $|z-y|\leq h(\gamma(z))$. Since $\|\zeta-\gamma\|_{\infty,\hat{K}} < \epsilon$ we get $|z-y|\leq h(\zeta(z)+\epsilon)$. Hence, at least one of the $v\in I_z$, let's say $v_0$, satisfies $|v_0-y|<h(\zeta(z)+\epsilon)-h(\zeta(z))+\delta$. Noting again that $\gamma(v_0) \leq  \dist(v_0,\spc(A))$, we get $\zeta(v_0)<\gamma(v_0)+\epsilon < h(\zeta(z)+\epsilon)-h(\zeta(z))+2\epsilon$. By the definition of $M_z$ this estimate now holds for all points $w \in M_z$ and we conclude that, for all $w \in M_z$, 
\begin{equation}\label{w_estimate}
\begin{split}
\dist(w,\spc(A)) &= h(g(\dist(w,\spc(A)))) \leq h(\gamma(w)) \\
&\leq h(\zeta(w)+\epsilon) \leq h(h(\zeta(z)+\epsilon)-h(\zeta(z))+3\epsilon).
\end{split}
\end{equation}
This observation holds for every $z\in G^\delta(K)$ and all $w \in M_z$, hence all points in $\Upsilon_K^\delta(\zeta)$ are closer to $\spc(A)$ than $u(\epsilon)$.

Conversely, take any $y\in\spc(A)\subset K$. Then there is a point $z\in G^\delta(K)$ with $|z-y|<\delta<\epsilon$, hence 
$
\zeta(z)<\gamma(z)+\epsilon \leq \dist(z,\spc(A))+\epsilon < 2\epsilon<1.
$
 Thus, $M_z$ is not empty and contains a point which is closer to $y$ than $h(\zeta(z))+\epsilon \leq h(2\epsilon)+\epsilon\leq u(\epsilon)$.
Finally notice that the mapping 
$$
(t,\xi)\mapsto h(h(t+\xi)-h(t)+3\xi)+\xi
$$ is continuous on the compact set $[0,1]\times [0,1]$, hence uniformly continuous. Moreover, for every fixed $t$ it tends to $0$ as $\xi\to 0$, thus we can conclude $u(\xi)\to 0$, and we have proved the claim.\end{proof}

Define the function
$\gamma_{m,n}(z,A):=\min\{\sigma_1(P_{n}(A-zI)P_{m}),\sigma_1(P_{n}(A^*-\bar{z}I)P_{m})\},$
and note that we can compute an approximation to $\gamma_{m,n}(z,A)$  from \textit{above} to within an accuracy of $1/m$ in finitely many arithmetic operations and comparisons using Proposition {\ref{REC_SING}} (this also includes the case of $\Delta_1$-information). Call this approximation function $\zeta_{{m},{n}}(z,A)$ and we can assume that it takes values in $\frac{1}{2{m}}\mathbb{N}$. As ${n}\rightarrow\infty$, $\gamma_{{m},{n}}(\cdot,A)$ converges to
$
\gamma_{{m}}(z,A):=\min\{\sigma_1((A-zI)P_{m}),\sigma_1((A^*-\bar{z}I)P_{m})\}
$
monotonically from below. By taking successive maxima over $n$ and then minima over $m$ if necessary:
$
\min_{1\leq j\leq m}\max_{1\leq k\leq n}\zeta_{j,k}(z,A),
$
we can assume that $\zeta_{{m},{n}}(\cdot,A)$ is non-decreasing in $n$ and non-increasing in $m$. Since $\gamma_{m,n}$ obeys these monotonicity relations, this preserves the error bound of $1/m$. It follows that $\zeta_{{m},{n}}(\cdot,A)$ converges to $\zeta_{m}(\cdot,A)$ which takes values in the set $\frac{1}{2{m}}\mathbb{N}$ (i.e. $\zeta_{{m},{n}}(z,A)$ is eventually constant for a given $z$) and such that
$\gamma_{{m}}(z,A)\leq \zeta_{{m}}(z,A)\leq \gamma_{{m}}(z,A)+1/m.$

Now let
\[
\hat{\Gamma}_{m,n}(A)=\Upsilon_{B_m(0)}^{1/2^m}(\zeta_{m,n}).
\]
To show that this provides an arithmetic tower of algorithms, note that the computation of $\Upsilon_{B_m(0)}^{1/2^m}(\zeta_{m,n})$ requires only finitely many evaluations of $\zeta_{m,n}$, and the finite number of constants $g(k/m)\leq 1$, $k=1,2,\ldots$. Since $G^{1/2^m}(B_m(0))$ is finite and we restricted values of $\zeta_{m,n}$ to $\frac{1}{2{m}}\mathbb{N}$, we must have that for large $n$, $\hat{\Gamma}_{m,n}(A)$ is constant and equal to $\Upsilon_{B_m(0)}^{1/2^m}(\zeta_{m})$. Denote this eventually constant set by $\hat{\Gamma}_{m}(A)$. We must now adapt $\hat{\Gamma}_{m,n}$ such that the output is non-empty and such that we gain the desired convergence in the Hausdorff topology yielding the $\Sigma_2^A$ classification.
For any $\hat{\Gamma}_{m,n}(A)$ let
$
S(m,n,A):=\max_{z \in \hat{\Gamma}_{m,n}(A)} \zeta_{{m},{n}}(z,A),
$
where we take the maximum over the empty set to be $+\infty$. Note that $\hat{\Gamma}_{m,n}(A)$ is empty if and only if $\zeta_{m,n}(z,A)>1$ for all $z\in G^{1/2^m}(B_m(0))$ and note also that $S(m,n,A)$ can be computed using finitely many arithmetic operations and comparisons from the given data.

For given $m,n$, if $n<m$ then set $\Gamma_{m,n}(A)=\{0\}$. Otherwise, compute $S(k,n,A)$ for $m\leq k\leq n$. If there exists such a $k$ with $S(k,n,A)\leq g(2^{-m})$, then choose a minimal such $k$ and set $\Gamma_{m,n}(A)=\hat{\Gamma}_{k,n}(A)$ (which must be non-empty by the definition of $S(m,n,A)$), otherwise set $\Gamma_{m,n}(A)=\{0\}$. It follows that this defines an arithmetic algorithm mapping into the appropriate metric space (in particular it outputs a non-empty compact set).
Since $\zeta_{m,n}$ increases to $\zeta_{m}$ and $g$ is continuous, if $\hat{\Gamma}_{l}(A)\neq\emptyset$ then $S(l,n,A)$ is finite for all $n\in\mathbb{N}$. For such an $l$, we must have $S(l,n,A)$ non-decreasing in $n$, convergent to
$
S_l(A):=\max_{z \in \hat{\Gamma}_l(A)} g(\zeta_l(z,A)).
$
On the other hand if $\hat{\Gamma}_{l}(A)=\emptyset$ then $\zeta_l(z,A)>1$ for all $z\in G^{1/2^l}(B_l(0))$ and the fact that $\zeta_{m,n}$ increases to $\zeta_{m}$ shows that $S(l,n,A)=+\infty$ for large $n$. 

Define the function 
$
\gamma(z) := \min\{\sigma_1(A-zI),\sigma_1(A^*-\overline{z}I)\}=\|(A-zI)^{-1}\|^{-1}.
$
To see why $\min\{\sigma_1(A-zI),\sigma_1(A^*-\overline{z}I)\}=\|(A-zI)^{-1}\|^{-1}$ see for example \cite{Hansen_JAMS}.
Now
\[
\sigma_1(P_n(A-zI)P_m) = \inf\{\|P_n(A-zI)P_m\xi\|: \xi \in \mathrm{Ran}(P_m), \|\xi\| = 1\}
\]
and
$
\sigma_1((A-zI)P_m) = \inf\{\|(A-zI)P_m\xi\|: \xi \in \mathrm{Ran}(P_m), \|\xi\| = 1\}.
$
Thus, since $P_m \rightarrow I$ strongly and $P_{m+1} \geq P_m$, then $\gamma_m \rightarrow \gamma$ pointwise and monotonically from above, and by Dini's Theorem the convergence is 
uniform on every compact set, in particular on the ball $K:=B_{m_0}(0)$, with a fixed $m_0> 2\|A\|+4$.
Also, $\gamma_{m,n} \rightarrow \gamma_m$ pointwise monotonically from below as $n \rightarrow \infty$, hence again by Dini's Theorem it follows that the convergence is uniform on the ball $K=B_{m_0}(0)$. Outside this ball we have, for $n>m$, by a Neumann argument
\begin{align*}
\gamma_{m,n}(z)&=\min\{\sigma_1(P_n(A-zI)P_nP_m),\sigma_1(P_n(A^*-\overline{z}I)P_nP_m)\}\\
	&\geq\min\{\sigma_1(P_n(A-zI)P_n),\sigma_1(P_n(A^*-\overline{z}I)P_n)\}\\
	&=\|(P_n(A-zI)P_n)^{-1}\|^{-1} = |z|\|(P_n-z^{-1}P_nAP_n)^{-1}\|^{-1} \geq 2.
\end{align*}

For all $n>m> m_0$, the points in the finite set $G^{1/2^m}(B_m(0))\setminus K$ lead to function values of $\zeta_{m,n}$ being larger than $1$ (since $\zeta_{m,n}$ approximates $\gamma_{m,n}$ to within $1/m$), hence $\hat\Gamma_{m,n}(A)=\Upsilon_{K}^{1/2^m}(\zeta_{m,n})$. Fix $\epsilon\in (0,1/2)$. Then there is an $m_1>m_0$ with $m_1> 3/\epsilon$ such that $\|\gamma-\zeta_m\|_{\infty,\hat{K}}<\epsilon/3$ on $\hat{K}:=B_{h(\diam(K)+2\epsilon)+\epsilon}(0)$ for all $m>m_1$. Moreover, for every $m$ there is an $n_1(m)$ such that $\|\gamma_m-\gamma_{m,n}\|_{\infty,\hat{K}}<\epsilon/3$ for all $n>n_1(m)$. This yields
\begin{equation}\label{EGammaZeta}
\begin{split}
\|\gamma-\zeta_{m,n}\|_{\infty,\hat{K}}
&\leq\|\gamma-\gamma_m\|_{\infty,\hat{K}}+\|\gamma_m-\gamma_{m,n}\|_{\infty,\hat{K}}
		+\|\gamma_{m,n}-\zeta_{m,n}\|_{\infty,\hat{K}}\\
&\leq \epsilon/3+\epsilon/3+1/m< \epsilon
\end{split}
\end{equation}
whenever $m>m_1$ and $n>n_1(m)$. Hence, by the above claim, we must have that
$
d_\mathrm{H}(\hat\Gamma_{m,n}(A),\spc(A))\leq u(\epsilon)
$
whenever $m>m_1$ and $n>n_1(m)$.
Since this bound tends to zero as $\epsilon\to 0$, it is proved that
\[
\lim_{m\to\infty}\limsup_{n\to\infty} d_\mathrm{H}(\hat\Gamma_{m,n}(A),\spc(A))=0.
\]
It follows that there exists $N_0\in\mathbb{N}$ minimal such that $S_{N_0}(A)<+\infty$, equivalently such that $\hat{\Gamma}_{N_0}(A)\neq\emptyset$. Monotonicity of $\zeta_m$ in $m$ and the fact that the grid refines itself now ensures that if $m\geq N_0$ then $S_m(A)<+\infty$. Furthermore, the above claim (as well as continuity in $g$) shows that $\lim_{m\rightarrow\infty}S_m(A)=0$. Let $N_1(m)\geq m$ be minimal such that $S_{N_1(m)}\leq g(2^{-m})$. It follows that we must have
$
\lim_{n\rightarrow\infty} \Gamma_{m,n}(A)=\hat \Gamma_{N_1(m)}(A).
$
We must also have $\lim_{m\rightarrow\infty}\Gamma_{m}(A)=\mathrm{sp}(A)$. Furthermore,
\begin{equation}
\max_{z \in \Gamma_{m}(A)} g(\mathrm{dist}(z,\mathrm{sp}(A)))\leq \max_{z \in \Gamma_{m}(A)} \gamma(z,A) \leq S_{N_1(m)}(A)\leq g(2^{-m}).
\end{equation}
But $g$ is strictly increasing so that we must have $\Gamma_{m}(A)\subset\mathrm{sp}(A)+B_{2^{-m}}(0)$ and hence $\Sigma_2^A$ convergence.

{\bf Step III:} $\{\Xi_{\mathrm{sp}},\Omega_{fg}\}\in\Sigma_1^A$. This is very similar to Step II, but now we use the function $f$ to collapse the first limit. We can approximate
$$
F_n(z,A):=\min\{\sigma_1(P_{f(n)}(A-zI)P_{n}),\sigma_1(P_{f(n)}(A^*-\bar{z}I)P_{n})\}+c_n,
$$
from \textit{above} to within an accuracy $1/n$ in finitely many arithmetic operations and comparisons using Proposition {\ref{REC_SING}} (this also includes the case of $\Delta_1$-information). Call this approximation function $\tilde F_n(z,A)$ and assume that $\tilde F_n(z,A)\in\mathbb{Q}$. Note that by definition of $D_{f,n}$ and the fact that $D_{f,n}(A)\leq c_n$, we must have $\tilde F_n(z,A)\geq \gamma_n(z,A)$ and without loss of generality (take successive minima if necessary) we can assume that $\tilde F_n$ converges locally uniformly to $\gamma$ monotonically from above.
Now let
$
{\Gamma}_{n}(A)=\Upsilon_{B_n(0)}^{1/2^n}(\tilde F_n).
$
Arguing as before, we see that this provides an arithmetic tower of algorithms, is non-empty for large $n$ (so we can assume this holds for all $n$ without loss of generality) and has $\lim_{n\rightarrow\infty}{\Gamma}_{n}(A)=\mathrm{sp}(A)$. Hence we only need to argue for the $\Sigma_1^A$ error control. Define
\begin{equation}\label{eq:def_E_n}
E_n(A)=\sup_{z\in {\Gamma}_{n}(A)} h_{2^{-n}}(\tilde F_n(z,A)),
\end{equation}
then since $h_{2^{-n}}\geq h$, we must have 
$
E_n(A)\geq \sup_{z\in {\Gamma}_{n}(A)} \mathrm{dist}(z,\mathrm{sp}(A)).
$
 Moreover, $\sup_{z\in {\Gamma}_{n}(A)}\tilde F_n(z,A)$ converges to $0$ as $n\rightarrow\infty$. Since $h_{2^{-n}}\leq h+ 2^{-n}$, it follows that $E_n(A)\rightarrow 0$ and hence (by the usual argument of taking subsequences if necessary) we have $\{\Xi_{\mathrm{sp}},\Omega_{fg}\}\in\Sigma_1^A$.

{\bf Step IV:} $\{\Xi_{\mathrm{sp}},\Omega_\mathrm{SA}\}\notin\Delta_2^G$. This is almost the same argument as the pseudospectrum case. Assume that there is a sequence $\{\Gamma_k\}$ of general algorithms such that $\Gamma_k(A) \rightarrow \mathrm{sp}(A)$ for all $A\in\Omega_\mathrm{SA}$, and consider operators of the type \eqref{Eq_Blocks}. The spectrum is $\{0,2\}$ disjoint to $B_{\frac{1}{2}}(1)$, independently of the choice of $\{l_r\}$. By exactly the same procedure as before, one obtains again that $1$ belongs to the partial limiting set of $\Gamma_k(A)$ for a certain $A$, hence a contradiction.

{\bf Step V:} $\{\Xi_\mathrm{sp},\Omega_{f}\}\notin\Delta_2^G$. Recall that $\Omega_f$ denotes the set of bounded operators on $l^2(\mathbb{N})$ whose dispersion is bounded by $f$. Thus, to show the claim, it suffices to show that for any height one general tower of algorithms $\{\Gamma_n\}_{n\in\mathbb{N}}$ for $\Xi_{\mathrm{sp}}$, there exists a weighted shift $S$, with $(Su)_1=0$ for all $u\in l^\infty(\N)$ and $Se_n = \alpha_n e_{n+1}$ where $\alpha = \{\alpha_n\}_{n\in\mathbb{N}} \in l^{\infty}(\mathbb{N})$, such that $\Gamma_m(S) \nrightarrow \mathrm{sp}(S)$ when $m \rightarrow \infty$, Obviously $S \in \Omega_{f}$ (recall $f(n)\geq n+1$). To construct such an $S$ we let 
 $$
 \alpha = \{0,0, \hdots, 0,1,0,0,\hdots, 0,1,1,0,0,\hdots,0,1,1,1,0,\hdots \}, \qquad \alpha_{l_j+1}, \alpha_{l_j + 2}, \hdots, \alpha_{l_j+j} = 1,
 $$
 for some sequence $\{l_j\}_{j\in\mathbb{N}}$ where $l_{j+1} > l_{j}+2j$ that we will determine. Observe that regardless of the choice of $\{l_j\}_{j\in\mathbb{N}}$ we have that $\mathrm{sp}(S) = B_{1}(0)$ (the closed disc centred at zero with radius one). Indeed, on the one hand $\|S\|=1$, hence $\mathrm{sp}(S) \subset B_{1}(0)$. On the other hand, one can define the elementary shift operator $V:e_n\mapsto e_{n+1}$, $n\in\Nb$, and its left inverse $V^-:e_{n+1}\mapsto e_n$, $n\in\Nb$, $e_1\mapsto0$. Then the shifted copies $(V^-)^{l_j}S V^{l_j}$ converge  strongly to the limit operator $V$ whose spectrum $\mathrm{sp}(V)=B_{1}(0)$ is necessarily contained in the essential spectrum of $S$ (cf. \cite{Roch_Silbermann_LimitOps} or \cite{Lindner}).

To construct $S$ we will inductively choose $\{l_j\}_{j\in\mathbb{N}}$ with the help of another sequence $\{m_j\}_{j\in\mathbb{Z}_+}$ that will also be chosen inductively. Before we start, define, for any $A\in\Omega_f$ and $m \in \mathbb{N}$,  $N(A,m)$ to be the smallest integer so that $\Lambda_{\Gamma_m}(A)$ only includes matrix entries $A_{ij}=\langle Ae_j,e_i \rangle$ with $i,j\leq N(A,m)$. Now let $m_0 = 1$ and choose $l_1 > N(0,m_1)$. Suppose that $l_1,\ldots,l_n$ and $m_0,\ldots,m_{n-1}$ are already chosen. Note that $\mathrm{sp}(P_{r}S) = \{0\}$, since $P_{r}S = P_{r}SP_{r}$ can be regarded as a $r\times r$-triangular matrix with zero-diagonal. Thus, since by assumption $\{\Gamma_m\}_{m\in\mathbb{N}}$ is a General tower of algorithms for $\Xi_1$, there is an $m_n$ such that $\Gamma_m(P_{l_n +n+1}S) \subset B_{\frac{1}{2}}(0),$ for all  $m \geq m_n.$ Let 
\begin{equation}\label{choosing_l}
l_{n+1} > N(P_{l_n +n+1}S,m_n) \text{ such that also } l_{n+1} > l_{n}+2n.
\end{equation}
Then, it follows that  
$
\Gamma_{m_n}(S) = \Gamma_{m_n}(P_{l_{n+1}}S)=\Gamma_{m_n}(P_{l_n +n+1}S).
$ 
Indeed, since any evaluation function $f_{i,j} \in \Lambda$ just provides the $(i,j)$-th matrix element, it follows by (\ref{choosing_l}) that for any evaluation functions $f_{i,j} \in \Lambda_{\Gamma_{m_n}}(S)$ we have that 
$f_{i,j}(S) = f_{i,j}(P_{l_{n+1}}S) = f_{i,j}(P_{l_n +n+1}S)$.
Thus, 
by assumption (iii) in the definition of a General algorithm (Definition \ref{alg}), we get that 
$
\Lambda_{\Gamma_{m_n}}(S) = \Lambda_{\Gamma_{m_n}}(P_{l_{n+1}}S) =  \Lambda_{\Gamma_{m_n}}(P_{l_n +n+1}S)
$
which, by assumption (ii) in Definition \ref{alg} implies the assertion.
Thus, by the choice of the sequences $\{l_j\}_{j\in\mathbb{N}}$ and $\{m_j\}_{j\in\mathbb{Z}_+}$, it follows that $\Gamma_{m_n}(S) = \Gamma_{m_n}(P_{l_n +n+1}S) \subset B_{\frac{1}{2}}(0)$ for every $n.$ Since $\mathrm{sp(S)} = B_{1}(0)$ we observe that $\Gamma_m(S) \nrightarrow \mathrm{sp}(S)$.

{\bf Step VI:} $\{\Xi_\mathrm{sp},\Omega_\mathrm{B}\}\notin\Delta_3^G$. To prove this we shall need one of the results from \S \ref{dec_sec}. Namely, if we define $\Omega'$ to be the collection of all infinite matrices $\{a_{i,j}\}_{i,j\in\Zb}$ with entries $a_{i,j}\in\{0,1\}$ and consider
\begin{align*}
\Xi'&:\Omega'\ni\{a_{i,j}\}_{i,j \in \mathbb{Z}}\mapsto\left(\exists D \forall j 
\left(\left(\forall i \sum_{k=-i}^i a_{k,j}<D\right) \vee
\left(\forall R \exists i \sum_{k=0}^i a_{k,j}>R \wedge \sum_{k=-i}^0 a_{k,j}>R\right)\right)\right)\\
&\quad\quad\text{(``there is a bound $D$ such that every column has either less than
$D$ $1$s or is two-sided infinite'')}
\end{align*}
(where we map into the discrete space $\{\Yes,\No\}$), then $\mathrm{SCI}(\Xi',\Omega')_{\mathrm{G}}=3$.

We may identify $\Omega_\mathrm{B} = \mathcal{B}(l^2(\Nb))$ with $\Omega = \mathcal{B}(X)$, where $X = \oplus_{n=-\infty}^{\infty} X_n$ in the $l^2$-sense and where $X_n = l^2(\mathbb{Z})$. Consider sequences $a=\{a_i\}_{i\in\Zb}$ over $\Zb$ with $a_i \in \{0,1\}$, and define respective operators $B_a\in\mathcal{B}(l^2(\Zb))$ with matrix representation $B_a=\{b_{k,i}\}$ by
\[b_{k,i}:=\begin{cases}
1:& k=i \text{ and } a_k=0\\
1:& k<i \text{ and } a_k=a_i=1 \text{ and } a_j=0 \text{ for all }k<j<i\\
0:& \text{ otherwise.}
\end{cases}\]
Then $B_a$ is again a shift on a certain subset of basis elements and the identity on the
other basis elements, hence we have the following possible spectra:
\begin{itemize}
\item $\spc(B_a)\subset\{0,1\}$ if $\{a_i\}$ has finitely many $1$s.
\item $\spc(B_a)=\mathbb{T}$, the unit circle, if there are infinitely many $i>0$ 
			with $a_i=1$ and infinitely many $i<0$ with $a_i=1$ (we say $\{a_i\}$ is
			two-sided infinite).
\item $\spc(B_a)=\mathbb{D}$, the unit disc, if $\{a_i\}$ has infinitely many $1$s, 
			but only finitely many for $i<0$ or finitely many for $i>0$ (we say $\{a_i\}$ 
			is one-sided infinite in that case).
\end{itemize}
Next for a matrix $\{a_{i,j}\}_{i,j\in\Zb}$ we define the operator 
\begin{equation}\label{the_C}
C:= \bigoplus_{k=-\infty}^{\infty} B_k
\end{equation}
on $X$, where $B_k=B_{\{a_{i,k}\}_{i\in\Zb}}$ corresponds to the column 
$\{a_{i,k}\}_{i\in\Zb}$ in the above sense. Concerning its spectrum we have
$\bigcup_{k\in\Zb}\spc(B_k)\subset\spc(C) \subset \mathbb{D}$ since $\|C\|=1$. 
Clearly, if one of the columns is one-sided infinite then 
$\spc(C) = \mathbb{D}$. The same holds true if for every $k\in\Nb$ there is 
a finite column with at least $k$ $1$s.
Otherwise (that is if there is a number $D$ such that for every column it 
holds that it either has less than $D$ $1$s or is two-sided infinite) the 
spectrum $\spc(C)$ is a subset of $\{0\}\cup\mathbb{T}$.

Suppose for a contradiction that there exists a height two tower, $\Gamma_{n_2,n_1}$ solving $\{\Xi_\mathrm{sp},\Omega_\mathrm{B}\}$. Consider the intervals
$
J_1=[0,1/8],
$ and 
$
J_2=[3/8,\infty).
$
Set $\alpha_{n_2,n_1}=\mathrm{dist}(1/2,\Gamma_{n_2,n_1}(A))$. Let $k(n_2,n_1)\leq n_1$ be maximal such that $\alpha_{n_2,k}(A)\in J_1\cup J_2$. If no such $k$ exists or $\alpha_{n_2,k}(A)\in J_1$ then set $\tilde{\Gamma}_{n_2,n_1}(\{a_{i,j}\})=\No$. Otherwise set $\tilde{\Gamma}_{n_2,n_1}(\{a_{i,j}\})=\Yes$. It is clear from the construction of the matrix $C$ from $\{a_{i,k}\}_{i\in\Zb}$ that this defines a generalised algorithm. In particular, given $N$ we can evaluate $\{f_{k,l}(C):k,l\leq N\}$ using only finitely many evaluations of $\{a_{i,j}\}$, where we can use a bijection between the canonical bases to view $C$ as acting on $l^2(\mathbb{N})$. The point of the intervals $J_1,J_2$ is that we can show $\lim_{n_1\rightarrow\infty}\tilde{\Gamma}_{n_2,n_1}(\{a_{i,j}\})=\tilde{\Gamma}_{n_2}(\{a_{i,j}\})$ exists (the distance to the point $1/2$ cannot oscillate infinitely often between $J_1$ and $J_2$). If $\Xi'(\{a_{i,j}\})=\No$ then for large $n_2$ we have $\lim_{n_1\rightarrow\infty}\alpha_{n_2,n_1}(A)<1/8$ and hence $\lim_{n_2\rightarrow\infty}\tilde{\Gamma}_{n_2}(\{a_{i,j}\})=\No$. Similarly, if $\Xi'(\{a_{i,j}\})=\Yes$ then for large $n_2$ we have $\lim_{n_1\rightarrow\infty}\alpha_{n_2,n_1}(A)>3/8$ and hence $\lim_{n_2\rightarrow\infty}\tilde{\Gamma}_{n_2}(\{a_{i,j}\})=\Yes$. Hence $\tilde{\Gamma}_{n_2,n_1}$ is a height two tower of general algorithms solving $\{\Xi',\Omega'\}$, a contradiction.
\end{proof}

\begin{remark}
We note that in the case of self-adjoint bounded operators the spectrum $\mathrm{sp}(A)$ is real and the function $g$ can be chosen as $x\mapsto x$. Thus, in the definition of $\Upsilon_K^\delta(\zeta)$ it suffices to consider compact $K\subset \Rb$, the real grid $G^\delta(K):=(K+[-\delta,\delta])\cap(\delta\Zb)$, and for all $z\in G^\delta(K)$ only the two points $z_{1/2}:=z\pm\zeta(z)$ in $I_z$. Also in the case of normal operators, where $g:x\mapsto x$ does the job again, the construction simplifies. In particular, for a given function $\zeta:\Cb\to[0,\infty)$ we may define sets $\Upsilon_K^\delta(\zeta)$ as follows:
For $z\in G^\delta(K)$ consider $I_z:=\{z+\zeta(z)e^{j\delta\ii}: j=0,1,\ldots,\left\lceil 2\pi\delta^{-1}\right\rceil$\} and define $\Upsilon_K^\delta(\zeta)$ again as in $\eqref{Upsilon}$. The proof is then the same, up to some obvious adaptations. 
\end{remark}

\begin{proof}[Proof that $\{\Xi_{\mathrm{sp}},\Omega_{f}\cap \Omega_{\mathrm{N}}\}\in\Sigma_1^{A,\mathrm{eigv}}$]
Since $\Omega_{f}\cap \Omega_{\mathrm{N}}\subset \Omega_{fg}$, the only part left of the proof is the result concerning approximate eigenvectors. Let $\{\Gamma_n\}_{n\in\mathbb{N}}$ denote the sequence of arithmetic algorithms defined in Step III of the above proof. By the now standard argument of taking subsequences, it is enough to show that given $z\in\Gamma_n(A)$ and $\delta\in\mathbb{Q}_{>0}$ with $\delta<1$, we can compute in finitely many arithmetic operations and comparisons a vector $\psi_n$ such that
$
\max\left\{\|A\psi_n-z\psi_n\|,|1-\|\psi_n\||\right\}\leq E_n(A)+2c_n+2\delta,
$
where $E_n(A)$ is defined in \eqref{eq:def_E_n}. 
Without loss of generality, we can assume that $z=0$ by an appropriate shift of the operator $A$. By construction of the algorithm, we must have that
$
\sigma_1(P_{f(n)}\tilde{A}P_n)+c_n< E_n(A)+\delta,
$
where $\tilde{A}$ is the approximation of the matrix $A$ used when computing $\Gamma_n(A)$ (recall we deal with $\Delta_1$ information). We assume without loss of generality that
$
\|AP_n-P_{f(n)}\tilde{A}P_n\|\leq c_n+\delta/2.
$
Let $\epsilon=(E_n(A)+\delta)^2$ and consider the matrix 
$$
B=\left(P_{f(n)}\tilde{A}P_n\right)^*\left(P_{f(n)}\tilde{A}P_n\right)-\epsilon I\in\mathbb{C}^{n\times n}
$$
which is Hermitian but not positive definite. It follows that $B$ can be put into the form
$
PBP^T=LDL^*,
$
where $L$ is lower triangular with $1$'s along its diagonal, $D$ is block diagonal with block sizes $1$ or $2$ and $P$ is a permutation matrix. This can be computed in finitely many arithmetic operations. Without loss of generality we assume that $P=I$. Let $x$ be an eigenvector of $B$ with non-positive eigenvalue then set $y=L^*x$. Such an $x$ exists by assumption. Note that, since $\delta>0$,
\begin{equation}
\label{strtcc}
\langle y,Dy\rangle=\langle L^*x,DL^*x\rangle=\langle x,Bx\rangle<0.
\end{equation}
It follows that there exists a non-zero vector $y_n$ with $\langle y_n,Dy_n\rangle\leq0$. Since the inequality in (\ref{strtcc}) is strict, such a vector is easy to compute using arithmetic operations by considering determinants and traces of $1$ blocks or $2$ blocks in the block diagonal matrix $D$. $L^*$ is invertible and upper triangular so we can solve for $\tilde{\psi}_n=(L^*)^{-1}y_n$. We can then approximately normalise $\tilde{\psi}_n$ to $\psi_n$ using finitely many arithmetic operations (e.g. by approximating the norm of $\tilde{\psi}_n$) so that
$
1-\delta<\|\psi_n\|\leq 1.
$
Note also that
\begin{equation*}
\|P_{f(n)}\tilde{A}P_n{\psi}_n\|^2=\langle {\psi}_n,B{\psi}_n\rangle+{\|\psi_n\|^2}\epsilon=\frac{\|\psi_n\|^2}{\|\tilde{\psi}_n\|^2}\langle y_n,Dy_n\rangle+\|\psi_n\|^2\epsilon\leq\epsilon.
\end{equation*}
It follows that (using ${\psi}_n$ to also denote the zero padding of ${\psi}_n$ to form a vector in $l^2(\mathbb{N})$)
\begin{equation*}
\begin{split}
\left\|AP_n{\psi}_n\right\|&\leq E_n(A)+\delta+\|AP_n-P_{f(n)}\tilde{A}P_n\|\|\psi_n\|\\
&\leq E_n(A)+\delta+ (c_n+\delta/2)(1+\delta)\leq E_n(A) +2c_n +2\delta
\end{split}
\end{equation*}
since $\delta<1$. The result now follows.
\end{proof}

\subsection{Essential Spectrum}

In this section, we prove the results for the essential spectrum. 
Since $\Omega_\mathrm{D}\subset\Omega_{fg}\subset\Omega_f$ and $\Omega_{\mathrm{SA}}\subset\Omega_{\mathrm{N}}\subset\Omega_{g}\subset\Omega_\mathrm{B}$, we only need to prove that $\{\Xi_{\mathrm{e}\textrm{-}\mathrm{sp}},\Omega_\mathrm{D}\}\notin\Delta_2^G$, $\{\Xi_{\mathrm{e}\textrm{-}\mathrm{sp}},\Omega_{\mathrm{SA}}\}\notin\Delta_3^G$, $\{\Xi_{\mathrm{e}\textrm{-}\mathrm{sp}},\Omega_\mathrm{B}\}\in\Pi_3^A$ and $\{\Xi_{\mathrm{e}\textrm{-}\mathrm{sp}},\Omega_f\}\in\Pi_2^A$. 

\begin{proof}[Proof of Theorem \ref{spec_thm_main} for the essential spectrum] 
{\bf Step I:}  $\{\Xi_{\mathrm{e}\textrm{-}\mathrm{sp}},\Omega_\mathrm{D}\}\notin\Delta_2^G$. To see this, suppose for a contradiction that a height one tower $\Gamma_n$ solves the computational problem. For the contradiction we will construct $A\in\Omega_\mathrm{D}$ with diagonal entries in $\{0,1\}$ such that $\Gamma_n(A)$ does not converge. Let $A_n=\mathrm{diag}(0,0,...,0)\in\mathbb{C}^{n\times n}$ and $B_n=\mathrm{diag}(1,1,...,1)\in\mathbb{C}^{n\times n}$ (we let $A_\infty$ and $B_\infty$ be the obvious infinite analogues). We will construct
$$
A=\bigoplus_{n\in\mathbb{N}}A_{a_n}\oplus B_{b_n},
$$
for $a_n,b_n\in\mathbb{N}$ inductively. Suppose that $a_1,b_1,a_2,b_2,...,a_m,b_m$ have been chosen. Then the operator
$$
C_m:=\big(\bigoplus_{n=1}^mA_{a_n}\oplus B_{b_n}\big)\oplus A_\infty
$$
has essential spectrum $\{0\}$. Hence there exists $n_m\geq m$ such that $\Gamma_{n_m}(C_m)\subset B_{1/4}(0)$. But by the definition of a general tower there must exist some $N(m)$ such that $\Gamma_{n_m}(C_m)$ only uses the evaluations of matrix elements $f_{i,j}(C_m)$ with $i,j\leq N(m)$. Now choose $a_{m+1}\geq \max\{N_m-(a_1+b_1+...+a_m+b_m),1\}$ then we must have $\Gamma_{n_m}(A)=\Gamma_{n_m}(C_m)$. Similarly, if $a_1,b_1,a_2,b_2,...,b_m,a_{m+1}$ have been chosen then we consider
$$
D_m:=\big(\bigoplus_{n=1}^mA_{a_n}\oplus B_{b_n}\big)\oplus A_{m+1}\oplus B_\infty
$$
and choose $b_{m+1}$ large so that $\Gamma_{\hat n_m}(A)=\Gamma_{\hat n_m}(D_m)\subset B_{1/4}(1)$ for some $\hat n_m\geq n_m$. This then gives the required contradiction, since the sequence $\Gamma_n(A)$ does not converge.

{\bf Step II:} $\{\Xi_{\mathrm{e}\textrm{-}\mathrm{sp}},\Omega_\mathrm{SA}\}\notin\Delta_3^G$. Suppose for a contradiction that $\Gamma_{n_2,n_1}$ is a height two tower solving this problem. Let $(\mathcal{M},d)$ be the discrete space $\{\Yes,\No\}$, let $\Omega'$ denote the collection of all infinite matrices $\{a_{i,j}\}_{i,j\in\mathbb{N}}$ with entries $a_{i,j}\in\{0,1\}$ and consider the problem function
\begin{equation*}
\Xi'(\{a_{i,j}\}):\text{ Does $\{a_{i,j}\}$ have (only) finitely many columns with (only) finitely many $1$s?}
\end{equation*}
In Section \ref{dec_sec} we prove that $\mathrm{SCI}(\Xi',\Omega')_{G} = 3$. We will gain a contradiction by using the supposed height two tower for $\{\Xi_{\mathrm{e}\textrm{-}\mathrm{sp}},\Omega_\mathrm{SA}\}$, $\Gamma_{n_2,n_1}$, to solve $\{\Xi',\Omega'\}$.

Without loss of generality, identify $\Omega_\mathrm{SA}$ with self adjoint operators in $\mathcal{B}(X)$ where $X=\bigoplus_{j=1}^{\infty}X_j$ in the $l^2$-sense with $X_j=l^2(\mathbb{N})$. Now let $\{a_{i,j}\}\in\Omega'$ and for the $j$th column define $B_j\in\mathcal{B}(X_j)$ with the following matrix representation:
$$
B_j=\bigoplus_{r=1}^{M_j} A_{l_r^j},\quad A_{m}:=\begin{pmatrix}
1& & & &1\\
 &0& & & \\
& &\ddots& & \\
 & & &0& \\
1& & & &1\\
\end{pmatrix}
\in\mathbb{C}^{m\times m},
$$
where if $M_j$ is finite then $l_{M_j}^j=\infty$ with $A_{\infty}=\mathrm{diag}(1,0,0,...)$. The $l_r^j$ are defined by the relation
\begin{equation}
\label{crazy_indices}
\sum_{r=1}^{\sum_{i=1}^m a_{i,j}} l_r^j=m+\sum_{i=1}^m a_{i,j},
\end{equation}
and measure the lengths ($+1$) of successive gaps between $1$'s in the $j$th column. Define the self-adjoint operator
$
A=\bigoplus_{j=1}^{\infty}B_j.
$
We then have that
$$
\mathrm{sp}_{\mathrm{ess}}(A)= \begin{cases}
        \{0,1,2\}, & \text{if } \Xi'(\{a_{i,j}\})=\No\\
        \{0,2\}, & \text{otherwise }.
        \end{cases}
$$

Consider the intervals
$
J_1=[0,1/2],
$ and 
$
J_2=[3/4,\infty).
$
Set $\alpha_{n_2,n_1}=\mathrm{dist}(1,\Gamma_{n_2,n_1}(A))$. Let $k(n_2,n_1)\leq n_1$ be maximal such that $\alpha_{n_2,k}(A)\in J_1\cup J_2$. If no such $k$ exists or $\alpha_{n_2,k}(A)\in J_1$ then set $\tilde{\Gamma}_{n_2,n_1}(\{a_{i,j}\})=\No$. Otherwise set $\tilde{\Gamma}_{n_2,n_1}(\{a_{i,j}\})=\Yes$. It is clear from (\ref{crazy_indices}) and the definition of the $A_m$ that this defines a generalised algorithm. In particular, given $N$ we can evaluate $\{A_{k,l}:k,l\leq N\}$ using only finitely many evaluations of $\{a_{i,j}\}$, where we can use a bijection between the canonical bases to view $A$ as acting on $l^2(\mathbb{N})$. Again, the point of the intervals $J_1,J_2$ is that we can show $\lim_{n_1\rightarrow\infty}\tilde{\Gamma}_{n_2,n_1}(\{a_{i,j}\})=\tilde{\Gamma}_{n_2}(\{a_{i,j}\})$ exists. If $\Xi'(\{a_{i,j}\})=\No$ then for large $n_2$ we have $\lim_{n_1\rightarrow\infty}\alpha_{n_2,k}(A)<1/2$ and hence $\lim_{n_2\rightarrow\infty}\tilde{\Gamma}_{n_2}(\{a_{i,j}\})=\No$. Similarly, if $\Xi'(\{a_{i,j}\})=\Yes$ then for large $n_2$ we have $\lim_{n_1\rightarrow\infty}\alpha_{n_2,k}(A)>3/4$ and hence $\lim_{n_2\rightarrow\infty}\tilde{\Gamma}_{n_2}(\{a_{i,j}\})=\Yes$. Hence $\tilde{\Gamma}_{n_2,n_1}$ is a height two tower of general algorithms solving $\{\Xi',\Omega'\}$, a contradiction.

{\bf Step III:} $\{\Xi_{\mathrm{e}\textrm{-}\mathrm{sp}},\Omega_\mathrm{B}\}\in\Pi_3^A$. 
We start by defining the following functions on $\mathbb{C}$, where $Q_n:=I-P_n$,
\begin{align*}
\mu_{m,n,k}&:z\mapsto \min\{\sigma_1(P_k(A-z I)Q_mP_n), \sigma_1(P_k(A-z I)^*Q_mP_n)\}\\
\mu_{m,n}&:z\mapsto \min\{\sigma_1((A-z I)Q_mP_n), \sigma_1((A-z I)^*Q_mP_n)\}\\
\mu_m&:z\mapsto \min\{\sigma_1((A-z I)Q_m), \sigma_1((A-z I)^*Q_m)\}.
\end{align*}
Here $P_k(A-z I)Q_mP_n$ is considered as operator on $\mathrm{Ran}(Q_mP_n)$, etc. as usual.
Recall from the previous proofs that, for every $n,m$, $\mu_{m,n,k}\to \mu_{m,n}$ pointwise and monotonically
from below as $k\to\infty$ and for every $m$ $\mu_{m,n}\to \mu_{m}$ pointwise and 
monotonically from above as $n\to\infty$. Furthermore, $\{\mu_m\}_{m \in \mathbb{N}}$ is pointwise increasing 
and bounded, hence converges as well. By Proposition \ref{REC_SING} we can compute with finitely many arithmetic operations and comparisons, for any given $z$, an approximation $\tilde\mu_{m,n,k}(z)\in\mathbb{Q}$ with $\left|\mu_{m,n,k}(z)-\tilde\mu_{m,n,k}(z)\right|\leq1/k$. Furthermore, we can approximate from below and assume without loss of generality (by taking successive maxima) that $\tilde\mu_{m,n,k}(z)$ converges to $\mu_{m,n}$ pointwise and monotonically from below (again, this also includes the case of $\Delta_1$-information).

Next, we define the finite grids
\[G_n:=\left\{\frac{s+\ii t}{2^n}: s,t\in\{-2^{2n},\ldots,2^{2n}\}\right\},\]
and, for $A\in \mathcal{B}(l^2(\mathbb{N}))$,  
\begin{equation*}
\hat\Gamma_{m,n,k}(A):=\left\{z\in G_n: \tilde\mu_{m,n,k}(z)\leq\frac{1}{m}\right\}\\
\end{equation*} 
\begin{align}\label{lim1}
\hat\Gamma_{m,n}(A)&:=\bigcap_{k\in\mathbb{N}}\hat\Gamma_{m,n,k}(A) = \lim_{k\to\infty}\hat\Gamma_{m,n,k}(A),\\
\label{lim2}
\hat\Gamma_{m}(A)&:=\overline{\bigcup_{n\in\mathbb{N}}\hat\Gamma_{m,n}(A)}= \lim_{n\to\infty}\hat\Gamma_{m,n}(A),\\
\label{lim3}
\hat\Gamma(A)&:=\bigcap_{m\in\mathbb{N}}\hat\Gamma_{m}(A) = \lim_{m\to\infty}\hat\Gamma_{m}(A).
\end{align}
It easily follows again that all $\hat\Gamma_{m,n,k}$ are general algorithms in the sense of Definition \ref{alg} that require only finitely many arithmetic operations. We shall show that for large enough $n$, the above sets are non-empty and establish the limits in (\ref{lim1}), (\ref{lim2}) and (\ref{lim3}) and that $\hat\Gamma(A)$ equals $\mathrm{sp}_{\mathrm{ess}}(A)$. We will show that it is possible to choose a subsequences of $n$ such that this holds (each output and any limits must never empty since we require convergence in the Hausdorff metric) allowing us to construct a height three arithmetic tower. The final limit will be from above and hence the $\Pi_3^A$ classification. 

To do that we abbreviate $\mathcal{H}:=l^2(\mathbb{N})$ and first show that 
\begin{equation}\label{EPLowNorm}
\mu(z):=\lim_{m\to\infty}\mu_m(z)\quad\text{equals}\quad\|(A-z I+\Kc(\mathcal{H}))^{-1}\|^{-1}
\quad\text{for all}\quad z \in \mathbb{C},
\end{equation}
where $A-z I+\Kc(\mathcal{H})$ denotes the element in the Calkin algebra $\mathcal{B}(\mathcal{H})/\mathcal{K}(\mathcal{H})$ and where we use the 
convention $\|b^{-1}\|^{-1}:=0$ if the element $b$ is not invertible.
Clearly it suffices to consider $z=0$. The estimate ``$\leq$'' is trivial in 
case $\mu(0)=0$. So, let $\mu(0)>\epsilon>0$. Choose $m\in\Nb$ such that 
$\mu_m(0)\geq\mu(0)-\epsilon$. The operator
$A_0:=AQ_m:\mathrm{Ran} Q_m\to \mathrm{Ran}(AQ_m)$ is invertible, hence the kernel of $A=AQ_m+AP_m$
has finite dimension. $\sigma_1(A^*Q_m)>0$ yields that $\mathrm{Ran} A$ has finite codimension, 
hence both $A$ and $AQ_m$ are Fredholm. Let $R$ be the orthogonal projection onto 
$\mathrm{Ran} AQ_m$, $B_0$ the inverse of $A_0$ and $B:=B_0R$. Then
\begin{align*}
BA-I&=(BA-I)P_m+(BA-I)Q_m=(BA-I)P_m\quad\text{and}\quad\\
AB-I&=(AB-I)(I-R)+(AB-I)R=(AB-I)(I-R)
\end{align*} 
are compact, i.e. $B$ is a regulariser for $A$. Now
\begin{align*}
\|(A+\Kc(\mathcal{H}))^{-1}\|^{-1}&\geq\|B\|^{-1}=\|B_0R\|^{-1} \\
&\geq(\|B_0\|\|R\|)^{-1}=\|B_0\|^{-1} =\sigma_1(AQ_m) \geq \mu(0)-\epsilon
\end{align*}
gives the estimate ``$\leq$'' since $\epsilon$ is arbitrary.

Conversely, there is nothing to prove if $A$ is not Fredholm, so
let $\epsilon>0$ and $B\in(A+\Kc(\mathcal{H}))^{-1}$ be a regulariser with
$\|B\|\leq \|(A+\Kc(\mathcal{H}))^{-1}\|+\epsilon$. Since the operator
$K:=BA-I$ is compact we get for all sufficiently large $m$ that
$\|Q_mBAQ_m-Q_m\|=\|Q_mKQ_m\|$ is so small such that 
$Q_m+Q_mKQ_m$ is invertible in $\mathcal{B}(\mathrm{Ran}(Q_m))$, 
\[\underbrace{Q_m(Q_m+Q_mKQ_m)^{-1}Q_mB}_{\displaystyle \quad \quad \quad \quad \quad \quad =:B_1\in\mathcal{B}(\mathcal{H})}AQ_m=Q_m
\quad\text{and}\quad\|Q_mB-B_1\|<\epsilon.\]
We get that $\sigma_1(AQ_m)>0$, hence the compression
$AQ_m:\mathrm{Ran}(Q_m) \to \mathrm{Ran}(AQ_m)$ is invertible and the compression
$B_1|_{\mathrm{Ran}(AQ_m)}:\mathrm{Ran}(AQ_m) \to \mathrm{Ran}(Q_m)$ is its (unique) inverse.
Thus, we have $\|B_1\|\geq\|B_1|_{\mathrm{Ran}(AQ_m)}\|=\sigma_1(AQ_m)^{-1}$ and further
$\|B\|\geq \|Q_mB\|\geq \|B_1\|-\|Q_mB-B_1\|\geq \sigma_1(AQ_m)^{-1}-\epsilon.$
We conclude for sufficiently large $m$ that 
$\sigma_1(AQ_m)^{-1}\leq \|B\|+\epsilon\leq \|(A+\Kc(\mathcal{H}))^{-1}\|+2\epsilon.$
Since $\epsilon>0$ is arbitrary we arrive at 
$
\lim_{m\to\infty}\sigma_1(AQ_m)\geq\|(A+\Kc(\mathcal{H}))^{-1}\|^{-1}.
$ 
Applying this observation to $A^*$ we also find
$$
\lim_{m\to\infty}\sigma_1(A^*Q_m)\geq\|(A^*+\Kc(\mathcal{H}^*))^{-1}\|^{-1}=
\|(A+\Kc(\mathcal{H}))^{-1}\|^{-1},
$$
which finishes the proof of \eqref{EPLowNorm}.
In particular we now can apply that all of the above functions $\mu_{m,n,k}$, $\mu_{m,n}$, $\mu_{m}$, $\mu$ are continuous 
with respect to $z$, and together with the already discussed pointwise monotone convergence 
results, Dini's Theorem gives that the convergences are even locally uniform.

We can now establish the limits in (\ref{lim1}), (\ref{lim2}) and (\ref{lim3}) for large enough $n$. 
Obviously, $\{\hat\Gamma_{m,n,k}(A)\}_k$ is decreasing. If $\hat\Gamma_{m,n}(A)=\emptyset$ then there must exist some finite $k$ with $\hat\Gamma_{m,n,k}(A)=\emptyset$ since the sets are nested, closed and uniformly bounded. Furthermore, $\{\hat\Gamma_{m,n}(A)\}_n$ is increasing since, for every $k$,
$\hat\Gamma_{m,n}(A)\subset \hat\Gamma_{m,n,k}(A) \subset \hat\Gamma_{m,n+1,k}(A)$. Let $z\in\mathrm{sp}_{\mathrm{ess}}(A)$. For $m\in\mathbb{N}$, $\mu_m(z)=0$ and furthermore, there is an $n_0(m)$ and a $z_m\in G_{n_0(m)}$ with $|z-z_m|<1/m$, $\mu_m(z_m)<1/(2m)$ and $\mu_{m,n}(z_m)<1/m$ for every $n\geq n_0(m)$. Then, for every $k$, $\hat\mu_{m,n,k}(z_m)<1/m$ as well. Since the essential spectrum of a bounded linear operator is non-empty it follows that there exists a minimal $N(m)$ such that if $n\geq N(m)$ then $\hat\Gamma_{m,n}(A)\neq\emptyset$.

We now alter $\hat\Gamma_{m,n,k}$ as follows. For a given $m,n$ and $k$ we successively compute $\hat\Gamma_{m,n,k}(A),\hat\Gamma_{m,n+1,k}(A),...$ and choose $N(m,n,k)\geq n$ minimal such that $
\hat\Gamma_{m,N(m,n,k),k}(A)\neq\emptyset.$ By the above remarks, it follows that this process must terminate. We also have that
$
\Gamma_{m,n}(A):=\lim_{k\rightarrow\infty}\Gamma_{m,n,k}(A)
$
exists (in fact $\Gamma_{m,n,k}(A)$ is eventually constant as we increase $k$ since $\hat\mu_{m,n,k}$ is increasing) and also that $\Gamma_{m,n}(A)=\hat\Gamma_{m,\max\{n,N(m)\}}(A).$ Since $\Gamma_{m,n}(A)$ are increasing in $n$, it then follows that
$$
\Gamma_{m}(A):=\lim_{n\rightarrow\infty}\Gamma_{m,n}(A)=\overline{\bigcup_{n\in\mathbb{N}}\Gamma_{m,n}(A)}=\overline{\bigcup_{n\in\mathbb{N}}\hat\Gamma_{m,n}(A)}.
$$
Finally, $\{\Gamma_{m}(A)\}_m$ is decreasing. To see this, choose $z\in\Gamma_m(A)$ and a sequence
$(z_n)$ with $z_n\to z$ and $z_n\in \hat\Gamma_{m,n}(A)$ (for large $n$), respectively. The functions $\mu_{m,n}$ are non-decreasing in $m$ and hence we have
$$
\hat\Gamma_{m,n}(A)=\left\{z\in G_n: \mu_{m,n}(z)\leq\frac{1}{m}\right\}\subset\hat\Gamma_{m-1,n}(A)
$$
from which we conclude $z_n\in \hat\Gamma_{m-1,n}(A)$, hence $z\in\Gamma_{m-1}(A)$. It follows that the limit $\Gamma(A):=\lim_{m\rightarrow\infty}\Gamma_{m}(A)$ exists.

We are left with proving that $\Gamma(A)=\mathrm{sp}_{\mathrm{ess}}(A)$. Let $z\in\mathrm{sp}_{\mathrm{ess}}(A)$. Arguing as before, for $m\in\mathbb{N}$, $\mu_m(z)=0$ and furthermore, there is an $n_0(m)$ and a $z_m\in G_{n_0(m)}$ with $|z-z_m|<1/m$, $\mu_m(z_m)<1/(2m)$ and $\mu_{m,n}(z_m)<1/m$ for every $n\geq n_0(m)$. Then for every $k$ $\mu_{m,n,k}(z_m)<1/m$ as well. We 
conclude that $z_m\in\Gamma_m(A)\subset\Gamma_l(A)$, $l=1,\ldots,m$. Thus the limit 
$z$ of the sequence $\{z_m\}$ belongs to all $\Gamma_l(A)$ and hence $\mathrm{sp}_{\mathrm{ess}}(A)\subset\Gamma(A)$.
Conversely, let $z\notin\mathrm{sp}_{\mathrm{ess}}(A)$. Then $\mu(z)>\epsilon>0$ for a certain
$\epsilon>0$ and for all $z$ in a certain neighbourhood $U$ of $z$. Moreover 
there is an $m_0>3/\epsilon$ such that $\mu_m(z)>\epsilon/2$ for all $m\geq m_0$ and 
$z\in U$, hence $\mu_{m,n}(z)>\epsilon/2$ for all $m\geq m_0$, all $n$ and all $z\in U$. 
Further, for every $m>m_0$ and $n$ there is a $k_0(m,n)$ such that 
$\mu_{m,n,k}(z)>\epsilon/3>1/m_0> 1/m$ for all $k\geq k_0(m,n)$ and $z\in U$. 
Thus, the intersection of $U$ and $\Gamma(A)$ is empty, in particular 
$z\notin\Gamma(A)$.

{\bf Step IV:} $\{\Xi_{\mathrm{e}\textrm{-}\mathrm{sp}},\Omega_f\}\in\Pi_2^A$. 
Knowing a bound $f$ on the dispersion of $A$ obviously suggests to plug it into the previously defined algorithms and define 
\begin{align*}
{\kappa}_{m,n}&:z\mapsto \min\{\sigma_1(P_{f(n)}(A-z I)Q_mP_n), \sigma_1(P_{f(n)}(A-z I)^*Q_mP_n)\}\\
\tilde{\Gamma}_{m,n}(A)&:=\left\{z\in G_n: \hat{\kappa}_{m,n}(z)\leq\frac{1}{m}\right\}.
\end{align*}
Where, as usual, we will approximate $\kappa_{m,n}$ to within $1/n$ by a function $\hat\kappa_{m,n}$ taking rational values that can be computed (using Proposition \ref{REC_SING} to cope with $\Delta_1$-information if needed) at any point using finitely many arithmetic operations and comparisons. Unfortunately, all we know about the functions ${\kappa}_{m,n},\mu_m$ is that they are Lipschitz continuous with Lipschitz constant $1$ and that ${\kappa}_{m,n}$ converge pointwise to $\mu_m$, but not, whether or when this convergence is monotone. Therefore we have to make a modification in order to guarantee the existence of the desired limiting sets. The following idea is similar to the use of the intervals $J_1$ and $J_2$ in Step II and avoids possible oscillations at the boundary. 

Let $V_m$  denote the square $V_m:=\{z\in\mathbb{C}: |\Re(z)|,|\Im(z)|\leq 2^{-(m+1)}\}$ and $V_m(z):= z+V_m$ the respective shifted copies. Moreover, set $Z_m:=\{\frac{s+\ii t}{2^m}: s,t\in\Zb\}$ and
\begin{align*}
S_{m,n}(z)&:=\{i=m+1,\ldots,n:\exists z\in V_m(z)\cap G_i:\hat{\kappa}_{m,i}(z)\leq 1/m\}\\
T_{m,n}(z)&:=\{i=m+1,\ldots,n:\exists z\in V_m(z)\cap G_i:\hat{\kappa}_{m,i}(z)\leq 1/(m+1)\},
\end{align*}
as well as
\begin{align*}
E_{m,n}(z)&:=|S_{m,n}(z)|+|T_{m,n}(z)| - n\\
I_{m,n}&:=\{z\in Z_m: E_{m,n}(z)>0\text{ and }\left|z\right|\leq n\}\\
\hat\Gamma_{m,n}(A)&:=\bigcup_{z\in I_{m,n}} V_m(z).
\end{align*}
Roughly speaking, $\hat\Gamma_{m,n}(A)$ is the union of a family of squares $V_m(z)$ with $E_{m,n}(z)$ being positive, which is the case if ``most of the $\hat{\kappa}_{m,i}$ are small on $V_m(z)$''.

To make this precise, we first notice that all $\hat{\kappa}_{m,i}(z)$, $i\geq m+1$, with $z$ outside the compact ball $K:=B_{2\|A\|+2}(0)$ are larger than one, $I_{m,n}$ are finite, and all $\hat\Gamma_{m,n}(A)$ are contained in $K$, due to a simple Neumann series argument.
Furthermore, $\hat{\kappa}_{m,n}\to\mu_m$ uniformly on $K$ due to the Lipschitz continuity (uniform in $n$) of $\hat{\kappa}_{m,n}$ and $\mu_m$.

We now show that for each $m\geq 5$ the sign of $E_{m,n}(z)$ are eventually constant with respect to $n$ for every $z\in Z_m\cap K$, if $n$ is sufficiently large. That is, for every $z$ there is an $n(z)$ such that either $E_{m,n}(z)\leq 0$ or $E_{m,n}(z)> 0$ for all $n\geq n(z)$. For fixed $z$ and $m\geq 5$ we have to consider three possible cases: The first one is $\mu_m(w)>1/m$ for all $w\in V_m(z)$. Then there exists an $n_0$ such that $\hat{\kappa}_{m,n}(w)>1/m$ for all $n\geq n_0$ and all $w\in V_m(z)$ (take into account that $V_m(z)$ is compact and $\hat{\kappa}_{m,n}\to\mu_m$ locally uniformly), hence $|S_{m,n}(z)|+|T_{m,n}(z)|$ is constant and $E_{m,n}(z)$ is monotonically decreasing. Secondly, assume that $\mu_m(w)<1/m$ for all $w\in V_m(z)$.
Then there exists an $n_0$ such that $\hat{\kappa}_{m,n}(w)<1/m$ for all $n\geq n_0$ and all $w\in V_m(z)$, hence $|S_{m,n}(z)|=n-c$ with a certain constant $c$, and $E_{m,n}(z)=|T_{m,n}(z)|-c$ is monotonically increasing.
Finally, assume that $1/m$ belongs to the interval 
$$
[\min\{\mu_m(w):w\in V_m(z)\},\max\{\mu_m(w):w\in V_m(z)\}]
$$
 and notice that the length of that interval is at most $2^{-m}$ which is less than $1/m-1/(m+1)$ for $m\geq 5$. Then there exists an $n_0$ such that $\hat{\kappa}_{m,n}(w)>1/(m+1)$ for all $n\geq n_0$ and all $w\in V_m(z)$, hence $\{|T_{m,n}(z)|\}_{n\geq n_0}$ is constant, and 
\[
E_{m,n}(z)=(|S_{m,n}(z)|-n) + |T_{m,n}(z)|
\] is monotonically decreasing. 

Taking the maximum $N$ of the finite set $\{n(z):z\in Z_m\cap K\}$ then yields that the $\hat\Gamma_{m,n}(A)$, $n\geq N$, are constant, hence converge (if this constant set is non empty) as $n\to\infty$.
If $z_0\in\mathrm{sp}_{\mathrm{ess}}(A)$ then $\mu(z_0)=0$, hence $\mu_m(z_0)=0$ for all $m$. So, for fixed $m$, we have $\hat{\kappa}_{m,n}(z)<1/(m+1)$ for all sufficiently large $n$ and all $z$ in the neighbourhood $U_{1/(2m)}(z_0)$. Choose $z\in Z_m$ such that $z_0\in V_m(z)\subset U_{1/(2m)}(z_0)$. This is possible since $m\geq 5$. Then it is immediate from the definitions that $E_{m,n}(z)=n-c$ with a constant $c$ for all sufficiently large $n$, hence $z_0\in\Gamma_{m,n}(A)$ for $n$ large. Now given $m,n$, successively compute $\hat\Gamma_{m+5,n}(A),\hat\Gamma_{m+5,n+1}(A),...$ and let $N(m,n)\geq n$ be minimal such that $\hat\Gamma_{m+5,N(m,n)}(A)\neq\emptyset$. Define
$
\Gamma_{m,n}(A)=\hat\Gamma_{m+4,N(m,n)}(A).
$
The above arguments, in particular the fact that $\mathrm{sp}_{\mathrm{ess}}(A)\neq\emptyset$, demonstrate that this sequence of computations halts and $\Gamma_{m,n}$ is an arithmetical algorithm. Note also that
$
\Gamma_m(A):=\lim_{n\rightarrow\infty}\Gamma_{m,n}(A)
$
exists and the above argument shows that it contains the essential spectrum. Note also that $\Gamma_{m,n}(A)$ is in fact equal to $\Gamma_m(A)$ for large $n$.

We claim that $\{\Gamma_m(A)\}_m$ is a decreasing nested sequence, hence converges as well. Indeed, let $z\in\Gamma_{m+1}(A)$, then $z\in\hat\Gamma_{m+5,n}(A)$ for large $n$, i.e. $z\in V_{m+5}(w)$ for a $w\in I_{m+5,n}$, i.e. $w\in Z_{m+5}$ and $E_{m+5,n}(w)>0$. Clearly, (for large enough $n$) there exists a $w_0\in Z_{m+4}$ with $V_{m+5}(w)\subset V_{m+4}(w_0)$, and further (since we can assume without loss of generality by computing maxima over successive $m$ that $\hat{\kappa}_{m+4,i}(z)\leq\hat{\kappa}_{m+5,i}(z)$ holds whenever $n>m+5$)
\begin{align*}
S_{m+5,n}(w) & = \{i=m+6,\ldots,n:\exists z\in V_{m+5}(w)\cap G_i:\hat{\kappa}_{m+5,i}(z)\leq 1/(m+5)\} \\
&\subset \{i=m+5,\ldots,n:\exists z\in V_{m+4}(w_0)\cap G_i:\hat{\kappa}_{m+4,i}(z)\leq 1/(m+4)\} = S_{m+4,n}(w_0)
\end{align*}
and analogously $T_{m+5,n}(w)\subset T_{m+4,n}(w_0)$. Therefore $E_{m+5,n}(w)\leq E_{m+4,n}(w_0)$, which shows that $w_0\in I_{m+4,n}$ and thus $z\in\Gamma_{m}(A)$.

It remains to prove that the final limiting set $\lim_{m\rightarrow\infty}\Gamma_{m}(A)$ coincides with the essential spectrum. We have already proven that it must contain the essential spectrum.
Conversely, let $z_0\notin\mathrm{sp}_{\mathrm{ess}}(A)$, i.e. $\mu(z_0)>0$. Then, for large $m_0$, there exists an $\epsilon>3/m_0$ such that $\mu_m(z_0)>\epsilon$ and $\hat{\kappa}_{m,n}(z_0)>\epsilon/2$ for $m\geq m_0$ and large $n$, and then also $\hat{\kappa}_{m,n}(z)>\epsilon/3>1/m_0$ for all $z$ in a certain neighbourhood $U$ of $z_0$. For all sufficiently large $m\geq m_0$ all $V_m(z)$ which contain $z_0$ are subsets of $U$, hence $E_{m,n}(z)=d-n$ with a constant $d$ for large $n$, that is $\lim_{n\rightarrow\infty}\hat\Gamma_{m,n}(A)$ and $\{z_0\}$ are separated. But since the $\{\Gamma_m(A)\}_m$ are nested, it follows $z_0$ is not in the limiting set $\lim_{m\rightarrow\infty}\Gamma_{m}(A)$. This finishes the proof.
\end{proof}

\subsection{Determining if a point $z$ lies in $\mathrm{sp}(A)$}

Recall that for this problem, we restrict to $z\in\mathbb{R}$ when considering $\Omega_{\mathrm{D}}$ or $\Omega_{\mathrm{SA}}$. We also restrict to $z\neq 0$ when considering $\Omega_{\mathrm{C}}$. Since $\Omega_\mathrm{D}\subset\Omega_{fg}\subset\Omega_{f}$, $\Omega_\mathrm{C}\subset\Omega_{f}$ and $\Omega_{\mathrm{SA}}\subset\Omega_{\mathrm{N}}\subset\Omega_g\subset\Omega_B$ it is enough to prove that $\{\Xi^z_{\mathrm{sp}},\Omega_{\mathrm{SA}}\}\not\in\Delta_3^G$, $\{\Xi^z_{\mathrm{sp}},\Omega_B\}\in\Pi_3^A$, $\{\Xi^z_{\mathrm{sp}},\Omega_f\}\in\Pi_2^A$, $\{\Xi^z_{\mathrm{sp}},\Omega_{\mathrm{D}}\}\not\in\Delta_2^G$ and $\{\Xi^z_{\mathrm{sp}},\Omega_{\mathrm{C}}\}\not\in\Delta_2^G$.

\begin{proof}[Proof of Theorem \ref{spec_thm_main} for determining if a point lies in the spectrum]
{\bf Step I:} 
$\{\Xi^z_{\mathrm{sp}},\Omega_{\mathrm{SA}}\}\not\in\Delta_3^G$. By considering the shift $A-zI$, we can without loss of generality assume that $z=0$. Suppose for a contradiction that $\Gamma_{n_2,n_1}$ is a height two tower solving $\{\Xi^0_{\mathrm{sp}},\Omega_{\mathrm{SA}}\}$. Let $(\mathcal{M},d)$ be the discrete space $\{\Yes,\No\}$, let $\Omega'$ denote the collection of all infinite matrices $\{a_{i,j}\}_{i,j\in\mathbb{N}}$ with entries $a_{i,j}\in\{0,1\}$ and consider the problem function
\begin{equation*}
\Xi'(\{a_{i,j}\}):\text{ Does $\{a_{i,j}\}$ have (only) finitely many columns with (only) finitely many $1$s?}
\end{equation*}
In Section \ref{dec_sec} we prove that $\mathrm{SCI}(\Xi',\Omega')_{G} = 3$. Our strategy will be the same as the proof that $\{\Xi_\mathrm{sp},\Omega_\mathrm{B}\}\notin\Delta_3^G$ - we will gain a contradiction by using the supposed height two tower $\Gamma_{n_2,n_1}$ to solve $\{\Xi',\Omega'\}$.

First we need a certain periodic semi-infinite Jacobi matrix which gives rise to spectral pollution when applying the finite section method. Define
$$
A_{\infty}:=\begin{pmatrix}
0&3 & & & &\\
3 &0&1 & & &\\
& 1&0& 3& &\\
& & 3&0&1 &\\
& & & 1&0& \ddots\\
& & & & \ddots & \ddots
\end{pmatrix}
$$
It is well known that $\mathrm{sp}(A_\infty)=[-4,-2]\cup[2,4]$ (see for instance \cite{denisov2003zeros}). However, an easy check shows that $0$ is an eigenvalue of the finite truncated matrix $P_nA_\infty P_n$ whenever $n$ is odd. With an abuse of notation we also define
$
A_n:=P_nA_\infty P_n\oplus C_\infty \in\mathcal{B}(l^2(\mathbb{N})),
$
where $C_n$ denotes the $n\times n$ diagonal matrix with diagonal entries equal to $-4$.

Without loss of generality, we identify $\Omega_\mathrm{SA}$ with self adjoint operators in $\mathcal{B}(X)$ where $X=\bigoplus_{j=1}^{\infty}X_j$ in the $l^2$-sense with $X_j=l^2(\mathbb{N})$. Now let $\{a_{i,j}\}\in\Omega'$ and for the $j$th column define $B_j\in\mathcal{B}(X_j)$ as follows. Let $I_j=\{i\in\mathbb{N}:a_{i,j}=1\}$ and $J_j=\{i\in\mathbb{N}:a_{i,j}=0\}$. We partition $\mathbb{N}$ into two sets:
$$
N_1(j)=\{1\}\cup\{2k,2k+1:k\in I_j\},\quad N_2(j)=\{2k,2k+1:k\in J_j\}.
$$
On $\mathrm{span}\{e_k:k\in N_1(j)\}$ we let $B_j$ act as $A_{\left|N_1(j)\right|}$, whereas on $\mathrm{span}\{e_k:k\in N_2(j)\}$ we let $B_j$ act as $C_{\left|N_2(j)\right|}$ (both with respect to the natural bases and ordering). It is clear that $B_j$ is unitarily equivalent to $
A_{\left|N_1(j)\right|}\oplus C_{\left|N_2(j)\right|}$. Hence $\mathrm{sp}(B_j)$ is equal to $[-4,-2]\cup[2,4]\cup K_j$, where $K=\{0\}$ if $\sum_{i}a_{i,j}<\infty$ and $K_j=\emptyset$ otherwise.

Next we define the operator 
$$
C:= \bigoplus_{j=1}^{\infty} \left(B_j+\frac{1}{2j}I\right)
$$
on $X$. Concerning its spectrum, we note that any non-zero point of $\mathrm{sp}(C)$ inside the interval $[-1,1]$ is equal to $1/(2j)$ corresponding to precisely when the column $\{a_{i,j}\}_{i\in\mathbb{N}}$ has finitely many $1$'s. It is also clear that $0\in\mathrm{sp}(C)$ precisely when this happens infinitely many times ($0$ is a limit point of a descending sequence in the spectrum). Hence $\Xi^0_{\mathrm{sp}}(C)=\Yes$ if and only if $\Xi'(\{a_{i,j}\})=\No$.

We then define $\tilde{\Gamma}_{n_2,n_1}(\{a_{i,j}\})=\Yes$ if $\Gamma_{n_2,n_1}(C)=\No$ and $\tilde{\Gamma}_{n_2,n_1}(\{a_{i,j}\})=\No$ if $\Gamma_{n_2,n_1}(C)=\Yes$. Given $N$, we can evaluate $\{f_{k,l}(C):k,l\leq N\}$ using only finitely many evaluations of $\{a_{i,j}\}$, where we can use a bijection between the canonical bases to view $C$ as acting on $l^2(\mathbb{N})$. This follows since given any finite $i$, we can compute the sets $\{1,...,i\}\cap N_1(j)$ and $\{1,...,i\}\cap N_2(j)$. Hence $\tilde{\Gamma}_{n_2,n_1}$ defines a generalised algorithm and provides a height two tower of general algorithms solving $\{\Xi',\Omega'\}$, a contradiction.

{\bf Step II:} $\{\Xi^z_{\mathrm{sp}},\Omega_B\}\in\Pi_3^A$.
By considering the shift $A-zI$, we can without loss of generality assume that $z=0$ (note also that only having $\Delta_1$-information regarding $z$ is captured by only having $\Delta_1$-information on matrix entries after this shift). Define the numbers
\begin{align*}
\gamma := \min\{\sigma_1(A),\sigma_1(A^*)\},\quad 
\gamma_m := \min\{\sigma_1(AP_m),\sigma_1(A^*P_m)\},
\end{align*}
\begin{align*}
\gamma_{m,n} &:= \min\{\sigma_1(P_nAP_m),\sigma_1(P_nA^*P_m)\}\\
\delta_{m,n} &:= \min\{2^{-m}k:k\in\Nb, 2^{-m}k\geq\sigma_1(P_nAP_m)\text{ or }2^{-m}k\geq\sigma_1(P_nA^*P_m)\}.
\end{align*}
As pointed out before, $A$ is invertible if and only if $\gamma>0$.
Furthermore, note that $\gamma_m\downarrow_m\gamma,$ and that
$\gamma_{m,n}\uparrow_n\gamma_m$ for every fixed $m$. 
The sequences $\{\delta_{m,n}\}_n$ are bounded and monotonically non-decreasing, and 
$\gamma_{m,n}\leq\delta_{m,n}\leq\gamma_{m,n}+2^{-m}\leq\gamma_{m}+2^{-m}$. Thus, for $\epsilon>0$ there is an $m_0$, and for every $m\geq m_0$ there is an $n_0=n_0(m)$ such that
\begin{equation}\label{Egammadelta}
|\gamma-\delta_{m,n}|\leq |\gamma-\gamma_m|+|\gamma_m-\gamma_{m,n}|+|\gamma_{m,n}-\delta_{m,n}|
\leq \epsilon/3+\epsilon/3+2^{-m}\leq \epsilon
\end{equation}
whenever $m\geq m_0$ and $n\geq n_0(m)$. So we see that the numbers $\delta_{m,n}$ converge monotonically from below for every $m$ as $n\to\infty$, and the respective limits form a non-increasing sequence with respect to $m$, tending to $\gamma$. Moreover, each $\delta_{m,n}$ can be computed with finitely many arithmetic operations by Proposition \ref{PCholesky}. Thus, if we define 
$\Gamma_{k,m,n}(A):=\text{(}\delta_{m,n}<k^{-1}\text{)}$, the monotonicity ensure that
$
\Gamma_{k}(A):=\lim_{m\to\infty} \lim_{n\to\infty} \Gamma_{k,m,n}(A)
$
exists. Moreover, if $\gamma<k^{-1}$ then $\Gamma_k(A)=\Yes$. If $\Gamma_k(A)=\No$ then we must have that $\gamma\geq k^{-1}$ and hence $\Xi^0_{\mathrm{sp}}(A)=\No$. Finally, if $\gamma>k^{-1}$ then $\Gamma_k(A)=\No$. Hence $\Gamma_{k,m,n}$ provides a $\Pi_3^A$ tower.

{\bf Step III:} $\{\Xi^z_{\mathrm{sp}},\Omega_f\}\in\Pi_2^A$.
Again, by considering the shift $A-zI$, we can without loss of generality assume that $z=0$. If one considers operators for which a bound $f$ on their dispersion is known, then choosing $n=f(m)$ turns \eqref{Egammadelta} into
\begin{equation}\label{Egammadelta2}
|\gamma-\delta_{m,f(m)}|\leq |\gamma-\gamma_m|+|\gamma_m-\gamma_{m,f(m)}|+|\gamma_{m,f(m)}-\delta_{m,f(m)}|
\leq \epsilon/3+\epsilon/3+2^{-m}\leq \epsilon
\end{equation}
for large $m$ taking $|\sigma_1(BP_m)-\sigma_1(P_{f(m)}BP_m)|\leq\|(I-P_{f(m)})BP_m\|$ into account. Therefore, a natural first guess for our general algorithms could be $\tilde{\Gamma}_{k,m}(A):=\text{(}\delta_{m,f(m)}<k^{-1}\text{)}$. Unfortunately, although $\delta_{m,f(m)}$ converges to $\gamma$ as $m\to\infty$ by \eqref{Egammadelta2}, this is not monotone in general. Hence, it might be the case that $\gamma=k^{-1}$, but $\delta_{m,f(m)}$ oscillates around $k^{-1}$ such that $\{\tilde{\Gamma}_{k,m}(A)\}_m$ may not converge. To overcome this drawback, we can use the same interval trick as before. Define $J_k^1=[0,k^{-1}]$ and $J_k^2=[2k^{-1},\infty)$. For any given $m$, let $j(m)\leq m$ be maximal such that $\delta_{j,f(j)}\in J^1_{k}\cup J^2_{k}$. If no such $j$ exists or $\delta_{j,f(j)}\in J^2_{k}$ then set $\Gamma_{k,m}(A)=\No$, otherwise set $\Gamma_{k,m}(A)=\Yes$. By our now standard argument, this converges as $m\to\infty$. If $\gamma>0$, then for large enough $k$ (such that $2k^{-1}<\gamma$), $\Gamma_{k,m}(A)=\No$ for large $m$. Conversely, if $\gamma=0$ then for any $k$, $\delta_{m,f(m)}\in J^1_k$ for large $m$ and hence $\Gamma_{k,m}(A)=\Yes$ for large $m$. This gives $\Pi_2^A$ convergence.


{\bf Step IV:} $\{\Xi^z_{\mathrm{sp}},\Omega_{\mathrm{D}}\}\not\in\Delta_2^G$.
Again, by considering the shift $A-zI$, we can without loss of generality assume that $z=0$. If we assume that there is a general height-one-tower of algorithms $\{\Gamma_n\}$ over $\Omega_{\mathrm{D}}$ then we can again construct counterexamples very easily: For a decreasing sequence $\{a_i\}$ of positive numbers we consider the diagonal operator $A:=\diag\{a_i\}$. Clearly, $0$ belongs to the spectrum of $A$ if and only if the $a_i$s tend to zero.
As a start, set $\{a_i^1\}:=\{1,1,\ldots\}$, choose $n_1$ such that $\Gamma_{n}(\diag\{a_i^1\})=\No$
for all $n\geq n_1$, and $i_1$ such that $\max\{i,j \, \vert \, f_{i,j} \in \Lambda_{\Gamma_{n_1}}(\diag\{a_i^1\})\}<i_1.$ This is possible by (iii) in Definition \ref{alg} and our now standard argument. Then set $\{a_i^2\}:=\{1,1,\ldots,1,1/2,1/2,\ldots\}$ with 
$1/2$s starting at the $i_1$th position. If $n_1,\ldots,n_{k-1}$ and $i_1,\ldots,i_{k-1}$ are already chosen then pick $n_k$ such that $\Gamma_{n}(\diag\{a_i^k\})=\No$ for all $n\geq n_k$, and $i_k$ such that $\max\{i,j \, \vert \, f_{i,j} \in \Lambda_{\Gamma_{n_k}}(\diag\{a_i^k\})\}<i_k$, and modify $\{a_i^{k}\}$ to $\{a_i^{k+1}\}:=\{1,\ldots,2^{-k},2^{-k},\ldots\}$ with $2^{-k}$s starting at the $i_k$th position. Now, the contradiction is as in the previous proofs and we see that $\{\Xi^0,\Omega_{\mathrm{D}}\}\not\in\Delta_2^G$.
%

{\bf Step V:} $\{\Xi^z_{\mathrm{sp}},\Omega_{\mathrm{C}}\}\not\in\Delta_2^G$.
Recall in this case that $z\neq0$. By scaling any $A\in\Omega_{\mathrm{C}}$ by the factor $3/(2z)$, we can assume without loss of generality that $z=3/2$. Suppose for a contradiction that a general height-one-tower of algorithms $\{\Gamma_n\}$ solves $\{\Xi^{\frac{3}{2}}_{\mathrm{sp}},\Omega_{\mathrm{C}}\}$. Consider the arrowhead matrix:
$$
A_{n}(\epsilon):=\begin{pmatrix}
1&\epsilon &\epsilon^2 &\cdots &\epsilon^n\\
\epsilon &0& & & \\
\epsilon^2& &\ddots& & \\
 \vdots& & &0& \\
\epsilon^n& & & &0\\
\end{pmatrix},
$$
where $\epsilon\in(0,1)$. A simple calculation yields that the eigenvalues of $A_{n}(\epsilon)$ are $\{0,1/2\pm\sqrt{1+4a_n(\epsilon)}/2\}$, where
$
a_n(\epsilon)=\frac{\epsilon^2(1-\epsilon^{2n})}{1-\epsilon^2}.
$
In particular, we choose $\epsilon=\sqrt{3/7}$ for which the only positive eigenvalue is
$
b_n:=\frac{1+\sqrt{1+3(1-\frac{3^n}{7^n})}}{2}.
$
We now choose an increasing sequence of integers (greater than $1$) $r_1,r_2,...$ inductively, and define $A\in\Omega_{\mathrm{C}}$ such that when projected onto the span of the basis vectors $\{e_1,e_{r_1},...,e_{r_n}\}$ (with the natural order), with projection denoted by $Q_n$, $Q_nAQ_n$ has matrix $A_{n}(\sqrt{3/7})$. We also enforce that if $j\notin\{r_n\}_{n\in\mathbb{N}}\cup\{1\}$, then the $j$the column and row of $A$ are zero. In other words, $A_{1,r_n}=A_{r_n,1}=(\sqrt{3/7})^n$, $A_{1,1}=1$ and all other entries are $0$. It follows that $\mathrm{sp}(A)=\{0,1/2\pm1\}$ and hence $\Xi^{\frac{3}{2}}_{\mathrm{sp}}(A)=\Yes$. However, we choose $\{r_n\}$ such that there is an increasing sequence $\{c_n\}$ with $\Gamma_{c_n}(A)=\No$, yielding the contradiction.

Suppose that $r_1,...,r_n$ have been chosen. Then let $B_n$ be the infinite matrix with $Q_nB_nQ_n$ having matrix $A_{n}(\sqrt{3/7})$ and zeros elsewhere. Clearly the only positive eigenvalue of $B_n$ is $b_n<3/2$ and hence $\Xi^{\frac{3}{2}}_{\mathrm{sp}}(B_n)=\No$. So there exists $c_n>r_n$ with $\Gamma_{c_n}(B_n)=\No$. But by our now standard argument using the Definition \ref{alg} of a general algorithm, we can choose $r_{n+1}>r_n$ large such that $\Gamma_{c_n}(A)=\Gamma_{c_n}(B_n)$.
\end{proof}

\begin{remark}
To deal with $\Delta_1$-information in Step II of the above proof, we can approximate $\delta_{m,n}$ from below to accuracy $1/n$ (taking rational values) and take successive maxima to preserve monotonicity as $n\rightarrow\infty$. In Step III, we simply approximate $\delta_{m,f(m)}$ to accuracy $1/m$ (taking rational values). In both cases we use Proposition \ref{REC_SING}.
\end{remark}

\subsection{Techniques for proving lower bounds}\label{dec_sec}

Here we collect two results concerning decision making problems which are used to show lower bounds for two of our spectral problems. Within this section we exclusively deal with problems (functions) 
\[\Xi:\Omega\to \mathcal{M}:=\{\Yes, \No\},\]
where $\mathcal{M}$ is equipped with the discrete metric. This means that for such problems 
we search for General algorithms $\Gamma_{n_k,\ldots,n_1}:\Omega \to \mathcal{M}$ which, for a given
input $\omega\in\Omega$, answer $\Yes$ or $\No$. We will refer to such problems as decision making problems. Clearly, a sequence 
$\{m_i\}\subset \mathcal{M}$ of such ``answers'' converges to $m \in \mathcal{M}$ if and only if 
finitely many $m_i$ are different from $m$. Let $\Omega_1$ denote the collection of all infinite matrices $\{a_{i,j}\}_{i,j\in\Nb}$  with entries $a_{i,j}\in\{0,1\}$ and let $\Omega_2$ denote the collection of all infinite matrices $\{a_{i,j}\}_{i,j\in\Zb}$ with entries $a_{i,j}\in\{0,1\}$. Consider the following two problems:
\begin{align*}
\Xi_1&:\Omega_{1}\ni\{a_{i,j}\}_{i,j \in \mathbb{N}}\mapsto\text{ Does $\{a_{i,j}\}$ have (only) finitely many columns with (only) finitely many $1$s?}\\
\Xi_2&:\Omega_{2}\ni\{a_{i,j}\}_{i,j \in \mathbb{Z}}\mapsto\left(\exists D \forall j 
\left(\left(\forall i \sum_{k=-i}^i a_{k,j}<D\right) \vee
\left(\forall R \exists i \sum_{k=0}^i a_{k,j}>R \wedge \sum_{k=-i}^0 a_{k,j}>R\right)\right)\right)\\
&\quad\quad\text{(``there is a bound $D$ such that every column has either less than
$D$ $1$s or is two-sided infinite'')}
\end{align*}

\begin{theorem}[Decision making problems]\label{decision1}
Given the set-up above we have
\begin{equation*}
\begin{split}
&\mathrm{SCI}(\Xi_1,\Omega_1)_{\mathrm{G}}=\mathrm{SCI}(\Xi_1,\Omega_1)_{\mathrm{A}}= 3,\\
&\mathrm{SCI}(\Xi_2,\Omega_2)_{\mathrm{G}}=\mathrm{SCI}(\Xi_2,\Omega_2)_{\mathrm{A}}= 3.
\end{split}
\end{equation*}
\end{theorem}

\begin{remark}
Note that the SCI of the decision problems above are considered with respect to general and arithmetic towers. These towers do not assume any computability model, but only a model on the mathematical tools allowed (arithmetic operations in the case of an arithmetic tower) and the way the algorithm can read the available information (only finite amount of input). However, the SCI framework with towers of algorithms fit naturally into the classical theory of computability and the Arithmetical Hierarchy.
\end{remark}

To prove Theorem \ref{decision1}, we need to introduce some helpful background.
Equip the set of all sequences $\{x_i\}_{i\in\Nb}\subset\{0,1\}$ with the
following metric:
\begin{equation}\label{metric_Baire}
d_\mathrm{B}(\{x_i\},\{y_i\}):=\sum_{n\in\Nb}3^{-n}|x_n-y_n|.
\end{equation}
The resulting metric space is known as the Cantor space. By the usual enumeration
of the elements of $\Nb^2$ this metric translates to a metric on the set $\Omega_1$ 
of all matrices  $A=\{a_{i,j}\}_{i,j\in\Nb}$ with entries in $\{0,1\}$. Similarly, we do this for the set 
$\Omega_2$ of all matrices  $A=\{a_{i,j}\}_{i,j\in\Zb}$ with entries in $\{0,1\}$. In each case this gives a
complete metric space, hence a so called Baire space, i.e. it is of second category 
(in itself). To make this precise we recall the following definitions:

\begin{definition}[Meager set]
A set $S\subset\Omega$ in a metric space $\Omega$ is nowhere dense if every open set $U\subset\Omega$ has an 
open subset $V\subset U$ such that $V\cap S=\emptyset$, i.e. if the interior of the
closure of $S$ is empty.
A set $S\subset\Omega$ is meager (or of first category) if it is an at most countable
union of nowhere dense sets. Otherwise, $S$ is non-meager (or of second category). 
\end{definition}
Notice that every subset of a meager set is meager, as is every countable union of meager 
sets. By the Baire category theorem, every (non-empty) complete metric space is non-meager.

\begin{definition}[Initial segment]
We call a finite matrix $\sigma \in \mathbb{C}^{n \times m}$ an initial segment for an infinite matrix 
$A\in\Omega_1$ and say that $A$ is an extension of $\sigma$ if $\sigma$ is in the 
upper left corner of $A$. In particular, $\sigma = P_nAP_m$ for some $n,m \in \mathbb{N}$, where we, with slight abuse 
of notation, consider $P_nAP_m \in \mathbb{C}^{n \times m}$. $P_n$ is as usual the projection onto $\mathrm{span}\{e_j\}_{j=1}^n$, 
where $\{e_j\}_{j \in \mathbb{N}}$ is the canonical basis for $l^2(\mathbb{N})$.  

Similarly, a finite matrix $\sigma \in \mathbb{C}^{(2n+1) \times (2m +1)}$ is an initial segment for an infinite matrix 
$B \in\Omega_2$ if $\sigma$ is in the centre of $B$ i.e. $\sigma = \tilde{P}_nB\tilde{P}_m$ where $\tilde{P}_n$ is the projection onto $\mathrm{span}\{e_j\}_{j=-n}^n$, 
where $\{e_j\}_{j \in \mathbb{Z}}$ is the canonical basis for $l^2(\mathbb{Z})$.
We denote that $A$ is an extension of $\sigma$ by $\sigma \subset A$, and the set of all extensions of $\sigma$ by $E(\sigma)$. The notion of extension extends in an obvious way to finite matrices.
\end{definition}

Notice that the set $E(\sigma)$ of all extensions of 
$\sigma$ is a non-empty open and closed neighbourhood for every extension of $\sigma$.

\begin{lemma}\label{initial}
Let $\{\Gamma_n\}_{n\in\Nb}$ be a sequence of General algorithms mapping $\Omega_1\to\mathcal{M}$, 
$T\subset\Omega_1$ be a non-empty closed set, and $S\subset T$ be a non-meager set (in $T$) 
such that $\xi=\lim_{n\to\infty}\Gamma_n(A)$ exists and is the same for all $A\in S$.
Then there exists an initial segment $\sigma$ and a number $n_0$ such that 
$E^T(\sigma):=T\cap E(\sigma)$ is not empty, and such that $\Gamma_n(A)=\xi$ 
for all $A\in E^T(\sigma)$ and all $n\geq n_0$. The same statement is true if we consider $\Omega_2$ instead of $\Omega_1$.
\end{lemma}

\begin{proof}
We are in a complete metric space $T$. Since $S=\bigcup_{k\in\Nb}S_k$ with 
$S_k:=\{A\in S:\Gamma_n(A)=\xi \;\forall n\geq k\}$ and $S$ is non-meager, not all of 
the $S_k$ can be meager, hence there is a non-meager $S_k$, and we set $n_0:=k$. 
Now, let $A$ be in the closure $\overline{S_{n_0}}$, i.e. there is a sequence 
$\{A_j\}\subset S_{n_0}$ converging to $A$. Note that by assumption (i) in Definition \ref{alg} and the fact that $\Gamma_n$ are General algorithms, we have that, for every fixed $n\geq n_0$, $|\Lambda_{\Gamma_n}(A)| < \infty$. Thus, by (ii) in Definition \ref{alg}, the General algorithm 
$\Gamma_n$ only depends on a finite part of $A$, in particular $\{A_f\}_{f \in \Lambda_{\Gamma_n}(A)}$ where $A_f = f(A).$
Since each $f \in \Lambda_{\Gamma_n}(A)$ represents a coordinate evaluation of $A$ and by the definition of the metric $d_B$ in (\ref{metric_Baire}), it follows that for all sufficiently large $j$, $f(A) = f(A_j)$ for all $f \in \Lambda_{\Gamma_n}(A)$. By assumption (iii) in Definition  \ref{alg}, it then follows that $\Lambda_{\Gamma_n}(A_j) =  \Lambda_{\Gamma_n}(A)$ for all sufficiently large $j$. Hence, by assumption (ii) in Definition  \ref{alg}, we have that  
 $\Gamma_n(A)=\Gamma_n(A_j)=\xi$ for all sufficiently large $j$.
Thus, $\Gamma_n(A)=\xi$ for all $n\geq n_0$ and all $A\in \overline{S_{n_0}}$.
Since $S_{n_0}$ is not nowhere dense, we can choose a point $\tilde A$ in the interior of 
$\overline{S_{n_0}}$ and fix a sufficiently large initial segment $\sigma$ of $\tilde A$ such that $E^T(\sigma)$ is a subset of $\overline{S_{n_0}}$. The assertion of the lemma now follows.
The extension of the proof to $\Omega_2$ is clear.
\end{proof}

Roughly speaking, this shows that there is a nice open and closed non-meager subspace 
of $T$ for which $\lim_{n\to\infty}\Gamma_n(A)$ exists even in a uniform manner. 
Note that this result particularly applies to the case $T=\Omega$.

 \begin{proof}[Proof of Theorem \ref{decision1}]
{\bf Step I:} $\mathrm{SCI}(\Xi_1,\Omega_2)_{\mathrm{G}} \geq 3$. We argue by contradiction and assume that there is a height two tower $\{\Gamma_r\}$, $\{\Gamma_{r,s}\}$ for $\Xi_1$, where 
$\Gamma_r$ denote, as usual, the pointwise limits $\lim_{s\to\infty} \Gamma_{r,s}$. We will inductively construct initial segments $\{\sigma_n\}$ with $\sigma_{n+1} \supset \sigma_n$ yielding an infinite matrix $A \supset \sigma_n$ for all $n \in \mathbb{N},$ such that $\lim_{r\rightarrow \infty} \Gamma_r(A)$ does not exist. We construct $\{\sigma_n\}$ with the help of two sequences of subsets $\{T_n\}$ and $\{S_n\}$ of $\Omega$, with the properties that $T_{n+1} \subset T_n$, each $T_n$ is closed, and either $T_n = \Omega_1$ or there is an initial segment $\sigma \in \mathbb{C}^{m\times m}$ where $m \geq n$ such that $T_n$ is the set of all extensions of $\sigma$ with all the remaining entries in the first $n$ columns being zero. 

Suppose that we have chosen $T_{n}$. Note that the subset of all matrices in $T_n$ with one particular entry being fixed is closed in $T_n$, hence the set of all matrices with one particular column being fixed is closed (as an intersection of closed sets). The latter set has no interior points in $T_n$, hence is nowhere dense in $T_n$. This provides that the set of all matrices in $T_n$ for which a particular column has only finitely many $1$s is a countable union of nowhere dense sets in $T_n$, hence 
is meager in $T_n$. Hence the set of all matrices in $A\in T_n$ with $\Xi_1(A)=\No$ (i.e. matrices with infinitely many ``finite columns'') is meager in $T_n$ as well. Let $R$ be its complement in $T_n$, i.e. the non-meager set of all matrices $A \in T_n$ with $\Xi_1(A)=\Yes$.

Clearly, $R=\bigcup_{r\in\Nb}R_r$ with $R_r:=\{A\in R:\Gamma_k(A)=\Yes \;\forall k\geq r\}$, 
and there is an $r_n$ such that $S_n:=R_{r_n}$ is non-meager in $T_n$.
Note that $\Gamma_{r_n,s}$ are General algorithms and $\Gamma_{r_n}(A) =\lim_{s \rightarrow \infty} \Gamma_{r_n,s}(A) = \Yes$ for all $A \in S_n$. 
Thus, Lemma \ref{initial} applies and yields an initial segment $\sigma_n$, such that
\begin{equation}\label{uniform}
E^{T_n}(\sigma_n) \neq \emptyset \quad\text{and } \Gamma_{r_n}(A)=\Yes \text{ for all } A\in E^{T_n}(\sigma_n).
\end{equation}
Now, let $T_{n+1} \subset T_n$ be the (closed) set of all matrices in $E^{T_n}(\sigma_n)$ 
with all remaining \footnote{I.e. outside the initial segment $\sigma_n$.} entries in the first $n+1$ columns being zero. Letting $T_0 = \Omega_1$ we have completed 
the construction. 

The nested initial segments $\sigma_{n+1}\supset\sigma_n$ obviously yield a matrix $A \in \cap_{n=0}^{\infty} T_n$ and this $A$ has only finitely many $1$s in each of its columns. However, by the construction of $\{T_n\}$, we have that $A \in E^{T_n}(\sigma_n)$ for all $n \in \mathbb{N}$. Thus, $\Xi_1(A) = \No$, but by (\ref{uniform}), $\Gamma_k(A) = \Yes$ for infinitely many $k$. 

{\bf Step II:} $\mathrm{SCI}(\Xi_2,\Omega_2)_{\mathrm{G}} \geq 3$.
The proof is very similar to the proof of Step I. In particular, we argue by contradiction and assume that there is a height two tower $\{\Gamma_r\}$, $\{\Gamma_{r,s}\}$ for $\Xi_2$. As above, we inductively construct initial segments $\{\sigma_n\}$ with $\sigma_{n+1} \supset \sigma_n$ yielding an infinite matrix $A \supset \sigma_n$ for all $n \in \mathbb{N},$ such that $\lim_{r\rightarrow \infty} \Gamma_r(A)$ does not exist. We construct $\{\sigma_n\}$ with the help of two sequences of subsets $\{T_n\}$ and $\{S_n\}$ of $\Omega_2$, with the properties that $T_{n+1} \subset T_n$, each $T_n$ is closed, and either $T_n = \Omega_2$ or there is an initial segment $\sigma \in \mathbb{C}^{(2m+1)\times (2m+1)}$ where $m \geq n$ such that $T_n$ is the set of all extensions of $\sigma$ with all $\pm n$th semi-columns being filled by $n$ additional $1$s and infinitely many $0$s, and and all the other $k$th columns, $|k|\leq n-1$, are being filled with zeros. In particular, if $\{a_{i,j}\}_{i,j \in \mathbb{Z}} \in T_n$ then 
\begin{equation}\label{extension}
\begin{split}
\{a_{i,\pm n}\}_{i \in \mathbb{Z}} &= \{\hdots, 0,\underbrace{1, \hdots, 1}_{\displaystyle \quad n \, \text{times}}, \sigma_{-m,\pm n}, \hdots, \sigma_{m,\pm n}, \underbrace{1, \hdots, 1}_{\displaystyle \quad n \, \text{times}}, 0, \hdots\}^T, \\
\{a_{i, k}\}_{i \in \mathbb{Z}} &= \{\hdots, 0, \sigma_{-m, k}, \hdots, \sigma_{m, k}, 0, \hdots\}^T, \quad k \in \mathbb{Z}_+, \, |k| \leq n-1.
\end{split}
\end{equation}

Suppose that we have chosen $T_{n}$. We argue as in Step I and deduce that for $k\in\mathbb{Z}$
the set of all matrices in $T_n$ with one of the two $k$th semi-columns being 
fixed is nowhere dense in $T_n$, hence the set of all matrices in $T_n$ with (one of the two) 
$k$th semi-columns having finitely many $1$s is meager in $T_n$. 
We conclude that
the set of all matrices in $T_n$ with one semi-column having finitely many $1$s is meager,
thus its complement in $T_n$, the set of all matrices with all semi-columns having 
infinitely many $1$s, is non-meager. Therefore the same holds for the superset
$\{A \in T_n:\Xi_2(A)=\Yes\}$. Denoting this set by $R$ we obviously have 
$R=\bigcup_{r\in\Nb}R_r$ with $R_r:=\{A\in R:\Gamma_k(A)=\Yes \;\forall k\geq r\}$, 
and there is an $r_n$ such that $S_n:=R_{r_n}$ is non-meager in $T_n$.
Note that $\Gamma_{r_n,s}$ are General algorithms and $\Gamma_{r_n}(A) =\lim_{s \rightarrow \infty} \Gamma_{r_n,s}(A) = \Yes$ for all $A \in S_n$. 
Thus, Lemma \ref{initial} applies and yields an initial segment $\sigma_n$, such that
\begin{equation}\label{uniform2}
E^{T_n}(\sigma_n) \neq \emptyset \quad\text{and } \Gamma_{r_n}(A)=\Yes \text{ for all } A\in E^{T_n}(\sigma_n).
\end{equation}
Now, let $T_{n+1} \subset T_n$ be the (closed) set of all matrices $\{a_{i,j}\}_{i,j \in \mathbb{N}}$ in $E^{T_n}(\sigma_n)$ with the property that \eqref{extension} holds with $\sigma = \sigma_n$. Letting $T_0 = \Omega_2$ concludes the construction. The nested sequence $\{\sigma_n\}$again defines a matrix $A \in \cap_{n=0}^{\infty} T_n$ with the property that $A$ has finitely many but at least $k$ non-zero
entries in the each of its $k$th semi-column which gives $\Xi_2(A) = \No$, but, by (\ref{uniform2}), $\Gamma_k(A) = \Yes$ for infinitely many $k$, a contradiction.  

{\bf Step III:} $\mathrm{SCI}(\Xi_1,\Omega_1)_{\mathrm{A}} \leq 3$ and $\mathrm{SCI}(\Xi_2,\Omega_2)_{\mathrm{A}} \leq 3$. This can again be proved by defining an appropriate tower of height $3$ directly. For $\Xi_1$ we define
\[\Gamma_{k,m,n} (\{a_{i,j}\}_{i,j \in \Nb})=\Yes \quad\Leftrightarrow\quad |\{j=1,\ldots,m:\sum_{i=1}^n a_{i,j}<m\}|<k.\]
For $\Xi_2$ we define
\[\Gamma_{k,m,n} (\{a_{i,j}\}_{i,j \in \mathbb{Z}})=\Yes \quad\Leftrightarrow\quad |\{j=-m,\ldots,m:k<\sum_{i=1}^n a_{i,j}<m\text{ or }k<\sum_{i=-n}^{-1} a_{i,j}<m\}|=0.\]
It is straightforward to show these provide height three arithmetical towers.
\end{proof}

The lower bounds of the SCI of the decision problems $\Xi_1$ and $\Xi_2$ allow us to obtain the lower bounds of the SCI of spectra and essential spectra of operators.

\section{Proofs of Theorem \ref{main_self_adjoint} and Theorem \ref{thm:comp-res}}

\begin{remark}[Fourier Transform]
In this section we require the Fourier transform on $\mathrm{L}^2(\R^d)$, which will be denoted by $\Fcal:\mathrm{L}^2(\R^d)\to\mathrm{L}^2(\R^d)$. Our definition of $\Fcal$ is as follows:
	\begin{equation*}
	[\Fcal\psi](\xi)
	=
         \int_{\R^d}\psi(x)e^{-2\pi ix\cdot\xi}\,dx.
	\end{equation*}
For brevity we may write $\hat\psi$ instead of $\Fcal\psi$. 
With this definition $\Fcal$ is unitary on $\mathrm{L}^2(\R^d)$.
\end{remark}

\begin{remark}[The Attouch--Wets Topology]\label{rem:topology}
In \eqref{eq:attouch-wets-metric} we introduced the Attouch--Wets metric $d_{\mathrm{AW}}$ on the space $\mathcal{M}$ of non-empty closed subsets of $\C$. Since it is not convenient to work with $d_{\mathrm{AW}}$ directly, we make note of the following simple characterisation of convergence w.r.t. $d_{\mathrm{AW}}$. Let $A\subset\C$ and $A_n\subset\C$ be a sequence of closed and non-empty sets. Then:
	\begin{equation}
	d_{\mathrm{AW}}(A_n,A)\to0\quad\text{if and only if}\quad d_\mathcal{K}(A_n,A)\to0\text{ for any compact }\mathcal{K}\subset\C,
	\end{equation}
where
	\begin{equation}\label{compact_conv}
	d_\mathcal{K}(S,T)=\max\left\{\sup_{s\in S\cap \mathcal{K}}d(s,T),\sup_{t\in T\cap \mathcal{K}}d(t,S)\right\},
	\end{equation}
	where we use the convention that $\sup_{s\in S\cap \mathcal{K}}d(s,T) = 0$ if $ S\cap \mathcal{K} = \emptyset$.
We refer to \cite[Chapter 3]{Beer} for details and further discussion. Equivalently, we observe that
	\begin{equation} \label{EqSetConv}
	\begin{split}
	&d_{\mathrm{AW}}(A_n,A)\to0\\&\qquad\text{if and only if}\\ &\forall\delta>0,\ \mathcal{K}\subset\C\text{ compact, }\exists N\text{ s.t. }\forall n>N,\ A_n\cap\mathcal K\subset \mathcal{N}_{\delta}(A) \text{ and }A\cap\mathcal K\subset \mathcal{N}_{\delta}(A_n)
	\end{split}
	\end{equation}
where $\mathcal{N}_{\delta}(X)$ is the usual open $\delta$-neighbourhood of the set $X$. In this section we will simply use the notation $A_n \rightarrow A$ to denote this convergence, since there is no room for confusion.
\end{remark}

\subsection{The case of bounded potential $V$: The proof of Theorem \ref{main_self_adjoint}}\label{bounded_potential}

We will split the proof of Theorem \ref{main_self_adjoint} into two sections:
\begin{itemize}
\item[a.]
\emph{$\SCI(\Xi_{\mathrm{sp}},\Omega_{\phi,g})_{\mathrm{A}} = 1$:} Whilst the proof of this is somewhat long and technical (extra care has been taken to deal with $\Delta_1$-information and arithmetic algorithms over $\mathbb{Q}$), it is done via similar steps to the proof of Theorem \ref{spec_thm_main} in \S \ref{spec_proof}, namely through approximations of the resolvent norm. However, some work is needed to convert point samples of $V$ into approximations of the relevant matrices with respect to a Gabor basis. Lemmas \ref{lem:loc-unif} and \ref{bounds_on_h} are technical lemmas needed to achieve this, whereas Lemma \ref{lemma_gamma} concerns the approximations obtained via discretisations of the relevant inner products (and is need to gain the $\Sigma_1^A$ classification).
\item[b.]
\emph{Error control and rest of proof:} Lemma \ref{lemma_gamma} is used to prove $\{\Xi_{\mathrm{sp}},\Omega_{\phi,g}\}\in\Sigma_1^A$ and we extend the argument in \S \ref{spec_proof} to prove $\{\Xi_{\mathrm{sp}},\Omega_{\phi,\mathrm{SA}}\}\in  \Sigma^{A,\mathrm{eigv}}_1$. To prove the rest of the theorem, we argue that it is enough to prove $\{\Xi_{\mathrm{sp,\epsilon}},\Omega_{\phi}\}\in\Sigma_1^A$. This is done via Lemma \ref{pseudospec_unbounded} which uses the approximations of $\gamma(z)$ constructed in part (a).
\end{itemize}

Before we embark on the proof, the reader unfamiliar with the concept of Halton sequences may want to review this material. An excellent reference is \cite{Niederreiter} (see p. 29 for definition). We will also be needing the following definition and theorem in order to prove Theorem \ref{main_self_adjoint}. 

\begin{definition}
Let $\{t_1,\hdots t_N\}$ be a sequence in $[0,1]^d$. Then we define the \emph{star discrepancy} of 
$\{t_1,\hdots t_N\}$  to be 
$$
D^*_N(\{t_1,\hdots t_N\}) = \sup_{K \in \mathcal{K}}\left|\frac{1}{N}\sum_{k=1}^N\chi_{K}(t_k) - \nu(K)\right|,
$$  
where 
$\mathcal{K}$ denotes the family of all subsets of $[0,1]^d$ of the form $\prod_{k=1}^d[0,b_k),$ $\chi_K$ denotes the characteristic function on $K$, $b_k \in (0,1]$ and $\nu$ denotes the Lebesgue measure. 
\end{definition}

\begin{theorem}[\cite{Niederreiter}]\label{disc_bound}
If $\{t_k\}_{k\in\mathbb{N}}$ is the Halton sequence in $[0,1]^d$ in the pairwise relatively prime 
bases $b_1,\hdots , b_d$, then 
\begin{equation}\label{star}
D^*_N(\{t_1,\hdots t_N\}) \leq \frac{d}{N} + \frac{1}{N}\prod_{k=1}^d\left(\frac{b_k-1}{2\log(b_k)}\log(N) + \frac{b_k+1}{2}\right) \qquad N \in \mathbb{N}.
\end{equation}
\end{theorem}

For a proof of this theorem see \cite{Niederreiter}, p. 29. Note that as the right-hand side of (\ref{star}) is somewhat cumbersome to work with, it is convenient to define the following constant.

\begin{definition}\label{C*}
Define $C^*(b_1,\hdots,b_d)$ to be the smallest integer such that for all $N \in \mathbb{N}$
$$
\frac{d}{N} + \frac{1}{N}\prod_{k=1}^d\left(\frac{b_k-1}{2\log(b_k)}\log(N) + \frac{b_k+1}{2}\right) \leq C^*(b_1,\hdots,b_d)\frac{\log(N)^d}{N}
$$
where $b_1,\dots,b_d$ are as in Theorem \ref{disc_bound}.
\end{definition}

Further to these definitions, we shall require a Gabor basis which is the core in the discretisation carried out to produce the tower of algorithms. In particular, 
let
\begin{equation}\label{def_psi}
\psi_{k,l}(x) = e^{2\pi i k x}\chi_{[0,1]}(x-l), \qquad k,l \in \mathbb{Z}.
\end{equation}
It is well-known that $\psi_{k,l}$ form an orthonormal basis for $\mathrm{L}^2(\mathbb{R})$. Thus, by applying the Fourier transform,
\begin{equation}\label{the_SSS}
\{\hat \psi_{k_1,l_1} \otimes \hat \psi_{k_2,l_2} \otimes \cdots \otimes \hat \psi_{k_d,l_d}: k_1,l_1, \hdots, k_d, l_d \in \mathbb{Z}\}
\end{equation}
forms an orthonormal basis for $\mathrm{L}^2(\mathbb{R}^d)$ since the Fourier transform $\mathcal{F}$ is unitary. 
Let $\{\varphi_j\}_{j\in\mathbb{N}}$ be an enumeration of the collection of functions above, define 
\begin{equation}\label{the_SS}
\mathcal{S} = \mathrm{span}\{\varphi_j\}_{j\in\mathbb{N}}
\end{equation}
 and let 
\begin{equation}\label{enumeration}
\theta:\mathbb{N} \ni j \mapsto (k_1,l_1) \times \hdots \times (k_d,l_d) \in  \mathbb{Z}^{2d}
\end{equation}
be the bijection used in this enumeration.
Define
\begin{equation}\label{k_and_l}
\begin{split}
\tilde k(m,d) &:= \max\{|k_p|:(k_p, l_p) =  \theta(j)_p, p \in \{1,\hdots, d\}, j \in \{1,\hdots, m\} \},\\
\tilde l(m,d) &:= \max\{|l_p|:(k_p, l_p) =  \theta(j)_p, p \in \{1,\hdots, d\}, j \in \{1,\hdots, m\} \},
\end{split}
\end{equation}
and let
\begin{equation}\label{C_1}
C_1(m,d,a): = d^2\left(4\frac{(\max\{\tilde l(m,d)^2 + \tilde l(m,d)  +1/3,1\})^2}{|a-\tilde k(m,d)| + 1}\right)^d, \qquad m,d,a \in \mathbb{N},
\end{equation}
\begin{equation}\label{C_2}
C_2(m,d) :=  2^d \left(2((\tilde l(m,d)+1)^4 + \tilde l(m,d)^4)^2(2(\tilde k(m,d)+1) + 2)\right)^d, \qquad m,d \in \mathbb{N}.
\end{equation}
The quantities $C_1(m,d,a)$ and $C_2(m,d)$ may seem to come out of the blue. They stem from 
Lemma \ref{lem:loc-unif} and Lemma \ref{bounds_on_h} that are technical lemmas needed in order to construct the tower of algorithms. However, $C_1(m,d,a)$ and $C_2(m,d)$ occur in the main proof and thus it is advantageous to introduce them here to prepare the reader.

\subsubsection{\bf Proof that $\SCI(\Xi_{\mathrm{sp}},\Omega_{\phi,g})_{\mathrm{A}} = 1$}

The proof will make clear that we do not need to worry about the algorithm outputting the empty set - given $m$, simply compute $\Gamma_{j(m)}(V)$ with $j(m)\geq m$ minimal such that $\Gamma_{j(m)}(V)\neq \emptyset$. 

\begin{proof}[Proof of $\SCI(\Xi_{\mathrm{sp}},\Omega_{\phi,g})_{\mathrm{A}} = 1$.]
\textbf{Step I: Defining $
\Gamma_{m}(\{V_{\rho}\}_{\rho \in \Lambda_{\Gamma_{m}}(V)})
$ and $\Lambda_{\Gamma_{m}}(V)$}. To do so recall 
$\mathcal{S}$ from (\ref{the_SS}).
Note that since $\mathcal{D}(H) = \mathrm{W}^{2,2}(\mathbb{R}^d)$, it is easy to show that $\mathcal{S}$ is a core for $H$. 
Let $P_{m}$, $m \in \mathbb{N}$, be the projection onto $\mathrm{span}\{\varphi_j\}_{j=1}^{m}$, and let $z \in \mathbb{C}$.  Define
\[
S_{m}(V,z) := (-\Delta + V - zI)P_{m}
\quad
\text{and}
\quad
\tilde S_{m}(V,z) := (-\Delta + \overline{V} - \overline{z}I)P_m.
\]
Let 
$$
\sigma_1(S_{m}(V,z)) := \min\{(\langle S_{m}(V,z) f, S_{m}(V,z) f\rangle)^{\frac{1}{2}}: f \in \mathrm{Ran}(P_{m}), \|f\| =1 \}
$$
and 
$
\sigma_1(\tilde S_{m}(V,z)) := \min\{( \langle \tilde S_{m}(V,z) f, \tilde S_{m}(V,z) f\rangle)^{\frac{1}{2}}: f \in \mathrm{Ran}(P_{m}), \|f\| =1\},
$
and define
\begin{equation}\label{the_gammas11}
\gamma_m(z):= \min\{ \sigma_1(S_{m}(V,z)),
\sigma_1(\tilde S_{m}(V,z))\}\}.
\end{equation}
Note that if we could evaluate $\gamma_m$ at any point $z$ using only finitely many arithmetic operations of elements of the form $V(x)$, $x \in \mathbb{R}^d$, we could have defined a general algorithm as desired by using 
$\Upsilon_{B_{m}(0)}^{1/m}(\gamma_{m})$ where $\Upsilon_{B_{m}(0)}^{1/m}$ is defined in (\ref{Upsilon}).
Unfortunately, such evaluation is not possible ($\gamma_m$ may depend on infinitely many samples of $V$),  and we will now focus on finding an approximation to $\gamma_m$.

Let $S =\{t_k\}_{k\in\mathbb{N}}$, where $t_k \in [0,1]^d$ is a Halton sequence (see \cite{Niederreiter} p. 29 for definition) in the pairwise relatively prime 
bases $b_1,\hdots, b_d$ (note that the particular choice of the $b_j$s is not important). Define, for $a>0$ and $N \in \mathbb{N}$, the discrete inner product
\begin{equation}\label{innerp}
\langle f,u\rangle_{a,N} = \frac{(2a)^d}{N}\sum_{k=1}^Nf^a(t_k)\overline{u^a(t_k)}, \qquad f,u \in \mathrm{L}^2(\mathbb{R}^d) \cap \mathrm{BV}_{\mathrm{loc}}(\mathbb{R}^d)),
\end{equation}
where we have defined the rescaling function on $[0,1]^d$ by 
\begin{equation}\label{the_a}
f^a = f(a(2\cdot-1),\hdots,a(2\cdot-1))\vert_{[0,1]^d},
\end{equation}
(we will throughout the proof use the superscript $a$ on a function to indicate (\ref{the_a})), 
where $\mathrm{BV}_{\mathrm{loc}}(\mathbb{R}^d)) = \{f : \mathrm{TV}(f\vert_{[-b,b]^d}) < \infty, \, \forall b > 0\}$ and $\mathrm{TV}(f\vert_{[-b,b]^d})$ denotes the total variation, in the sense of Hardy and Krause (see \cite{Niederreiter}), of $f$ restricted to $[-b,b]^d$. Note that since $V \in \mathrm{L}^{\infty}(\mathbb{R}^d) \cap  \mathrm{BV}_{\mathrm{loc}}(\mathbb{R}^d)$ and any $f \in  \mathrm{Ran}(P_{m})$ is smooth we have that $S_{m}(V,z) f \in \mathrm{L}^2(\mathbb{R}^d) \cap \mathrm{BV}_{\mathrm{loc}}(\mathbb{R}^d))$. Hence, we can define 
for $n,m \in \mathbb{N}$
\begin{equation}\label{a_and_N}
\begin{split}
\sigma_{1,n}(S_{m}(V,z)) &:= \min \{(\langle S_{m}(V,z) f, S_{m}(V,z) f\rangle_{n,N(n)})^{\frac{1}{2}} : f \in \mathrm{Ran}(P_{m}),  \|f\|=1\}\\
\sigma_{1,n}(\tilde S_{m}(V,z)) &:= \min \{(\langle \tilde S_{m}(V,z) f, \tilde S_{m}(V,z) f\rangle_{n,N(n)})^{\frac{1}{2}} : f \in \mathrm{Ran}(P_{m}),  \|f\|=1\},
\end{split}
\end{equation}
where $N(n) := \lceil n\phi(n)^4\rceil$ and where $\phi$ comes from the definition of $\Omega_{\phi}$. We also set
\begin{equation}\label{the_Zs}
\begin{split}
Z_{m}(z)_{ij} &= \langle S_{m}(V,z) \varphi_j, S_{m}(V,z) \varphi_i\rangle_{n,N(n)}, \quad i,j \leq m, \\
\tilde Z_{m}(z)_{ij} &= \langle \tilde S_{m}(V,z) \varphi_j, \widetilde S_{m}(V,z) \varphi_i\rangle_{n,N(n)}, \quad i,j \leq m.
\end{split}
\end{equation}
We have the following expansion
\begin{equation}\label{tedious2}
\begin{split}
\langle  S_{m}(V,z)\varphi_j,  S_{m}(V,z)\varphi_i\rangle_{n,N} =& 
\langle  \Delta \varphi_j,  \Delta \varphi_i\rangle_{n,N} - \langle
V\varphi_j,  \Delta \varphi_i\rangle_{n,N} - \langle  \Delta \varphi_j,
V\varphi_i\rangle_{n,N}\\  &+ \langle  V\varphi_j,  V\varphi_i\rangle_{n,N} - 2\Re(z)\langle \Delta\varphi_j,
\varphi_i\rangle_{n,N} \\&+ \langle 2\Re(z\overline{V})\varphi_j,
\varphi_i\rangle_{n,N}  +
|z|^2\langle  \varphi_j,  \varphi_i\rangle_{n,N},
\end{split}
\end{equation}
with a similar expansion holding for the matrix entries of $\tilde Z_{m}(z)$. Recall that the $\varphi_j$s are an enumeration of the Fourier transforms of the basis
$
\psi_{k,l}(x) = e^{2\pi i k x}\chi_{[0,1]}(x-l),
$
$k,l \in \mathbb{Z}.$
It is easy to derive closed form expressions for $\hat \psi_{k,l}$ and $\frac{\partial^2 \hat \psi_{k,l}}{\partial \xi^{2}}$, and these expressions are variations of products of exponential functions and functions of the form $x \mapsto 1/x^p$ for $p = 1,2,3$. It follows that the matrix entries of $Z_{m}(z)$ and $\tilde Z_{m}(z)$ also have closed form expressions in terms of point evaluations of the potential $V$ (at the Halton nodes - see (\ref{innerp})). Note that the Halton nodes are rational. Using (\ref{tedious2}), it follows that given $\Delta_1$-information for $\Lambda$, we can compute in finitely many arithmetic operations and comparisons, approximations to $Z_{m}(z)$ to any required accuracy. The same holds true for $\tilde Z_{m}(z)$. From Proposition {\ref{REC_SING}}, it follows that $\sigma_{1,n}(S_{m}(V,z))$ and $\sigma_{1,n}(\tilde S_{m}(V,z))$ can be computed to any given accuracy using finitely many arithmetic operations and comparisons.

Consider the quantity
\begin{equation}\label{getting_N}
\begin{split}
\tilde \beta(m, n) &:= (m + 1)m C_1(m,d,n) \\ 
& \qquad +  (2n)^dd^2\left(m^2 + \sigma^2\phi^2(n) + 2(\sigma m +1)(\phi(n) + 1)\right) \\
& \qquad \times  \left(1+ \sigma^2 +  2\sigma \right) C_2(m,d) C^*(b_1,\hdots,b_d)\frac{\log(N(n))^d}{N(n)},
\end{split}
\end{equation}
where $\sigma = 3^d-2^{d+1} +2$, $C_1(m,d,n)$ is defined in (\ref{C_1}), $C_2(m,d)$ is defined in (\ref{C_2}) and $C^*(b_1,\hdots,b_d)$ is defined in Definition \ref{C*}. The function $\tilde \beta$ may seem to come somewhat out of the blue, however, it stems from certain bounds in (\ref{intbounds202}) (see also (\ref{tau2})) on errors of discrete integrals related to (\ref{a_and_N}). For any $m,n$, we can compute an upper bound in $\mathbb{Q}$ for $\tilde \beta(m, n)$ accurate to $1/m^4$ using finitely many arithmetic operations over $\mathbb{Q}$. Denote such an approximation by $\tilde \tau(m, n)$ and set
\begin{equation}\label{the_n}
n(m) := \min\left\{n: \tilde \tau(m,n) \leq \frac{1}{m^3}\right\}.
\end{equation}
First, note that the choice of $N(n)$ in (\ref{getting_N}) implies that $\tilde \beta(m,n) \rightarrow 0$ as $n \rightarrow \infty$. Thus, $n(m)$ is well defined since $\tilde \tau(m, n)<{m^{-3}}$ for large $n$. Second, note that it is clear that $\tilde \tau$, and hence also $n(m)$, can be evaluated by using finitely many arithmetic operations and comparisons.

We now let $\zeta_{m}(z)$ be a non-negative real valued function with 
\begin{equation}\label{totalrecall}
0 \leq \zeta_{m}(z) - \min\{ \sigma_{1,n(m)}(S_{m}(V,z)),
\sigma_{1,n(m)}(\tilde S_{m}(V,z))\} \leq \frac{1}{m}.
\end{equation}
Combining the above remarks shows that, given $\Delta_1$-information for $\Lambda$, we can compute such an approximation $\zeta_{m}(z)$ for any $z \in \mathbb{C}$ in finitely many arithmetic operations and comparisons over $\mathbb{Q}$. We can now define 
$$
\Gamma_{m}(V) :=  \Upsilon_{B_{m}(0)}^{1/m}(\zeta_{m}), 
$$
where $\Upsilon_{B_m(0)}^{1/m}(\zeta_{m})$ is defined in (\ref{Upsilon}). We conclude this step by noting that $\Gamma_{m}$ are arithmetic towers of algorithms using $\Delta_1$-information for $\Lambda$.

\textbf{Step II:}
We show that 
$
	\Gamma_{m}(V) \rightarrow 
	\Xi_{\mathrm{sp}}(V),
$
as $m \to \infty.$
Note that, by the properties of the Attouch--Wets topology, and as discussed in Remark \ref{rem:topology}, it suffices to show that for any compact set $\mathcal{K} \subset \mathbb{C}$
\begin{equation}\label{suff_convergence}
\begin{split}
&d_{\mathcal{K}}(\Gamma_{m}(V), \Xi_{\mathrm{sp}}(V)) \longrightarrow 0, \quad m \rightarrow \infty,
\end{split}
\end{equation}	
where $d_{\mathcal{K}}$ is defined in (\ref{compact_conv}). To show (\ref{suff_convergence}) we start by defining 
\begin{equation}\label{the_gamma}
\begin{split}
\gamma(z) &:= \min\big\{\inf\{\|(-\Delta + V - zI)\psi\| : \psi \in \mathrm{W}^{2,2}(\mathbb{R}^d), \|\psi\| = 1\},\\
&\quad \inf\{\|(-\Delta + \overline{V} - \overline{z}I)\psi\| : \psi \in \mathrm{W}^{2,2}(\mathbb{R}^d), \|\psi\| = 1\}\big\} = \|(-\Delta + V - zI)^{-1}\|^{-1},
\end{split}
\end{equation}
where we use the convention that $\|(-\Delta + V - zI)^{-1}\|^{-1} = 0$ when $z \in \spc(-\Delta + V)$ and proceed similarly to the proof of Theorem \ref{spec_thm_main} with the following claim. Before we state the claim recall $h$ from the definition of $\Upsilon_{K}^\delta(\zeta)$ in Step II of the proof of Theorem \ref{spec_thm_main} in \S \ref{spec_proof}.

{\bf Claim:} Let $\mathcal{K} \subset \mathbb{C}$ be any compact set, and let $K$ be a compact set containing $\mathcal{K}$ such that $\spc(-\Delta + V) \cap K \neq \emptyset$ and $0<\delta<\epsilon<1/2$.
Suppose that $\zeta$ is a function with $\|\zeta-\gamma\|_{\infty,\hat{K}}:=\|(\zeta-\gamma)\chi_{\hat{K}}\|_\infty < \epsilon$ on 
$\hat{K}:=(K+B_{h(\diam(K)+2\epsilon)+\epsilon}(0))$, where $\chi_{\hat{K}}$ denotes the characteristic function of $\hat K$ and $h$ is the inverse of $g$. Finally, let 
$u$ be defined as in (\ref{u_function}).
 Then $\lim_{\xi\to 0}u(\xi)=0$ and
\begin{equation}\label{claim_unbound}
d_{\mathcal{K}}(\Upsilon_{K}^\delta(\zeta),\spc(-\Delta + V)) \leq u(\epsilon). 
\end{equation}
To prove the claim, we first show that 
\begin{equation}\label{one_way}
\sup_{s\in \Upsilon_{K}^\delta(\zeta) \cap \mathcal{K}}\dist(s,\spc(-\Delta + V)) \leq u(\epsilon).
\end{equation}
If $\Upsilon_{K}^\delta(\zeta) \cap \mathcal{K}= \emptyset$ then there is nothing to prove, thus we assume that $\Upsilon_{K}^\delta(\zeta) \cap \mathcal{K} \neq \emptyset$.
Let $z\in G^\delta(K)$ and recall $G^\delta(K)$, $h_{\delta}$ and $I_z =B_{h_\delta(\zeta(z))}(z)\cap(\delta(\Zb+\ii\Zb))$ from the definition of $\Upsilon_{K}^\delta(\zeta)$ in Step II of the proof of Theorem \ref{spec_thm_main} in \S \ref{spec_proof}. Notice that we may argue exactly as in (\ref{I_z}) and deduce that $I_z\subset\hat{K}$.
Suppose that $M_z \neq \emptyset$. Note that from
$$
\|(-\Delta+V-zI)^{-1}\|^{-1} \geq g(\dist(z,\spc(H))),
$$ 
the monotonicity of $h$, and the compactness of $\spc(-\Delta + V) \cap K \neq \emptyset$ there is a $y\in \spc(-\Delta + V)$ of minimal distance to $z$ with $|z-y|\leq h(\gamma(z))$.
Since $\|\zeta-\gamma\|_{\infty,\hat{K}} < \epsilon$, and by using the monotonicity of $h$, we get $|z-y|\leq h(\zeta(z)+\epsilon)$.
Hence, at least one of the $v\in I_z$, say $v_0$, satisfies $|v_0-y|<h(\zeta(z)+\epsilon)-h(\zeta(z))+2\delta$.
Thus, by noting that $\gamma(v_0) \leq  \dist(v_0,\mathrm{sp}(-\Delta + V))$, and by the assumption that $\delta < \epsilon$, we get 
$
\zeta(v_0)<\gamma(v_0)+\epsilon < h(\zeta(z)+\epsilon)-h(\zeta(z))+3\epsilon.
$
 By the definition of $M_z$, this estimate now holds for all points $w \in M_z$. Thus, we may argue exactly as in (\ref{w_estimate}) and deduce that 
$$
\dist(w,\spc(-\Delta + V)) \leq h(h(\zeta(z)+\epsilon)-h(\zeta(z))+3\epsilon),
$$
which yields (\ref{one_way}).
To see that 
\begin{equation}\label{the_other_way}
 \sup_{t\in \spc(-\Delta + V)\cap \mathcal{K}}\dist(\Upsilon_{K}^\delta(\zeta) ,t) \leq u(\epsilon),
 \end{equation}
(where we assume that $\spc(-\Delta + V)\cap \mathcal{K} \neq \emptyset$) take any $y \in \spc(-\Delta + V)\cap \mathcal{K} \subset K$. Then there is a point $z\in G^\delta(K)$ with $|z-y|<\delta<\epsilon$, hence 
$$
\zeta(z)<\gamma(z)+\epsilon \leq \dist(z,\spc(-\Delta + V))+\epsilon < 2\epsilon<1.
$$ Thus, $M_z$ is not empty and contains a point which is closer to $y$ than $h(\zeta(z))+\epsilon \leq h(2\epsilon)+\epsilon\leq u(\epsilon)$, and this yields (\ref{the_other_way}). The fact that $\lim_{\xi\to 0}u(\xi)=0$ is shown in Step II of the proof of Theorem \ref{spec_thm_main} in \S \ref{spec_proof}, and we have proved the claim.

Armed with this claim we continue on the path to prove (\ref{suff_convergence}). We define 
\begin{equation}\label{the_gammas22}
\gamma_{m,n}(z):= \min\{ \sigma_{1,n}(S_{m}(V,z)),
\sigma_{1,n}(\tilde S_{m}(V,z))\}.
\end{equation}
Then $\left\|\zeta_m - \gamma_{m,n(m)}\right\|_{\infty}\leq 1/m$ where 
$n(m)$ is defined as in (\ref{the_n}).
By Lemma \ref{lemma_gamma} (below), $\zeta_{m} \rightarrow \gamma$ locally uniformly, when $m \rightarrow \infty$. 
Let $m_0$ be large enough so that for all $m \geq m_0$, 
$
\Gamma_{m}(V) \cap \mathcal{K}  = \Upsilon_{B_{m_0}(0)}^{1/m} (\zeta_{m}) \cap \mathcal{K}.
$
Choose $K = B_{m_0}(0)$ and $\epsilon\in (0,1/2)$ as in the claim. Then, by the claim, there is an $m_1>m_0$ such that for every $m>m_1$, by \eqref{claim_unbound}, 
$
d_{\mathcal{K}}(\Gamma_{m}(V),\Xi_{\mathrm{sp}}(V)) \leq u(\epsilon).
$ 
Since $\lim_{\xi\to 0}u(\xi)=0$ then (\ref{suff_convergence}) follows.
\end{proof}

To finish this step of the proof, we need to establish the convergence of the functions $\gamma_m,$ $\zeta_{m}$ and $\gamma_{m,n}$. 
\begin{lemma}\label{lem:loc-unif}
Consider the functions $\gamma_{m,n}$ and $\gamma_{m}$ defined in (\ref{the_gammas22}) and (\ref{the_gammas11}) respectively. 
Then 
$
\gamma_{m,n} \rightarrow \gamma_{m},
$
locally uniformly as $n \rightarrow \infty.$
\end{lemma}

\begin{proof}

Note that we will be using the notation $\mathrm{TV}_{[-a,a]^d}(f) =  \mathrm{TV}(f|_{[-a,a]^d})$.
Let, for $s,t \in \{0,1\}$, $i,j \leq m$ and $u \in \{V,\overline{V}, |V|^2\}$
$$
I(u,\Delta^s\varphi_j,\Delta^t\varphi_i) = \int_{\mathbb{R}^d}u(x)\sum_{p\in\Phi(s), q\in \Phi(t)}h_{i,j,p,q}(x) \, dx,
$$
where
\begin{equation}\label{h}
\begin{split}
h_{i,j,p,q}(x) &:= \left(\hat \psi_{\theta(j)_1}(x_1) \cdots \frac{\partial^{\tilde s}\hat \psi_{\theta(j)_p}}{\partial x_p^{\tilde s}} (x_p) \cdots \hat \psi_{\theta(j)_d}(x_d)\right)\\
& \qquad \qquad \qquad \qquad \times \left( \overline{ \hat \psi_{\theta(i)_1}(x_1) \cdots \frac{\partial^{\tilde t}\hat \psi_{\theta(i)_q}}{\partial x_q^{\tilde t}}(x_q) \cdots \hat \psi_{\theta(i)_d}(x_d)}\right),
\end{split}
\end{equation}
and
\begin{equation}\label{Omega}
\Phi(t) =  
\begin{cases}
\{1,\hdots,d\}, & t = 1 \\
\{1\},         & t = 0.
\end{cases}
\end{equation}
Observe that by the definition of  $\gamma_{m,n}$ and  $\gamma_{m}$  in (\ref{the_gammas22}) and (\ref{the_gammas11}) the lemma follows if we can show that 
\begin{equation}\label{key_eq}
I(u,\Delta^s\varphi_j,\Delta^t\varphi_i) -  \frac{(2n)^d}{N}\sum_{k=1}^N u^n(t_k) \sum_{p\in\Phi(s), q\in \Phi(t)} h^n_{i,j,p,q}(t_k)) \longrightarrow 0, \qquad n \rightarrow \infty,
\end{equation}
where  $N = N(n)$ is from (\ref{the_Zs}), $i,j \leq m$, $s,t \in \{0,1\}$ and $u$ is either $V,\overline{V}, |V|^2$ (recall the notation $V^a$ from (\ref{the_a})).
Note that, by the multi-dimensional Koksma--Hlawka inequality (Theorem 2.11 in \cite{Niederreiter}) it follows that 
\begin{equation}\label{intbounds44}
\begin{split}
&\left|I(u,\Delta^s\varphi_j,\Delta^t\varphi_i) -  \frac{(2n)^d}{N}\sum_{k=1}^Nu^n(t_k)  \sum_{p\in\Phi(s), q\in \Phi(t)} h^n_{i,j,p,q}(t_k)\right| \\
&\leq \left\|u  \sum_{p\in\Phi(s), q\in \Phi(t)}h_{i,j,p,q} \chi_{R(n)}\right\|_{\mathrm{L}^1} + (2n)^d\cdot \mathrm{TV}_{[-n,n]^d}\left(u  \sum_{p\in\Phi(s), q\in \Phi(t)} h_{i,j,p,q}\right) D^*_N(t_1,\hdots,t_N),
\end{split}
\end{equation}
where $R(n) = ([-n,n]^d)^c$.
To bound the first part of the right-hand side of (\ref{intbounds44}) note that 
\begin{equation}\label{g_bound}
\left\|u  \sum_{p\in\Phi(s), q\in \Phi(t)}h_{i,j,p,q}\chi_{R(n)}\right \|_{\mathrm{L}^1} \leq \|u\|_{\infty} K_{i,j}(n),
\end{equation}
where
$$
K_{i,j}(n) := \sum_{p\in\Phi(s), q\in \Phi(t)} 
\left \langle \left|\chi_{([-n,n]^d)^c}\hat \psi_{\theta(j)_1} \cdots \frac{\partial^{\tilde s}\hat \psi_{\theta(j)_p}}{\partial x_p^{\tilde s}}  \cdots \hat \psi_{\theta(j)_d}\right|, \left|\hat \psi_{\theta(i)_1} \cdots \frac{\partial^{\tilde t}\hat \psi_{\theta(i)_q}}{\partial x_q^{\tilde t}} \cdots \hat \psi_{\theta(i)_d}\right|
\right \rangle, 
$$
(recall $\theta$ from (\ref{enumeration})) 
where $\chi_{([-n,n]^d)^c}$ denotes the characteristic function on $([-n,n]^d)^c$. To bound $K_{i,j}(n)$, note that it follows by the definition of $\psi_{k,l}$ with $k,l \in \mathbb{Z}$ in (\ref{def_psi}) and some straightforward integration that for $1 \leq p \leq d$ and $(k_p, l_p) =  \theta(j)_p$ we have   
\begin{equation}\label{int_bounds11}
\left|\hat \psi_{k_p,l_p}(x_p)\right| \leq 
 \begin{cases}
 1 & \text{when} \, \,  k_p -1 \leq x_p \leq  k_p + 1,\\
 \frac{1}{|x_p - k_p| + 1} & \text {otherwise},
\end{cases}
\end{equation}
\begin{equation}\label{int_bounds12}
\left|\frac{\partial^2\hat \psi_{k_p,l_p}}{\partial x_p^2}(x_p)\right| \leq 
 \begin{cases}
 l_p^2 + l_p + \frac{1}{3} & \text{when} \, \,  k_p -1 \leq x_p \leq  k_p + 1,\\
 \frac{ l_p^2 + l_p + \frac{1}{3}}{|x_p - k_p| + 1} & \text {otherwise}.
\end{cases}
\end{equation}
Hence, if 
\begin{equation*}
\begin{split}
\tilde k = \tilde k(m,d) &:= \max\{|k_p|:(k_p, l_p) =  \theta(j)_p, p \in \{1,\hdots, d\}, j \in \{1,\hdots, m\} \},\\
\tilde l = \tilde l(m,d) &:= \max\{|l_p|:(k_p, l_p) =  \theta(j)_p, p \in \{1,\hdots, d\}, j \in \{1,\hdots, m\} \},
\end{split}
\end{equation*}
 and $n > \tilde k$, then it follows that
\begin{equation}\label{bound_K}
\begin{split}
K_{i,j}(n)&\leq  d^2 \max_{\substack{p\in\Phi(s)\\q\in \Phi(t)\\s,t \in \{0,1\}}} \left\{ \left\langle \left|\chi_{([-n,n]^d)^c}\hat \psi_{\theta(j)_1} \cdots \frac{\partial^{2 s}\hat \psi_{\theta(j)_p}}{\partial x_p^{2 s}}  \cdots \hat \psi_{\theta(j)_d}\right|,\left|\hat \psi_{\theta(i)_1} \cdots \frac{\partial^{2 t}\hat \psi_{\theta(i)_q}}{\partial x_q^{2t}} \cdots \hat \psi_{\theta(i)_d}\right| \right \rangle  \right\} \\
& \leq d^2\left(4\frac{(\max\{\tilde l^2 + \tilde l  +1/3,1\})^2}{|n-\tilde k| +1}\right)^d=:C_1(m,d,n).
\end{split}
\end{equation}
To bound the second part of the right-hand side of (\ref{intbounds44}) observe that, by Lemma \ref{bounds_on_h} we have 
\begin{equation}\label{tv_est}
\begin{split}
&(2n)^d\cdot\mathrm{TV}_{[-n,n]^d}\left(u  \sum_{p \in\Phi(s), q \in \Phi(t)} h_{i,j,p,q}\right) \\
& \quad \leq  (2n)^dd^2\max_{p \in\Phi(s), q \in \Phi(t)}\big(\|u\|_{\infty}\|h_{i,j,p,q}\|_{\infty}  + \sigma^2\mathrm{TV}_{[-n,n]^d}(u)\mathrm{TV}_{[-n,n]^d}(h_{i,j,p,q})  \\
& \qquad +  \sigma\left(\mathrm{TV}_{[-n,n]^d}(u)\|h_{i,j,p,q}\|_{\infty} + \mathrm{TV}_{[-n,n]^d}(h_{i,j,p,q})\|u\|_{\infty} \right)\big)\\
& \quad \leq (2n)^dd^2\max\left\{\|V\|_{\infty},\|V^2\|_{\infty},\mathrm{TV}_{[-n,n]^d}(V), \mathrm{TV}_{[-n,n]^d}(|V|^2)\right\}\left(1+ \sigma^2 +  2\sigma \right) C_2(m,d),
\end{split}
\end{equation}
where $\sigma = 3^d-2^{d+1} +2$ and $C_2(m,d)$ is defined in (\ref{C_2}).
Thus, by (\ref{intbounds44}), (\ref{g_bound}), (\ref{bound_K}), (\ref{tv_est}), Lemma \ref{bounds_on_h} and Theorem \ref{disc_bound} (recall that $\{t_k\}_{k\in \mathbb{N}}$ is a Halton sequence) we get 
\begin{equation}\label{intbounds202}
\begin{split}
&\left|I(u,\Delta^s\varphi_j,\Delta^t\varphi_i) -  \frac{(2n)^d}{N}\sum_{k=1}^N V^n(t_k)  \sum_{p\in\Phi(s), q \in \Phi(t)} h^n_{i,j,p,q}(t_k)\right| \\
&\leq \max\{\|V\|_{\infty}, \|V\|_{\infty}^2 \} C_1(m,d,n) + (2n)^dd^2\max\left\{\|V\|_{\infty},\|V^2\|_{\infty},\mathrm{TV}_{[-n,n]^d}(V), \mathrm{TV}_{[-n,n]^d}(|V|^2)\right\}  \\
& \qquad \qquad \times   \left(1+ \sigma^2 +  2\sigma \right) C_2(m,d)\left(\frac{d}{N} + \frac{1}{N}\prod_{k=1}^d\left(\frac{b_k-1}{2\log(b_k)}\log(N) + \frac{b_k+1}{2}\right)\right)\\
& \leq \beta(\|V\|_{\infty}, m, n),
\end{split}
\end{equation}
where the last inequality uses the bound on the total variation of $V$ from (\ref{BV_bound}) and
\begin{equation}\label{tau2}
\begin{split}
\beta(\|V\|_{\infty}, m, n) 
&:= (\|V\|_{\infty} + 1)\|V\|_{\infty} C_1(m,d,n) \\ 
& \qquad +  (2n)^dd^2\left(\|V\|^2_{\infty} + \sigma^2\phi^2(n) + 2(\sigma \|V\|_{\infty} +1)(\phi(n) + 1)\right) \\
& \qquad \times  \left(1+ \sigma^2 +  2\sigma \right) C_2(m,d) C^*(b_1,\hdots,b_d)\frac{\log(N)^d}{N}, \quad N(n)=\lceil n\phi(n)^4\rceil
\end{split}
\end{equation}
  (recall (\ref{a_and_N})) where $C^*(b_1,\hdots,b_d)$ is defined in Definition \ref{C*}.
  Finally, note that, by the definition of $C_1(m,d,n)$ and the fact that we have chosen $N(n)$ according to (\ref{tau2}), it follows that $ \beta(\|V\|_{\infty}, m, n) \rightarrow 0$
as $n \rightarrow \infty.$ Hence, (\ref{key_eq}) follows via (\ref{tau2}), and the proof is finished.
\end{proof}

\begin{lemma}\label{bounds_on_h}


For all $a > 0$, $i,j \leq n_2$ and $m,n \leq d$:
\begin{itemize}
\item[(i)]
$
\mathrm{TV}(h^a_{i,j,m,n})= \mathrm{TV}_{[-a,a]^d}(h_{i,j,m,n}) \leq C_2(m,d),
$
\item[(ii)]
$
\|h^a_{i,j,m,n}\|_{\infty} \leq C_2(m,d),
$
\item[(iii)] for $u \in \mathrm{BV}_{\mathrm{loc}}(\mathbb{R}^d)$ and $\sigma = 3^d-2^{d+1} +2$ we have that 
\begin{equation*}
\begin{split}
\mathrm{TV}(u^a h^a_{i,j,p,q})= \mathrm{TV}_{[-a,a]^d}(u h_{i,j,p,q}) &\leq  \|u\|_{\infty}\|h_{i,j,p,q}\|_{\infty}  + \sigma^2\mathrm{TV}_{[-a,a]^d}(u)\mathrm{TV}_{[-a,a]^d}(h_{i,j,p,q})  \\
& \quad+  \sigma\left(\mathrm{TV}_{[-a,a]^d}(u)\|h_{i,j,p,q}\|_{\infty} + \mathrm{TV}_{[-a,a]^d}(h_{i,j,p,q})\|u\|_{\infty} \right),
\end{split}
\end{equation*}
\item[(iv)] $\mathrm{TV}_{[-a,a]^d}(|g|^2) \leq \|g\|^2_{\infty} + \sigma^2\mathrm{TV}^2_{[-a,a]^d}(g) + 2\sigma \|g\|_{\infty}  \mathrm{TV}_{[-a,a]^d}(g)$
\end{itemize}
where
\begin{equation*}
\begin{split}
C_2(m,d) &:=  2^d \left(2((\tilde l+1)^4 + \tilde l^4)^2(2(\tilde k+1) + 2)\right)^d,\\
\end{split}
\end{equation*}
and $\tilde k, \tilde l$ are defined in (\ref{k_and_l}).
\end{lemma}

\begin{proof}
To prove both (i) and (ii) we will use the easy facts that 
$\mathrm{TV}(h^a_{i,j,p,q}) = \mathrm{TV}_{[-a,a]^d}(h_{i,j,p,q})$ and $\mathrm{TV}(g^a h^a_{i,j,p,q}) = \mathrm{TV}_{[-a,a]^d}(g h_{i,j,p,q})$.
To prove (i) of the claim let us first recall (see for example \cite{Niederreiter}, p. 19) that when $\psi \in C^1([-a,a]^d)$ then
\begin{equation}\label{VHK1}
\mathrm{TV}_{[-a,a]^d}(\psi) = \sum_{k=1}^d \sum_{1\leq i_1< \hdots < i_k \leq d}V^{(k)}(\psi; i_1,\hdots,i_k), 
\end{equation}
where $V^{(k)}(\psi; i_1,\hdots,i_k) = V^{(k)}(\psi_{i_1,\hdots,i_k})$
and
$$
\psi_{i_1,\hdots,i_k}: (y_1, \hdots, y_k) \mapsto \psi(\tilde y_1, \hdots , \tilde y_d), \qquad \tilde y_j = a, \, j \neq i_1,\hdots,i_k, \quad \tilde y_{i_j} = y_j, 
$$
$$
V^{(k)}(\varphi) = \int_{-a}^a \cdots \int_{-a}^a \left |\frac{\partial^k \varphi}{\partial x_1\cdots \partial x_k}\right |\, dx_1 \hdots dx_k, \qquad \varphi \in C^1([-a,a]^k).
$$
Note that from (\ref{h}) and (\ref{def_psi}) it follows that $h^a_{i,j,p,q} \in  C^{\infty}([0,1]^d)$, so by the definition of $h$ in (\ref{h}) we have that, for $k \in \{1,\hdots,d\}$ and $1\leq i_1< \hdots < i_k \leq d,$ 
\begin{equation}\label{VHK2}
\begin{split}
&V^{(k)}(h^a_{i,j,p,q}; i_1,\hdots,i_k) \\&
\qquad \leq \prod_{\mu=1}^d \max\Bigg[\max_{s,t = 0,2}\int_{-a}^a
\left|\frac{\partial}{\partial x_{\mu}}\left(\frac{\partial^s \hat \psi_{\theta(j)_{\mu}}}{\partial x_{\mu}^s}(x_{\mu})  \overline{\frac{\partial^{t}\hat \psi_{\theta(i)_{\mu}}}{\partial x_{\mu}^{t}}(x_{\mu})}  \right)\right|\, dx_{\mu},\\
& \qquad \qquad \qquad \qquad \max_{\stackrel{s,t = 0,2}{x_{\mu} \in [-a,a]}}
\left |
\frac{\partial^s \hat \psi_{\theta(j)_{\mu}}}{\partial x_{\mu}^s}(x_{\mu})  \overline{\frac{\partial^{t}\hat \psi_{\theta(i)_{\mu}}}{\partial x_{\mu}^{t}}(x_{\mu})}\right|
\Bigg], \quad \forall k,p, q \leq d.
\end{split}
\end{equation}
We will now focus on bounding the right-hand side of (\ref{VHK2}). Note that by using the definition of $\psi_{k,l}$ with $k,l \in \mathbb{Z}$ in (\ref{def_psi}) and some straightforward integration it follows that for $1 \leq p \leq d$ and $(k_p, l_p) =  \theta(j)_p$ we have 
\begin{equation}\label{int_bounds21}
\left|\frac{\partial \hat \psi_{k_p,l_p}}{\partial x_p}(x_p)\right| \leq 
 \begin{cases}
 l_p + \frac{1}{2} & \text{when} \, \,  k_p -1 \leq x_p \leq  k_p + 1,\\
 \frac{l + \frac{1}{2} }{|x_p - k_p| + 1} & \text {otherwise},
\end{cases}
\end{equation}
\begin{equation}\label{int_bounds22}
\left|\frac{\partial^3\hat \psi_{k_p,l_p}}{\partial x_p^3}(x_p)\right| \leq 
 \begin{cases}
 \frac{(l_p+1)^4 - l^4_p}{4} & \text{when} \, \,  k_p -1 \leq x_p \leq  k_p + 1,\\
 \frac{(l_p+1)^4 - l^4_p}{4(|x_p - k_p| + 1)} & \text {otherwise}.
\end{cases}
\end{equation}
Thus, by using (\ref{int_bounds11}), (\ref{int_bounds12}), (\ref{int_bounds21}) and (\ref{int_bounds22}) it follows that 
\begin{equation}\label{VHK3}
\begin{split}
&\max_{s,t = 0,2}\int_{-a}^a
\left|\frac{\partial}{\partial x_{\mu}}\left(\frac{\partial^s \hat \psi_{\theta(j)_{\mu}}}{\partial x_{\mu}^s}(x_{\mu})  \overline{\frac{\partial^{t}\hat \psi_{\theta(i)_{\mu}}}{\partial x_{\mu}^{t}}(x_{\mu})}  \right)\right|\, dx_{\mu} \\
& \qquad \qquad \qquad \qquad \leq 
 2\max_{s,t = 0,1,2,3}\int_{-\infty}^{\infty}
\left|\frac{\partial^s \hat \psi_{\theta(j)_{\mu}}}{\partial x_{\mu}^s}(x_{\mu})  \overline{\frac{\partial^{t}\hat \psi_{\theta(i)_{\mu}}}{\partial x_{\mu}^{t}}(x_{\mu})} \right|\, dx_{\mu}\\ 
& \qquad \qquad \qquad \qquad \leq 2((\tilde l+1)^4 + \tilde l^4)^2\left (2(\tilde k+1) + \int_{[-\infty,-1] \cup [1,\infty]} \frac{1}{y^2} \, dy   \right)\\
& \qquad \qquad \qquad \qquad =  2((\tilde l+1)^4 + \tilde l^4)^2\left (2(\tilde k+1) + 2\right),
\end{split}
\end{equation}
where $\tilde k := \max\{|k_p|:(k_p, l_p) =  \theta(j)_p, p \in \{1,\hdots, d\}, j \in \{1,\hdots, n\} \}$, 
$\tilde l := \max\{|l_p|:(k_p, l_p) =  \theta(j)_p, p \in \{1,\hdots, d\}, j \in \{1,\hdots, n\} \}$.
Moreover, by (\ref{int_bounds11}) and (\ref{int_bounds12})
\begin{equation}\label{VHK4}
 \max_{\stackrel{s,t = 0,2}{x_{\mu} \in [-a,a]}}
\left |
\frac{\partial^s \hat \psi_{\theta(j)_{\mu}}}{\partial x_{\mu}^s}(x_{\mu})  \overline{\frac{\partial^{t}\hat \psi_{\theta(i)_{\mu}}}{\partial x_{\mu}^{t}}(x_{\mu})}\right| \leq \max\{\tilde l^2 + \tilde l  +1/3,1\}, \qquad  i,j \leq m, \quad 1 \leq \mu \leq d.
\end{equation}
Hence, from (\ref{VHK2}), (\ref{VHK3}) and (\ref{VHK4}) it follows that for $k \in \{1,\hdots,d\}$ and $1\leq i_1< \hdots < i_k \leq d,$ 
$$
V^{(k)}(h^a_{i,j,p,q}; i_1,\hdots,i_k) \leq \left(2((\tilde l+1)^4 + \tilde l^4)^2(2(\tilde k+1) + 2)\right)^d
$$
and thus, by (\ref{VHK1}) we get that 
\begin{equation*}
\begin{split}
\mathrm{TV}_{[-a,a]^d}(h_{i,j,p,q}) &\leq  \left(2((\tilde l+1)^4 + \tilde l^4)^2(2(\tilde k+1) + 2)\right)^d \sum_{k=1}^d \binom{d}{k}\\
& \leq  2^d \left(2((\tilde l+1)^4 + \tilde l^4)^2(2(\tilde k+1) + 2)\right)^d,
\end{split}
\end{equation*}
and thus we have proved (i) in the claim.

To prove (ii) in the claim, we observe that by (\ref{def_psi}), (\ref{h})  and (\ref{VHK4}) it follows that  
$$
\|h^a_{i,j,p,q}\|_{\infty} \leq \prod_{\mu=1}^d  \max_{\stackrel{s,t = 0,2}{x_{\mu} \in [-\infty,\infty]}}
\left |
\frac{\partial^s \hat \psi_{\theta(j)_{\mu}}}{\partial x_{\mu}^s}(x_{\mu})  \overline{\frac{\partial^{t}\hat \psi_{\theta(i)_{\mu}}}{\partial x_{\mu}^{t}}(x_{\mu})}\right| \leq  \left(\max\{\tilde l^2 + \tilde l  +1/3,1\}\right)^d, 
$$
for $i,j \leq m$ and $p,q \leq d$. Obviously, the last part of the above inequality is bounded by $C_2(m,d),$
which yields the assertion.

To prove (iii) and (iv) we will use the fact (see \cite{Blumlinger}) that 
$$
\mathcal{A} = \{f \in \mathcal{M}([-a,a]^d): \|f\|_{\infty} + \mathrm{TV}_{[-a,a]^d}(f) < \infty\},
$$
where $\mathcal{M}([-a,a]^d)$ denotes the set of measurable functions on $[-a,a]^d$,
is a Banach algebra when $\mathcal{A}$ is equipped with the norm
$\|f\|_{\mathcal{A}} =  \|f\|_{\infty} + \sigma \mathrm{TV}_{[-a,a]^d}(f),$ where $\sigma > 3^d-2^{d+1} +1.$
We will let $\sigma = 3^d-2^{d+1} +2$. Hence, we get, by the Banach algebra property of the norm and (i) and (ii) that we already have proved, that 
\begin{equation*}
\begin{split}
\mathrm{TV}_{[-a,a]^d}(u h_{i,j,p,q}) &\leq \|u\|_{\infty}\|h_{i,j,p,q}\|_{\infty}  + \sigma^2\mathrm{TV}_{[-a,a]^d}(u)\mathrm{TV}_{[-a,a]^d}(h_{i,j,p,q})\\ 
& \quad +  \sigma\left(\mathrm{TV}_{[-a,a]^d}(u)\|h_{i,j,p,q}\|_{\infty} + \mathrm{TV}_{[-a,a]^d}(h_{i,j,p,q})\|u\|_{\infty} \right), \quad u \in \mathcal{A},
\end{split}
\end{equation*}
 finally proving (iii). The proof of (iv) is almost identical.
 \end{proof}

\begin{lemma}\label{lemma_gamma}
Recall $\zeta_m$ defined in (\ref{totalrecall}). Then,
$\zeta_m  \rightarrow \gamma$ locally uniformly, where $\gamma$ is defined in (\ref{the_gamma}).
Furthermore, if $m\geq \left\|V\right\|_{\infty}$ then we have
$$
\zeta_m(z)\geq \gamma_m(z)-\frac{2+\left|z\right|}{m},
$$
where  $\gamma_{m}$ is defined in (\ref{the_gammas11}).
\end{lemma}

\begin{proof}
Observe that $\gamma_{m} \rightarrow \gamma$ locally uniformly as $m \rightarrow \infty$. Indeed, let 
$
\mathcal{T} = \{\|(-\Delta + V + zI)\psi\| : \psi \in \mathrm{W}^{2,2}(\mathbb{R}^d), \|\psi\| = 1\}.
$
Then, since $\mathcal{S}$ is a core for $H$ (recall $\mathcal{S}$ from Step I of the proof of $\SCI(\Xi_{\mathrm{sp}},\Omega_{\phi,g})_{\mathrm{A}} = 1$) then every element in $\mathcal{T}$ can be approximated arbitrarily well by $\|(-\Delta + V + zI)\tilde \varphi\|$ for some $\tilde \varphi \in \mathcal{S}$, thus it follows from (\ref{the_gamma}) that we have convergence. Note that the convergence must be monotonically from above by the definition of $P_{m}$, and thus Dini's Theorem assures the locally uniform convergence.
Thus, it suffices to show that $|\zeta_m - \gamma_m| \rightarrow 0$ locally uniformly as $m \rightarrow \infty$.  

Note that if we define, for $z \in \mathbb{C}$, the operator matrices
\begin{equation}\label{the_ZsBBB}
\begin{split}
Z_{m}(z)_{ij} &= \langle S_{m}(V,z) \varphi_j, S_{m}(V,z) \varphi_i\rangle_{n,N}, \quad i,j \leq m, \\
\tilde Z_{m}(z)_{ij} &= \langle \tilde S_{m}(V,z) \varphi_j, \widetilde S_{m}(V,z) \varphi_i\rangle_{n,N}, \quad i,j \leq m, 
\quad N = \lceil n\phi(n)^4\rceil,
\end{split}
\end{equation}
where $n = n(m)$ is defined in (\ref{the_n})
and 
$$
W_{m}(z)_{ij} = \langle S_{m}(V,z) \varphi_j, S_{m}(V,z) \varphi_i\rangle, \quad i,j \leq m, 
$$ 
$$
\tilde W_{m}(z)_{ij} = \langle \tilde S_{m}(V,z) \varphi_j, \widetilde S_{m}(V,z) \varphi_i\rangle, \quad i,j \leq m,
$$ 
the desired convergence follows if we can show that $\|Z_{m}(z) - W_{m}(z)\|$ and $\|\tilde Z_{m}(z) - \tilde W_{m}(z)\|$ tend to zero as $m$ tends to infinity for all $z$ in some compact set. However, this follows by the choice of 
$n(m) = \min\{n: \tilde \tau(m,n) \leq \frac{1}{m^3}\}$ in (\ref{the_n}). In particular,  $\beta(m,m,n) = \tilde\beta(m,n)\leq\tilde\tau(m,n) $ and clearly $\beta(\|V\|_{\infty}, m, n) \leq \beta(m,m,n)$ for $\|V\|_{\infty} \leq m$ (recall $\beta$ from (\ref{tau2})). We also have
\begin{equation}\label{the_S}
\begin{split}
\langle  S_{m}(V,z)\varphi_j,  S_{m}(V,z)\varphi_i\rangle_{n,N} =& 
\langle  \Delta \varphi_j,  \Delta \varphi_i\rangle_{n,N} - \langle
V\varphi_j,  \Delta \varphi_i\rangle_{n,N} - \langle  \Delta \varphi_j,
V\varphi_i\rangle_{n,N}\\  &+ \langle  V\varphi_j,  V\varphi_i\rangle_{n,N} - 2\Re(z)\langle \Delta\varphi_j,
\varphi_i\rangle_{n,N} \\&+ \langle 2\Re(z\overline{V})\varphi_j,
\varphi_i\rangle_{n,N}  +
|z|^2\langle  \varphi_j,  \varphi_i\rangle_{n,N}.
\end{split}
\end{equation}
Thus it follows immediately by (\ref{intbounds202}) that
\begin{align*}
\max\big\{\left|Z_{m}(z)_{ij}-W_{m}(z)_{ij}\right|,|\tilde Z_{m}(z)_{ij}-\tilde W_{m}(z)_{ij}|\big\}&\leq (4(\left|z\right|+1)+\left|z\right|^2)\beta(\|V\|_{\infty}, m, n)\\
&\leq \frac{4(\left|z\right|+1)+\left|z\right|^2}{m^3}.
\end{align*}
Using the fact that the operator norm of a matrix is bounded by its Frobenius norm $\|\cdot\|_F$, it follows that for $z \in K \subset \mathbb{C}$, where $K$ is compact,  $\|Z_{m}(z) - W_{m}(z)\|_F = \mathcal{O}(\frac{1}{m^2})$ and $\|\tilde Z_{m}(z) - \tilde W_{m}(z)\|_F = \mathcal{O}(\frac{1}{m^2})$ for sufficiently large $m$. To see the explicit bound, note that the above shows for $\|V\|_{\infty} \leq m$ that
$$
\gamma_m(z)^2\leq \frac{4(\left|z\right|+1)+\left|z\right|^2}{m^2} + \zeta_m(z)^2\leq \left(\zeta_m(z)+ \frac{\sqrt{4(\left|z\right|+1)+\left|z\right|^2}}{m} \right)^2
$$
Taking square roots and re-arranging gives the result.
\end{proof}

\subsubsection{\bf Proof of the $\in \Sigma_1^A$ and $\in \Sigma_1^{A,\mathrm{eigv}}$ classifications in Theorem \ref{main_self_adjoint}}


In order to show the $\Sigma_1^A$ classification for $\{\Xi_{\mathrm{sp}},\Omega_{\phi,g}\}$, consider $\hat\Gamma_m(A)=\Gamma_{m+\left\lceil \left\|V\right\|_{\infty}\right\rceil }(V)$ where we now use the fact that an upper bound on $\left\|V\right\|_{\infty}$ is included in the evaluation functions. From Lemma \ref{lemma_gamma}, if $z\in\hat\Gamma_m(A)$ then
$$
\mathrm{dist}(z,\mathrm{sp}(-\Delta+V))\leq g^{-1}\left(\zeta_{m+\left\lceil \left\|V\right\|_{\infty}\right\rceil}(z)+\frac{2+\left|z\right|}{m}\right).
$$
This can be approximated from above to within an error that converges to zero as $m\rightarrow\infty$ using finitely many evaluations of the function $g$ at rational points. Taking the maximum over all $z\in \hat\Gamma_m(A)$ gives us an error bound which converges to $0$ uniformly on compact subsets of $\mathbb{C}$ as $m\rightarrow\infty$. The following shows this is enough for the $\Sigma_1^A$ error control.

\begin{lemma}
\label{sigma1_unbounded}
Let $\Xi:\Omega\rightarrow(\mathcal{C}(\mathbb{C}),d_{\mathrm{AW}})$ be a problem function and suppose that there is an arithmetic tower of algorithms $\{\Gamma_m\}$ for $\Xi$. Suppose also that there exists a function $E_m:\Gamma_m(A)\mapsto \mathbb{R}_{\geq 0}$ (which may depend on $A$) computed along with each $\Gamma_m$ (using finitely many arithmetic operations and comparisons) and converging uniformly to zero on compact subsets, such that
$$
\mathrm{dist}(z,\Xi(A))\leq E_m(z),\quad\forall z\in\Gamma_m(A).
$$ 
Suppose also that $\Gamma_m(A)$ is finite for each $m$ and $A$. Then we can compute in finitely many arithmetic operations and comparisons a sequence of non-negative numbers $b_m\rightarrow 0$ such that
$
\Gamma_m(A)\subset A_m
$ for some $A_m\in\mathcal{C}(\mathbb{C})$ with $d_{\mathrm{AW}}(A_m,\Xi(A))\leq b_m$. Hence, by taking subsequences if necessary, we can build an arithmetic $\Sigma_1^A$ tower for $\{\Xi,\Omega\}$.
\end{lemma}

\begin{proof}
Let $a^n_m=\sup\{E_m(z):z\in\Gamma_m(A)\cap B_n(0)\}$. Define
$$
A_m^n=\big((\Xi(A)+B_{a^n_m}(0))\cap B_n(0)\big)\cup(\Gamma_m(A)\cap\{z:\left|z\right|\geq n\}).
$$
It is clear that $\Gamma_m(A)\subset A_m^n$ and given $\{\Gamma_m(A),E_m(A)\}$ (we assume $\Gamma_m(A)\neq\emptyset$), we can easily compute a lower bound $n_1$ such that $\Xi(A)\cap B_{n_1}(0)\neq\emptyset$. Compute this from $\Gamma_1(A)$ and then fix it. Suppose that $n\geq 4n_1$, and suppose that $\left|z\right|<\left\lfloor n/4\right\rfloor$. Then the points in $A_m^n$ and $\Xi(A)$ nearest to $z$ must lie in $B_n(0)$ and hence
$
\mathrm{dist}(z,A_m^n)\leq \mathrm{dist}(z,\Xi(A))
$
and 
$
 \mathrm{dist}(z,\Xi(A))\leq \mathrm{dist}(z,A_m^n)+a^n_m.
$
It follows that $$d_{\mathrm{AW}}(A_m^n,\Xi(A))\leq a^n_m + 2^{-\left\lfloor n/4\right\rfloor}.$$
We now choose a sequence $n(m)$ such that setting $A_m=A_m^{n(m)}$ and $b_m=a^{n(m)}_n + 2^{-\left\lfloor n(m)/4\right\rfloor}$ proves the result. Clearly it is enough to ensure that $b_m$ is null. If $m<4n_1$ then set $n(m)=4n_1$, otherwise consider $4n_1\leq k\leq m$. If such a $k$ exists with $a_m^k\leq 2^{-k}$ then let $n(m)$ be the maximal such $k$ and finally if no such $k$ exists then set $n(m)=4n_1$. For a fixed $n$, $a^n_m\rightarrow 0$ as $m\rightarrow\infty$. It follows that for large $m$, we must have $a_m^{n(m)}\leq 2^{-n(m)}$ and that $n(m)\rightarrow\infty$.
\end{proof}

Finally, we extend the argument of \S \ref{spec_proof} for the approximate eigenvectors.

\begin{proof}[Proof that $\{\Xi_{\mathrm{sp}},\Omega_{\phi,\mathrm{SA}}\}\in  \Sigma^{A,\mathrm{eigv}}_1$]
We need only argue for the approximate eigenvectors and we sketch the proof, since it is a simple adaptation of the discrete case considered in \S \ref{spec_proof}. Consider a Schr\"odinger operator in $\Omega_{\phi,\mathrm{SA}}$ with potential $V$ and $z\in\hat\Gamma_m(V)$, where $\hat\Gamma_m$ is the constructed $\Sigma_1^A$ tower for $\Omega_{\phi,g}$. By taking subsequences if necessary, it suffices to show that we can compute a vector $\psi_m\in\mathbb{C}^m$ such that, for a given $\delta\in\mathbb{Q}_{>0}$ with $\delta<1$,
\begin{equation}
\label{approx_e_SCHRODINGER}
\langle Z_m(z)\psi_m,\psi_m\rangle \leq \sqrt{\sigma_1(Z_m(z))}+\delta,\quad 1-\delta<\|\psi_m\|<1,
\end{equation}
where $Z_m(z)$ is the Hermitian positive (semi-)definite matrix defined via (\ref{the_ZsBBB}). The vector $\psi_m$ will then correspond to the first $m$ coefficients with respect to the Gabor basis. To see why this is sufficient, note that if $T$ denotes the infinite matrix corresponding to $-\Delta+V-zI$ (with respect to the Gabor basis) and $P_m$ denotes the projection onto the span of the first $m$ basis functions, then (\ref{approx_e_SCHRODINGER}) implies that
$$
\|TP_m\psi_m\|^2=\langle T^*T\psi_m,\psi_m\rangle=\langle Z_m(z)\psi_m,\psi_m\rangle
$$
and that $\sqrt{\sigma_1(Z_m(z))}$ is bounded above by a computable null sequence since $z\in\hat\Gamma_m(V)$. We can then adapt the proof of $\{\Xi_{\mathrm{sp}},\Omega_{f}\cap \Omega_{\mathrm{N}}\}\in\Sigma_1^{A,\mathrm{eigv}}$, in \S \ref{spec_proof} with suitable approximations of $Z_m(z)$ (which can be computed with error control using $\Delta_1$ information by the above arguments) replacing the matrix $(P_{f(n)}\tilde{A}P_n)^*(P_{f(n)}\tilde{A}P_n)$.
\end{proof} 

\subsubsection{\bf Proof of the $\in \Pi_2^A$ classification in Theorem \ref{main_self_adjoint}}

Note that its is clear that none of the problems lie in $\Delta_1^G$. Hence to finish the proof of Theorem \ref{main_self_adjoint}, we must show that $\{\Xi_{\mathrm{sp},\epsilon},\Omega_{\phi}\}\in\Sigma_1^A$ since by taking $\epsilon\downarrow 0$ this will show $\{\Xi_{\mathrm{sp}},\Omega_{\phi}\}\in\Pi_2^A$ since we have $\Omega_{\phi,g}\subset\Omega_{\phi}$. Note that through the use of $\zeta_m$ and Lemma \ref{lemma_gamma} we can compute, using finitely many arithmetic operations and comparisons for any $z$, a function $\hat \gamma_m(z)$ that converges uniformly to $\gamma(z)$ from \eqref{the_gamma} on any compact subset of $\mathbb{C}$ with $\hat\gamma_{m}(z)\geq\gamma(z)$. The next Lemma then says that this is enough.

\begin{lemma}
\label{pseudospec_unbounded}
Suppose that $\hat\gamma_m(z)\geq\gamma(z)$ converge uniformly to $\left\|(-\Delta+V-zI)^{-1}\right\|^{-1}$ as $m \rightarrow \infty$ on compact subsets of $\mathbb{C}$. Set
$$
{\Gamma}_m(V)=(B_m(0)\cap\frac{1}{m}(\mathbb{Z}+i\mathbb{Z}))\cap\{z:\hat\gamma_m(z)<\epsilon\}.
$$
For large $m$, $\Gamma_m(V)\neq\emptyset$ so we can assume this without loss of generality. Also, $d_{\mathrm{AW}}({\Gamma}_m(V),\mathrm{sp}_{\epsilon}(-\Delta+V))\rightarrow 0$ as $m\rightarrow\infty$ and clearly $\Gamma_m(V)\subset\mathrm{sp}_{\epsilon}(-\Delta+V)$.
\end{lemma}

\begin{proof}
Since the pseudospectrum is non-empty, for large $m$, $\Gamma_m(V)\neq\emptyset$ so we may assume that this holds for all $m$ without loss of generality. We use the characterisation of the Attouch--Wets topology where it is enough to consider closed balls. Suppose that $n$ is large such that $B_n(0)\cap\mathrm{sp}_{\epsilon}(-\Delta+V)\neq\emptyset$. Since $\Gamma_m(V)\subset\mathrm{sp}_{\epsilon}(-\Delta+V)$, we must show that given $\delta>0$, there exists $N_1$ such that if $m>N_1$ then $\mathrm{sp}_{\epsilon}(-\Delta+V)\cap B_n(0)\subset{\Gamma_m(V)+B_{\delta}(0)}$. Suppose for a contradiction that this were false. Then there exists $z_j\in\mathrm{sp}_{\epsilon}(-\Delta+V)\cap B_n(0)$, $\delta>0$ and $m_j\rightarrow\infty$ such that $\mathrm{dist}(z_j,\Gamma_{m_j}(V))\geq \delta$. Without loss of generality, we can assume that $z_j\rightarrow z\in\mathrm{sp}_{\epsilon}(-\Delta+V)$. There exists some $w$ with $\left\|(-\Delta+V-wI)^{-1}\right\|^{-1}<\epsilon$ and $\left|z-w\right|\leq \delta/2$. Assuming $m_j>n+\delta$, there exists $y_{m_j}\in (B_{m_j}(0)\cap\frac{1}{{m_j}}(\mathbb{Z}+i\mathbb{Z}))$ with $\left|y_{m_j}-w\right|\leq 1/{m_j}$. It follows that
$$
\hat\gamma_{m_j}(y_{m_j})\leq \left|\hat\gamma_{m_j}(y_{m_j})-\gamma(y_{m_j})\right|+\left|\gamma(w)-\gamma(y_{m_j})\right|+\left\|(-\Delta+V-wI)^{-1}\right\|^{-1}.
$$
But $\gamma$ is continuous and $\hat\gamma_{m_j}$ converges uniformly to $\gamma$ on compact subsets. Hence for large $m_j$, $\hat\gamma_{m_j}(y_{m_j})<\epsilon$ so that $y_{m_j}\in\Gamma_{m_j}(V)$. But
$
\left|y_{m_j}-z\right|\leq \left|z-w\right|+\left|y_{m_j}-w\right|\leq \delta/2 + 1/{m_j}
$
which is smaller than $\delta$ for large $m_j$. This gives the required contradiction.
\end{proof}

\subsection{The case of unbounded potential $V$: The proof of Theorem \ref{thm:comp-res}}

In this section we prove Theorem \ref{thm:comp-res} on the SCI of spectra and pseudospectra of Schr\"odinger operators with unbounded potentials. First of all, we will build the $\Delta_2^A$ algorithms. Let us outline the steps of the proof first:
\begin{itemize}
\item[a.]
\emph{Compactness of the resolvent:} The assumptions on the potential imply that the operator $H$ has a compact resolvent $R(H,z)$ (see Proposition \ref{tarkan kompaktsuus}). Therefore the spectrum is countable, consisting of eigenvalues with finite-dimensional invariant subspaces.
\item[b.]
\emph{Finite-dimensional approximations:} The main part of the proof centres around showing that it is possible to construct, with finitely many evaluations of $V$, square matrices $\widetilde H_n$ whose resolvents (when suitably embedded into the large space) converge to $R(H,z_0)$ in norm  at a suitable point $z_0$  (see Theorem \ref{strong}). Note that this technique is very different from the techniques used so far in the paper and is only possible due to compactness.
\item[c.]
\emph{Convergence of the spectrum and pseudospectrum:} We use the convergence at $z_0$ to show convergence at other points $z$ in the resolvent set.
\end{itemize}
Once this is done, we prove that neither problem lies in $\Sigma_1^G\cup\Pi_1^G$.

As the argument is otherwise independent of the particular set-up, we start with a general discussion.  In the end, we demonstrate the construction of the matrices $\widetilde H_n$ and the convergence of the resolvents.  
We assume the following:

{\bf (i)  Assumptions on the operator $A$:
} Given a closed densely defined operator $A$ in a Hilbert space $\mathcal H$  such that at $z_0\in \mathbb C$ the resolvent operator
$R(z_0)=(A-z_0)^{-1}$ is compact  $R(z_0) \in \mathcal K (\mathcal H)$.  Thus  sp$(A)=\{ \lambda_j\}$,  the spectrum of $A$,  is at most countable  with no finite accumulation points.

 {\bf  (ii) Assumptions on the approximations $A_n$:}
 Suppose $A_n$ is a finite rank approximation to $A$  such that  if $E_n$ is the  orthogonal projection onto the range of $A_n$, then $A_n =A_nE_n$.   We put further $\mathcal H_n= E_n\mathcal H$ and  denote by $\widetilde A_n$ the matrix representing $A_n$ when restricted to the invariant subspace $\mathcal H_n$ w.r.t. some orthonormal basis.
Now, take the resolvent $(A_nE_n-zE_n)^{-1}$ of this restriction, extend it to $\mathcal H_n^\perp$ by zero, and denote this extension by $R_n(z)$.
Then $R_n(z)=R_n(z)E_n$, and $R_n(z)=(A_n-z)^{-1} +(I-E_n)z^{-1}$ for all $z\neq 0$ for which the inverse exists.
Finally we assume that $R_n(z_0)$ exist and
\begin{equation}\label{resolventtiapproksimaatio}
\lim_{n\rightarrow\infty}\| R_n(z_0) -R(z_0)\| =0.
\end{equation}

\subsubsection{\bf{Convergence of the spectrum and pseudospectrum.}}
The first step is to conclude that if the finite rank approximations to the resolvent converge in operator norm at one point $z_0$, then they also converge locally uniformly away from the spectrum of $A$.
To that end   denote by $U_r(\mu)$ the open disc at centre $\mu$ and radius $r$.

\begin{proposition}\label{resolventtienkovnvergenssi}
Suppose  $R(z)$ and $R_n(z)$ are as above and satisfy  \eqref{resolventtiapproksimaatio}.
Let $\mathcal K \subset \mathbb C$ be compact, $r>0$ and  define
$
\mathcal K_r = \mathcal K  \setminus  \bigcup_j  U_r(\lambda_j).
$
Then for large enough $n$, $R_n(z)$ exists for all $z \in \mathcal K_r$ and 
$
\sup_{z\in \mathcal K_r}\| R_n(z) -R(z)\| \rightarrow 0
$ 
as $n \rightarrow \infty$.
\end{proposition}
\begin{proof}
Clearly
$
R(z)=R(z_0) (I-(z-z_0)R(z_0))^{-1}
$
and
$
R_n(z)=R_n(z_0) (I-(z-z_0)R_n(z_0))^{-1}
$
for all $z$ in which $R(z)$, resp. $R_n(z)$, exist.
By (\ref{resolventtiapproksimaatio})  it suffices to prove the existence of $R_n(z)$ and
$$
\sup_{z\in \mathcal K_r}\| (I-(z-z_0)R_n(z_0))^{-1} -(I-(z-z_0)R(z_0))^{-1}\| \rightarrow 0.
$$
However, we know that $(I-(z-z_0)R(z_0))^{-1}$ is meromorphic  in the whole plane and hence analytic in the compact set $\mathcal K_r$  and in particular uniformly bounded.   But this means that it is sufficient to show that the inverses converge, which in turn is immediate from (\ref{resolventtiapproksimaatio}) since
$$
\sup_{z\in \mathcal K_r}\| (I-(z-z_0)R_n(z_0)) -(I-(z-z_0)R(z_0))\| \leq \| R_n(z_0) -R(z_0)\|+\sup_{z\in \mathcal K_r} |z-z_0| \ \ \|R_n(z_0)-R(z_0)\|.
$$
To see that this suffices,  write $T_n(z)= (I-(z-z_0)R_n(z_0))$, $T(z)=(I-(z-z_0)R(z_0))$ and
$
T_n(z)= T(z)[I+ T(z)^{-1}(T_n(z)-T(z))].
$
Then for large enough $n$ and $z\in \mathcal K_r$ by a Neumann series argument
$$
\|T_n(z)^{-1} -T(z)^{-1}\| \le \|T(z)^{-1}\| \  [(1-\|T(z)^{-1}\| \|T_n(z) - T(z)\|)^{-1} -1].
$$
\end{proof}

\begin{proposition}\label{spektrinkonvergenssi}
Let $\mathcal K \subset \mathbb C$ be compact and $\delta>0$.  Then, for  all large enough $n$,
$
{\rm sp}(A)\cap \mathcal K  \subset \mathcal N_\delta({\rm sp}(A_n))
$ 
and
$
{\rm sp}(A_n)\cap \mathcal K  \subset \mathcal N_\delta({\rm sp}(A)).
$
\end{proposition}

\begin{proof} Since the eigenvalues are precisely the  poles of the resolvents, the claim follows immediately from the previous proposition.
\end{proof}

The last proposition gives the convergence of the spectra.  
The discussion on pseudospectra is somewhat more involved.   We need to know that the norm of the resolvent is not constant in any open set. The following is a theorem due to J. Globevnik, E.B. Davies and E. Shargorodsky which we formulate here as a lemma:

\begin{lemma}[\cite{Globevnik} and \cite{Davies_Sharg}]\label{eivakio}
Suppose   $A$ is a closed  densely defined operator in $\mathcal H$ such that the resolvent $R(z)=(A-z)^{-1}$ is compact.   Let $\Omega \subset \mathbb C$ be open and connected,  and assume that,  for all $z\in \Omega$,
$
\|R(z)\| \le M .
$
Then, for all $z\in \Omega$,
$
\|R(z)\| <M.
$
This is particularly true if $\mathcal H$ is finite-dimensional.
\end{lemma} 

The theorem  in \cite{Davies_Sharg} is  formulated for Banach spaces $X$ with the extra assumption that   $X$ or its dual are complex strictly convex, a condition which holds for Hilbert spaces. The case $\mathcal H$ being of finite dimension is already settled by \cite{Globevnik}.
We put
$
\gamma(z)=1/\|R(z)\|
$
and
$
\gamma_n(z)=1/\|R_n(z)\|
$
and summarise the properties of $\gamma$ and $\gamma_n$ as follows:

\begin{lemma}\label{ominaisuudet}
If (i) and (ii)  hold, then  
$
\gamma_n(z) \rightarrow \gamma(z)
$
uniformly on compact sets. 
Neither $\gamma$, nor $\gamma_n$ is constant in any open set and they have local minima only where they vanish. Additionally, 
$
\gamma(z) \le {\rm dist}(z,{\rm sp}(A)).
$
Consequently, 
\[{\rm sp}_\epsilon(A)= \{z \ : \ \gamma(z) \le \epsilon\}= \mathrm{cl}\{z \ : \ \gamma(z) < \epsilon\},\quad
{\rm sp}_\epsilon(A_n)= \{z \ : \ \gamma_n(z) \le \epsilon\}= \mathrm{cl}\{z \ : \ \gamma_n(z) < \epsilon\}.\] 
\end{lemma}

\begin{proof}
Note that $\gamma(z) \le {\rm dist}(z,{\rm sp}(A))$  follows from a reformulation of a general property of resolvents. Next, notice that $\|R_n(z)\|=\|R(A_n,z)\|$ and that the norms of resolvents are subharmonic away from spectra and therefore $\gamma$ and $\gamma_n$ cannot have proper local minima,  except when they vanish. Furthermore, they cannot be constant in an open set by Lemma \ref{eivakio}. 

To conclude the  local uniform convergence, let $M$ be such that along the curve $\{|z|=M\}$ there are no eigenvalues of $A$ and choose $\mathcal K$ as the set $\{|z| \le M\}$. Choose any $\epsilon$, small enough so that the discs $\{|z-\lambda_j|\}\le \epsilon/3$  separate the eigenvalues inside $\mathcal K$. By Proposition \ref{resolventtienkovnvergenssi} we may assume that $n$ is large enough so that
 for $z\in \mathcal K_{\epsilon/3}$ (recall $\mathcal{K}_r$ from Proposition \ref{resolventtienkovnvergenssi})
  we have $|\gamma_n(z) -\gamma(z)| \le \epsilon/3$.
 On the other hand, if  $|z -\lambda_j|  \le \epsilon/3$  then $\gamma(z) \le \epsilon/3$ and, since $\gamma_n$ has to vanish also somewhere in that disc (again for large enough $n$) and $\gamma_n(z) \le {\rm dist}(z,{\rm sp}(A_n|_{\mathcal{H}_n}))$, we have  $\gamma_n(z) \le 2 \epsilon/3$ in that disc, hence
 $
 |\gamma_n(z) - \gamma(z) | \le \gamma_n(z) + \gamma(z) \le \epsilon.
 $
Thus we have $
 |\gamma_n(z) - \gamma(z) | \le \epsilon
 $ for all $z\in \mathcal K$.

Finally, to justify the equivalence of the characterisations of pseudospectra just notice that the level sets $\{z \ : \ \gamma(z) = \epsilon\}$ and $\{z \ : \ \gamma_n(z) = \epsilon\}$ cannot contain open subsets or isolated points.
 \end{proof}

\begin{lemma}\label{joukkojensuppeneminen}  
Assume $\varphi_n$ and $\varphi$ are continuous non-negative functions in $\mathbb C$  which have local minima only when they vanish,  are not constant in any open set and $\varphi_n$ converges to $\varphi$ uniformly in compact sets.  Set $\mathcal S:=\{ z \ : \ \varphi(z)\le 1\}$ and  $\mathcal S_n:=\{ z \ : \ \varphi_n(z)\le 1\}$.  Let $\mathcal K $ be compact and $\delta>0$. Then the following hold for all large enough $n$:
$
\mathcal S \cap \mathcal K \subset \mathcal N_\delta(\mathcal S_n), 
$
$
\mathcal S_n \cap \mathcal K \subset \mathcal N_\delta(\mathcal S),
$
where $\mathcal N_\delta(\cdot)$ denotes the open $\delta$ neighbourhood. 
\end{lemma} 
\begin{proof}
Consider
$
\mathcal S \cap \mathcal K \subset \mathcal N_\delta(\mathcal S_n), 
$
 and assume that the left hand side is not empty.  Due to compactness of $\mathcal S \cap \mathcal K$ there are points $z_i \in \mathcal S \cap \mathcal K$  for $i=1,\dots,m$  such that
$
\mathcal S \cap\mathcal K \subset \bigcup_{i=1}^m U_{\delta/2}(z_i).
$
Notice that $\varphi(z_i)\le 1$. If $\varphi(z_i)<1$,  set $y_i=z_i$. Otherwise,  $\varphi(z_i)=1$, in which case $z_i$ cannot be a local minimum, but since $\varphi$ is not constant in any  open set,     there exists a point $y_i \in U_{\delta/2}(z_i)$ such that $\varphi(y_i) <1$.
But since $\varphi_n$ converges uniformly in compact sets to $\varphi$ we conclude that for all large enough $n$ and all $i$ we have 
$\varphi_n(y_i) < 1$.  Hence $z_i \in \mathcal N_{\delta/2}(\mathcal S_n)$ and so
$
\mathcal S \cap\mathcal K \subset \bigcup_{i=1}^m U_{\delta/2}(z_i) \subset \mathcal N_\delta(\mathcal S_n).
$

Consider now 
$
\mathcal S_n \cap \mathcal K \subset \mathcal N_\delta(\mathcal S),
$.  If it would not hold, there would exist a sequence 
$\{n_j\}$ and points $z_{n_j} \in \mathcal S_{n_j} \cap \mathcal K$ such that $z_{n_j} \notin \mathcal N_\delta(\mathcal S)$.  Suppose $z_{n_{j_k}} \rightarrow \hat z$.  Then dist$(\hat z, \mathcal S)\ge \delta$ as well.  However, writing
$
\varphi(\hat z) \le |\varphi(\hat z)-\varphi(z_{n_{j_k}})|
+|\varphi(z_{n_{j_k}})-\varphi_{n_{j_k}}(z_{n_{j_k}})|+\varphi_{n_{j_k}}(z_{n_{j_k}})
$
we obtain $\varphi(\hat z)\le1$ as the  first term on the right tends to zero because  $\varphi$ is continuous,  the second term  converge to zero as $\varphi_n$ approximate $\varphi$  uniformly in compact sets,  and $\varphi_{n_{j_k}}(z_{n_{j_k}})\le 1$.  Hence $\hat z\in \mathcal S \cap \mathcal K$ which is a contradiction.
\end{proof}
Note that the same argument for Lemma \ref{joukkojensuppeneminen} holds when replacing $\leq 1$ by $\leq \epsilon$ in the definitions of $\mathcal S$ and $\mathcal S_n$. Combining the results of this section, we can state the following result. 

\begin{proposition}\label{pseudokonvergenssi}
Let $\mathcal K \subset \mathbb C$ be compact and $\delta>0$.  Then,  for  all large enough $n$,
$$
{\rm sp}_\epsilon(A)\cap \mathcal K  \subset \mathcal N_\delta({\rm sp}_\epsilon(A_n)), \qquad 
{\rm sp}_\epsilon(A_n)\cap \mathcal K  \subset \mathcal N_\delta({\rm sp}_\epsilon(A)).
$$
\end{proposition}

\subsubsection{\bf{The general algorithms.}}
Here $A$, $A_n$  are operators in $\mathcal H$ as in (i), (ii)  above, while $\widetilde A_n$ is the matrix representing $A_n$ when restricted to the finite-dimensional invariant subspace $\mathcal H_n=E_n \mathcal H$. In particular
$
\|R_n(z) \| = \|(\widetilde A_n -z)^{-1}\|.
$ 
Denoting by $\sigma_1$ the smallest singular value of a square matrix we have
$
\gamma_n(z) = 1/\|R_n(z)\| = \sigma_1(\widetilde A_n-zI).
$
Let $r>0$ and define $
G_r := B_r(0)\cap(\frac{1}{2r}{(\Zb+i\Zb)})$. Suppose that the matrices $\widetilde A_n$ are available with $\Delta_1$-information. From Proposition \ref{REC_SING} is follows that we can compute, in finitely many arithmetic operations and comparisons over $\mathbb{Q}$, an approximation to $\gamma_n(z)$ from above, accurate to $1/n^2$, and taking values in $\mathbb{Q}_{\geq 0}$. Call this approximation $\hat \gamma_n$ and let $\epsilon_n\in\mathbb{Q}$ be an approximation of $\epsilon$ from below accurate to $1/n^2$ and define $\Gamma^1_n$ and $\Gamma^2_n$ by 
\begin{equation}
\Gamma^1_n(A) =  \left\{z \in G_{n} : \hat \gamma_n(z)\leq \frac{1}{n}\right\},\qquad 
\Gamma^2_n(A) =  \left\{z \in G_{n} : \hat \gamma_n(z)\leq \epsilon_n\right\},
\end{equation}
which we shall prove to be the towers of algorithms for $\Xi_{\mathrm{sp}}$ and $\Xi_{\mathrm{sp},\epsilon}$ (as defined in Theorem \ref{thm:comp-res}), respectively.  
Observe that  $\Gamma_n^1(A)$ and $\Gamma_n^2(A)$ can be executed in a finite number of arithmetic operations over $\mathbb{Q}$ using $\Delta_1$-information. Also note that our proof will show that $\Gamma^i_n(A)\neq\emptyset$ for large $n$. Hence by our usual trick of searching for minimal $n(m)\geq m$ such that this is so, we can assume without loss of generality this holds for all $n$.

\begin{proposition}\label{yleinentapaus}
The algorithms satisfy the following:
\begin{equation}\label{H_n_conv}
\Gamma^1_n(A) \longrightarrow \mathrm{sp}(A), \qquad 
\Gamma^2_n(A) \longrightarrow \mathrm{sp}_{\epsilon}(A), \quad n \rightarrow \infty.
\end{equation}
\end{proposition}

\begin{proof}
\emph{We begin with the second part of \eqref{H_n_conv}.} It suffices to show that given $\delta$ and a compact  ball $\mathcal{K}$, for large $n$:
$$
\textrm{(i)} \,\, \Gamma^2_n(A) \cap \mathcal{K} \subset 
\mathcal{N}_{\delta}( \mathrm{sp}_{\epsilon}(A)), \qquad  \textrm{(ii)} \,\, \mathrm{sp}_{\epsilon}(A) \cap \mathcal{K} \subset \mathcal{N}_{\delta}( \Gamma^2_n(A)).
$$
Note that
$
\Gamma^2_n(A)\subset \mathrm{sp}_{\epsilon}(\widetilde{A}_n)\cap G_n
$
and hence the first inclusion follows immediately from Proposition \ref{pseudokonvergenssi}. To see (ii), we argue by contradiction and suppose not. Then by possibly passing to an increasing subsequence $\{k_n\}_{n\in\N}\subset\N$ there is a sequence $z_n \in (\mathrm{sp}_{\epsilon}(A) \cap \mathcal{K})\setminus\mathcal{N}_{\delta}(\Gamma^2_n(A))$ for all $n$. Since $\mathrm{sp}_{\epsilon}(A) \cap \mathcal{K}$ is a compact set, by possibly extracting a subsequence, we have that $z_n \rightarrow z_0 \in \mathrm{sp}_{\epsilon}(A) \cap \mathcal{K}$.
Consider the open ball $U_{\delta/3}(z_0)$ which must contain all $z_n$ for $n$ sufficiently large.
Since $\gamma(z)$ is continuous, positive, not constant in any open set and without nontrivial local minima, it follows that ${\rm sp}_\epsilon(A)$ equals the closure of its  interior points. In particular  
 $\mathrm{int}(\mathrm{sp}_{\epsilon}(A)) \cap U_{\delta/3}(z_0) \neq \emptyset$.  
Suppose then $r>0$ and $y_0$ are such that the closure of the open ball $U_r(y_0)$ is inside  this open set: $ B_r(y_0) \subset \mathrm{int}(\mathrm{sp}_{\epsilon}(A)) \cap U_{\delta/3}(z_0)$.  
 We claim that $\{z:\hat\gamma_n(z)\leq\epsilon\} \cap U_r(y_0) = U_r(y_0)$ for all large enough $n$. Indeed, since $U_r(y_0)$ bounded away from the boundary of the pseudospectrum of $A$,  we have $\gamma(z) \le \epsilon -s$ for some $s>0$ and for all $z\in U_r(y_0)$.  Now the claim follows from the locally uniform convergence of $\gamma_n$ and hence of $\hat\gamma_n $.
   By the definition of $G_n$ we have that  $U_r(y_0) \subset  \mathcal{N}_{\delta/3}(U_r(y_0) \cap G_n)$ for large $n$, so, by the claim, $U_r(y_0) \subset \mathcal{N}_{\delta/3}(\{z:\hat\gamma_n(z)\leq\epsilon\}\cap G_n).$ Hence, since $U_r(y_0) \subset U_{\delta/3}(z_0)$, it follows that  
$$
z_n \in U_{\delta/3}(z_0) \subset \mathcal{N}_{2\delta/3}(U_r(y_0)) \subset \mathcal{N}_{\delta}(\{z:\hat\gamma_n(z)\leq\epsilon\}\cap G_n),
$$
for large $n$, 
contradicting $z_n \notin \mathcal{N}_{\delta}(\Gamma^2_n(A))$.
To prove the first part of \eqref{H_n_conv} we argue as follows. Given $\delta>0$ and compact $\mathcal{K}$, we need to show that for large $n$:
$$\textrm{(iii)} \, \, \mathrm{sp}(A) \cap \mathcal{K} \subset \mathcal{N}_{\delta}(\{z:\hat\gamma_n(z)\leq1/n\} \cap G_n) \qquad 
\textrm{(iv)} \,\,  \{z:\hat\gamma_n(z)\leq1/n\} \cap G_n \cap \mathcal{K}\subset  \mathcal{N}_{\delta}(\mathrm{sp}(A)).
$$
For notational convenience, we let $a_n=1/n-1/n^2$.

To show (iii), we start by defining $\widetilde G_n := \frac{1}{2n}{(\Zb+i\Zb)}$ and note that for $\lambda_j \in \mathrm{sp}(\widetilde{A}_n)$ we have that $\mathcal{N}_{a_n}(\{\lambda_j\}) \cap \widetilde G_n \neq \emptyset$ for large $n$. Hence, $\mathrm{sp}(\widetilde{A}_n) \subset  \mathcal{N}_{1/n}\left(\mathcal{N}_{a_n}\left(\mathrm{sp}(\widetilde{A}_n)\right) \cap \widetilde G_n\right)$. Since $\mathcal{N}_{a_n}(\mathrm{sp}(\widetilde{A}_n)) \subset \mathrm{sp}_{a_n}(\widetilde{A}_n)$, it follows that $\mathrm{sp}(\widetilde{A}_n) \subset  \mathcal{N}_{1/n}\left(\mathrm{sp}_{a_n}(\widetilde{A}_n) \cap \widetilde G_n\right)$. Now by the first part of Proposition \ref{spektrinkonvergenssi} we have that $\mathrm{sp}(A) \cap \mathcal{K} \subset \mathcal{N}_{\delta/2}(\mathrm{sp}(\widetilde{A}_n))$ for large $n$. Thus, combining the previous observations, we have for large $n$ that
$$
\mathrm{sp}(A) \cap \mathcal{K} \subset  \mathcal{N}_{\delta/2+1/n}\left(\mathrm{sp}_{a_n}(\widetilde{A}_n) \cap \widetilde G_n\right)\subset  \mathcal{N}_{\delta/2+1/n}\left(\{z:\hat\gamma_n(z)\leq1/n\} \cap \widetilde G_n\right).
$$
However, since $\mathcal{K}$ is bounded we have that there exists an $r > 0$ such that if $\lambda \in \widetilde G_n \cap U_r(0)^c$ then $\mathcal{N}_{\delta}(\{\lambda\}) \cap \mathrm{sp}(A) \cap \mathcal{K} = \emptyset$ for all $n$. Hence,  
$\mathrm{sp}(A) \cap \mathcal{K} \subset \mathcal{N}_{\delta}\left(\{z:\hat\gamma_n(z)\leq1/n\} \cap G_n\right)$
as desired.

To see (iv), let $r > 0$ be so large that $\mathcal{N}_{\delta}(U_r(0)^c) \cap \mathcal{K} = \emptyset$. Note that $\mathrm{sp}_{\epsilon}(A) \rightarrow \mathrm{sp}(A)$ as $\epsilon \rightarrow 0$. Thus, $\mathrm{sp}_{\epsilon_1}(A) \cap B_r(0) \subset  \mathcal{N}_{\delta/2}(\mathrm{sp}(A))$ for a sufficiently small $\epsilon_1$. Also, by the second part of  Proposition \ref{pseudokonvergenssi} it follows that 
$\mathrm{sp}_{\epsilon_1}(\widetilde{A}_n) \cap  \mathcal{K} \subset \mathcal{N}_{\delta/2}(\mathrm{sp}_{\epsilon_1}(A))$ for large $n$. However, by the choice of $r$ we have that $\mathrm{sp}_{\epsilon_1}(\widetilde{A}_n) \cap  \mathcal{K} \subset \mathcal{N}_{\delta/2}(\mathrm{sp}_{\epsilon_1}(A) \cap B_r(0))$. Clearly, $\mathrm{sp}_{1/n}(\widetilde{A}_n) \cap  \mathcal{K} \subset \mathrm{sp}_{\epsilon_1}(\widetilde{A}_n) \cap  \mathcal{K} $ for large $n$. Thus, by patching the above inclusions together we get that 
$$
\{z:\hat\gamma_n(z)\leq1/n\} \cap G_n \cap \mathcal{K}\subset\mathrm{sp}_{1/n}(\widetilde{A}_n) \cap  \mathcal{K} \subset \mathrm{sp}_{\epsilon_1}(\widetilde{A}_n) \cap  \mathcal{K} \subset  \mathcal{N}_{\delta/2}(\mathrm{sp}_{\epsilon_1}(A) \cap B_r(0)) \subset  \mathcal{N}_{\delta}(\mathrm{sp}(A)),
$$
for large $n$, as desired. This finishes the proof of Proposition \ref{yleinentapaus}.
\end{proof}


Next, we pass from these general considerations to the Schr\"odinger case.

\subsubsection{\bf{Compactness of the resolvent.}}
We first show that the resolvent of the Schr\"odinger operator $H\in\Omega_\infty$ is compact. To prove this we recall some well known lemmas and definitions from \cite{Kato_pert95}.
\begin{definition}\label{m-sectorial}
An operator $A$ on the Hilbert space $\mathcal H$ is {\it accretive} if the ${\rm Re} \langle Ax,x\rangle \ge 0$ for $x \in \mathcal D(A)$.  It is called m-accretive if there exists no proper accretive extension.  If $A$ (possibly after shifting with a scalar) is m-accretive and additionally there exists $\beta <\pi/2$ such that $| \arg \langle  Ax,x \rangle | \le \beta$ for all $x \in\mathcal D(A)$, then $A$ is m-sectorial.
\end{definition}
\begin{lemma}[{\cite[VI-Theorem 3.3]{Kato_pert95}}]\label{Katon kirjasta}
Let $A$ be m-sectorial with $B={\rm Re} \ A$. $A$ has compact resolvent if and only if $B$ has.
\end{lemma}

\begin{lemma}[{\cite[V-Theorem 3.2]{Kato_pert95}}]\label{myoskirjasta}
If $T$ is closed and  the complement of $\mathrm{Num}(T)$ is connected, then for every $\zeta$ in the complement of the closure of $\mathrm{Num}(T)$ the following hold: the kernel of $T-\zeta$ is trivial and the range of 
$T-\zeta$  is closed  with constant codimension.
\end{lemma}

\begin{proposition}\label{tarkan kompaktsuus}
Suppose  $V$ is continuous $\mathbb R^d \rightarrow \mathbb C$ satisfying the following:
$
V(x) =|V(x)| e^{i\varphi(x)}
$
such that 
$
|V(x)| \rightarrow \infty  \text{ as }  x \rightarrow \infty,
$
and there exist non-negative $\theta_1,\theta_2$  such that $ \theta_1+\theta_2<\pi$ and
$
-\theta_2 \le \varphi(x) \le \theta_1.$
Denote  by $h$ the operator $h=-\Delta+V$ with domain $\mathcal D(h)= C_c^\infty (\mathbb R^d)$ and put in $\mathrm{L}^2(\mathbb R^d)$ $H=h^{**}$.
Then  $H=-\Delta + V$ is a densely defined operator with compact resolvent, whose spectrum lies in the sector $\{z:\mathrm{arg}(z) \in [-\theta_2,\theta_1]\}$.
\end{proposition}

\begin{proof}
The proof goes as follows:  Notice first that the numerical range of $H$  lies in a sector with opening $2 \beta<\pi$. Then we turn the sector into the symmetric position around the positive real axis to get the operator $a(\alpha)$.  It is clearly enough to show that $A(\alpha)=a(\alpha)^{**}$  is an m-sectorial operator with half-angle $\beta=(\theta_1+\theta_2)/2$ which has a compact resolvent.
Next, since the numerical range of $a(\alpha)$ is not the whole plane, the operator is closable. Then we conclude that every point away from the numerical range belongs to the resolvent set. This is done based on the fact that the adjoint shares the same key properties as $A(\alpha)$. Then the compactness of the resolvent follows by considering the resolvent of the real part of $A(\alpha)$.

Here is the notation.  Put $\alpha=(\theta_1-\theta_2)/2$ so that $|\alpha|<\pi/2$. 
Then with 
\begin{equation}\label{vartheta}
\vartheta(x)= \varphi(x)- \alpha
\end{equation}
 we have
$
a(\alpha):= e^{-i \alpha}h =-e^{-i \alpha} \Delta + |V(x)| e^{i\vartheta(x)}
$
and after extending $A(\alpha)=a(\alpha)^{**}$ , in particular
$
H(\alpha):=   {\rm Re} A(\alpha)= - {\cos \alpha} \ \Delta+  \cos\vartheta(x)|V(x)|.
$

We claim that the operator
$
A(\alpha):=e^{-i \alpha}H
$
is  m-sectorial with half-angle $\beta = (\theta_1+\theta_2)/2.$ Indeed, it is immediate that the numerical range satisfies the following
$
\mathrm{Num}(a(\alpha)) \subset \{z=re^{i\theta}  : \ |\theta| \le \beta, \ r\ge 0 \  \},
$
which is not the whole complex plane,
and we can therefore (by \cite[V-Theorem 3.4 on p. 268]{Kato_pert95}) consider the extended closed operator $A(\alpha)$ instead.  The next thing is to conclude that points away from this closed sector are  in the resolvent set of $A(\alpha)$. Take any point $\zeta=r e^{i\varphi}$ with $\beta < |\varphi| \le \pi$, $r>0$.  We need to conclude that $\zeta \notin {\rm sp}(A(\alpha))$. Since the complement of $\mathrm{Num}(A(\alpha))$ is connected, the following holds (by Lemma \ref{myoskirjasta}):  the operator  
$A(\alpha)-\zeta$ has closed range with constant codimension.  Thus, we need that the  range  is the whole space. Put for that purpose $T=A(\alpha) - \zeta.$
Suppose there is $g\not=0$ such that $ g \in {\rm Ran }(T)^{\perp}$.  Then for all $f\in \mathcal D(T)$  we have
$
\langle Tf,g \rangle =0
$
which means, as $\mathcal D(T)$ is dense, that $T^* g =0$.  But that is not the case as $A(\alpha)^* - \overline {\zeta} $  is also closed whose complement of the numerical range is connected and hence does not have a non-trivial kernel.  

The proof of Proposition \ref{tarkan kompaktsuus} can now be completed by invoking Lemma 
\ref{Katon kirjasta} since  it is well known (\cite{ReedSimonIV}, Theorem XIII.67) that (since $\alpha < \pi/2$) the self-adjoint operator $H(\alpha)$ has compact resolvent when the potential $|V(x)|$ tends to  infinity  with $x$.  
\end{proof}

We shall next consider the discretisation of  $H$ and of $A(\alpha)$.  It shall be clear that the discrete versions have their numerical ranges inside the same sectors, where the numerical range of an operator $T$ is denoted by $\mathrm{Num}(T)$. Thus all resolvents can be estimated using the fact that if $(T-\zeta)^{-1}$ is regular outside the closure of $\mathrm{Num}(T)$, then
$
\|(T-\zeta)^{-1}\| \le {1}/{{\rm dist}(\zeta, \mathrm{Num}(T))}.
$
  
\subsubsection{\bf{Discretizing the  Schr\"odinger operator.}}
We shall show how to assemble the matrices $\widetilde H_n$ mentioned above. The underlying Hilbert space is again $\mathrm{L}^2(\mathbb R^d)$ and 
we start with approximating the Laplacian.
Let $1 \leq j \leq d,\ t \in \mathbb{R}$ and define $U_{j,t}$ to be the one-parameter unitary group of translations
$$
U_{j,t}\psi(x_1,\hdots, x_d) = \psi(x_1,\hdots, x_j - t, \hdots, x_d)
$$ 
and let $P_j$ be the infinitesimal generator of $U_{j,t}$ so that 
$
U_{j,t} = e^{\ii tP_j}
$
and $P_j = \lim_{t \rightarrow 0}\frac{1}{\ii t}(U_{j,t} - I).$
Thus,  defining 
	$
	\Phi_n(x) = \frac{n}{\ii}(e^{\ii\frac{1}{n} x} -1)
	$
	with $n \in \mathbb{N}$ and $x \in \mathbb{R},$
it follows that
	\be\label{eq:laplacian-approx}
	|\Phi_n|^2(P_j)\psi(x) = n^2(-\psi(x_1,\hdots,x_j + 1/n,\hdots x_d) - \psi(x_1,\hdots,x_j -  1/n,\hdots x_d) + 2\psi(x))
	\en
is the discretised Laplacian in the $j$ direction. The full discretised Laplacian is therefore $\sum_{j=1}^d|\Phi_n|^2(P_j)$. Now we replace $V$ by an appropriate approximation. Consider the lattice $(\frac{1}{n}\mathbb{Z})^d$ as a subset of $\mathbb{R}^d$ and for $y \in (\frac{1}{n}\mathbb{Z})^d$ define the box
	\be\label{eq:qn}
	Q_n(y) = \left\{x = (x_1,\hdots, x_d):  x_j \in\left[ y_j - \frac{1}{2n}, y_j + \frac{1}{2n}\right), \, 1 \leq j \leq d\right\}.
	\en
	Let
	$
	S_n
	=
	\left[-\lfloor\sqrt{n}\rfloor, \lfloor\sqrt{n}\rfloor\right]^d
	\subset
	\R^d
	$
and define $E_n$ to be the orthogonal projection onto the subspace
	\begin{equation}\label{ESn}
	\begin{split}
	\left\{\psi \in  \mathrm{L}^2(\mathbb{R}^d) :  \psi = \sum_{y \in (\frac{1}{n}\mathbb{Z})^d\cap S_n} \alpha_y \chi_{Q_n(y)}, \, \alpha_y \in \mathbb{C}\right\},
	\end{split}
	\end{equation}
where $\chi_{Q_n(y)}$ denotes the characteristic function on $Q_n(y)$.
Define the approximate potential as
	$$
	V_n(x)
	=
	\begin{cases}
	V(y)&x\in Q_n(y)\cap S_n\text{ for some }y \in (\frac{1}{n}\mathbb{Z})^d,\\
	0&\text{otherwise}.
	\end{cases}
	$$
Note that $V_n=E_nV_nE_n$, but that, in general, $V_n\neq E_nVE_n$. Finally, we define the approximate Schr\"odinger operator $H_n: \mathrm{L}^2(\R^d)\to \mathrm{L}^2(\R^d)$ defined as
	\be\label{eq:approx-op-comp-resolv}
	H_n
	=
	E_n\sum_{j=1}^d|\Phi_n|^2(P_j)E_n+V_n.
	\en


\begin{remark}\label{rek:matrix-rep}
Note that the restriction $H_n\vert_{\mathrm{Ran}(E_n)}$ of $H_n$ to the image of $E_n$ has a matrix representation $\widetilde{H}_n\in\C^{m\times m}$ (where $m = \mathrm{dim}(\mathrm{Ran}(E_n))$) defined as follows. First, for $y_1,y_2\in(\frac{1}{n}\mathbb{Z})^d\cap S_n$,
\[\langle |\Phi_n|^2(P_j)E_n n^{d/2}\chi_{Q_n(y_1)}, n^{d/2}\chi_{Q_n(y_2)}\rangle =
\begin{cases}
2n^{2} & y_1=y_2\\
-n^{2} &  y_1-y_2=\pm 1/n e_j\\
0 &  \text{otherwise}
\end{cases}
\]
and $\langle V_n  n^{d/2}\chi_{Q_n(y_1)},  n^{d/2}\chi_{Q_n(y_2)} \rangle = V(y_1)$ when $y_1 = y_2$ and zero otherwise. Thus, we can form the matrix representation of $H_n\vert_{\mathrm{Ran}(E_n)}$ with respect to the orthonormal basis $\{n^{d/2}\chi_{Q_n(y)}\}_{y \in (\frac{1}{n}\mathbb{Z})^d \cap S_n}$. It is important to note that calculating the matrix elements of $\widetilde{H}_n$ requires knowledge only of $\{V_f\}_{f \in \Lambda_n}$ where we have
$\Lambda_n := \left\{f_y: y \in ({n}^{-1}\mathbb{Z})^d \cap S_n\right\}$ and $V_{f_y}=f_y(V) = V(y).$
\end{remark}

\subsubsection{\bf{Proof that $\{\Xi_{\mathrm{sp}},\Omega_\infty\}\in\Delta^A_2,\{\Xi_{\mathrm{sp},\epsilon},\Omega_\infty\}\in\Delta^A_2$}.}
 
We have so far shown that the Assumption (i)  holds, and we are left to show that the discretisation we have chosen satisfies Assumption (ii).  In particular, we need to demonstrate that our discretisation satisfies \eqref{resolventtiapproksimaatio}. That is the topic of the following theorem.

\begin{theorem}\label{strong}
Let $V\in C(\R^d)$ be sectorial as defined in \eqref{eq:sector} satisfying $|V(x)|\to\infty$ as $|x|\to\infty$, and let $h=-\Delta+V$ with $\mathcal D(h)=C_c^\infty(\R^d)$ and let $H=h^{**}$. Let $H_n$ be as in \eqref{eq:approx-op-comp-resolv}. Then  there exists $z_0$ such that 
$
\| (H-z_0)^{-1} -  (H_n -z_0)^{-1} E_n\| \rightarrow 0,
$
as $n \rightarrow \infty$.
\end{theorem}

Note that we immediately have 
$$
\text{ Theorem \ref{strong} } + \text{ Proposition \ref{yleinentapaus} } \Rightarrow \{\Xi_{\mathrm{sp}},\Omega_\infty\}\in\Delta^A_2,\{\Xi_{\mathrm{sp},\epsilon},\Omega_\infty\}\in\Delta^A_2.
$$
Thus, the rest of the section is devoted to prove Theorem \ref{strong}.

We shall treat the discretisations in a similar way as the continuous case, namely by ``rotating'' the operator into symmetric position with respect to the real axis and then, by taking the real part, we are dealing with a  sequence of self-adjoint invertible operators.
Before we prove this theorem, we will need a couple of lemmas. We recall the following definition.

\begin{definition}[Collectively compact]
A set $\mathcal T \subset B(\mathcal H)$ is called  {\it collectively compact} if  the set $\{ Tx \ :  \ T \in \mathcal T, \|x\| \le 1 \}$ has compact closure. 
\end{definition}
\begin{lemma}\label{LemCollComp}
Let $\{K_n\}$ be a collectively compact operator sequence and $K_n^*\to 0$ strongly. Then $\|K_n\|\to 0$.
\end{lemma}
\begin{proof}
It is well known that on any compact set $\mathcal{B}$ the strong convergence $K_n^*\to 0$ turns into norm convergence:
$\sup\{\|K_n^*x\|:x\in \mathcal{B}\}\to_n 0.$ Since $\mathcal{B}:=\clos\{K_nx:\|x\|\leq 1, n\in\Nb\}$ is compact, we get
\[\|K_n\|^2=\|K_n^*K_n\|=\sup\{\|K_n^*K_nx\|:\|x\|\leq 1\}\leq\sup\{\|K_n^*y\|:y\in\mathcal{B}\}\to 0\quad\text{as}\quad n\to\infty.\]
\end{proof}
We also need  a modification of Lemma  \ref{Katon kirjasta}.
\begin{lemma}\label{modifikaatioEn}
Let $\{A_n\}$ be m-sectorial  with common semi-angle $\beta<\pi/2$ and denote $B_n= {\rm Re} \ A_n$.  Assume that $\{E_n\}$ is a sequence of orthogonal projections, converging strongly to identity and such that $A_n E_n = E_nA_nE_n$ and $B_n E_n = E_nB_nE_n$.    
Assume further that $\{B_n^{-1}\}$ is uniformly bounded. 
If $\{B_n^{-1} E_n\}$ is collectively compact, then so is $\{A_n^{-1}E_n \}$.
\end{lemma}
\begin{proof}
Denote by $B_n^{1/2}$ the unique self-adjoint non-negative square root of $B_n$. By \cite[VI-Theorem 3.2 on p.337]{Kato_pert95} for each $A_n$ there exists a bounded symmetric operator $C_n$ satisfying $\|C_n\| \le \tan (\beta)$ and  such that
$
A_n= B_n^{1/2} (1+\ii \ C_n)B_n^{1/2}.
$
Writing 
$$A_n^{-1}= \int_0^\infty e^{-tA_n} dt$$ we conclude that $E_nA_n^{-1}E_n= A_n^{-1}E_n$ and likewise for $B_n^{-1}$. Assume now that $\{B_n^{-1}E_n\}$ is collectively compact.
But then so is  $\{(B_n+t)^{-1}E_n\} \ = \ \{B_n^{-1}E_n (I+t B_n^{-1})^{-1}E_n\}$  and  writing, compare \cite[V (3.43) on p.282]{Kato_pert95},
$$
B_n^{-1/2}E_n= \frac{1}{\pi} \int_0^\infty t^{-1/2}(B_n + t)^{-1}E_n dt
$$
we see that $\{B_n^{-1/2}E_n \}$ is also collectively compact and $B_n^{-1/2}E_n=E_nB_n^{-1/2}E_n$.  Finally $\{A_n^{-1}E_n\}$ is then collectively compact as well since  $A_n^{-1}E_n$ is of the form $  B_n^{-1/2}E_n T_n$ with $T_n$ uniformly bounded. 
\end{proof}

\begin{proof}[Proof of Theorem \ref{strong}]
Note that it is clear from the definition of $H_n$ and the assumption on $V$ that 
$
\mathrm{Num}(H_n) \subset \{r e^{i\rho} : -\theta_2 \leq \rho \leq \theta_1, r\geq 0\}
$ for all $n$. Thus, since $H_n$ is bounded and by Proposition \ref{tarkan kompaktsuus} we can choose any 
point  $z_0 \in \mathbb{C}$ such that $z_0$ has a positive distance $d$ to the closed sector
$\{r e^{i\rho} : -\theta_2 \leq \rho \leq \theta_1, r\geq 0\}$, and both $R(H,z_0) = (H-z_0)^{-1}$ and $R(H_n,z_0) = (H_n -z_0)^{-1}$ for every $n$ will exist. Moreover, $R(H_n,z_0)$ are uniformly bounded for all $n$, since for every $x$, $\|x\|=1$,
\[\|(H_n-z_0)x\|\geq|\langle(H_n-z_0)x,x\rangle|\geq|\langle H_nx,x\rangle-z_0|\geq d.\]
 Note that by Lemma \ref{LemCollComp} it suffices to show that (i) $R(H_n,z_0)^*E_n \rightarrow R(H,z_0)^*$ strongly, and
(ii) $\{R(H_n,z_0)E_n - R(H,z_0)\}$ is collectively compact, which follows if we can show that $\{R(H_n,z_0)E_n\}$ is collectively compact.

To see (i) observe that $C^{\infty}_c(\mathbb{R}^d)$ is a common core for $H $ and for $ H_n$. Hence  by \cite[VIII-Theorem 1.5 on p.429]{Kato_pert95}, the strong resolvent convergence $R(H_n,z_0)^*\to R(H,z_0)^*$ will follow if we show that 
$ H_n^*\psi \rightarrow H^*\psi$ as $n\to\infty$
for any $\psi \in C^{\infty}_c(\mathbb{R}^d).$ 
Then the strong convergence $R(H_n,z_0)^*E_n\to R(H,z_0)^*$ follows as well. 
Note that 
\begin{equation}\label{mmm}
\| H_n^*\psi - H^*\psi\| \leq \left\|\sum_{j=1}^d|\Phi_n|^2(P_j)E_n\psi -
\sum_{j=1}^d P^2_j\psi \right\| 
 + \|(\overline{V}_n - \overline{V})\psi\|.  
\end{equation}
Also, 
$
|\Phi_n|^2(P_j) = n(\tau_{-1/ne_j}-I)n(\tau_{1/ne_j}-I),$ 
where $\tau_{z}\psi(x) = \psi(x-z)$
 and $\{e_j\}$ is the canonical basis for $\mathbb{R}^d$. Moreover, for $\psi \in C^{\infty}_c(\mathbb{R}^d),$ 
$$
E_n\psi =  \sum_{y \in (\frac{1}{n}\mathbb{Z})^d\cap S_n} (\Psi_n *  \psi)(y)\chi_{Q_n(y)},  \quad \Psi_n = \rho_n \otimes \hdots \otimes \rho_n, \quad \rho_n = n \chi_{[-\frac{1}{2n},\frac{1}{2n})},
$$
where $S_n$ was defined in \eqref{ESn}. 
Thus, it follows from easy calculus manipulations and basic properties of convolution that $|\Phi_n|^2(P_j)E_n\psi = \sum_{y \in (\frac{1}{n}\mathbb{Z})^d} (\Psi_n * \tilde \rho_1 *_j \tilde \rho_2 *_j \psi'')(y)\chi_{Q_n(y)}$, where $\tilde \rho_1 = n\chi_{[-1/n,0]}$, $\tilde \rho_2 = n\chi_{[0,1/n]}$ and $*_j$ denotes the convolution operation in the $j$th variable. By standard properties of the convolution we have that $\Psi_n * \tilde \rho_1 *_j \tilde \rho_2 *_j \psi'' \rightarrow \psi''$ uniformly as $n \rightarrow \infty$. 
Thus, since $\psi \in C^{\infty}_c(\mathbb{R}^d),$ the first part of the right-hand side of (\ref{mmm}) tends to zero as $n \rightarrow \infty$.
Due to the continuity of $V$ and the bounded support of $\psi$ it also follows easily that $\|(\overline{V}_n-\overline{V})\psi\| \to 0$ as $n\to\infty$.

To see (ii) we use the same trick as in the proof of Proposition \ref{tarkan kompaktsuus}. In particular, first set $z_0 = -e^{i\alpha}$ (which is clearly in the resolvent set of $H_n$ for $\alpha=(\theta_1-\theta_2)/2$)
then let 
$
A_n( \alpha) = e^{-i \alpha}(H_n - z_0)
$ 
and further
$H_n(\alpha)={\rm Re} \ A_n(\alpha)$.
Note that, by Lemma \ref{modifikaatioEn}, we would be done if we could show that $\{H_n(\alpha)^{-1}\}$ is uniformly bounded and $\{H_n(\alpha)^{-1}E_n\}$ is collectively compact as that would yield collective compactness of $\{A_n( \alpha)^{-1}E_n\}$ and hence of 
$\{R(H_n,z_0)E_n\}$. 
To establish the uniform bound, note that 
\begin{equation}\label{hooenna}
H_n(\alpha) = \cos \alpha \ E_n \sum_{j=1}^d|\Phi_n|^2 (P_j) E_n + 
 \cos \vartheta(x) |V_n(x)| +1,
 \end{equation}
where $\vartheta$ is defined in (\ref{vartheta}). Thus $ \|H_n(\alpha)^{-1}\| \le 1$ and by applying Lemma \ref{collect} we are now done.
\end{proof}


\begin{lemma}\label{collect}
Let $H_n(\alpha)$ be given by (\ref{hooenna}). Then the set $\{H_n(\alpha)^{-1}E_n\}$ is collectively compact.
\end{lemma}

\begin{proof}
We shall show that if we choose an arbitrary sequence $\{\psi_n\} \subset \mathrm{L}^2(\mathbb R^d)$  satisfying $\|\psi_n\| \le 1$, then the sequence $\{\varphi_n\}$ where $\varphi_n = H_n(\alpha)^{-1}E_n\psi_n$, is relatively compact in $\mathrm{L}^2(\mathbb R^d)$.  
The compactness argument is based on the Rellich's criterion.

\begin{lemma}[Rellich's criterion (\cite{ReedSimonIV} Theorem XIII.65)]
Let $F(x)$ and $G(\omega)$ be two  measurable non-negative functions becoming larger than any constant for all large enough $\
|x|$ and $|\omega |$.  Then
$$
S=\{ \varphi \ : \ \int |\varphi(x)|^2 dx \le 1, \ \int F(x) |\varphi(x)|^2 dx  \le 1, \ \int G(\omega) |\mathcal F \varphi(\omega)|^2 d\omega  \le 1\}
$$
is a compact subset of $\mathrm{L}^2(\mathbb R^d)$.
\end{lemma}
To prove Lemma \ref{collect} we proceed as follows.  First we conclude that $\{\varphi_n\}$ is a bounded sequence itself.   Then, in order to be able to define suitable functions $F,G$ we need to approximate the sequence by another one of  the form $\Psi_n*\varphi_n$.
This approximation shall satisfy $\lim_{n\to\infty}\|\Psi_n * \varphi_n - \varphi_n\| = 0$ and this  is very similar to the standard result on local uniform convergence of mollifications of continuous functions.  Then the Rellich's criterion holds for $\Psi_n*\varphi_n$ with $F(x)$ essentially given by $|V(x)|$ and $G(\omega)$ by $|\omega|^2$.  We then conclude that the sequence $\{\Psi_n*\varphi_n\}$ is relatively compact.  But since $\lim_{n\to\infty}\|\Psi_n * \varphi_n - \varphi_n\| = 0$,  the sequence $\{\varphi_n\}$ is relatively compact as well, completing the argument.

More precisely, since $|\vartheta(x)| \le \alpha < \pi/2$ we have from (\ref{hooenna})
\begin{equation}\label{someestimate}
|\langle H_n(\alpha) \varphi_n, \varphi_n\rangle| \ge \cos \alpha \left(
\left\langle \sum_{ j =1}^{ d} |\Phi_n|^2(P_j)\varphi_n, \varphi_n\right\rangle + 
\left \langle |V_{{n}}|\varphi_n, \varphi_n \right\rangle\right) + \|\varphi_n\|^2.
\end{equation}
But $|\langle H_n(\alpha) \varphi_n, \varphi_n\rangle|$ is bounded not only from below but also from above.  Indeed, $|\langle H_n(\alpha) \varphi_n, \varphi_n\rangle| = |\langle E_n \psi_n, \varphi_n\rangle|\leq\|H_n(\alpha)^{-1}E_n\|\|\psi_n\|^2.$ Thus,  we conclude first from (\ref{someestimate}) that  the sequence $\{\varphi_n\}$ is  bounded. Next, in view of
\eqref{someestimate}, there exist constants $C_1,C_2 >0$   such that for all  $n \in \mathbb{N}$
	\begin{equation}\label{eq:bound-approx-laplacian1}
	\left\langle \sum_{j=1}^{  d} |\Phi_n|^2(P_j)\varphi_n, \varphi_n\right\rangle
	\leq
	C_1,
	\qquad 
	\langle |V_n|\varphi_n, \varphi_n \rangle
	\leq
	C_2.
	\end{equation}
First we use the bound in the first part of \eqref{eq:bound-approx-laplacian1}. Letting $\mathcal{F}$ denote the Fourier transform, we have that 
	$
	(\mathcal{F}\Phi_n(P_j) \varphi_n)(\omega) = 
	\Phi_n(\omega_j)(\mathcal{F}\varphi_n)(\omega),$ for a.e. $\omega$ and for $1\leq j\leq d.
	$
Letting $\Theta_n(\omega)=\frac{\sin(\omega/2n)}{\omega/2n}$, an application of the Fourier transform to \eqref{eq:bound-approx-laplacian1} along with Plancherel's theorem yield
	$$
	\int_{\mathbb{R}^d} |(\mathcal{F}\varphi_n)(\omega)|^2
	\sum_{1 \leq j \leq d}|\omega_j\Theta_n(\omega_j)|^2 \, d\omega \leq C_1.
	$$
Moreover, since $|\Theta_n(\omega)| \leq 1$ for all $\omega,$ we get 
\begin{equation}\label{aur}
\int_{\mathbb{R}^d}|\omega|^2|\Theta_n(\omega_1) \cdots \Theta_n(\omega_d)|^2 
|(\mathcal{F}\varphi_n)(\omega)|^2 \, d\omega \leq C_1.
\end{equation}

We now define the approximation $\Psi_n*\varphi_n$.  Let $\Psi_1(z) = \chi_{[-1/2,1/2]^d}(z)$ and further $\Psi_n(z) = n^d\Psi_1(nz)$,
where $\chi_A(z)$ is the usual characteristic function for the set $A$. We  shall prove below that $\lim_{n\to\infty}\|\Psi_n * \varphi_n -\varphi_n\| = 0$, which in particular shows that the sequence $\{\Psi_n*\varphi_n\}$ is bounded.
Observe then that $(\mathcal{F}\Psi_n)(\omega) = \Theta_n(\omega_1) \cdots \Theta_n(\omega_d)$. Therefore we obtain from (\ref{aur})
 $
 \int_{\mathbb R^d} |\omega|^2 |\mathcal F(\Psi_n * \varphi_n)(\omega)|^2 d\omega \le C_1,
 $
 which  shows that we can choose $G(\omega)$ to be (a constant times) $|\omega|^2$.  
 
 We still need to establish the growth function $F(x)$ for  $\Psi_n*\varphi_n$.
 Consider $\varphi_n$.  It is of the form $\varphi_n= (E_n+E_n B_nE_n)^{-1}E_n\psi_n$ and hence $E_n\varphi_n= \varphi_n$.  Therefore  $\varphi_n$ vanishes outside $S_n$ and we can  essentially replace $V_n$ by $V$ in the inequality in the last part of \eqref{eq:bound-approx-laplacian1}. To that end, put
 $
 F(x) = \min_{|y| \ge |x| } |V(y)|.
 $
Then  with some constant $C_3$
\begin{equation}\label{uusi}
\int_{\mathbb R^d} F(x) |(\Psi_n *\varphi_n)(x)| ^2 dx \le C_3.
\end{equation}
 
In view of the bounds \eqref{aur}, (\ref{uusi}) and since the sequence $\{\Psi_n*\varphi_n\}_{n\in\N}$ is bounded in $\mathrm{L}^2$, Rellich's criterion   implies that $\{\Psi_n*\varphi_n\}_{n\in\N}$ is a {relatively compact} sequence and it therefore follows 
that $\{\varphi_n\}_{n\in\N}$ is {relatively compact}, thus finishing the proof.  Hence, our only remaining obligation is to show that $\lim_{n\to\infty}\|\Psi_n * \varphi_n - \varphi_n\| = 0$.  This result is very similar to the standard result on local uniform convergence of mollifications of continuous functions.

Let $z \in \mathbb{R}^d$ and define the shift operator $\tau_{z}$ 
on $\mathrm{L}^2(\mathbb{R}^d)$ by 
$\tau_{z}f(x) = f(x-z)$. Now 
observe that 
by Minkowski's inequality for integrals it follows that 
\begin{equation}\label{integral}
\begin{split}
\|\Psi_n * \varphi_n - \varphi_n\| \leq \int_{\mathbb{R}^d} \|\tau_{\frac{1}{n}z}\varphi_n -
\varphi_n\||\Psi_1(z)|\, dz
= \int_{[-1/2,1/2]^d} \|e^{i\frac{z_d}{n}P_d} \hdots e^{i\frac{z_1}{n}P_1}\varphi_n - \varphi_n\|
\, dz.
\end{split}
\end{equation}
The claim follows from an $\epsilon/d$ argument and \eqref{integral} combined with the dominated convergence theorem (recall that $\{\varphi_n\}$ is bounded): we need to show that for fixed $z \in [-1/2,1/2]^d$ and for any $ 1< j \leq d$,
\begin{equation}\label{trott}
\begin{split}
\lim_{n\to\infty}\left\|e^{i\frac{z_j}{n}P_j} \hdots e^{i\frac{z_1}{n}P_1}\varphi_n -
e^{i\frac{z_{j-1}}{n}P_{j-1}} \hdots e^{i\frac{z_1}{n}P_1}\varphi_n\right\| = 0,
\quad \lim_{n\to\infty} \left\|e^{i\frac{z_1}{n}P_1}\varphi_n - \varphi_n\right\| = 0.
\end{split}
\end{equation}
Since 
$
e^{i\frac{z_j}{n}P_j}e^{i\frac{z_k}{n}P_k} = e^{i\frac{z_k}{n}P_k}e^{i\frac{z_j}{n}P_j}
$
and
$\|e^{i\frac{z_j}{n}P_j} \cdots e^{i\frac{z_1}{n}P_1}\| \leq 1
$
for $1 \leq j,k \leq d,$
\eqref{trott} will follow if we can show that 
$\|(e^{i\frac{z_j}{n}P_j} -I)\varphi_n\| \rightarrow 0$ as $n \rightarrow \infty.$
Note that, by the choice of the projections $E_n$, it follows that for $1 \leq j \leq d,$
$
|((e^{i\frac{z_j}{n}P_j} -I)\varphi_n)(x)| \leq |((e^{i\frac{1}{n}P_j}
-I)\varphi_n)(x)|, 
$ for $0 \leq z_j \leq 1/2$ and $x \in \mathbb{R}^d
$. Also,  
\[
|((e^{i\frac{z_j}{n}P_j} -I)\varphi_n)(x)| \leq |((e^{-i\frac{1}{n}P_j}
-I)\varphi_n)(x)|, \qquad-1/2 
\leq z_j <0.
\]  
However the bound  $\sum_{1 \leq j \leq d}
\|\Phi_n(P_j)\varphi_n\|^2 \leq C_1$ 
implies that 
$
\lim_{n\to\infty}\|(e^{\pm i\frac{1}{n}P_j} -I)\varphi_n)\| = 0,
$
which proves the claim.
\end{proof}

\subsubsection{\bf{Proof that neither problem lies in $\Sigma_1^G\cup\Pi_1^G$.}}
Finally, we shall complete the proof of Theorem \ref{thm:comp-res} by showing that $\{\Xi_{\mathrm{sp}},\Omega_\infty\}\not\in\Sigma_1^G\cup\Pi_1^G$ and $\{\Xi_{\mathrm{sp},\epsilon},\Omega_\infty\}\not\in\Sigma_1^G\cup\Pi_1^G$.

\begin{proof}
\textbf{Step I: $\{\Xi_{\mathrm{sp}},\Omega_\infty\}\notin \Sigma_1^G$}.
Suppose for a contradiction that there exists a $\Sigma_1^G$ tower $\Gamma_n$ which solves the computational problem $\{\Xi_{\mathrm{sp}},\Omega_\infty\}$. Now let $V$ be any (real-valued) positive potential in the class $\Omega_\infty$ such that the corresponding Schr{\"o}dinger operator is self-adjoint and has a unique ground state (the operator must be bounded below). Call the associated operator $H_0$. For instance, in one dimension this could be the quantum harmonic oscillator $V(x)=x^2$, and examples in arbitrary dimension (the harmonic oscillator in $d>1$ dimensions does not have a unique ground state) are well known in the physics literature. In this case, let $\phi_0$ be the normalised ground state and $E$ be the orthogonal complement of the span of this function intersected with the domain of $H_0$. Assume that $H_0\phi_0=c\phi_0$. Denoting the standard $\mathrm{L}^2(\mathbb{R}^d)$ inner product by $\langle \cdot,\cdot\rangle$, it follows that there exists some $\eta>0$ such that
$$
\langle H_0 \phi,\phi\rangle \geq (c+\eta)\left\|\phi\right\|^2,\quad \forall\phi\in E.
$$

There exists $n$ such that there is a point $z_n\in\Gamma_n(V)$ with $\left|z_n-c\right|\leq \eta/20$ and such that $\Gamma_n(V)$ guarantees there is a point in the spectrum $\Xi_{\mathrm{sp}}(V)$ of distance at most $\eta/20$ to $z_n$. Hence $\Gamma_n(V)$ guarantees there is a point in the spectrum $\Xi_{\mathrm{sp}}(V)$ of of distance at most $\eta/10$ from $c$. There also exists a finite set $S=\{x^1,...,x^{M(n)}\}$ such that the computation of $\Gamma_n(V)$ only depends on the potential $V$ evaluated at points in $S$. Let $V_m$ be a sequence of real-valued continuous potentials such that $0\leq V_m(x)\leq 1$, $V_m(x^j)=0$ $\forall x^j\in S$ and such that $V_m$ converges pointwise almost everywhere to $1$ as $m\rightarrow\infty$. By construction and the definition of a general algorithm (Definition \ref{alg}) we must have for all $a\in\mathbb{R}_{+}$ that $\Gamma_n(V+aV_m)=\Gamma_n(V)$. In particular, this includes the guarantee of a point in the spectrum $\Xi_{\mathrm{sp}}(V+aV_m)$ of distance at most $\eta/10$ from $c$. We will show that this gives rise to a contradiction for a choice of $a\in\mathbb{R}_{+}$ and $m$.

Indeed, choose $m$ large such that
$
\langle V_m \phi_0,\phi_0\rangle\geq \frac{10}{11},
$
and set $a=\eta/2$. It is well known that the minimum of the spectrum $\Xi_{\mathrm{sp}}(V+aV_m)$ is given by
$$
\inf_{\phi\in \mathcal{D}(H_0):\left\|\phi\right\|=1}\langle (H_0+aV_m) \phi,\phi\rangle.
$$
In particular, $H_0+aV_m$ and $H_0$ have the same domain as $V_m$ is bounded. Now let $\phi\in\mathcal{D}(H_0)$ of norm $1$. Without loss of generality by a change of phase, we can write 
$
\phi=\delta \phi_0+\sqrt{1-\delta^2}\phi_1,
$
with $\phi_1\in E$ of unit norm and $\delta\in[0,1]$. Using the fact that $H_0\phi_0=c\phi_0$ and $H_0$ is self-adjoint and $\langle \phi_0,\phi_1\rangle=0$, we have that
\begin{align*}
\langle (H_0+aV_m) \phi,\phi\rangle&=\delta^2c+(1-\delta^2)\langle H_0 \phi_1,\phi_1\rangle+\delta^2a\langle V_m \phi_0,\phi_0\rangle\\
&\quad +a(1-\delta^2)\langle V_m \phi_1,\phi_1\rangle+2\mathrm{Re}(a\delta\sqrt{1-\delta^2}\langle V_m \phi_0,\phi_1\rangle)\\
&\geq c+(1-\delta^2)\eta+\frac{10}{11}\delta^2 a-2a\delta\sqrt{1-\delta^2},
\end{align*}
where we have used that $V_m$ is positive to throw away the $\langle V_m \phi_1,\phi_1\rangle$ term. It follows that the minimum of the spectrum of $H_0+aV_m$ is at least
\[
c+\inf_{\delta\in[0,1]}\eta(1-(1-5/11)\delta^2-\delta\sqrt{1-\delta^2})> c+\frac{\eta}{10},
\]
yielding the required contradiction.

\textbf{Step II: $\{\Xi_{\mathrm{sp}},\Omega_\infty\}\notin \Pi_1^G$}.
We argue as in Step I but now the proof is less involved. Suppose for a contradiction that there exists a $\Pi_1^G$ tower $\Gamma_n$ which solves the computational problem $\{\Xi_{\mathrm{sp}},\Omega_\infty\}$. We let $H_0$, $V$, $\phi_0$ and $E$ be as in Step I, where we also assume as before that $H_0\phi_0=c\phi_0$. We also assume that $c\geq 0$ and $V(x)\geq 1$.

Arguing as before, there exists some $n$ such that $\Gamma_n(V)$ guarantees that the spectrum is disjoint from the interval $[c-3/2,c-1/2]$. Again, there exists a finite set $S=\{x^1,...,x^{M(n)}\}$ such that the computation of $\Gamma_n(V)$ only depends on the potential $V$ evaluated at points in $S$. Let $V_m$ be a sequence of real-valued continuous potentials such that $-1\leq V_m(x)\leq 0$, $V_m(x^j)=0$ $\forall x^j\in S$ but now such that $V_m$ converges pointwise almost everywhere to $-1$ as $m\rightarrow\infty$. Note that we must have $V+V_m\in\Omega_\infty$ since we assume the pointwise inequality $V(x)\geq 1$. By construction and the definition of a general algorithm (Definition \ref{alg}) we must have that $\Gamma_n(V+V_m)=\Gamma_n(V)$. In particular, this includes the guarantee that the spectrum of $H_0+V_m$ is disjoint from the interval $[c-3/2,c-1/2]$. But we have that
$$
\langle (H_0+V_m-(c-1))\phi_0,\phi_0\rangle=\langle V_m\phi_0,\phi_0\rangle +1\rightarrow 0,
$$
as $m\rightarrow\infty$. It follows for some large $m$ that $\left\|R(c-1,H_0+V_m)\right\|^{-1}\leq 1/4$ and hence that the spectrum of $H_0+V_m$ intersects the interval $[c-3/2,c-1/2]$, since the operator is self-adjoint. But this contradicts the $\Pi_1^G$ guarantee.

\textbf{Step III: $\{\Xi_{\mathrm{sp},\epsilon},\Omega_\infty\}\notin \Pi_1^G\cup\Sigma_1^G$}.
The arguments are the same as in Steps I and II. We note that the pseudospectrum is simply the $\epsilon$ ball neighbourhood of the spectrum in these self-adjoint cases. The arguments work once we scale the operators by $N/\epsilon$ for some large $N$ in order to gain the relevant separations.
\end{proof}

\section{Proofs of Theorem \ref{linear_systems_thrm} and Theorem \ref{thrm:norm_inverse}}\label{linear_systems_proofs}

\begin{proof}[Proof of Theorem \ref{linear_systems_thrm}]
We have that $\mathrm{SCI}(\Xi_{\mathrm{inv}},\Omega_1)_{\mathrm{A}} \geq\mathrm{SCI}(\Xi_{\mathrm{inv}},\Omega_1)_{\mathrm{G}} \geq \mathrm{SCI}(\Xi_{\mathrm{inv}},\Omega_2)_{\mathrm{G}}$ and $\mathrm{SCI}(\Xi_{\mathrm{inv}},\Omega_3)_{\mathrm{G}} \geq 1$. It is also clear that $\{\Xi_{\mathrm{inv}},\Omega_4\}\not\in\Delta_0^G$. Hence it is enough to prove that $\mathrm{SCI}(\Xi_{\mathrm{inv}},\Omega_2)_{\mathrm{G}} \geq 2$, $\mathrm{SCI}(\Xi_{\mathrm{inv}},\Omega_1)_{\mathrm{A}} \leq 2$, $\mathrm{SCI}(\Xi_{\mathrm{inv}},\Omega_3)_{\mathrm{A}} \leq 1$ and $\{\Xi_{\mathrm{inv}},\Omega_4\}\in\Delta_1^A$.

{\bf Step I:} We start by showing that $\mathrm{SCI}(\Xi_{\mathrm{inv}},\Omega_2)_{\mathrm{G}} \geq 2$.
For $n,m\in\Nb\setminus\{1\}$ let
\[B_{n,m}:=\begin{pmatrix}
1/m& & & &1\\
 &1& & & \\
 & &\ddots& & \\
 & & &1& \\
1& & & &1/m\\
\end{pmatrix}
\in\Cb^{n\times n}\]
and for a sequence $\{l_n\}_{n \in \mathbb{N}} \subset\Nb\setminus\{1\}$ set 
$$
A:= \bigoplus_{n=1}^{\infty} B_{l_n,n+1}.
$$ 
Clearly, $A$ defines an invertible operator on $l^2(\mathbb{N})$ with bounded inverse.
Furthermore, we define $b=\{b_j\} \in l^2(\mathbb{N})$ 
such that 
$$
b_j = 
\begin{cases}
\displaystyle\frac{1}{n+2} & \displaystyle j = 1+\sum_{i=1}^{n}l_i, \quad n\in\Nb_0\\
0 & \text{otherwise}.
\end{cases}
$$
Let  also $C_m:=\diag\{1/m,1,1,\ldots\}$ and note that its inverse is given by 
$\diag\{m,1,1,\ldots\}$. We argue by contradiction and suppose that there is a General tower of algorithms $\Gamma_n$ of height one such that $\Gamma_n(A,b) \rightarrow \Xi_{\mathrm{inv}}(A,b)$ as $n \rightarrow \infty$ for $(A,b)\in\Omega_2$. For such $A$, $b$ and $k \in \mathbb{N}$ let $N(A,b,k)$ denote the smallest integer such that the evaluations from $\Lambda_{\Gamma_k}(A,b)$ only include matrix entries $A_{ij}=\langle Ae_j,e_i \rangle$ with $i,j\leq N(A,b,k)$ and the entries
$b_i$ with $i\leq N(A,b,k)$. To obtain a particular counterexample $(A,b)$ we construct sequences $\{l_n\}_{n \in \mathbb{N}}$ and $\{k_n\}_{n \in \mathbb{Z}_+}$ inductively such that $A$ and $b$ are given by $\{l_n\}$ as above but $\Gamma_{k_n}(A,b) \nrightarrow \Xi_{\mathrm{inv}}(A,b)$. As a start, set $k_0=l_0:=1$. 
The sequence $\{x^{(1)}_j\}_{j\in\Nb} := (C_{2})^{-1}P_{1}b$ has a $1$ as its first entry and since, by assumption, $\Gamma_k \rightarrow \Xi_{\mathrm{inv}}$, there is a $k_1$ such that, for all $k \geq k_1$, the first entry of $\Gamma_{k}(C_2,P_1b)$ is closer to $1$ than $1/2$. Then, choose $l_1>N(C_2,P_1b,k_1)-l_0$.
Now, for $n>1$, suppose that $l_0,\ldots,l_{n-1}$ and $k_0, \hdots, k_{n-1}$ are already chosen. Set $s_{n}:=\sum_{i=0}^{n-1} l_i$. Then also $P_{s_n}b$ is already determined and
$$
x^{(n)}_{s_n+1} = 1, \quad\text{where}\quad \{x^{(n)}_j\}_{j\in\Nb} := (B_{l_1,2} \oplus B_{l_2,3} \oplus \ldots \oplus B_{l_{n-1},n} \oplus C_{n+1})^{-1}P_{s_n}b.
$$ 
Since, by assumption, $\Gamma_k \rightarrow \Xi_{\mathrm{inv}}$, there is a $k_n$ such that for all $k \geq k_n$
\begin{equation*}
|x^{(n,k)}_{s_n+1} -1| \leq 1/2, \quad\text{where}\quad \{x^{(n,k)}_j\}_{j\in\Nb} := \Gamma_{k}(B_{l_1,2} \oplus B_{l_2,3} \oplus \ldots \oplus B_{l_{n-1},n} \oplus C_{n+1},P_{s_n}b).
\end{equation*}
Now, choose $l_{n}>N(B_{l_1,2} \oplus B_{l_2,3} \oplus \ldots \oplus B_{l_{n-1},n} \oplus C_{n+1},P_{s_n}b,k_n)-l_0-l_1-\ldots-l_{n-1}$.
By this construction we get for the resulting $A$ and $b$ that for every $n$
$$
\Gamma_{k_n}(A,b) = \Gamma_{k_n}(B_{l_1,2} \oplus B_{l_2,3} \oplus \ldots \oplus B_{l_{n-1},n} \oplus C_{n+1},P_{s_n}b).
$$
In particular $\lim_{k\rightarrow \infty}\Gamma_{k}(A,b)$ does not exist in $l^2(\mathbb{N})$, a contradiction. 

{\bf Step II:} $\{\Xi_{\mathrm{inv}},\Omega_1\}\in\Delta_3^A$. Let $A$ be invertible and $Ax=b$ with the unknown $x$. Since $P_m$ are compact projections converging strongly to the identity, we get that the ranks $\rk P_m = \rk (AP_m) = \rk(P_nAP_m)$ for every $m$ and all $n\geq n_0$ with an $n_0$ depending on $m$ and $A$.

We let $\zeta_{m,n}$ be an approximation of $\sigma_1(P_mA^*P_nAP_m)$ from above, accurate to $n^{-1}$ and computed using finitely many arithmetic operations and comparisons (over $\mathbb{Q}$) with $\Delta_1$-information (Proposition \ref{REC_SING}). We also let $B_{m,n}$ denote the corresponding rational approximation of $P_nAP_m$ used in the computation, which without loss of generality, we assume to be accurate to at least $(2n)^{-1}$ in the operator norm. Let $b_n$ denote a rational approximation of $P_nb$ correct to $n^{-1}$. We can define
\[\Gamma_{m,n}(A,b):=
\begin{cases}
\quad\quad\quad \{0\}_{j\in\Nb} & \text{ if } \zeta_{m,n}\leq\frac{1}{m}\\
(B_{m,n}^*B_{m,n})^{-1}B_{m,n}^*b_n& \text{ otherwise.}
\end{cases}\]
Note that for every $A$, $b$, $m$, $n$ in view of Proposition \ref{REC_SING} and any standard algorithm for finite dimensional linear problems, these approximate solutions can be computed by finitely many arithmetic operations on finitely many entries of $A$ and $b$ (with $\Delta_1$-information), hence $\Gamma_{m,n}$ are general algorithms in the sense of Definition \ref{alg} and require only a finite number of arithmetic operations.

Moreover, they converge to $y_m:=(P_mA^*AP_m)^{-1}P_mA^*b$ as $n\to\infty$. It is well known that $y_m$ is also a (least squares) solution of the optimisation problem $\|AP_my-b\|\to\min$, that is \[\|AP_my_m-b\|\leq \|AP_mx-b\|\leq\|A\|\|P_mx-A^{-1}b\|=\|A\|\|P_mx-x\|\to 0\] as $m\to\infty$. Therefore $\|y_m-x\|=\|P_my_m-x\|$ is not greater than \[\|A^{-1}\|\|A(P_my_m-x)\|=\|A^{-1}\|\|AP_my_m-b\|\leq\|A^{-1}\|\|A\|\|P_mx-x\|\to 0,\] which yields the convergence $y_m\to x$ and finishes the proof of Step II.

{\bf Step III:} $\{\Xi_{\mathrm{inv}},\Omega_3\}\in\Delta_2^A$. Let $f$ be a bound on the dispersion of $A$. The smallest singular values of the operators $AP_m$ are uniformly bounded below by $\|A^{-1}\|^{-1}$ which, together with $\|P_{f(m)}AP_m-AP_m\|\to0$, yields that the limit inferior of the smallest singular values of $P_{f(m)}AP_m$ is positive, hence the inverses of the operators $C_m:=P_mA^*AP_m$ and $D_m:=P_mA^*P_{f(m)}AP_m$ on the range of $P_m$ exist for sufficiently large $m$ and have uniformly bounded norm. Moreover, $\|C_m^{-1}-D_m^{-1}\|\leq \|C_m^{-1}\|\|D_m-C_m\|\|D_m^{-1}\|$ tend to zero as $m\to\infty$.

This particularly implies that the norms $\|y_m-(P_mA^*P_{f(m)}AP_m)^{-1}P_mA^*b\|$ with $y_m$ as above tend to zero as $m\to\infty$, and we easily conclude that the norms $\|y_m-\Gamma_{m,f(m)}(A,b)\|$ tend to zero as well. With the convergence $\|y_m-x\|\to 0$ from the previous proof, now also 
$\|x-\Gamma_{m,f(m)}(A,b)\|\to 0$ holds as $m\to\infty$, which is the assertion $\mathrm{SCI}(\Xi_{\mathrm{inv}},\Omega_3)_{\mathrm{A}} \leq 1$.

{\bf Step IV:} We prove that $\{\Xi_{\mathrm{inv}},\Omega_4\}\in\Delta_1^A$. To do this we take the algorithm constructed in Step III, and note that by increasing $m$ if necessary, we can assume that $\zeta_{m,f(m)}\geq\sigma_1(P_mA^*P_{f(m)}AP_m)> 1/m$. Hence we only need to bound the error of the approximation. We have that
\begin{align*}
\|A\Gamma_{m,f(m)}(A,b)-b\|&\leq\|P_{f(m)}AP_m\Gamma_{m,f(m)}(A,b)-P_mb\|\\
&\quad\quad\quad\quad+\|P_mb-b\|+\|(I-P_{f(m)})AP_m\|\|\Gamma_{m,f(m)}(A,b)\|\\
&\leq\|P_{f(m)}AP_m\Gamma_{m,f(m)}(A,b)-P_mb\|+c_m(1+\|\Gamma_{m,f(m)}(A,b)\|)
\end{align*}
and hence the bound
$$
\|\Gamma_{m,f(m)}(A,b)-A^{-1}b\|\leq M\big[\|P_{f(m)}AP_m\Gamma_{m,f(m)}(A,b)-P_mb\|+c_m(1+\|\Gamma_{m,f(m)}(A,b)\|)\big].
$$
Note that this final bound converges to zero and it is also clear that we can approximate it to arbitrary accuracy using finitely many arithmetic operations and comparisons.
\end{proof}

\begin{remark}
The technique used with uneven sections to obtain the bound $\mathrm{SCI}(\Xi_{\mathrm{inv}},\Omega_1)_{\mathrm{A}} \leq 2$ is also referred to as asymptotic Moore--Penrose inversion as well as modified (or non-symmetric) finite section method  in the literature, although written in a different form, and is widely used (see e.g. \cite{Heinig_Hellinger,Silbermann_mod,Silbermann22,Charly1,SeSi_mod}). Also the idea to exploit bounds on the off diagonal decay is considered e.g. in \cite{Charly2} or in the theory of band-dominated operators and operators of Wiener type (cf. \cite{Roch_Silbermann_LimitOps,Lindner,SeiSurvey}).
\end{remark}

\begin{proof}[Proof of Theorem \ref{thrm:norm_inverse}]
Clearly $\{\Xi_{\mathrm{norm}},\Omega_3\}\not\in\Delta_1^G$ and $\Omega_2\subset\Omega_1$, so it is enough to prove $\{\Xi_{\mathrm{norm}},\Omega_1\}\in\Pi_2^A$, $\{\Xi_{\mathrm{norm}},\Omega_3\}\in\Pi_1^A$ and $\{\Xi_{\mathrm{norm}},\Omega_2\}\not\in\Delta_2^G$. We start with the latter. Let $\{l_n\}_{n\in\mathbb{N}}$ be some sequence of integers $l_n\geq2$.
Define
 \[
A:=\bigoplus_{n=1}^{\infty} B_{l_n} - I, \qquad B_n:=\begin{pmatrix}
1& & & &1\\
 &0& & & \\
 & &\ddots& & \\
 & & &0& \\
1& & & &1\\
\end{pmatrix}
\in\Cb^{n\times n}.\]
Clearly, such $A$ are invertible and their inverses have norm one. Suppose that $\{\Gamma_k\}$ is a height-one General tower of algorithms and that $\Lambda_{\Gamma_k}(A)$ only includes matrix entries $A_{ij}=\langle Ae_j,e_i \rangle$ with $i,j\leq N(A,k)$. In order to find a counterexample we again construct an appropriate sequence $\{l_n\}\subset\Nb\setminus\{1\}$ by induction: For $C:=\diag\{1,0,0,0,\ldots\}$ one obviously has $\|(C-I)^{-1}\|^{-1}=0$. As a start, choose $k_0:=1$ and $l_1>N(C-I,k_0)$.
Now, suppose that $l_1,\ldots,l_n$ are already chosen. Then the operator given by the matrix $B_{l_1} \oplus \ldots \oplus B_{l_n} \oplus C-I$ is not invertible, hence there exists a $k_n$ such that, for every $k\geq k_n$,
$$
\Gamma_k(B_{l_1} \oplus \ldots \oplus B_{l_n} \oplus C-I)<\frac{1}{2}.
$$
Now finish the construction by choosing 
$
l_{n+1}>N(B_{l_1} \oplus \ldots \oplus B_{l_n} \oplus C-I,k_n)-l_1-l_2-\ldots-l_n.
$

So, we see that
$$
\Gamma_{k_n}(A)=\Gamma_{k_n}(B_{l_1} \oplus \ldots \oplus B_{l_n} \oplus C-I) \nrightarrow \|A^{-1}\|^{-1}=1, \quad n \rightarrow \infty,
$$ 
a contradiction. Thus $\{\Xi_{\mathrm{norm}},\Omega_2\}\not\in\Delta_2^G$.

In order to prove $\{\Xi_{\mathrm{norm}},\Omega_1\}\in\Pi_2^A$ we introduce the numbers
\begin{align*}
\gamma &:= \|A^{-1}\|^{-1}=\min\{\sigma_1(A),\sigma_1(A^*)\}\\
\gamma_m &:= \min\{\sigma_1(AP_m),\sigma_1(A^*P_m)\}\\
\gamma_{m,n} &:= \min\{\sigma_1(P_nAP_m),\sigma_1(P_nA^*P_m)\}
\end{align*}
and note that $\gamma_m\downarrow_m\gamma,$ and $\gamma_{m,n}\uparrow_n\gamma_m$ for every fixed $m$. We let $\delta_{m,n}$ be an approximation to $\gamma_{m,n}$ from above to accuracy $n^{-1}$, computed using finitely many arithmetic operations and comparisons (over $\mathbb{Q}$) with $\Delta_1$-information (Proposition \ref{REC_SING}). Since $\gamma_{m,n}\leq\delta_{m,n}\leq\gamma_{m,n}+1/n$, it follows that $\{\delta_{m,n}\}_n$ converges to $\gamma_m$ for every $m$, and for $\epsilon>0$ there is an $m_0$, and for every $m\geq m_0$ there is an $n_0=n_0(m)$ such that
\begin{equation}\label{ENorm}
|\gamma-\delta_{m,n}|\leq |\gamma-\gamma_m|+|\gamma_m-\gamma_{m,n}|+|\gamma_{m,n}-\delta_{m,n}|
\leq \epsilon/3+\epsilon/3+1/n\leq \epsilon
\end{equation}
whenever $m\geq m_0$ and $n\geq n_0(m)$. Since $\delta_{m,n}$ and hence 
$\Gamma_{m,n}(A):=\delta_{m,n}$ can again be computed with finitely many arithmetic operations by Proposition \ref{REC_SING}, this provides an arithmetic tower of algorithms of height two, with the final convergence from above and hence easily completes the proof that $\{\Xi_{\mathrm{norm}},\Omega_1\}\in\Pi_2^A$.
On $\Omega_3$ we apply \eqref{ENorm} with $n=f(m)$ and straightforwardly check that $\Gamma_{m}(A):=\delta_{m,f(m)}$ provides a height $1$ tower. If we wish to have $\Pi_1^A$ convergence (i.e. convergence from above) then we need to use the sequence $\{c_n\}$ that bounds the dispersion to bound the difference between $\delta_{m,f(m)}$ and $\gamma_m$ and choose $\Gamma_{m}(A):=\delta_{m,f(m)}+c_m$.
\end{proof}

\section{Smale's problem on roots of polynomials and Doyle-McMullen towers}\label{roots_pols}

In this section, we recall the definition of a tower of algorithms from \cite{Doyle_McMullen}. We will name this type of tower a Doyle--McMullen tower and demonstrate how the results in \cite{McMullen1} and \cite{Doyle_McMullen} can be put into the framework of the SCI. In particular, we will demonstrate how the construction of the Doyle--McMullen tower in \cite{Doyle_McMullen} can be viewed as a tower of algorithms defined in Definition \ref{tower_funct}. 
Note that one can compute zeros of a polynomial if one allows arithmetic operations and radicals and can pass to a limit. However, what if one cannot use radicals, but rather iterations of a rational map? A natural choice of such a rational map would be Newton's method. The only problem is that the iteration may not converge, and that motivated the question by Smale quoted in the introduction.

As we now know from \cite{McMullen1}, the answer is no. However, the results in \cite{Doyle_McMullen} show that the quartic and the quintic can be solved with several rational maps and limits, while this is not the case for higher degree polynomials. Below we first quote their results and then specify a particular tower of height three in the form that it can be viewed as a tower of algorithms in the sense of this paper.

\subsection{Doyle--McMullen towers}
A \emph{purely iterative algorithm} \cite{smale_question} is a rational map \footnote{I.e. it's a rational map of the coefficients of $p$.} 
	$$
	T:\mathbb P_d\to\rat_m,\; p\mapsto T_p
	$$
which sends any polynomial $p$ of degree $\leq d$ to a rational function $T_p$ of a certain degree $m$.
An important example of a purely iterative algorithm is \emph{Newton's method}. Furthermore, Doyle and McMullen call a purely iterative algorithm \emph{generally convergent} if
	$$
	\lim_{n\to\infty}T_p^n(z)\text{ exists for }(p,z)\text{ in an open dense subset of }\mathbb P_d\times\hat\C.
	$$
Here $T_p^n(z)$ denotes the $n$th iterate $T_p^n(z)=T_p(T_p^{n-1}(z))$ of $T_p$.
For instance, Newton's method is generally convergent \emph{only} when $d=2$. However, given a cubic polynomial $p\in\mathbb P_3$ one can define an appropriate rational function $q\in\rat_3$ whose roots coincide with the roots of $p$, and for which Newton's method \emph{is} generally convergent (see \cite{McMullen1}, Proposition 1.2). In \cite{Doyle_McMullen} the authors provide a definition of a tower of algorithms, which we quote verbatim:
	\begin{definition}[Doyle--McMullen tower]\label{DMcMullen}
	A tower of algorithms is a finite sequence of generally convergent algorithms, linked together serially, so the output of one or more can be used to compute the input to the next. The final output of the tower is a single number, computed rationally from the original input and the outputs of the intermediate generally convergent algorithms.
	\end{definition}
\begin{theorem}[McMullen \cite{McMullen1}; Doyle and McMullen \cite{Doyle_McMullen}]\label{SCI_McMullen}
For $\mathbb P_d$ there exists a generally convergent algorithm only for $d\le 3$. Towers of algorithms exist additionally for $d=4$ and $d=5$ but not for $d\geq6$.

\end{theorem}

Note that, as shown in \cite{Smale3}, there are generally convergent algorithms if, in addition, one allows the operation of complex conjugation.  
In the following, we present how the Doyle--McMullen towers can be recast in the form of a general tower as defined in Definition \ref{tower_funct}.

\subsection{A height three tower for the quartic}

In the following $X,Y,\dots$ denote variables in the polynomials while $x,y,\dots \in \mathbb C$.
We build the tower following the standard reduction path, see e.g. \cite{Dickson59}.
Given 
$$
p(X):=X^4+a_1 X^3+a_2 X^2+a_3 X + a_4
$$
one first transforms the equation by change of variable $Y=X+a_1/4$ to arrive into the polynomial
$$
q(Y):= Y^4+b_2Y^2+b_3  Y +b_4,
$$
which one writes, with help of a parameter $z$, as 
$
q(Y) = (Y^2+z)^2 -r(Y,z)
$
where
$$
r(Y,z)= (2z-b_2)Y^2-b_3 Y+z^2-b_4.
$$
Here one wants a value of $z$ such that $r(Y,z)$ becomes a square which requires the discriminant to vanish:
$
4(2z-b_2)(z^2-b_4) - b_3^2=0.
$
Viewing this as polynomial in $Z$, making a change of variable $W=Z+ (1/6) b_2$ and scaling the polynomial to monic we arrive at asking for a root of 
\begin{equation}\label{kolmas}
s(W):= W^3+ c_2 W+c_3.
\end{equation}
As all these are rational computations on the coefficients of $p$, we shall not express them explicitly.

We denote by $N(f,\xi_0)$ the function in Newton's iteration with initial value $\xi_0$:
$$
\xi_{j+1} := N(f(\xi_j))   \     \text{ where } \   N(f(\xi))=\xi - \frac{f(\xi)}{f'(\xi)}
$$
and further by $N_j$ the mapping from initial data to the $j^{th}$ iterate
$
N_j :  \  (f,\xi_0)\mapsto \xi_j.
$
We shall apply Newton's iteration to the rational function \cite{Doyle_McMullen}
$$
t(W):=\frac{s(W)}{3c_2 W^2+9c_3W -c_2^2}.
$$
 Thus
$
w_j=N_j(t,w_0)
$
denotes the $j^{th}$ iterate $w_j$ for a zero for $s(w)=0$. This iteration converges in an open dense set of initial data. Denote $w_\infty:= \lim_{j\rightarrow \infty} w_j$. Now we change the variable $Z=W-(1/6)b_2$ and, denoting by $z_j$ and $z_\infty$ the corresponding values, we obtain $r(Y,z_\infty)$ as a square:
$$
r(Y,z_\infty)=(2z_\infty-b_2)\left(Y-\frac{b_3}{2(2z_\infty-b_2)}\right)^2.
$$

To find a zero of $q(Y)$ we shall need to have a generally convergent iteration for $\sqrt{2z-b_2}$. Thus, we set 
$
u_j(V):= V^2+b_2-2z_j 
$
and apply Newton's method for this, starting with initial guess $v_0$ and iterating $k$ times and set
$
v_{k, j}:= N_{k}(u_j, v_0).
$
From 
$q(Y)= (Y^2+ z_\infty)^2- r(Y,z_\infty)=0
$ we move to solve 
one of the factors
$$
Q(Y)=Y^2+z_\infty -\sqrt{2z_\infty-b_2}\left(Y-\frac{b_3}{2(2z_\infty-b_2)}\right)=0.
$$
However, we can do this only based on approximative values for the parameters, so we set
$$
Q_{k,j}(Y)=Y^2+z_j -v_{k,j}\left(Y-\frac{b_3}{2(2z_j-b_2)}\right)=0.
$$
Now apply Newton's iteration to this, say $n$ times, using starting value $y_0$ and denote the output by $y_{n,k,j}$:
$$
y_{n,k,j}=N_n(Q_{k,j}, y_0).
$$
Finally, we set
$
x_{n,k,j}=y_{n,k,j}-a_1/4
$
in order to get an approximation to a root of $p$.
Suppose now $j=n_1, k=n_2, n=n_3$. If $n_1\rightarrow \infty$ then $w_{n_1}\rightarrow w_\infty$ and hence $z_{n_1}\rightarrow z_\infty$, too.
It is natural to denote $u(V):= V^2+b_2-2z_\infty$ and correspondingly
$
v_{n_2}:=N_{n_2}(u,v_0)
$
and
$$
Q_{n_2}(Y)=Y^2+z_\infty -v_{n_2}\left(Y-\frac{b_3}{2(2z_\infty-b_2)}\right)=0.
$$
Then in an obvious manner
$
x_{n_3,n_2}= N_{n_3}(Q_{n_2}, y_0)  -a_1/4.
$
Then we have
$
\lim_{n_1\rightarrow \infty} x_{n_3,n_2,n_1} = x_{n_3,n_2}.
$
If we denote
$
x_{n_3} = N_{n_3} (Q, y_0)-a_1/4,
$
then clearly
$
\lim_{n_2\rightarrow \infty} x_{n_3,n_2} = x_{n_3}.
$
Finally 
$
x_\infty= \lim_{n_3\rightarrow \infty}x_{n_3}
$
is a root of $p$.

\subparagraph{\textbf{The link to the $\mathrm{SCI}$}}

One special feature of these towers, which are built on generally convergent algorithms, is the following: in addition to the polynomial $p$, the initial values for the iterations have to be read into the process via evaluation functions.
Denoting the initial values for the three different Newton's iterations by $d_0=(w_0, v_0,y_0) \in \mathbb C^3$ we can now put this Doyle--McMullen tower in the form of a general tower as defined in Definition \ref{tower_funct}, with the slight weakening that, for each $p\in\mathbb P_4$, the tower might converge only at a dense subset of initial values. In particular, set
\begin{equation*}
\begin{split}
\Gamma_{n_3} &:  \mathbb P_4 \times  \mathbb C^3  \rightarrow \mathbb C,  \text{  by  }
(p, d_0) \mapsto x_{n_3},\\
\Gamma_{n_3,n_2} &:  \mathbb P_4 \times  \mathbb C^3  \rightarrow \mathbb C \text{  by  }
(p, d_0) \mapsto x_{n_3,n_2},\\
\Gamma_{n_3,n_2,n_1} &:  \mathbb P_4 \times \mathbb C^3 \rightarrow \mathbb C  \text{  by  }
(p, d_0) \mapsto x_{n_3,n_2,n_1}.
\end{split}
\end{equation*}
Thus, if we let $\Omega = \mathbb P_4 \times \mathbb C^3$ and $\Xi,\mathcal{M}$ be as in Example \ref{Ex1} (III), and complement $\Lambda$ by the mappings that read $w_0, v_0,y_0$ from the input, then by the construction above and Theorem \ref{SCI_McMullen} we have that 
$$
\mathrm{SCI}(\Xi,\Omega)_{\mathrm{DM}} \in \{2,3\}.
$$

 \subsection{A height three tower for the quintic}
Let
$$
p(X)=X^5+a_1X^4+a_2X^3+a_3X^2+a_4X+a_5
$$
be the given quintic. Doyle and McMullen \cite{Doyle_McMullen} give a generally convergent algorithm for the quintic in Brioschi form. Thus, one needs first to bring the general quintic to Brioschi form, then apply the iteration and finally construct at least one root for $p(X)$. In the following, we outline a path for doing this, which follows L. Kiepert \cite{Kiepert} except that the Brioschi quintic is solved by Doyle--McMullen iteration rather than by using Jacobi sextic. This path can be found in \cite{King}.  

One begins applying a Tschirnhaus transformation
$
Y=X^2 -uX+v 
$
to arrive into {\it principal} form
$$
q(Y)= Y^5+b_3Y^2+b_4Y+b_5.
$$
Here $v$ is obtained from a linear equation but to solve $u$ one needs to solve a quadratic equation 
$
 Q(U)=U^2 +\alpha U + \beta,
$
 where the coefficients $\alpha, \beta$ are rational expressions of the coefficients of $p(X)$, (see for example p. 100, eq. (6.2-9) in \cite{King}).
  
Here is the first application of Newton's method. We are given an initial value $u_0$ and iterate $j$ times $u_j= N_{j}(Q,u_0).$ We may assume that $v$ is known exactly but we only have an approximation $u_j$ to make the transformation. So, suppose the Newton iteration converges to $u_\infty$. Thus, we make the transformation using $u_j$ and {\it force} the coefficients $b_{2,j}=b_{1,j}=0$ while keep the others as they appear. The transformation being continuous yields polynomials
$$
q_j(Y)= Y^5+b_{3,j}Y^2+b_{4,j}Y+b_{5,j},
$$  
whose roots shall converge to those of $q(Y)$. The next step is to transform $q_j(Y)$ into Brioschi form. Let the Brioschi form corresponding to the exact polynomial $q(Y)$ be denoted by $B(Z)$
\begin{equation}\label{Brioschi form}
B(Z)=Z^5-10 C Z^3 +45 C^2 Z -C^2=0,
\end{equation}
while with $B_j(Z)$ we denote the exact Brioschi form corresponding to $q_j(Y)$. 
The transformation from $q(Y)$ to $B(Z)$ is of the form
\begin{equation}\label{Tsch}
Y= \frac{\lambda+\mu Z}{(Z^2/C)-3}.
\end{equation}
Here $\lambda$ satisfies a quadratic equation with coefficients being polynomials of the coefficients in the principal form (p. 107, eq. (6.3-28) in \cite{King}). Let us denote that quadratic by $R(L)$ when it comes from $q(Y)$ and by $R_j(L)$ when it comes from $q_j(Y)$ respectively. Thus here we meet our second application of Newton's method. So, we denote by 
$$\lambda_{k,j}:= N_k(R_j, \lambda_0)
$$ the output of iterating $k$ times for a solution of 
$
R_j(L)=0.
$ 
And, in a natural manner, we denote also
$$\lambda_k=N_k(R,\lambda_0) \quad \text{and} \quad \lambda=\lim_{k\rightarrow \infty} N_k(R,\lambda_0).$$
The corresponding values of $\mu_{k,j}, \mu_k$ and $\mu$ are then obtained by simple substitution (p. 107, eq. (6.3-30) in \cite{King}). The Tschirnhaus transformation with exact values $(\lambda,\mu)$ transforms the equation not yet to the Brioschi form with just one parameter $C$ but such that the constant term may be different. However, the last step is just a simple scaling, and then one is in the Brioschi form (\ref{Brioschi form}). However, when we apply the transformation with the approximated values
$(\lambda_{k,j},\mu_{k,j})$ or with $(\lambda_{k},\mu_{k})$ we do not arrive at the Brioschi form. So, we {\it force} the coefficients of the fourth and second powers to vanish and replace the coefficients of the first power to match with the coefficients in the third power. Finally, after scaling the constant terms we have the Brioschi quintics $B_{k,j}$ and $B_k$, e.g.
\begin{equation}\label{Brioschi-kj}
B_{k,j}(Z)=Z^5-10 C_{k,j} Z^3 +45 C_{k,j}^2 Z -C_{k,j}^2=0.
\end{equation}
Provided that the Newton iterations converge, that is, the initial values $
 (u_0, \lambda_0)$ are generic, these quintics converge to the exact one.

Here we apply the generally convergent iteration by Doyle and McMullen \cite{Doyle_McMullen}.
They specify a rational function 
$$
T_C(Z)=z-12 \frac{g_C(Z)}{g_C'(Z)}
$$
where $g$ is a polynomial of degree 6 in the variable $C$ and of degree 12 in $Z$.
Starting from an initial guess $w_o$ from an initial guess $w_{n+1}=T_C(T_C(w_n))$ to convergence and applying $T_C$ still once, we obtain, after a finite rational computation with these two numbers, two roots of the Brioschi, say $z_I$ and $z_{II}$. If applied to the approximative quintics and if the iteration is truncated after $n$ steps, together with the corresponding post-processing, we have obtained e.g. a pair $(z_{I, n , k , j},z_{II,n,k,j})$.
 
What remains is to invert the Tschirnhaus transformations. Suppose $z$ is a root of the exact Brioschi form (\ref{Brioschi form}). Then the corresponding root of the principal quintic is obtained 
immediately from (\ref{Tsch})
$$t
y= \frac{\lambda+\mu z}{(z^2/C)-3}.
$$
Naturally, we can only apply this using approximated values for the parameters. Finally, one needs to transform the (approximative roots) of the principal quintic to (approximative) roots for the original  general quintic $p(X)$. This is done by a rational function
$
X= r(Y)
$
where $r(Y)$ is of second order in $Y$ and  the coefficients are polynomials of the coefficients if the original $p(X)$ and $u$ and $v$ (p. 127, eq. (6.8-3) in \cite{King}). Again, we would be using only approximative values $u_j$ in place of the exact $u$. In any case, at the end we obtain a pair of approximations to the roots of the original quintic. If we put $n_1=j, n_2=k$ and $n_3=n$, then this pair could be denoted by $(x_{I, n_3,n_2,n_1}, x_{II, n_3,n_2,n_1})$.

\subparagraph{\textbf{The link to the $\mathrm{SCI}$}}
In the same way as with the quartic, we assume that the initial value $d_0 =(u_0, \lambda_0, w_0) \in \mathbb C^3$ is generic, so that all iterations converge for large enough values and since the transformations are continuous functions of the parameters in it, all necessary limits exist and match with each other. The functions $\Gamma_{n_3,n_2,n_1}$ can then be identified in a natural manner:
\begin{equation*}
\begin{split}
&\Gamma_{n_3}  \ :  \mathbb P_5 \times \mathbb C^3  \rightarrow \mathbb C^2,   \text{  by  }
(p,d_0) \mapsto (x_{I, n_3}, x_{II, n_3}),\\
&\Gamma_{n_3,n_2}  \ :  \mathbb P_5 \times \mathbb C^3 \rightarrow \mathbb C^2  \text{  by  }
(p,d_0) \mapsto (x_{I, n_3,n_2}, x_{II, n_3,n_2}),\\
&\Gamma_{n_3,n_2,n_1}  \ :  \mathbb P_5 \times \mathbb C^3 \rightarrow \mathbb C^2  \text{  by  }
(p,d_0) \mapsto (x_{I, n_3,n_2,n_1}, x_{II, n_3,n_2,n_1}),
\end{split}
\end{equation*}
where $(x_{I, n_3,n_2}, x_{II, n_3,n_2})$ and $(x_{I, n_3}, x_{II, n_3})$ are the limits as $n_1 \rightarrow \infty$ and $n_2 \rightarrow \infty$ respectively. These limits exist for initial values in an open dense subset of $\mathbb C^3$.
Hence, we let $\Omega = \mathbb P_5 \times \mathbb C^3$, and $\Xi,\mathcal{M},\Lambda$ be as in case of the quartic. Then, by the construction above and Theorem \ref{SCI_McMullen} we have, again in a slightly weakened sense, that 
$$
\mathrm{SCI}(\Xi,\Omega)_{\mathrm{DM}} \in \{2,3\}.
$$

\subsection{Particular initial guesses and height one towers}
The special feature of the above mentioned Doyle--McMullen towers is that they address the question of whether one can achieve converge to the roots of a polynomial $p$ for (almost) arbitrary initial guesses. With a slight change of perspective, one might also ask the question of how large the SCI gets if one applies purely iterative algorithms \textit{after a suitable clever choice} of initial values. And indeed, the answer to this question is very satisfactory: For polynomials of arbitrary degree, one can compute the whole set of roots (more precisely: approximate it in the sense of the Hausdorff distance) by a tower of height one which just consists of Newton's method.

\medskip 

The key tool for the choice of the initial values is the main theorem of \cite{HubbSS}:
\begin{theorem}[Hubbard, Schleicher and Sutherland \cite{HubbSS}]\label{THubSchleiSuth}
For every $d\geq 2$ there is a set $S_d$ consisting of at most $1.11 d \log^2 d$ points in $\Cb$ with the property that for every polynomial $p$ of degree $d$ and every root $z$ of $p$ there is a point $s\in S_d$ such that the sequence of Newton iterates $\{s_n\}_{n\in\Nb}:=\{N_p^n(s)\}_{n\in\Nb}$ converges to $z$.
In particular, the proof is constructive, and these sets $S_d$ can easily be computed.
\end{theorem}

A further important property of Newton's method is that, in the case of convergence, the speed is at least linear: If $z_n:=N_p^n(s)$ tend to a root $z$ of $p$ then there exists a constant $c$ such that $|z_n-z|\leq c/n$.
Finally we have the following.
\begin{proposition} \label{PRootDist}
Let $p$ be a polynomial of degree $d$, $\epsilon>0$ and $z_n:=N_p^n(s)$. If $|z_n-z_{n+1}|<\frac{\epsilon}{d}$ then there is a root $z$ of $p$ with $|z_n-z|<\epsilon$.
\end{proposition}
\begin{proof}
We have $\left|\frac{p(z_n)}{p' (z_n)}\right|=|z_n-z_{n+1}|<\frac{\epsilon}{d},
\text{ hence } |p(z_n)|<\frac{\epsilon |p' (z_n)|}{d}$.
Decompose $p(x)=a\Pi_{i=1}^d(x-x_i)$, notice that $p' (x)=a\sum_{j=1}^d\Pi_{i=1,i\neq j}^d(x-x_i)$, choose $j$ such that
$|\Pi_{i=1,i\neq j}^d(z_n-x_i)|$ is maximal, and conclude that 
\[|a\Pi_{i=1}^d(z_n-x_i)|=|p(z_n)|<\frac{\epsilon |p' (z_n)|}{d}\leq \epsilon |a\Pi_{i=1,i\neq j}^d(z_n-x_i)|,\] 
thus $|z_n-x_j|<\epsilon$. Now $z=x_j$ is a root as asserted.
\end{proof}

Let $p$ be a polynomial of degree $d$. For each $s\in S_d$ let $s_n$ denote the $n$th Newton iterates of $s$, and define
\begin{equation}\label{EPurelyIterative}
\Gamma_n(p):=\left\{s_n: s\in S_d, |s_n-s_{n+1}|<\frac{1}{\sqrt{n}}\right\}.
\end{equation}
Then $(\Gamma_n(p))$ converges to the set $\mathcal{Z}(p)$ of all zeros of $p$ in the Hausdorff metric.
Indeed, let $z$ be a zero of $p$. By Theorem \ref{THubSchleiSuth} there is an initial value $s\in S_d$ such that $s_n=N_p^n(s)$ tend to $z$ with at least linear speed, i.e.
\[|s_n-s_{n+1}|\leq |s_n-z|+|s_{n+1}-z|\leq \frac{2c}{n} <\frac{1}{\sqrt{n}}\]
for all large $n$, hence $s_n\in \Gamma_n(p)$ for all large $n$.
Conversely, each $s_n\in\Gamma_n(p)$ has the property that its distance to the set $\mathcal{Z}(p)$ is less than $\epsilon=\frac{d}{\sqrt{n}}$ by Proposition \ref{PRootDist}.

Therefore we define $\Omega_d=\mathbb P_d$ to be the set of polynomials of degree $d$, $\mathcal M$ the set of finite subsets of $\Cb$ equipped with the Hausdorff metric, and $\Xi:\Omega_d\to\mathcal M$ be the mapping that sends $p\in\Omega_d$ to the set of its zeros. Further, $\Lambda_d$ shall consist of the evaluation functions that read the coefficients of the polynomial $p\in\Omega_d$, and the constant functions with the values $s\in S_d$. Note again that these values can be effectively constructed.

\begin{theorem}
Consider $(\Xi,\Omega_d,\mathcal M, \Lambda_d)$ as above. Then the algorithms \eqref{EPurelyIterative} define an arithmetic tower of height one for the computation of the roots of each input polynomial $p$ with error control. Thus this tower realises $\{\Xi,\Omega_d,\mathcal M, \Lambda_d\}\in\Sigma_1^A$. Moreover, this tower employs just Newton's Method, i.e. a purely iterative algorithm.
\end{theorem}

\section{Computational Examples}\label{numerics}
The purpose of this section is to demonstrate that the new towers of algorithms developed to yield the sharp classifications in the SCI hierarchy yield implementable algorithms that are efficient. In the case of $\Sigma^A_1$ classifications, up to a user specified error tolerance, they will never produce incorrect output. This fact makes the algorithms particularly suited for computer assisted proofs. Moreover, they provide the first computations done of spectra of several types of operators that before were out of reach. Convergent algorithms that never make mistakes are of course sought after in the sciences, and the reader may consult \cite{PRL} to see the algorithms used in practice for large scale problems in physics. 

\subsection{Toeplitz operators}
\begin{figure}
\includegraphics[width=1\textwidth,clip,trim= 40mm 0mm 45mm 0mm]{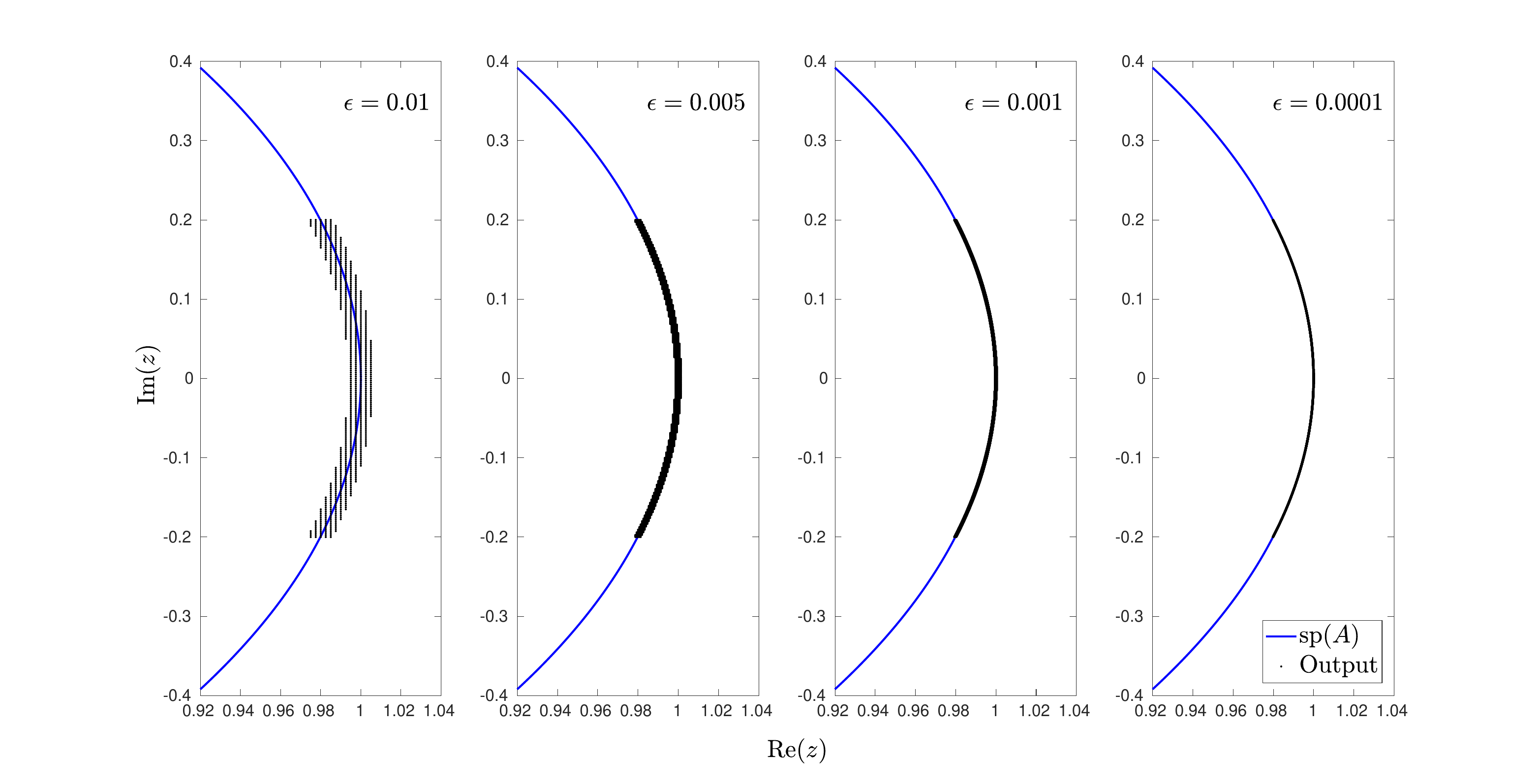}
\caption{The figure shows a $\Gamma_n(A) \cap K$ (black) for a compact set $K \subset \mathbb{C}$ on top of a part of $\mathrm{sp}(A)$ (blue) for different increasing values of $n$ corresponding to the chosen $\epsilon$, where $A$ is the shift operator on $l^2(\mathbb{Z})$.}
\label{shift}
\end{figure}

Toeplitz and Laurent operators are familiar test objects given that their spectra are very well understood \cite{bottcher2006analysis, Bottcher_book}. In this first example, we are concerned with operators that are banded with known growth on their resolvents. In particular, the problem of computing the spectrum lies in $\Sigma_1^A$ and has $\mathrm{SCI} = 1$. Since the problem does not lie in $\Pi_1^G$, we monitor the changes of $\Gamma_n(A)$ as $n \rightarrow \infty$. This is common practice in computations when error control is not available. In particular, we choose an $\epsilon > 0$ and $K \in \mathbb{N}$ and stop the iteration when 
\begin{equation}\label{stopping}
\max\{E_n(A),d(\Gamma_n(A),\Gamma_{n+k}(A))\} \leq \epsilon \, \text{ for all } \, k \leq K.
\end{equation}  
Here $E_n(A)$ refers to the error guarantee $\Gamma_n(A)\subset \mathrm{sp}(A)+B_{E_n(A)}(0)$ provided by the algorithm. To visualise the convergence, we tested the tower of height one on the shift operator in Figure \ref{shift}. Note that it is crucial to know the SCI of the problem so that one can apply the tower of algorithms with the correct height. In particular, trying to solve this problem with a tower of height two would make the computation incredibly more complex. Compare, for example, with the experiment in \S \ref{fQ}.

\subsection{Spectra and approximate eigenvectors of operators on aperiodic tilings}
\begin{figure}
\centering 
\includegraphics[width=1\textwidth,clip,trim= 0mm 0mm 0mm 0mm]{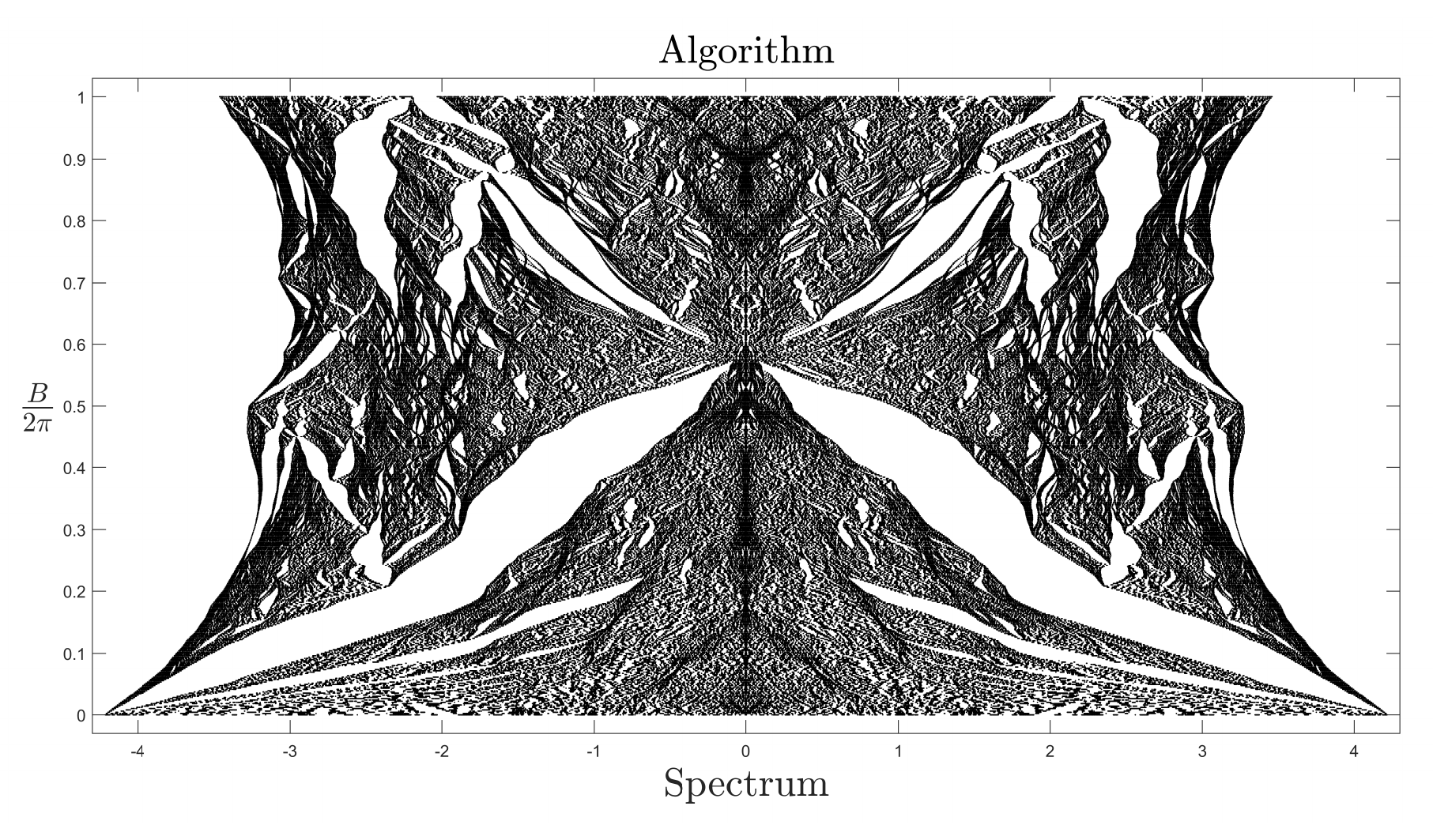}
\includegraphics[width=1\textwidth,clip,trim= 0mm 0mm 0mm 0mm]{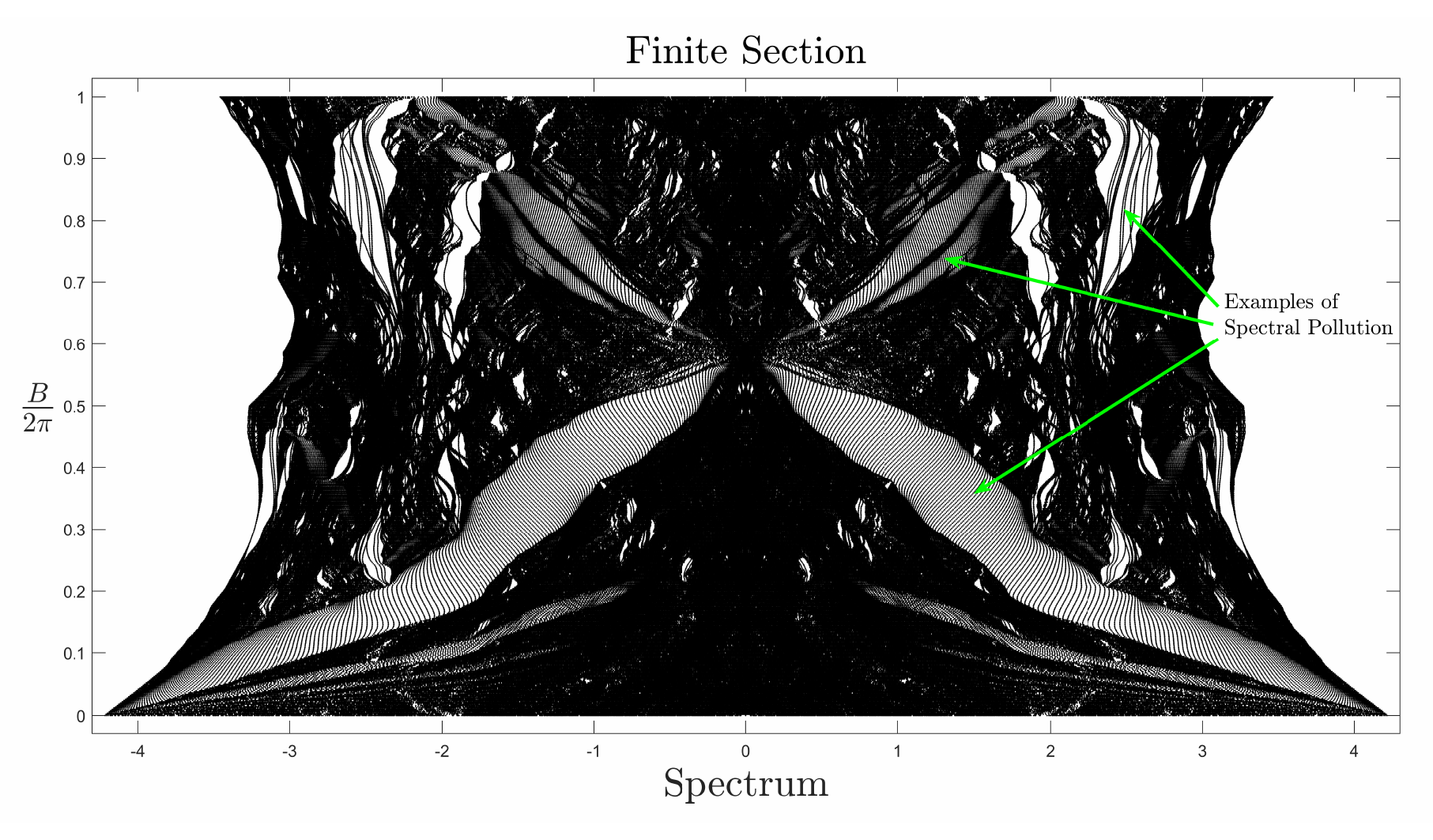}
\caption{Top: Output of the algorithm providing $\Sigma^A_1$ classification computing spectra of the Hamiltonian in \eqref{eq:mag_ham} with error tolerance parameter $10^{-2}$ and different strengths of the magnetic field. The algorithm correctly leaves out the gaps and shows the fractal nature of the spectrum. Bottom: Output of the finite section method ($4000$ basis sites) showing severe spectral pollution.}
\label{Quas2}
\end{figure}
\begin{figure}
\centering
\includegraphics[height=52mm,clip,trim= 55mm 90mm 55mm 90mm]{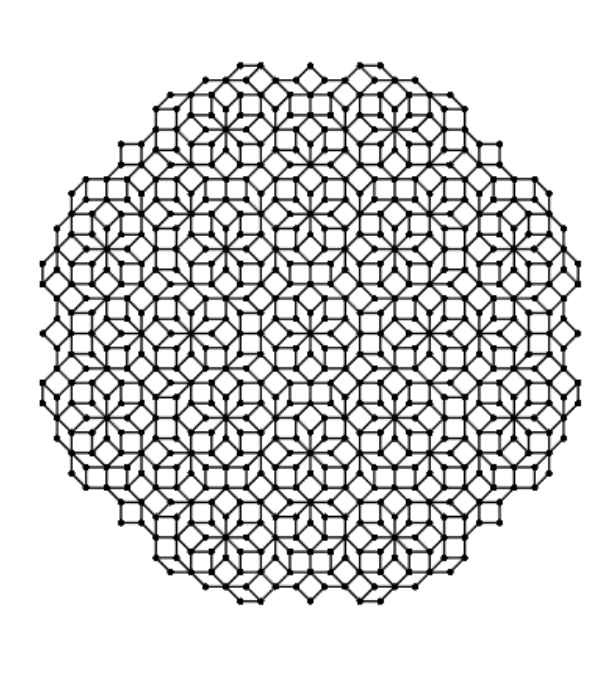}
\includegraphics[height=52mm,clip,trim= 20mm 0mm 34.5mm 0mm]{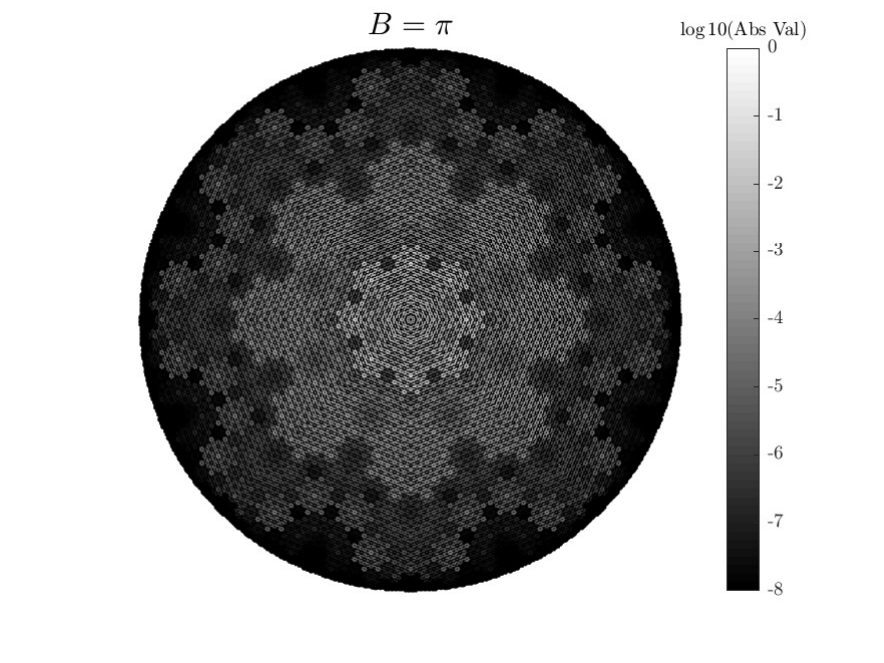}
\includegraphics[height=52mm,clip,trim= 20mm 0mm 5mm 0mm]{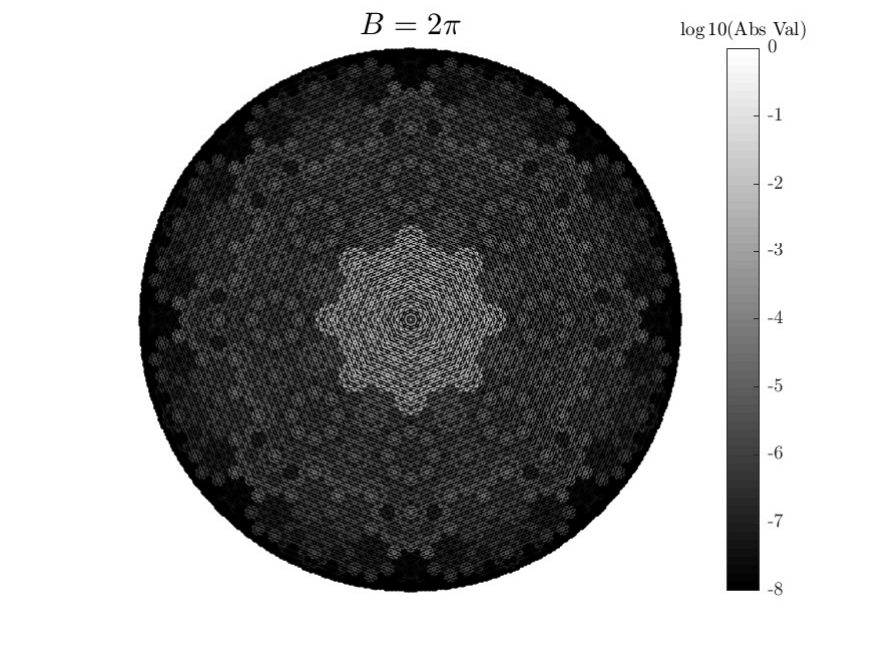}
\caption{Left: Finite portion of the Ammann--Beenker tiling. The vertices correspond to the sites. Middle and Right: Approximate states (eigenvectors) $\psi$ corresponding to the value $\lambda=0$ for $B = \pi, 2\pi$ (logarithm of absolute value shown). These have bounds of $\|(H- \lambda)\psi\|$ by $3.3\times 10^{-7}$ and $1.5\times 10^{-6}$ respectively and were computed using $10^5$ basis sites.}
\label{Quas1}
\end{figure}

Quasicrystals,\footnote{Discovered in 1982 by D. Shechtman who was awarded the Nobel prize in 2011 for his discovery.} and more generally aperiodic systems, have generated considerable interest due to their often exotic physical/spectral properties \cite{PhysRevLett.53.1951,stadnik2012physical}. However, the lack of reliable algorithms have limited the insight obtained from computations. We present the first rigorous spectral computational study with error bounds on an Ammann--Beenker tiling, a standard 2D model of a quasicrystal \cite{ammann1992aperiodic,steurer2004twenty}. Such models are difficult to deal with due to the lack of translational symmetry. The tiling has eight-fold rotational symmetry, shown in Figure \ref{Quas1} (left), which has been found to occur in real quasicrystals, e.g. in \cite{wang1987two}. We consider a magnetic Hamiltonian
\begin{equation}\label{eq:mag_ham}
(H\psi)_a=-\sum_{ a \sim b }e^{\mathrm{i}\alpha_{b,a}}\psi_b,
\end{equation}
where the $a\sim b$ means vertices $a$ and $b$ are connected by an edge, and hence the summation is over connected sites. A constant perpendicular magnetic field with potential $\textbf{A}(x,y,z)=(0,xB,0)$ with $B \in \mathbb{R}$ is applied, leading to the Peierls phase factor between sites $a$ and $b$:
$$
\alpha_{b,a}=\int_{b}^{a} \textbf{A}\cdot d\textbf{l},
$$
where $\textbf{l}$ is the arclength. 
Figure \ref{Quas2} shows the output of the algorithm providing $\Sigma^A_1$ classification computing spectra of the Hamiltonian in \eqref{eq:mag_ham} for different values of $B \in [0,2\pi]$ using the stopping criterion \eqref{stopping} and an error tolerance of $10^{-2}$ . The algorithm correctly leaves out the gaps in the spectrum, avoiding spectral pollution. We also show the output of the finite section method which suffers from severe spectral pollution. One can study quasiperiodic tilings via periodic approximates \cite{duneau1989approximants}. However, it is not clear how these approximations affect the spectrum \cite{massatt2017electronic} and in the case of magnetic field, this imposes severe restrictions on the values of $B$ allowed \cite{tran2015topological}. In contrast, there is no such limitation for the new algorithm, which also provides rigorous error bounds and is guaranteed to converge. The new algorithm can also be used for non-constant magnetic fields. Finally, in Figure \ref{Quas1}, we have also shown approximate eigenvectors for different values of $B$.

\subsection{Non-Hermitian Hamiltonians}
Non-Hermitian Hamiltonians have been standard in open systems, however, they have also found their way to quantum mechanics of closed systems due to the seminal work of C. Bender \cite{Bender, Bender3}. There are also other variants of non-Hermitian quantum mechanics pioneered by N. Hatano  and D. R. Nelson \cite{Hatano_96, Hatano_97}.  The non-self-adjointness make spectral computations incredibly difficult, and algorithms have typically not been available for rigorous computations. 
As an example of computing pseudospectra and to demonstrate generality, we consider a non-normal operator $A$ on $l^2(\mathbb{N})$ given by
\begin{equation}\label{pp}
(Ax)_n=\begin{cases}x_{n-1}+\mathrm{i}\sin(n)x_n-x_{n+1},\quad &\text{if }n+1\text{ is prime}\\
x_{n-1}+\mathrm{i}\sin(n)x_n+x_{n+1},\quad &\text{otherwise},
\end{cases},
\end{equation}
with the convention that $x_0=0$. Figure \ref{non_herm} shows pseudospectra computed using the new algorithm providing $\Sigma^A_1$ classification and attempts of computing pseudospectra using square finite section truncations for $2000$ basis vectors. We see that taking square truncations gives rise to over estimates of the resolvent norm resulting in incorrect spectral information.  

\begin{figure}
\centering
\includegraphics[width=0.49\textwidth,clip,trim= 0mm 0mm 0mm 0mm]{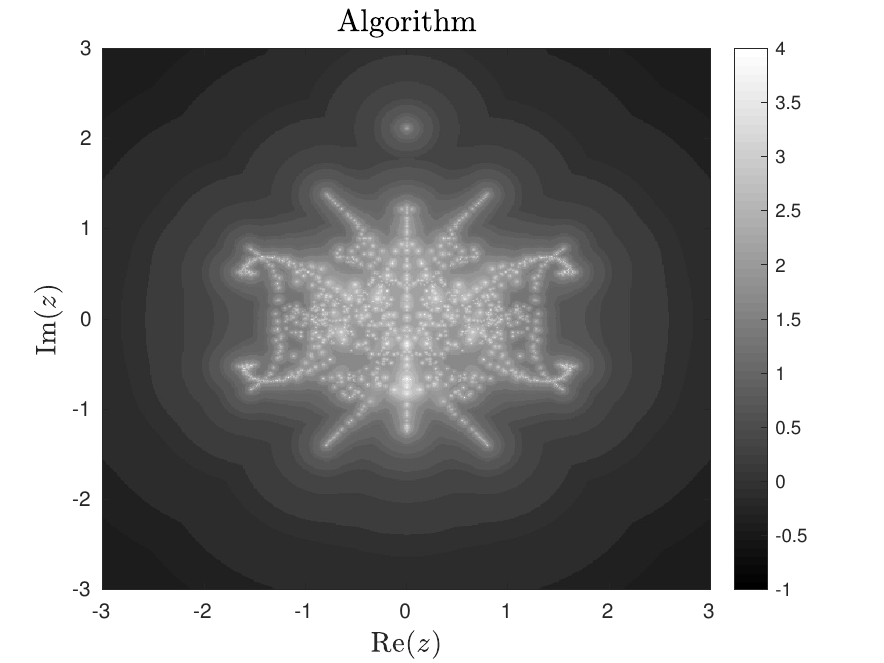}
\includegraphics[width=0.49\textwidth,clip,trim= 0mm 0mm 0mm 0mm]{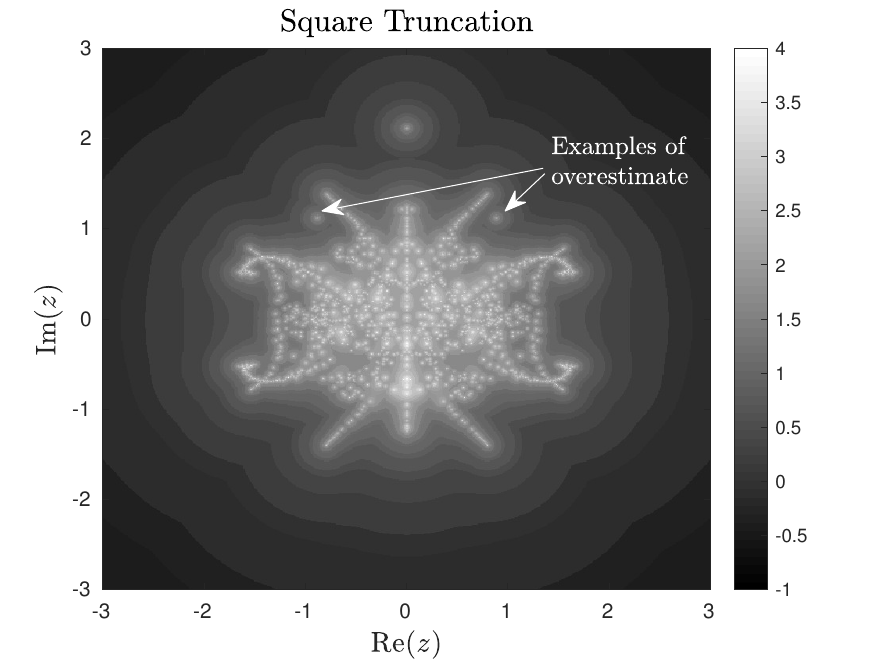}
\caption{Left: Pseudospectra for the operator given by (\ref{pp}) computed by the algorithm providing $\Sigma^A_1$ classification of the problem of computing pseudospectra. The colourbars correspond to the logarithm (base $10$) of the resolvent norm (truncated at $4$ for visibility). Right: failed attempt of computing pseudospectra with classical square truncation of the operator.}
\label{non_herm}
\end{figure}

\subsection{Schr\"odinger operator on $\mathbb{R}$}

\begin{figure}
\centering
\includegraphics[width=1.05\textwidth,clip,trim= 25mm 0mm 0mm 0mm]{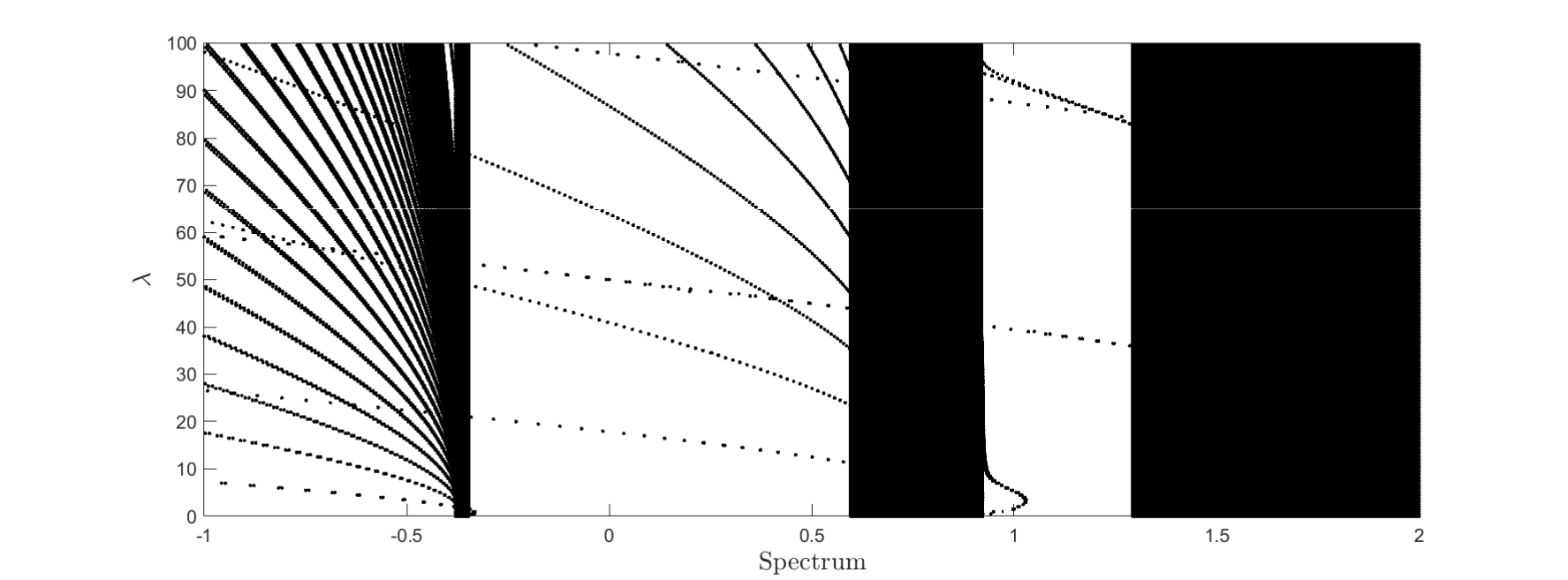}
\caption{A portion of the computed spectrum of the one dimensional Schr{\"o}dinger operator with potential $V_{\lambda}$ from \eqref{sinc_pot} computed with an algorithm, providing the  $\Sigma_1^A$ classification, with error bound $\epsilon=0.01$.}
\label{cts_model}
\end{figure}

We now test the algorithm that computes spectra of Schr\"odinger operators 
\begin{equation*}
	 H=-\Delta+V,\qquad V:\R^d\rightarrow \mathbb{R},
\end{equation*} 
acting on $\mathrm{W}^{2,2}(\Rbb)$ (i.e. a continuum model) with bounded potential. This example demonstrates the power of the $\Sigma^A_1$ classification of Theorem \ref{main_self_adjoint}. Recall that our algorithm uses only evaluations of the potential itself. As the class of problems considered are $\in \Sigma_1^A$ and have SCI $=1$, we shall use the stopping criterion in (\ref{stopping}) with $\epsilon=0.01$. 
We chose the slowly decaying potentials
\begin{equation}
\label{sinc_pot}
V_{\lambda}(x)=\cos(x)+\lambda \frac{\sin(x)}{x}.
\end{equation}
When $\lambda=0$, the operator is periodic, and computation of the spectrum reduces to computing spectra of two differential operators on a compact interval (with periodic and anti-periodic boundary conditions) which have compact resolvent. However, when $\lambda\neq 0$, eigenvalues appear in the gaps of the essential spectrum (and below), and methods based on finite section produce spectral pollution. Additionally to this, the slow decay of the potential makes this extremely difficult to detect via other means. In Figure \ref{cts_model} we display the computation of a portion of the spectrum in $[-1,2]$ for various choice of $\lambda$ and the error bound $\epsilon=0.01$. The algorithm allows us to track eigenvalues in the gaps with error control guaranteeing the error bound.

\subsection{The operator $f(Q)$}\label{fQ}
\begin{figure}
\includegraphics[height=105mm]{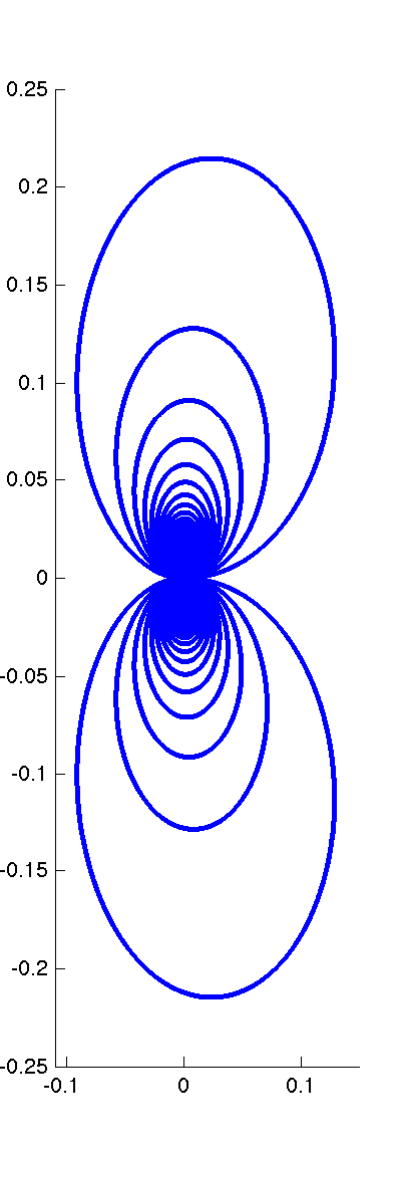}
\includegraphics[height=105mm]{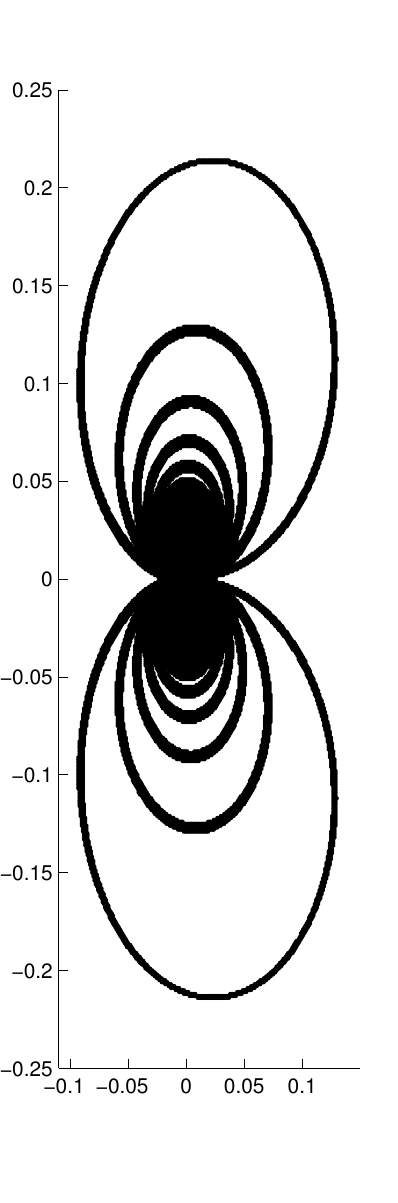}
\includegraphics[height=105mm]{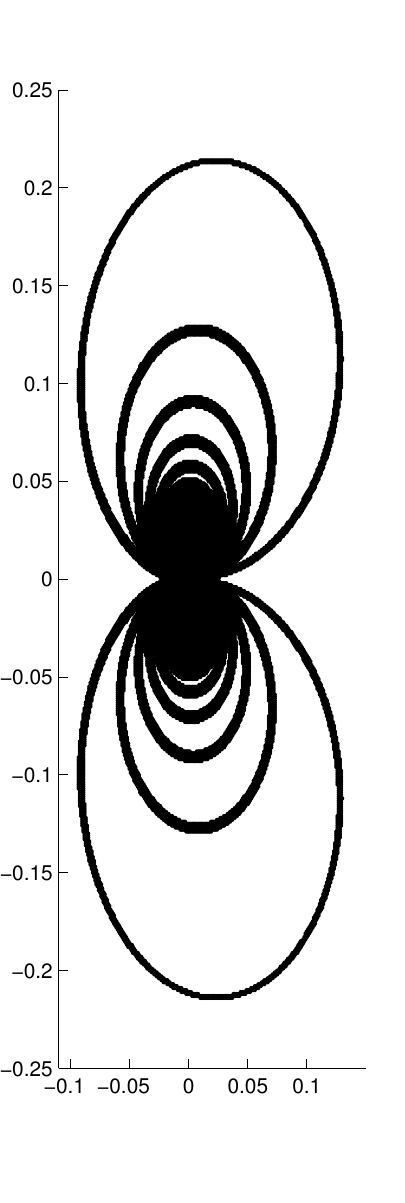}
\includegraphics[height=105mm]{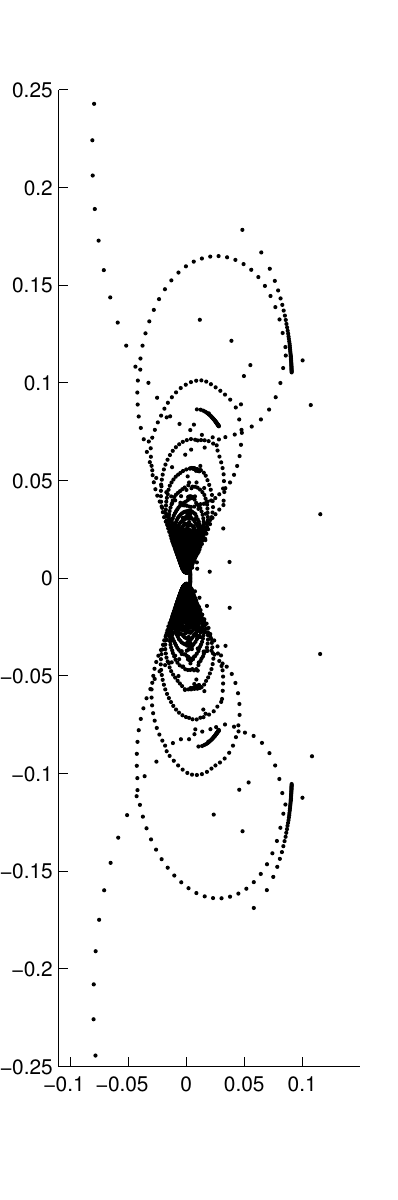}
\caption{The left figure is a zoomed in part of $\mathrm{sp}(f(Q))$. The two following figures are  $\Gamma_{n,m}(A)$ and $\Gamma_{n+p,m+s}(A)$ (restricted to the zoomed in part) to visualise the stopping criterion in (\ref{stopping2}).  To get a better approximation a smaller $\epsilon$ must be chosen. $A$ is the matrix representation of $f(Q)$. The right figure is the result of the finite section method trying to compute $\mathrm{sp}(f(Q))$.}
\label{fQ_fig}
\end{figure}
If we consider the multiplication operator $(Qg)(x) = xg(x)$ on $\mathrm{L}^2(\mathbb{R})$, then, for a bounded continuous function $f:\mathbb{R} \rightarrow \mathbb{R}$, the spectrum of $f(Q)$ is the range of the function $f$.
In this example we use $f(x) = \frac{i(\exp(-2\pi ix) - 1)}{2\pi x}$. To create an infinite matrix representation of $f(Q)$ we 
first consider the following Gabor basis for $\mathrm{L}^2(\mathbb{R})$:
$$
e^{2\pi i m x}\chi_{[0,1]}(x-n), \qquad m,n \in \mathbb{Z},
$$
(where $\chi$ is the characteristic function) and then chose some enumeration of $\mathbb{Z} \times \mathbb{Z}
$ into $\mathbb{N}$ to obtain a basis $\{\psi_j\}$ that is just
indexed over $\mathbb{N}.$ To get our basis we let $\varphi_j =
\mathcal{F}\psi_j$, where $\mathcal{F}$ is the Fourier Transform.
Finally we obtain the infinite matrix representation 
$
A_{ij} = \langle f(Q)\varphi_j,\varphi_i \rangle.
$
Note that this becomes a full infinite matrix, however, we know the growth of the resolvent of the operator, thus, this is a problem in the class $\Sigma_2^A$ with SCI $= 2$.  As there are now two limits, our algorithm depends on two parameters, namely $m$ and $n$, and we compute $\Gamma_{n,m}(A)$. This means that the stopping criterion from (\ref{stopping}) becomes as follows. Choose $\epsilon > 0$ and $K \in \mathbb{N}$.  Define, for any $n,l \in \mathbb{N}$,
\begin{equation}\label{stopping2}
\begin{split}
\tilde \Gamma_n(A) &:=  \Gamma_{n,m}(A),\quad
m = \min\{p: d(\Gamma_{n,p}(A),\Gamma_{n,p+k}(A) \leq \epsilon \, \text{ for all } \, k \leq K\}\\
\tilde \Gamma(A) &:= \tilde \Gamma_l(A), \qquad l = \min\{p: d(\tilde \Gamma_{p}(A), \tilde \Gamma_{p+k}(A) \leq \epsilon \, \text{ for all } \, k \leq K\},
\end{split}
\end{equation}  
and let the output be $\tilde \Gamma(A)$. This stopping criterion is obviously a generalisation of (\ref{stopping}) and extends in an obvious way to several limits. Note however how incredibly more complex it gets by adding one more limit. In Figure \ref{fQ_fig} we have plotted $\Gamma_{n,m}(A)$ and $\Gamma_{n+p,m+s}(A)$ visualizing an output based on the two limit stopping criterion in (\ref{stopping2}). We also plotted the result of the finite section method. As we are computing within the class of problems with SCI $=2$, there is of course no way that the finite section method could work.

\appendix

\section{Proof of Proposition \ref{thrm:prop_SCI} and generalisations} 
\label{KS_results}

\subsection{Proof of Proposition \ref{thrm:prop_SCI} parts (i) and (ii)}

Let $(\mathcal{M},d)$ be a metric space with the Attouch--Wets or Hausdorff topology induced by another metric space $(\mathcal{M}',d_{\mathcal{M}'})$. For the Attouch--Wets topology and any fixed $x_0\in\mathcal{M}'$ we set
$$
d_{\mathrm{AW}}(C_1,C_2)=\sum_{n=1}^{\infty} 2^{-n}\min\{1,{\sup}_{d_{\mathcal{M}'}(x_0,x)\leq n}\left|\mathrm{dist}(x,C_1)-\mathrm{dist}(x,C_2)\right|\},
$$
for $C_1,C_2\in\mathrm{Cl}(\mathcal{M}'),$ where $\mathrm{Cl}(\mathcal{M}')$ denotes the set of non-empty closed subsets of $\mathcal{M}'$. In the case that $\mathcal{M}'=\mathbb{C}$ with the usual metric we take $x_0=0$. Using the notation of \S \ref{sec:background}, we have the following `sandwich' lemma.

\begin{lemma}
\label{sandwich}
Suppose that $(\mathcal{M},d)$ is the Hausdorff or Attouch--Wets topology induced by a metric space $(\mathcal{M}',d_{\mathcal{M}'})$. Let $\epsilon>0$. Suppose also that $A,A',B,B',C\in\mathcal{M}$ with $A\mathop{\subset}_{\mathcal{M}^{\prime}} A'$, $C\mathop{\subset}_{\mathcal{M}^{\prime}} B'$, $d(C,A')\leq \epsilon$ and $d(B,B')\leq \epsilon$. Then
$
d(A,C)\leq d(A,B)+2\epsilon.
$
\end{lemma}

\begin{proof}
Suppose first that $(\mathcal{M},d)$ is the Hausdorff topology. If $x\in C$ then $x\in B'$ and $
\mathrm{dist}(x,A)\leq d(B',A)\leq d(A,B)+\epsilon.$ On the other hand, if $x\in A$ then $x\in A'$ and $\mathrm{dist}(x,C)\leq d(A',C)\leq\epsilon$. The result now follows.

Suppose now that $(\mathcal{M},d)$ is the Attouch--Wets topology and let $x\in\mathcal{M}'$. Since $C\mathop{\subset}_{\mathcal{M}^{\prime}} B'$ we must have
$$
\mathrm{dist}(x,A)-\mathrm{dist}(x,C)\leq \mathrm{dist}(x,A)-\mathrm{dist}(x,B')\leq \left|\mathrm{dist}(x,A)-\mathrm{dist}(x,B)\right|+\left|\mathrm{dist}(x,B)-\mathrm{dist}(x,B')\right|.
$$
Similarly, since $A\mathop{\subset}_{\mathcal{M}^{\prime}} A'$ we must have
$$
\mathrm{dist}(x,C)-\mathrm{dist}(x,A)\leq \mathrm{dist}(x,C)-\mathrm{dist}(x,A')\leq \left|\mathrm{dist}(x,C)-\mathrm{dist}(x,A')\right|.
$$
It follows that
$$
\left|\mathrm{dist}(x,A)-\mathrm{dist}(x,C)\right|\leq \left|\mathrm{dist}(x,A)-\mathrm{dist}(x,B)\right|+\left|\mathrm{dist}(x,B)-\mathrm{dist}(x,B')\right|+\left|\mathrm{dist}(x,C)-\mathrm{dist}(x,A')\right|.
$$
The result now follows.
\end{proof}

\begin{proposition}
\label{eps_tower}
Let $(\mathcal{M},d)$ be either a metric space with the Attouch--Wets or Hausdorff topology induced by another metric space $(\mathcal{M}',d_{\mathcal{M}'})$ or a totally ordered metric space with order respecting metric. Suppose we have a computational problem
$$
\Xi:\Omega\rightarrow \mathcal{M},
$$
with a corresponding convergent $\Sigma_k^\alpha$ tower $\Gamma_{n_k,...,n_1}^1$ and a corresponding convergent $\Pi_k^\alpha$ tower $\Gamma_{n_k,...,n_1}^2$ (either both arithmetic or both general). Suppose also that $1\leq k\leq 3$ and that, in the case of arithmetic towers, we can compute for every $A\in\Omega$ the distance $d(\Gamma_{n_k,...,n_1}^1(A),\Gamma_{n_k,...,n_1}^2(A))$ to arbitrary precision using finitely many arithmetic operations and comparisons. Then $\{\Xi,\Omega\}\in\Delta_k^{\alpha}$.
\end{proposition}

\begin{remark}
This proposition essentially says that we can combine the two notions of error control $\Pi_k$ and $\Sigma_k$ to reduce the number of limits needed by one.
\end{remark}

\begin{proof}[Proof of Proposition \ref{eps_tower}]
\textbf{Step I:} For $k=1$ and the case that $(\mathcal{M},d)$ is either a metric space with the Attouch--Wets or Hausdorff topology, this is a trivial consequence of Lemma \ref{sandwich}. Let $\delta_{n_1}$ be an approximation of
$$
d(\Gamma_{n_1}^1(A),\Gamma_{n_1}^2(A))+2\cdot2^{-n_1}
$$
from above to accuracy $1/n_1$. Note that suitable approximations can easily be generated using approximations of $d(\Gamma_{n_1}^1(A),\Gamma_{n_1}^2(A))$. Let $\epsilon>0$, then simply choose $n_1\in\mathbb{N}$ minimal such that $\delta_{n_1}\leq\epsilon$. In the case that $(\mathcal{M},d)$ is totally ordered with order respecting metric
$$
d(\Gamma_{n_1}^1(A),\Xi(A))\leq d(\Gamma_{n_1}^1(A),\Gamma_{n_1}^2(A)),
$$
and we can take $n_1$ large such that the right-hand side is less than the given $\epsilon$ (recall we can compute the right-hand side to arbitrary precision). Set $\Gamma(A)=\Gamma^1(A)$, then we have
$$
d(\Gamma_(A),\Xi(A))\leq\epsilon.
$$

\textbf{Step II:} For larger $k$ we use the same idea, but we must be careful to ensure the first $k-1$ limits exist. For the rest of the proof, $\tilde d$ will denote an approximation of $d$ to accuracy $1/n_1$ (which by assumption can always be computed).

We first deal with the case $k=2$. Let $\epsilon>0$ and consider the intervals $J^1_{\epsilon}=[0,\epsilon]$ and $J^2_{\epsilon}=[2\epsilon,\infty)$. Let $\delta_{n_2,n_1}(A)$ be an approximation of
$$
d(\Gamma_{n_2,n_1}^1(A),\Gamma_{n_2,n_1}^2(A))+2\cdot2^{-n_2}
$$
from above to accuracy $1/n_1$. Again note that we can easily construct such approximations. It is clear that $\lim_{n_1\rightarrow\infty}\delta_{n_2,n_1}(A)=d(\Gamma_{n_2}^1(A),\Gamma_{n_2}^2(A))+2\cdot2^{-n_2}=:\delta_{n_2}(A)$ and that $d(\Gamma_{n_2}^1(A),\Xi(A))\leq \delta_{n_2}(A)$ (again appealing to Lemma \ref{sandwich} if we are in the case of the Attouch--Wets or Hausdorff topologies). Given $n_1,n_2$, let $l(n_2,n_1)\leq n_1$ be maximal such that $\delta_{n_2,l}(A)\in J_{\epsilon}^1\cup J_{\epsilon}^2$. If no such $l$ exists or $\delta_{n_2,l}(A)\in J_{\epsilon}^1$ then define $\mathrm{Osc}(\epsilon;n_1,n_2,A)=1$ otherwise define $\mathrm{Osc}(\epsilon;n_1,n_2,A)=0$. Since $\delta_{n_2,n_1}(A)$ cannot oscillate infinitely often between the two intervals $J^1_{\epsilon}$ and $J^2_{\epsilon}$, it follows that
$$
\mathrm{Osc}(\epsilon;n_2,A):=\lim_{n_1\rightarrow\infty}\mathrm{Osc}(\epsilon;n_1,n_2,A)
$$
exists. Define $\Gamma_{n_1}^{\epsilon}(A)$ as follows. Choose $j\leq n_1$ minimal such that $\mathrm{Osc}(\epsilon;n_1,j,A)=1$ if such a $j$ exists, and define $\Gamma_{n_1}^{\epsilon}(A)=\Gamma_{j,n_1}(A)$. If no such $j$ exists then define $\Gamma_{n_1}^{\epsilon}(A)=C_0$ where $C_0$ is any fixed member of $(\mathcal{M},d)$. In particular, $\Gamma_{n_1}^{\epsilon}$ is a type $\alpha$ algorithm. Now for large $n_2$, we must have $\delta_{n_2}(A)<\epsilon$ and hence $\mathrm{Osc}(\epsilon;n_2,A)=1$. It follows that $\Gamma^{\epsilon}(A)=\lim_{n_1\rightarrow\infty}\Gamma_{n_1}^{\epsilon}(A)$ exists and is equal to $\Gamma_{N}^1(A)$ where $N\in\mathbb{N}$ is minimal with $\mathrm{Osc}(\epsilon;N,A)=1$. It follows that $d(\Gamma^{\epsilon}(A),\Xi(A))\leq 2 \epsilon$.

We will use the $\Gamma_{n_1}^{\epsilon}(A)$ to construct a height one tower. Observe first of all that by our assumptions we can compute $\tilde d(\Gamma_{m}^{\epsilon_1}(A),\Gamma_{n}^{\epsilon_2}(A))$ for $m,n\in\mathbb{N}$ and $\epsilon_1,\epsilon_2>0$. Given $n_1$, choose $j=j(n_1)\leq n_1$ maximal such that for all $1\leq l\leq j$ we have
\begin{equation}
\label{key_diag_arg}
\tilde d(\Gamma_{n_1}^{2^{-j}}(A),\Gamma_{n_1}^{2^{-l}}(A))\leq 4(2^{-j}+2^{-l}).
\end{equation}
If no such $j$ exists then set $\Gamma_{n_1}=C_0$, otherwise set $\Gamma_{n_1}(A)=\Gamma_{n_1}^{2^{-j(n_1)}}(A)$. Again, this is easily seen to be a type $\alpha$ algorithm. Pick any $N\in\mathbb{N}$, then by the convergence of the $\Gamma_{n_1}^{\epsilon}(A)$ and $d(\Gamma^{\epsilon}(A),\Xi(A))\leq 2 \epsilon$, (\ref{key_diag_arg}) must hold for $j=N$ and $1\leq l\leq N$ if $n_1$ is large enough. Hence by definition of $j(n_1)$,
$$
\limsup_{n_1\rightarrow\infty}d(\Gamma_{n_1}(A),\Xi(A))\leq \limsup_{n_1\rightarrow\infty}d(\Gamma_{n_1}^{2^{-N}}(A),\Xi(A)) + 2^{3-N}\leq 2^{4-N}.
$$
Since $N$ was arbitrary we must have convergence to $\Xi(A)$.

\textbf{Step III:} We now deal with $k=3$. The strategy will be similar to the $k=2$ case but now we construct $\Gamma_{n_2,n_1}^{\epsilon}(A)$ such that $\Gamma_{n_2}^{\epsilon}(A):=\lim_{n_1\rightarrow\infty}\Gamma_{n_2,n_1}^{\epsilon}(A)$ exists and is $3\epsilon$ close to $\Xi(A)$ for large $n_2$, but may not converge in $(\mathcal{M},d)$. Using this, we will construct a height two type $\alpha$ tower.

As in Step II, let $\epsilon>0$ and consider the intervals $J^1_{\epsilon}=[0,\epsilon]$ and $J^2_{\epsilon}=[2\epsilon,\infty)$. Let $\delta_{n_3,n_2,n_1}(A)$ be an approximation of
$$
d(\Gamma_{n_3,n_2,n_1}^1(A),\Gamma_{n_3,n_2,n_1}^2(A))+2\cdot2^{-n_3},
$$
from above to accuracy $1/n_1$. Again, we have
$$\lim_{n_2\rightarrow\infty}\lim_{n_1\rightarrow\infty}\delta_{n_3,n_2,n_1}(A)=d(\Gamma_{n_3}^1(A),\Gamma_{n_3}^2(A))+2\cdot2^{-n_3}=:\delta_{n_3}(A)$$
exists with $d(\Gamma_{n_3}^1(A),\Xi(A))\leq \delta_{n_3}(A)$. Given $n_1,n_2$ and $j$, let $l(j,n_2,n_1)\leq n_1$ be maximal such that $\delta_{j,n_2,l}(A)\in J_{\epsilon}^1\cup J_{\epsilon}^2$. If no such $l$ exists or $\delta_{j,n_2,l}(A)\in J_{\epsilon}^1$ then define $\mathrm{Osc}(\epsilon;n_1,n_2,j,A)=1$ otherwise define $\mathrm{Osc}(\epsilon;n_1,n_2,j,A)=0$. Arguing as in Step I we have
$$
\mathrm{Osc}(\epsilon;n_2,j,A):=\lim_{n_1\rightarrow\infty}\mathrm{Osc}(\epsilon;n_1,n_2,j,A)
$$
exists. Now consider $\mathrm{Osc}(\epsilon;n_1,n_2,j,A)$ for $j\leq n_2$. If such a $j$ exists with $\mathrm{Osc}(\epsilon;n_1,n_2,j,A)=1$ then let $j(n_1,n_2)$ be the minimal such $j$ and set $\Gamma_{n_2,n_1}^{\epsilon}(A)=\Gamma_{j(n_1,n_2),n_2,n_1}^{1}(A)$. Otherwise set $\Gamma_{n_2,n_1}^{\epsilon}(A)=C_0$, where again $C_0$ is some fixed member of $(\mathcal{M},d)$. Since we only deal with finitely many $j\leq n_2$, it is clear that $\Gamma_{n_2,n_1}^{\epsilon}$ is a type $\alpha$ algorithm. Furthermore, we must have that $\Gamma_{n_2}^{\epsilon}(A):=\lim_{n_1\rightarrow\infty}\Gamma_{n_2,n_1}^{\epsilon}(A)$ exists and is defined as follows. Let $j(n_2)\leq n_2$ be minimal with $\mathrm{Osc}(\epsilon;n_2,j,A)=1$ (if such a $j$ exists). If such a $j$ exists then $\Gamma_{n_2}^{\epsilon}(A)=\Gamma_{j(n_2),n_2}^{1}(A)$, otherwise $\Gamma_{n_2}^{\epsilon}(A)=C_0$.

Now there exists $N\in\mathbb{N}$ such that $\delta_{N}(A)<\epsilon/2$ and hence $\delta_{N,n_2}(A)<\epsilon$ for large $n_2$. But this implies that $\mathrm{Osc}(\epsilon;n_2,N,A)=1$. Hence for $n_2$ large we must have $j(n_2)\leq N$. If $\delta_{l}(A)>2\epsilon$ then for large $n_2$ we must have $\delta_{l,n_2}(A)>2\epsilon$ and hence $\mathrm{Osc}(\epsilon;n_2,l,A)=0$. As $n_2$ increases, $j(n_2)$ may not converge. However, the above arguments show that for large $n_2$ it can take only finitely many values, say in the set $S=\{s_1,...,s_m\}$, all of which must have $\delta_{s_i}(A)\leq 2\epsilon$. It follows that for large $n_2$ we must have
\begin{equation}
\label{phew}
d(\Gamma_{n_2}^{\epsilon}(A),\Xi(A))\leq 3\epsilon.
\end{equation}

Now we get to work using these `towers' (which don't necessarily converge in the last limit) and the trick to avoid oscillations. Define
\begin{align*}
F(n_1,n_2,j,l,A)&:=\tilde d(\Gamma_{n_2,n_1}^{2^{-j}}(A),\Gamma_{n_2,n_1}^{2^{-l}}(A)),\\
F(n_2,j,l,A)&:=\lim_{n_1\rightarrow\infty}F(n_1,n_2,j,l,A)=d(\Gamma_{n_2}^{2^{-j}}(A),\Gamma_{n_2}^{2^{-l}}(A))
\end{align*}
and the intervals $J_{j,l}^1=[0,4(2^{-j}+2^{-l})],J_{j,l}^2=[8(2^{-j}+2^{-l}),\infty)$. Given $j,l,n_1$ and $n_2$, let $i(j,l,n_2,n_1)\leq n_1$ be maximal such that $F(i,n_2,j,l,A)\in J_{j,l}^1\cup J_{j,l}^2$. If no such $i$ exists or if it does and $F(i,n_2,j,l,A)\in J_{j,l}^1$ then define $\widehat{\mathrm{Osc}}(n_1,n_2,j,l,A)=1$ otherwise define $\widehat{\mathrm{Osc}}(n_1,n_2,j,l,A)=0$. Choose $j=j(n_1,n_2)\leq n_2$ maximal such that for all $1\leq l\leq j$ we have $\widehat{\mathrm{Osc}}(n_1,n_2,j,l,A)=1$. If no such $j$ exists then set $\Gamma_{n_2,n_1}=C_0$, otherwise set $\Gamma_{n_2,n_1}(A)=\Gamma_{n_2,n_1}^{2^{-j(n_1,n_2)}}(A)$. Again, this is easily seen to be a type $\alpha$ algorithm. 

Arguing as before, we have the existence of
$$
\widehat{\mathrm{Osc}}(n_2,j,l,A):=\lim_{n_1\rightarrow\infty}\widehat{\mathrm{Osc}}(n_1,n_2,j,l,A).
$$
Now define $h=h(n_2)\leq n_2$ maximal such that for all $1\leq l\leq h$ we have $\widehat{\mathrm{Osc}}(n_2,h,l,A)=1$. If no such $h$ exists then we must have
$$
\Gamma_{n_2}(A):=\lim_{n_1\rightarrow\infty}\Gamma_{n_2,n_1}(A)=C_0,
$$
otherwise we must have
$$
\Gamma_{n_2}(A):=\lim_{n_1\rightarrow\infty}\Gamma_{n_2,n_1}(A)=\Gamma_{n_2}^{2^{-h(n_2)}}(A).
$$
By (\ref{phew}), for any fixed $j,l$ we have $\widehat{\mathrm{Osc}}(n_2,j,l,A)=1$ for large $n_2$ and hence $h(n_2)$ exists for large $n_2$ and diverges to $\infty$. Now let $N\in\mathbb{N}$ then it follows that
\begin{align*}
\limsup_{n_2\rightarrow\infty}d(\Gamma_{n_2}^{2^{-h(n_2)}}(A),\Xi(A))&\leq \limsup_{n_2\rightarrow\infty}d(\Gamma_{n_2}^{2^{-N}}(A),\Xi(A)) + d(\Gamma_{n_2}^{2^{-h(n_2)}}(A),\Gamma_{n_2}^{2^{-N}}(A))\\
&\leq 3\cdot 2^{-N}+\limsup_{n_2\rightarrow\infty}8(2^{-h(n_2)}+2^{-N})\leq 11\cdot2^{-N}.
\end{align*}
Since $N$ was arbitrary we must have convergence to $\Xi(A)$.
\end{proof}


\begin{proof}[Proof of Proposition \ref{thrm:prop_SCI} parts (i) and (ii)]
The statement regarding intersections follows directly from Proposition \ref{eps_tower} and the following remark - no assumptions on being able to compute distances between output of algorithms is necessary when considering general towers. For the sharpness result in (i), we deal with $X=\Sigma$ and the $X=\Pi$ follows from an identical argument. Suppose that $\Delta_k^G\not\ni\{\Xi,\Omega\}\in \Sigma_k^{\alpha}$. If $\{\Xi,\Omega\}\in\Pi_k^{\alpha}$, we would have $\{\Xi,\Omega\}\in\Sigma_k^\alpha\cap\Pi_k^\alpha\subset \Sigma_k^G\cap\Pi_k^G=\Delta_k^G$, a contradiction. 
\end{proof}

\subsection{Proof of Proposition \ref{thrm:prop_SCI} part (iii)}\label{sec:Cucker}
To prove this part, we consider the following alternative definition in the case that $\mathcal{M} = \{0,1\}$. Note that if we restricted to recursivity in the Turing sense with $\Xi$ describing subsets of $\mathbb{N}$, this would correspond to the classical arithmetical hierarchy. However, it is much more general and also encompasses the work of F. Cucker \cite{Cucker_AH_real} in the BSS model. 

\begin{definition}[SCI hierarchy, $\mathcal{M} = \{0,1\}$ (alternative definition)]\label{general_stuff}
Suppose that $\mathcal{M} = \{0,1\}$. We define the following
\begin{itemize}
\item[(i)] We say that $\Xi:\Omega\to\mathcal{M}$ permits a representation by an alternating 
quantifier form of length $m$ if 
$$
\Xi=(Q_mn_m)\cdots (Q_1n_1)\Gamma_{n_m,\ldots,n_1},
$$
where $(Q_i)$ is a list of alternating quantifiers $(\forall)$ and $(\exists)$, and all $\Gamma_{n_m,\ldots,n_1}:\Omega\to\mathcal{M}$ are general algorithms in the sense of
Definition \ref{alg}.
\item[(ii)]
We say that $\{\Xi,\Omega\}$ is $\Sigma_m$ if an alternating quantifier form of length $m$ exists
with $Q_m$ being $(\exists)$, and that $\{\Xi,\Omega\}$ is $\Pi_m$ if an alternating quantifier form 
of length $m$ exists with $Q_m$ being $(\forall)$. 
\item[(iii)]
We say that $\{\Xi,\Omega\}$ is $\Delta_m$ if $\{\Xi,\Omega\}$ is $\Sigma_m$ and $\Pi_m$.
\end{itemize}
\end{definition}

It is not clear from the wordings of Definition \ref{def:tot_ord} and Definition \ref{general_stuff} that they are equivalent. However, the next proposition provides the link. 

\begin{proposition}[The SCI hierarchy encompasses the arithmetical hierarchy]\label{thrm:prop_SCI_00}
When $\mathcal{M} = \{0,1\}$, Definition \ref{def:tot_ord} and Definition \ref{general_stuff} are equivalent and hence the SCI encompasses generalisations of the arithmetical hierarchy. This also holds for arithmetic towers which extends the arithmetical hierarchy to arbitrary domains.
\end{proposition}

This immediately implies part (iii) or Proposition \ref{thrm:prop_SCI} and hence the rest of this subsection is devoted to proving Proposition \ref{thrm:prop_SCI_00}.

\begin{remark}
 In classical hierarchies the $\Delta_k$ class is defined by $\Delta_k = \Sigma_k \cap \Pi_k$. This is not the case in the SCI hierarchy. The $\Delta^{\alpha}_k$ classes form the core of the hierarchy, and only when there is extra structure on the metric space does it makes sense to define the $\Sigma^{\alpha}_k$ and the $\Pi^{\alpha}_k$. Moreover, in the general SCI hierarchy, we may have that 
$$
\Delta^{\alpha}_k \neq \Sigma^{\alpha}_k \cap \Pi^{\alpha}_k.
$$
Of course, in the special cases of the SCI hierarchy such as the arithmetical hierarchy, then $\Delta_k = \Sigma_k \cap \Pi_k$. Also, we show that $\Delta^{\alpha}_k = \Sigma^{\alpha}_k \cap \Pi^{\alpha}_k$ for $k = 1,2,3$ and $\alpha = G,A$ in the computational spectral problem case, however, there is no reason that this should hold for $k > 3$ in general. 
Moreover, classical hierarchies have that $\Sigma_k\setminus \Delta_{k-1} \neq \emptyset$ and $\Pi_k\setminus \Delta_{k-1} \neq \emptyset$. This does not have to be the case in general SCI hierarchies. Indeed, one may have that 
\[
\Sigma^{\alpha}_k\setminus \Delta^{\alpha}_{k-1} = \emptyset \text{ or } \Pi^{\alpha}_k\setminus \Delta^{\alpha}_{k-1} = \emptyset.
\]
 This happens, for example, in the SCI hierarchy for the computational spectral problem. 
\end{remark}

To prove Proposition \ref{thrm:prop_SCI_00} we make the following definition, which corresponds to the SCI hierarchy in the main text.

\begin{definition}[Limit forms]\label{general_stuff2}
If $\mathcal{M} = \{0,1\}$, we define the following with respect to a given type of tower of algorithms (arithmetical, radical general etc.):
\begin{itemize}
\item[(i)] We say that $\{\Xi,\Omega\}$ is $\tilde\Sigma_m$ if there exists a height $m$ tower solving the computational problem such that the final limit is monotonic from below. We say that $\{\Xi,\Omega\}$ is $\tilde\Pi_m$ if there exists a height $m$ tower solving the computational problem such that the final limit is monotonic from above.
\item[(ii)] We say that $\{\Xi,\Omega\}$ is $\tilde\Delta_{m+1}$ if there exists a height $m$ tower solving the computational problem.
\end{itemize}
\end{definition}

The following theorem demonstrates how the SCI framework can be viewed, in the special case of $\mathcal{M} = \{0,1\}$, as a generalisation of the Arithmetical Hierarchy to arbitrary computational problems. In particular, one can define a hierarchy for any kind of tower. Here we do this for a general tower, and obviously, this can be done for any tower. We will call the hierarchy described below a {\it General Hierarchy}.

\begin{proposition}[General Hierarchy] \label{theorem_general}
Suppose that $\mathcal{M} = \{0,1\}$. Following Definitions \ref{general_stuff} and \ref{general_stuff2}, for any $m\geq 1$ we have that
$$
\tilde\Sigma_m=\Sigma_m, \quad \tilde\Pi_m=\Pi_m\quad\text{and}\quad\tilde\Delta_m=\Delta_m.
$$
\end{proposition}

\begin{proof}[Proof of Proposition \ref{theorem_general}]
{\bf Step I:} We show that if $\mathrm{SCI}(\Xi,\Omega)_{\mathrm{G}}\leq m$ then $\Xi$ is $\Delta_{m+1}$.
Let $p=\lim_i p_i$. Then
\[p=\tr \quad\Leftrightarrow\quad
\forall n \exists k (k\geq n \wedge p_k) \quad\Leftrightarrow\quad
\exists n \forall k (k\leq n \vee p_k).\]
Further, let $\varphi:\Nb\to\Nb\times\Nb$, $k\mapsto(\varphi_1(k),\varphi_2(k))$ 
be a bijection which enumerates all pairs of natural numbers, and note that 
$$
\exists n \exists m (p_{n,m})\Leftrightarrow\exists k (p(\varphi_1(k),\varphi_2(k))),
\qquad 
\forall n \forall m (p_{n,m})\Leftrightarrow\forall k (p(\varphi_1(k),\varphi_2(k))),
$$
for any family $(p_{n,m})_{n,m\in\Nb}\subset\mathcal{M}$.
Thus, every limit in a tower of height $m$ can be converted alternately into an expression with two quantifiers 
($\forall \exists$ or $\exists\forall$), and then $m-1$ doubles $\exists\exists$ or 
$\forall\forall$ can be replaced by single quantifiers. This easily gives the claim.

{\bf Step II:} We show that if $\Xi$ is $\Sigma_m$ or $\Pi_m$ then $\mathrm{SCI}(\Xi,\Omega)_{\mathrm{G}}\leq m$. In fact we show that $\Sigma_m\subset\tilde\Sigma_m$ and $\Pi_m\subset\tilde\Pi_m$.
As a start let $(p_i)\subset\mathcal{M}$ be a sequence. Then 
\[
(\forall i (p_i))=\tr \quad\Leftrightarrow 
		\quad \left(\lim_{n\to\infty} \bigwedge_{i=1}^n p_i\right) = \tr,\quad\quad
(\exists i (p_i))=\tr \quad\Leftrightarrow 
		\quad \left(\lim_{n\to\infty} \bigvee_{i=1}^n p_i\right) = \tr.
\]
Furthermore, the conjunction (disjunction) of limits coincides with the limit of 
the elementwise conjunction (disjunction), hence
\[\forall n_m\exists n_{m-1}\cdots \forall n_1 \Gamma_{n_m,\cdots,n_1}
= \lim_{k_m} \lim_{k_{m-1}}\cdots\lim_{k_1}\bigwedge_{i_m=1}^{k_m}\bigvee_{i_{m-1}=1}^{k_{m-1}} \cdots\bigwedge_{i_1=1}^{k_1}\Gamma_{i_m,i_{m-1},\cdots,i_1}\]
and similarly for any other possible alternating quantifier form. 
Since the $\Gamma_{n_m,\cdots,n_1}$ in the alternating quantifier form at the
left-hand side are General algorithms, the right-hand side obviously yields a tower
of algorithms of height $m$. Moreover, we obtain the required monotonic final limits.

{\bf Step III:} We show that $\tilde\Delta_m=\Delta_m$. Let $m\in\Nb$ be the smallest number with $\Xi$ being $\Delta_{m+1}$.
In the above steps we have already seen that 
$m\leq \mathrm{SCI}(\Xi,\Omega)_{\mathrm{G}} \leq m+1$, and we next prove the following:
If
\[\Xi(y) = \exists i \forall j (g_0(i,j,y)) = \forall n \exists m (g_1(n,m,y))\]
then $\Xi(y)=\lim_{k\to\infty} g(k,y)$ with a function $g$ being easily derivable
from $g_0$, $g_1$.
The following construction is adopted from \cite[Proofs of Theorems 1 and 3]{Gold65}.
Fix $y$ and define a function $h_0:\Nb\to\mathcal{M}$ recursively as follows:

{\ttfamily
$i(1):= 1$, $j(1):=1$, $h_0(1):=g_0(i(1),j(1),y)$.

If $h_0(l)=\tr$ 

then: $i(l+1):=i(l)$, $j(l+1):=j(l)+1$

else: $i(l+1):=i(l)+1$, $j(l+1):=1$.

$l:=l+1$.

$h_0(l):=g_0(i(l),j(l),y)$.
}\\
We observe that, if $\Xi(y)=\tr$ then $h_0(l)$ converges as $l\to\infty$ with limit $\tr$.
Otherwise, the limit does not exist or is $\f$.
The same construction applies to 
$\neg(\forall n \exists m (g_1(n,m,y))) = \exists n \forall m \neg(g_1(n,m,y))$ and yields a
function $h_1$ which converges to $\tr$ if and only if $\Xi(y)=\f$.
Clearly, exactly one of the functions $h_0$, $h_1$ converges to $\tr$.
Now we derive the desired $g$ from $h_0$ and $h_1$ as follows:

{\ttfamily
$\alpha(1)=0$.

If $h_{\alpha(k)}(k)=\tr$ 

then: $\alpha(k+1):=\alpha(k)$

else: $\alpha(k+1):=1-\alpha(k)$.

$k:=k+1$.

If $\alpha(k)=0$ 

then: $g(k,y):=\tr$

else: $g(k,y):=\f$.
}\\
This provides $\Xi(y)=\lim_{k\to\infty} g(k,y)$.

Next, let $g_0$ and $g_1$ be of the form 
$g_s(i,j,y)=\lim_{r} f^s_{i,j,r}(y)$, $s\in\{0,1\}$. Fix $y$.
Then for every pair $(i,j)$ there is an $r(i,j)$ such that $f^s_{u,v,r}(y)=g_s(u,v,y)$
for all $u\leq i$, $v\leq j$, $s\in\{0,1\}$ and $r\geq r(i,j)$. Thus, $g$ is also of 
the form $g(k,y)=\lim_{r} f_{k,r}(y)$
with $f_{k,r}$ being defined by the above procedure applied to the functions 
$(i,j,y)\mapsto f^s_{i,j,k}(y)$ instead of $g_s(i,j,y)$ ($s\in\{0,1\}$).

Now we are left with iterating this argument: If both functions $g_s$ ($s\in\{0,1\}$) 
are of the form 
$g_s(i,j,y)=\lim_{k_{m-1}} \lim_{k_{m-2}}\cdots\lim_{k_1} f^s_{i,j,k_{m-1},\cdots,k_1}(y)$
with certain General algorithms $f^s_{i,j,k_{m-1},\cdots,k_1}$, then also
$g$ is of the form
\[g(k,y)=\lim_{k_{m-1}} \lim_{k_{m-2}}\cdots\lim_{k_1} f_{k,k_{m-1},\cdots,k_1}(y)\]
with $f_{k,k_{m-1},\cdots,k_1}$ being defined by the same 
procedure as before applied to the functions 
$(i,j,y)\mapsto f^s_{i,j,k_{m-1},\cdots,k_1}(y)$ instead of $g_s(i,j,y)$ ($s\in\{0,1\}$).
The resulting functions $y\mapsto f_{k,k_{m-1},\cdots,k_1}(y)$ are General algorithms 
for every $k$, since their evaluation requires only finitely 
many evaluations of the General algorithms $f^s_{i,j,k_{m-1},\cdots,k_1}$.

{\bf Step IV:} It remains to show that $\tilde\Sigma_m\subset\Sigma_m$ and $\tilde\Pi_m\subset\Pi_m$. Suppose that $\Xi\in\tilde\Sigma_m$($\in\tilde\Pi_m$) then by considering the first $m-1$ limits there exists a family $\Xi_{n_m}\in\tilde\Delta_{m}=\Delta_{m}$ (this is also trivially true if $m=1$) such that
$$
\Xi(y)=\lim_{n_m\rightarrow\infty}\Xi_{n_m}(y)
$$
with the final limit monotonic from below (above). But then we must have $\Xi(y)=\exists n_m\Xi_{n_m}(y)$ ($\Xi(y)=\forall n_m\Xi_{n_m}(y)$). But $\Xi_{n_m}\in\Sigma_m$($\in\Pi_m$) and we can collapse the double quantifier $\exists\exists$ ($\forall\forall$) to a single $\exists$($\forall$).
\end{proof}

\subsection{The Baire hierarchy}
To end this appendix, we also make some remarks on the Baire hierarchy. The Baire hierarchy \cite{kechris1995classical}, which is closely related to the Borel hierarchy \cite{kechris1995classical}, in descriptive set theory, has similarities to the SCI hierarchy, however, is fundamentally different.  However, it is worth mentioning, as the Baire hierarchy does include classes of functions that are obtained as limits of functions from lower levels in the hierarchy, hence the two hierarchies share some similarities. 

Recall that given metrisable spaces $X,Y$ and a continuous function $f : X \rightarrow Y$ we say that $f$ is of Baire class $0$. We define a function $g: X \rightarrow Y$ to be in Baire class $1$ if there is a sequence of functions $\{g_n\}$, all of Baire class $0$, such that $g(x) = \lim_{n\rightarrow \infty}g_n(x)$ for all $x \in X$. In general, for $1 < \rho < \omega_1$ we define a function $f: X \rightarrow Y$ to be of Baire class $\rho$ if it is the pointwise limit of a sequence of functions $f_n:X \rightarrow Y$, where $f_n$ is of Baire class $\rho_n < \rho$. In order to understand the similarities and differences between the two hierarchies, we provide a short discussion below. 

{\it Similarities between the SCI and Baire hierarchies.} The main similarity between the hierarchies is the concept of pointwise limits. Indeed, for the integer values of the Baire classes, this number indeed resembles the SCI.     

{\it Differences between the SCI and Baire hierarchies.} The differences between the hierarchy are due to the fact that they describe very different problems. This can be summed up as follows.
\begin{itemize}[noitemsep]
\item[(i)] ({\it Generality}). The SCI hierarchy is designed to be able to handle all types of computational problems such as Smale's problem on iterative algorithms for polynomial root-finding, Doyle--McMullen towers, the insolvability of the quintic etc. This is obviously not within the scope of the Baire hierarchy, however, this was never the intention for this hierarchy. 
\item[(ii)] ({\it Refinements}). An important difference between the hierarchies is that the SCI hierarchy, when extra structure on $\mathcal{M}$ is available, allows for the refinements in terms of the $\Sigma^{\alpha}_k$ and $\Pi^{\alpha}_k$ classes. This type of refinement is not captured by the Baire hierarchy, however, that has never been the motivation. 

\item[(iii)] ({\it Topology vs information}). A striking difference is that the Baire hierarchy is based on metrisable topologies, whereas the SCI hierarchy is based on the information $\Lambda$ (see Definition \ref{def:comp_prob}) available to the algorithm. The computational spectral problem is a good example to illustrate the issue. Let $\Xi: \Omega \ni A \mapsto \mathrm{sp}(A) \in \mathcal{M}$ where $\Omega$ is the set of self-adjoint operators in $\mathcal{B}(l^2(\mathbb{N}))$ and $\mathcal{M}$ is the collection of non-empty compact subsets of $\mathbb{C}$ with the Hausdorff metric. If we equip $\Omega$ with the operator norm topology, then $\Xi$ is Baire class $0$. Yet, the SCI $= 2$ for $\Xi$. If one changes the metric on $\Omega$, the Baire class will change, yet the SCI remains unchanged. Also, as a side note, the algorithms used in this paper to show that the SCI $= 2$ are not continuous in any metrisable topology. Thus, there is no metric on $\Omega$ such that these become Baire class $0$. 

Finally, if we consider self-adjoint Schr\"odinger operators on $\mathrm{L}^2(\mathbb{R}^d)$ with bounded potential $V$ such that $V \in \mathrm{BV}_{\mathrm{loc}}(\mathbb{R}^d)$, then the SCI of the spectral map is $1$ if we can access point samples of $V$. Also, if we equip this set of operators with the natural graph metric (equivalent to norm convergence in the bounded case) the spectral map is Baire class $0$. However, if one changes $\Lambda$, such that we are given matrix elements of the operator with respect to some orthonormal basis of the domain, we may get that the SCI $= \infty$, as the matrix representation may not uniquely determine the spectrum. Thus, the SCI changes with $\Lambda$ (see Definition \ref{def:comp_prob}) that determines which information is available whereas the Baire class changes with the metric. 
\end{itemize}


\bibliography{bib_file_SCI}
\bibliographystyle{abbrv}
\end{document}